\documentclass[11pt]{article}
\usepackage{hyperref} 
\hypersetup{colorlinks=true, linkcolor=blue!90!black, citecolor=orange!80!black, urlcolor=blue!90!black} 
\usepackage[english]{babel}
\usepackage{mathtools, amssymb} 

\usepackage{tocloft} 
 
\usepackage{mathtext} 
\usepackage{bm} 
\usepackage{amsthm} 
\usepackage{dsfont} 
\usepackage{tikz,subcaption, float, pgfplots}
\usetikzlibrary{decorations.markings, arrows.meta, snakes}

\DeclareSymbolFontAlphabet{\mathbb}{AMSb}

\definecolor{spec}{RGB}{0,50,230} 
\definecolor{spec2}{RGB}{229, 113, 27} 

\tikzset{arrow/.style={-{To[length=1.2mm, width=2.2mm]}}}
\tikzset{->-/.style={line width=0.8pt,decoration={
			markings,
			mark=at position 0.56 with {\arrow{Stealth[length=1.7mm,width=1.3mm]}}},postaction={decorate}}}

\numberwithin{equation}{section}

\renewcommand{\phi}{\varphi}
\renewcommand{\imath}{\mathrm{i}}
\renewcommand{\epsilon}{\varepsilon}
\DeclareMathOperator{\sh}{sh}

\theoremstyle{definition}
\newtheorem{remark}{Remark}
\newtheorem{example}{Example}

\theoremstyle{plain}
\newtheorem{theorem}{Theorem}[section]
\newtheorem{corollary}{Corollary}[section]
\newtheorem{proposition}{Proposition}[section]
\newtheorem{lemma}{Lemma}[section]

\newtheorem{innercustomthm}{Theorem}
\newenvironment{customthm}[1]
	{\renewcommand\theinnercustomthm{#1}\innercustomthm}
	{\endinnercustomthm}
	
\newtheorem{innercustomprop}{Proposition}
\newenvironment{customprop}[1]
	{\renewcommand\theinnercustomprop{#1}\innercustomprop}
	{\endinnercustomprop}
	
\newtheorem{innercustomcor}{Corollary}
\newenvironment{customcor}[1]
	{\renewcommand\theinnercustomcor{#1}\innercustomcor}
	{\endinnercustomcor}
	
\newtheorem{innercustomprops}{Propositions}
\newenvironment{customprops}[1]
	{\renewcommand\theinnercustomprops{#1}\innercustomprops}
	{\endinnercustomprops}

\renewcommand{\epsilon}{\varepsilon}

\newcommand{\bbLambda}{\reflectbox{\raisebox{\depth}{\scalebox{1}[-1]{$\mathbb V$}}}}

\newcommand{\T}{\mathbb{T}}
\newcommand{\A}{\mathbb{A}}
\newcommand{\B}{\mathbb{B}}
\newcommand{\C}{\mathbb{C}}
\newcommand{\D}{\mathbb{D}}
\newcommand{\LLambda}{\bbLambda}
\newcommand{\QQ}{\mathbb{Q}}
\newcommand{\QQr}{\mathbb{Q}'}

\def\a{\alpha}
\def\b{\beta}

\def\bx{\bm{x}}
\def\by{\bm{y}}
\def\bl{\bm{\lambda}}
\def\bg{\bm{\gamma}}
\def\bt{\bm{t}}
\def\bgg{\underline{\bm{\gamma}} }
\def\beq{\begin{equation}}
	\def\eeq{\end{equation}}
\def\beqq{\begin{equation*}}
	\def\eeqq{\end{equation*}}

\def\d{\partial}
\def\g{\gamma}

\def\l{\lambda}
\def\R{\mathbb{R}}
\def\Res{\operatorname{Res}}
\newcommand{\rf}[1]{(\ref{#1})}
\newcommand{\bcdmu}{\hat{{\mu}}_{BC}}

\newcommand{\bdelta}{\delta}
\def\HA{\mathrm{H}}
\def\Hdual{\hat{\mathrm{H}}}
\def\HBCdual{\hat{\mathbb{H}}}
\newtheorem*{definition*}{Definition}
\def\ve{\varepsilon}
\def\vf{\varphi}

\let\Re\relax
\let\Im\relax
\DeclareMathOperator{\Re}{Re}
\DeclareMathOperator{\Im}{Im}

\begin{document}
	
\begin{center}
	
	{\bf \Large $BC$ Toda chain II: symmetries. Dual picture}
	
	\vspace{0.4cm}
	
	{N. Belousov$^{\dagger}$, S. Derkachov$^{\diamond\dagger}$, S. Khoroshkin$^{\ast\circ\dagger}$}
	
	\vspace{0.4cm}
	
	{\small \it
		$^\dagger$Beijing Institute of Mathematical Sciences and Applications, \\
		Huairou district, Beijing, 101408, China \vspace{0.2cm}\\
		$^\diamond$Steklov Mathematical Institute, Fontanka 27, \\St.~Petersburg, 191023, Russia \vspace{0.2cm}\\
		$^\ast$Department of Mathematics, Technion,
		Haifa, Israel\vspace{0.2cm}\\
		$^\circ$Skolkovo Institute of Science and Technology,\\Skolkovo, 121205, Russia
	}
	
\end{center}

\begin{abstract} 
	In the previous paper we derived Gauss--Givental integral representation for the wave functions of quantum $BC$ Toda chain and also introduced Baxter operators for this model. In the present paper we prove commutativity of Baxter operators, as well as show that the constructed wave functions are symmetric with respect to signed permutations of spectral parameters and diagonalize Baxter operators. Furthermore, we derive Mellin--Barnes integral representation for the wave functions. With its help we show that wave functions satisfy dual system of difference equations with respect to spectral parameters and coincide with hyperoctahedral Whittaker functions. Finally, we give heuristic proofs of orthogonality and completeness of the wave functions. 
\end{abstract}

\vspace{-0.3cm}

\tableofcontents

\section{Introduction} 

The quantum Toda chain of $BC_n$ type is governed by the Hamiltonian~\cite{Skl, I}
\begin{align}\label{H-bc}
	\mathbb{H}_{BC} = - \sum_{j = 1}^n \partial_{x_j}^2 + 2 \sum_{j = 1}^{n - 1} e^{x_j - x_{j + 1}} + 2\alpha \, e^{-x_1} + \beta^2  e^{-2x_1},
\end{align}
which acts on functions of $n$ spatial variables $x_j \in \mathbb{R}$ and contains two parameters $\alpha, \beta$.
In~\cite{BDK} we found the reflection operator, which acts on functions of one variable and can be used to diagonalize $BC_1$ Toda Hamiltonian. Furthermore, combining reflection operator with Sklyanin's intertwining operators~\cite{Skl2} we constructed the monodromy operators satisfying reflection equation associated with $BC_n$ system. Reductions of monodromy operators produce the so-called raising and Baxter operators. With the help of raising operators we derived the explicit expression for the wave function of $BC_n$ system in a form of Gauss--Givental iterated integrals. Besides, we proved the commutativity of Baxter operators with $BC_n$ Toda Hamiltonians and Baxter equation characterising their spectra.

The same technique was used in \cite{ADV, DKM, BKP} for the analysis of $q$-Toda chain and Heisenberg spin chains, while for $GL_n$ Toda chain ($\alpha = \beta = 0$) the analogous construction reproduces the known results of \cite{GLO} obtained using representation theory of classical Lie groups $GL(n,\R)$. 
 
In the present paper we further study the wave functions of $BC_n$ Toda chain using other tools known in the theory of quantum integrable systems. First, we prove commutativity and exchange relations for Baxter and raising operators constructed in \cite{BDK}. These relations correspond to certain integral identities, which can be described and 
derived in a graphical language using a few basic transformations called star-triangle and flip relations. In addition, we show that our raising and Baxter operators, as well as their products, are well defined on the spaces of polynomially bounded continuous functions.

Second, using the relations between the Baxter and raising operators, we diagonalize the Baxter operators and establish the symmetry of the wave functions with respect to signed permutations of the spectral parameters.
 
Next, we compute the scalar product between $GL_n$ and $BC_n$ wave functions. This computation leads us to the Mellin--Barnes integral representation of $BC_n$ wave functions, generalizing Iorgov and Shadura results~\cite{IS}  for $B_n$ system ($\beta = 0$). Further, we show that Gauss--Givental and Mellin--Barnes integral representations are essentially sufficient to prove the orthogonality and completeness relations for the $BC_n$ wave functions.
 
It is known that properly normalized wave functions of $BC_n$ Toda chain, also called hyperoctahedral Whittaker functions, enjoy difference equations in spectral parameters~\cite{DE}. To verify the dual equations for our wave functions we note that the kernel of Mellin--Barnes representation can be regarded as Ruijsenaars $BC_n$ --- $GL_n$ kernel function, which satisfies corresponding difference equations. In addition, we found the second kernel function and another Mellin--Barnes integral representation for $BC_n$ wave functions.
  
Finally, following the strategy of \cite{HR3,BCDK} we establish asymptotics of our wave function and identify it with the hyperoctahedral Whittaker function by van Diejen and Emsiz. The Mellin--Barnes iterative integral can be calculated by residues and the result is the precise formula for the corresponding Harish-Chandra series.

In the following subsections we briefly recall the construction of commuting Hamiltonians and their eigenfunctions for both $GL_n$ and $BC_n$ Toda chains (see~\cite{BDK} and references therein for details), and then state the main results of this paper.
 
\subsection{Hamiltonians}

For the description of quantum Toda chains of $GL$ and $BC$ types we need two basic blocks: Lax matrix associated with the $j$-th particle~\cite{F, Gaud, S}
\begin{align}\label{QI1}
	L_j(u) = 
	\begin{pmatrix}
		u + \imath \partial_{x_j} & e^{-x_j} \\[4pt]
		-e^{x_j} & 0
	\end{pmatrix}
\end{align}
and the reflection (boundary) matrix 
\begin{align}\label{QI2}
	K(u) =  
	\begin{pmatrix}
		-\alpha & u - \frac{\imath}{2} \\[6pt]
		- \beta^2 \bigl(u - \frac{\imath}{2} \bigr) & -\alpha
	\end{pmatrix}.
\end{align}
With their help we construct monodromy matrix $T_n(u)$ 
\begin{align}\label{QI3}
	T_n(u) = L_n(u) \cdots L_1(u) = 
	\begin{pmatrix}
		A_n(u) & B_n(u) \\[4pt]
		C_n(u) & D_n(u)
	\end{pmatrix},
\end{align}
which serves for $GL$ Toda chain, and monodromy matrix $\T_n(u)$
\begin{align}\label{QI4}
	\begin{aligned}
		\T_n(u) 
		 = T_n(u) \, K(u) \, \sigma_2 \, T_n^t(-u) \, \sigma_2 =
		\begin{pmatrix}
			\A_n(u) & \B_n(u) \\[4pt]
			\C_n(u) & \D_n(u)
		\end{pmatrix},\qquad \sigma_2=\begin{pmatrix}
		0& -\imath \\[4pt]\imath&0\end{pmatrix}
	\end{aligned},
\end{align}
which desribes $BC$ Toda system (to distinguish objects related to $BC$ system we often use the {\tt mathbb} font). The elements $A_n(u)$ and $\B_n(u)$ are generating functions for commuting Hamiltonians of $GL$ and $BC$ Toda chains correspondingly
\begin{align}
	& A_n(u) = u^n + \sum_{s = 1}^n u^{n - s} \, \HA_s, && \hspace{-1cm} [\HA_s, \HA_r] = 0, \\[6pt]
	& 	\B_n(u) = (-1)^n \biggl(u - \frac{\imath}{2} \biggr) \biggl( u^{2n} + \sum_{s = 1}^n u^{2(n - s)} \, \mathbb{H}_s \biggr), && \hspace{-1cm}  [\mathbb{H}_s, \mathbb{H}_r] = 0.
\end{align}
The first coefficients are related to the quadratic Hamiltonian~\eqref{H-bc}, namely, $\mathbb{H}_1 = -\mathbb{H}_{BC}$ and
\begin{align}
	\HA_1 = \sum_{j = 1}^n \imath \partial_{x_j}, \qquad \HA_2 = \frac{1}{2} \bigl( \HA_1^2 - \HA_{GL} \bigr),
\end{align}
where
\begin{align}\label{HA}
	\HA_{GL} = \mathbb{H}_{BC} \bigr|_{\alpha = \beta = 0} = - \sum_{j = 1}^n \partial_{x_j}^2 + 2 \sum_{j = 1}^{n - 1} e^{x_j - x_{j + 1}}.
\end{align}
Two other special cases of $BC$ Toda system are $B$ Toda chain ($\alpha \not= 0$, $\beta =0 $) and $C$ Toda chain ($\alpha = 0$, $\beta \not= 0$).

\subsection{Wave functions and Baxter operators}

Denote tuples of $n$ variables and sums of their components
\begin{align}
	\bm{x}_n = (x_1, \dots, x_n), \qquad \underline{\bm{x}}_n = x_1 + \ldots + x_n.
\end{align}
The wave functions of $GL$ and $BC$ systems are parametrized by zeroes of the generating functions eigenvalues
\begin{align}\label{A-eigen}
	& A_n(u) \, \Phi_{\bm{\lambda}_n}(\bm{x}_n) = \prod_{j = 1}^n (u - \lambda_j) \,  \Phi_{\bm{\lambda}_n}(\bm{x}_n), \\[6pt]
\label{B-eigen}
	& \B_n(u) \, \Psi_{\bm{\lambda}_n}(\bm{x}_n) = (-1)^n \biggl(u - \frac{\imath}{2} \biggr) \prod_{j = 1}^n (u^2 - \lambda_j^2) \, \Psi_{\bm{\lambda}_n}(\bm{x}_n).
\end{align}
As shown in~\cite{BDK}, they can be constructed using the so-called raising operators. To introduce these operators denote 
\begin{align}
	g = \frac{1}{2} + \frac{\alpha}{\beta}.
\end{align}
Besides, in what follows we always assume conditions 
\begin{align}
	g > 0, \qquad \beta > 0.
\end{align}
The raising operators are integral operators acting on functions $\phi(\bm{x}_{n - 1})$ by the formulas
\begin{align}
	 & \bigl[ \Lambda_{n}(\lambda) \, \phi \bigr] (\bm{x}_{n})= \int_{\mathbb{R}^{n - 1}} d\bm{y}_{n - 1} \; \Lambda_\lambda(\bm{x}_{n} | \bm{y}_{n - 1}) \, \phi(\bm{y}_{n - 1}), \\[6pt]
	 & \bigl[ \LLambda_n(\lambda) \, \phi \bigr] (\bm{x}_n) = \int_{\mathbb{R}^{n - 1}} d\bm{y}_{n - 1} \; \LLambda_\lambda(\bm{x}_{n} | \bm{y}_{n - 1}) \, \phi(\bm{y}_{n - 1}),
\end{align}
where the kernels are given by
\begin{align} \label{Lker-expl}
	& \Lambda_\lambda(\bm{x}_{n} | \bm{y}_{n - 1}) =  \exp \biggl( \imath \lambda \bigl(\underline{\bm{x}}_{n} - \underline{\bm{y}}_{n - 1} \bigr) 
	- \sum_{j = 1}^{n - 1} (e^{x_j - y_j} + e^{y_j - x_{j + 1}}) \biggr), \\[10pt]
	& \begin{aligned}
		& \LLambda_\lambda(\bm{x}_n | \bm{y}_{n - 1}) = \frac{(2\beta)^{\imath \lambda}}{\Gamma(g - \imath \lambda)} \int_{\mathbb{R}^n} d\bm{z}_n \;
		(1 + \beta e^{-z_1})^{- \imath \lambda - g} \; (1 - \beta e^{-z_1})^{- \imath \lambda + g - 1}\, \theta(z_1 - \ln \beta) \\[6pt]
		& \quad \times  \exp \biggl( \imath \lambda \bigl( \underline{\bm{x}}_n + \underline{\bm{y}}_{n - 1} - 2 \underline{\bm{z}}_n \bigr) - \sum_{j = 1}^{n - 1} (e^{z_j - x_j} + e^{z_j - y_j} + e^{x_j - z_{j + 1}} + e^{y_j - z_{j + 1}} )  - e^{z_n - x_n} \biggr).
	\end{aligned}
\end{align}
Here and in what follows by $\theta(x)$ we denote the Heaviside step function
\begin{align}
	\theta(x) = \left\{ 
	\begin{aligned}
		& 1, && x \geq 0, \\[4pt]
		& 0, && x < 0.
	\end{aligned}
	\right.
\end{align}
 
The wave functions are defined by iterative integrals
\begin{align} \label{Phi-GG-0}
	& \Phi_{\bm{\lambda}_n}(\bm{x}_n) =  \Lambda_n(\lambda_n) \cdots \Lambda_1(\lambda_1) \cdot 1, \\[6pt] \label{Psi-GG-0}
	& \Psi_{\bm{\lambda}_n}(\bm{x}_n) =  \LLambda_n(\lambda_n) \cdots \LLambda_1(\lambda_1) \cdot 1,
\end{align}
which by tradition are called Gauss--Givental representations. Note that from the explicit formula~\eqref{Lker-expl} we have
\begin{align}
	\Lambda_n(\lambda + \rho) = e^{\imath \rho \underline{\bm{x}}_{n}} \, \Lambda_n(\lambda) \, e^{- \imath \rho \underline{\bm{x}}_{n - 1}},
\end{align}
and consequently,
\begin{align} \label{Phi-shift}
	\Phi_{\bm{\lambda}_n + \rho \bm{e}_n}(\bm{x}_n) = e^{\imath \rho \underline{\bm{x}}_{n}} \, \Phi_{\bm{\lambda}_n}(\bm{x}_n) , \qquad \bm{e}_n = (1, \dots, 1).
\end{align}

\begin{example} \label{ex:PhiPsi1}
	In the case $n = 1$ we have $\Phi_{\lambda}(x) = e^{\imath \lambda x}$ and
	\begin{align}
		\begin{aligned}
			\Psi_{\lambda}(x) & = \frac{e^{\frac{x}{2}}}{\sqrt{2\beta}} \, W_{\frac{1}{2}- g, - \imath \lambda}(2\beta e^{-x}) \\[6pt]
			& = \frac{(2\beta)^{\imath \lambda}}{\Gamma ( g - \imath \lambda )} \, \int_{\ln \beta}^\infty dz \; e^{\imath \lambda(x - 2z) - e^{z - x}} \, (1 + \beta e^{-z})^{-\imath \lambda - g} \, (1-\beta e^{-z})^{-\imath \lambda + g - 1},
		\end{aligned}
	\end{align}
	where $W_{\kappa,\mu}(z)$ is the Whittaker function~\cite[\href{https://dlmf.nist.gov/13}{Chapter 13}]{DLMF}. The above integral representation for it is well known~\cite[\href{http://dlmf.nist.gov/13.16.E5}{(13.16.5)}]{DLMF}.
\end{example}

The close relatives of the raising operators are Baxter operators. These are integral operators acting on functions $\phi(\bm{x}_n)$ by the formulas
\begin{align}
	& \bigl[ Q_n(\lambda) \, \phi \bigr] (\bm{x}_n) = \int_{\mathbb{R}^{n}} d\bm{y}_{n} \; Q_\lambda(\bm{x}_n | \bm{y}_{n}) \, \phi(\bm{y}_{n}), \\[6pt]
	& \bigl[ \QQ_n(\lambda) \, \phi \bigr] (\bm{x}_n) = \int_{\mathbb{R}^{n}} d\bm{y}_{n} \; \QQ_\lambda(\bm{x}_n | \bm{y}_{n}) \, \phi(\bm{y}_{n}),
\end{align}
where the kernels are given by
\begin{align} \label{Q-ker}
	& Q_\lambda(\bm{x}_n | \bm{y}_n) = \exp \biggl( \imath \lambda \bigl(\underline{\bm{x}}_n - \underline{\bm{y}}_{n} \bigr) 
	- \sum_{j = 1}^{n - 1} (e^{x_j - y_j} + e^{y_j - x_{j + 1}}) - e^{x_n - y_n} \biggr), \\[10pt]
	& \begin{aligned}
		\QQ_\lambda(\bm{x}_n | \bm{y}_{n}) & = \frac{(2\beta)^{\imath \lambda}}{\Gamma(g - \imath \lambda)} \int_{\mathbb{R}^{n + 1}} d\bm{z}_{n + 1}
		\; (1 + \beta e^{-z_1})^{- \imath \lambda - g} \; (1 - \beta e^{-z_1})^{- \imath \lambda + g - 1}\,	\theta(z_1 - \ln \beta)
		 \\[6pt]
		&  \times \exp \biggl( \imath \lambda \bigl( \underline{\bm{x}}_n + \underline{\bm{y}}_{n} - 2 \underline{\bm{z}}_{n + 1} \bigr) - \sum_{j = 1}^{n} (e^{z_j - x_j} + e^{z_j - y_j} + e^{x_j - z_{j + 1}} + e^{y_j - z_{j + 1}} ) \biggr).
	\end{aligned}
\end{align}
These operators commute with the Hamiltonians of the corresponding systems
\begin{align}
	\bigl[ Q_n(\lambda), A_n(u) \bigr] = 0, \qquad \bigl[ \QQ_n(\lambda), \B_n(u) \bigr] = 0,
\end{align}
and also satisfy Baxter equations
\begin{align} \label{Bax-eq}
	Q_n(\lambda) \, A_n(\lambda) = \imath^{-n} \, Q_n(\lambda - \imath), \quad \quad \QQ_n(\lambda) \, \B_n(\lambda) = - \frac{\beta ( g + \imath \lambda)}{2\lambda} \, \QQ_n(\lambda - \imath),
\end{align}
see~\cite{BDK}. 

As shown in Appendix~\ref{sec:GG-bounds}, all of the above integral operators are well defined on the spaces of polynomially bounded continuous functions
\begin{align} 
	\mathcal{P}_n = \bigl\{ \phi \in C(\mathbb{R}^n) \colon \quad | \phi(\bm{x}_n) | \leq P(|x_1|, \dots, |x_n|), \quad P \text{ --- polynomial} \bigr\}.
\end{align}
Namely, by Propositions~\ref{prop:QL-A-space},~\ref{prop:LQQr-space} we have
\begin{align}
	& \Lambda_n(\lambda) \colon \; \mathcal{P}_{n - 1} \; \to \; \mathcal{P}_n, && \hspace{-2cm} \lambda \in \mathbb{R}, \\[6pt]
	& \LLambda_n(\lambda) \colon \; \mathcal{P}_{n - 1} \; \to \; \mathcal{P}_n, &&  \hspace{-2cm} \lambda \in \mathbb{R}, \\[6pt]
	& Q_n(\lambda) \colon \; \mathcal{P}_{n} \; \to \; \mathcal{P}_n, &&  \hspace{-2cm} \Im \lambda < 0, \\[6pt]
	& \QQ_n(\lambda) \colon \; \mathcal{P}_{n} \; \to \; \mathcal{P}_n, &&  \hspace{-2cm} \Im \lambda \in (-g, 0).
\end{align}
As a result, the products of these operators are well defined. In particular, the Gauss--Givental representations~\eqref{Phi-GG-0},~\eqref{Psi-GG-0} converge and wave functions $\Phi_{\bm{\lambda}_n}(\bm{x}_n), \Psi_{\bm{\lambda}_n}(\bm{x}_n)$ are polynomially bounded. 

In the previous paper~\cite{BDK} we also proved that the wave function $\Psi_{\bm{\lambda}_n}(\bm{x}_n)$ decays rapidly in the classically forbidden regions $x_{j + 1} \ll x_{j}$, as well as $x_1 \ll 0$, see Proposition~\ref{prop:bc-bound} here for the precise statement.
 
\begin{remark}
In~\cite{BDK} all relations with operators hold on the spaces of exponentially tempered smooth functions, which are suited for the action of monodromy matrix entries. In this paper we work with the spaces $\mathcal{P}_{n}$ that are more suitable for the products of integral operators: note that the Baxter operators do not preserve the space $\mathcal{E}_n(\mathbb{R}^n)$ considered in~\cite[Corollary~5]{BDK}.
\end{remark}

\subsection{Main results} \label{sec:results}

In what follows we use shorthand notation for the products of gamma functions
\begin{align}
	\Gamma(a \pm b) = \Gamma(a + b) \, \Gamma(a - b), \qquad \Gamma(\pm a \pm b) = \Gamma(a + b) \, \Gamma(a - b) \, \Gamma(-a + b) \, \Gamma(-a-b).
\end{align}
In Section \ref{section2} using diagram technique we deduce the following relations for $BC$ raising and Baxter operators.

\begin{customprops}{\ref{prop:Lrefl}, \ref{prop:QL-rel}}
	The relations
	\begin{align*}
		& \LLambda_{n}(\lambda) = \LLambda_{n}(-\lambda), && \hspace{-1cm} \lambda \in \mathbb{R}, \\[6pt]
		& \LLambda_{n + 1}(\lambda) \, \LLambda_{n}(\rho) = \LLambda_{n + 1}(\rho) \, \LLambda_{n}(\lambda), && \hspace{-1cm} \lambda, \rho \in \mathbb{R}, \\[6pt]
		& \QQ_n(\lambda) \, \QQ_n(\rho) = \QQ_n(\rho) \, \QQ_n(\lambda), && \hspace{-1cm} \Im \lambda, \, \Im \rho \in (-g, 0), \\[6pt]
		& \QQ_n(\lambda) \, \LLambda_n(\rho) = \Gamma(\imath \lambda \pm \imath \rho) \, \LLambda_n(\rho) \, \QQ_{n - 1}(\lambda), && \hspace{-1cm} \Im \lambda \in (-g, 0), \, \rho \in \mathbb{R},
	\end{align*}
hold on the spaces $\mathcal{P}_{n - 1}$ and $\mathcal{P}_n$ correspondingly. 
\end{customprops}

In the last identity for $n = 1$ we denote
\begin{align}
	\QQ_0(\lambda) = \frac{(2\beta)^{-\imath \lambda} \, \Gamma(2\imath \lambda)}{\Gamma(g + \imath \lambda)} \,  \mathrm{Id}.
\end{align}
The above relations imply that the $BC$ wave function~\eqref{Psi-GG-0} is the eigenfuction of Baxter operators, symmetric with respect to the action of the Weyl group of the root system $BC_n$.

\begin{customthm}{\ref{thm:Psi-sym}} Let $\sigma \in S_n$, $\bm{\epsilon}_n \in \{1, -1\}^n$. Then
	\begin{align} \label{Psi-sym-0}
		\Psi_{\l_1,\ldots, \l_n}(\bm{x}_n) = \Psi_{\ve_1\l_{\sigma(1)},\ldots, \ve_n\l_{\sigma(n)}}(\bm{x}_n).
	\end{align}
\end{customthm}
 
\begin{customthm}{\ref{theorem2.2}}
	Let $\bm{\lambda}_n \in \mathbb{R}^n$ and $\Im \lambda \in (-g, 0)$. Then
	\begin{align*}
		& \QQ_n(\lambda) \, \Psi_{\bm{\lambda}_n}(\bm{x}_n) = \frac{(2\beta)^{-\imath \lambda} \, \Gamma(2\imath \lambda) }{\Gamma(g + \imath \lambda)} \, \prod_{j = 1}^n \Gamma(\imath \lambda \pm \imath \lambda_j) \, \Psi_{\bm{\lambda}_n}(\bm{x}_n).
	\end{align*}
\end{customthm}

Note that the above symmetry implies non-degeneracy of the Hamiltonians eigenvalues~\eqref{B-eigen}, while the last formula is in accordance with the Baxter equation~\eqref{Bax-eq}. Besides, for comparison, let us write the relations between $GL$ raising and Baxter operators
\begin{align}
	& \Lambda_{n + 1}(\lambda) \, \Lambda_n(\rho) = \Lambda_{n + 1}(\lambda) \, \Lambda_n(\rho), \\[6pt] \label{QQ-gl}
	& Q_n(\lambda) \, Q_n(\rho) = Q_n(\rho) \, Q_n(\lambda), \\[6pt]
	& Q_n(\lambda) \, \Lambda_n(\rho) = \Gamma(\imath \lambda - \imath \rho) \, \Lambda_n(\rho) \, Q_{n - 1} (\lambda), 
\end{align}
which lead to the analogous statements for the $GL$ wave functions~\eqref{Phi-GG-0}
\begin{align}
	& \Phi_{\lambda_1, \dots, \lambda_n}(\bm{x}_n) = \Phi_{\lambda_{\sigma(1)}, \dots, \lambda_{\sigma(n)}}(\bm{x}_n), \qquad \sigma \in S_n, \\[6pt] \label{Q-Phi}
	& Q_n(\lambda) \, \Phi_{\bm{\lambda}_n}(\bm{x}_n) = \prod_{j = 1}^n \Gamma(\imath \lambda - \imath \lambda_j) \, \Phi_{\bm{\lambda}_n}(\bm{x}_n),
\end{align}
see~\cite[Theorem 2.3]{GLO1}.

With the help of the above Baxter operators and diagram technique in Section~\ref{sec:gl-bc-prod} we calculate the following scalar product between $GL$ and $BC$ wave functions. 

\begin{customprop}{\ref{prop:gl-bc-prod}}
	Let $\bm{\lambda}_n \in \mathbb{R}^n$, $\bm{\g}_n \in \mathbb{C}^n$ such that $	\Im \g_1 = \ldots = \Im \g_n < 0$. Then 
	\begin{align*} 
		\int_{\mathbb{R}^n} d\bm{x}_n \; \Phi_{-\bm{\g}_n}(\bm{x}_n)  \, \Psi_{\bm{\lambda}_n}(\bm{x}_n) \, \exp (- \beta e^{-x_1}) 
		=   \frac{(2\beta)^{-\imath \bgg_n} \prod\limits_{j, k=1}^n \Gamma(\imath \gamma_j \pm \imath \lambda_k) }{ \prod\limits_{1\leq j < k\leq n} \Gamma(\imath \gamma_j + \imath \gamma_k) \, \prod\limits_{j = 1}^n \Gamma (g + \imath \gamma_j ) }.
	\end{align*}
\end{customprop}

It is known that $GL$ wave functions satisfy orthogonality and completeness relations
\begin{align} \label{gl-orth}
	& \int_{\mathbb{R}^n} d\bm{x}_n \; \overline{\Phi_{\bm{\lambda}_n}(\bm{x}_n)} \, \Phi_{\bm{\rho}_n}(\bm{x}_n) = \frac{1}{\hat{\mu}(\bm{\lambda}_n)} \, \frac{1}{n!} \sum_{\sigma \in S_n} \delta(\lambda_1 - \rho_{\sigma(1)}) \cdots \delta(\lambda_n - \rho_{\sigma(n)}), \\[6pt] \label{gl-compl}
	& \int_{\mathbb{R}^n} d\bm{\lambda}_n \; \hat{\mu}(\bm{\lambda}_n) \, \overline{\Phi_{\bm{\lambda}_n}(\bm{x}_n)} \, \Phi_{\bm{\lambda}_n}(\bm{y}_n) = \delta(x_1 - y_1) \cdots \delta(x_n - y_n)
\end{align}
with the spectral measure
\begin{align} \label{gl-measure-0}
	\hat{\mu}(\bm{\lambda}_n) = \frac{1}{n! \, (2\pi)^n} \prod_{1 \leq j \not= k \leq n} \frac{1}{\Gamma(\imath \lambda_j - \imath \lambda_k)},
\end{align}
see~\cite{STS, W, Kozl}. Combining completeness relation~\eqref{gl-compl} with the above proposition we obtain Mellin--Barnes integral representation for $BC$ wave functions, the details are given in Section~\ref{sec:gg-mb-equiv}.

\begin{customthm}{\ref{theorem2.3}}
	For $\bm{\lambda}_n \in \mathbb{R}^n$ and $\epsilon > 0$
	\begin{align} \label{Psi-MB-0} 
		\Psi_{\bm{\lambda}_n}(\bm{x}_n)  = e^{\beta e^{-x_1}}  \int_{(\mathbb{R} - \imath \epsilon)^n}  d\bm{\gamma}_n \; \hat{\mu}(\bm{\gamma}_n) \, \frac{(2\beta)^{-\imath \bgg_n} \prod\limits_{j, k=1}^n \Gamma(\imath \gamma_j \pm \imath \lambda_k) }{ \prod\limits_{1\leq j < k\leq n} \Gamma(\imath \gamma_j + \imath \gamma_k) \, \prod\limits_{j = 1}^n \Gamma (g + \imath \gamma_j ) } \, \Phi_{\bm{\gamma}_n}(\bm{x}_n).
	\end{align}
\end{customthm}

Note that the symmetry with respect to signed permutations of spectral parameters~\eqref{Psi-sym-0} becomes evident from this representation.

The formula~\eqref{Psi-MB-0} generalizes Iorgov and Shadura result for $B$ Toda chain~\cite[(26)]{IS}, which can be obtained in the limit $\beta \to 0$. In this limit the integrand simplifies
\begin{align}
	\frac{(2\beta)^{-\imath \bgg_n}}{\prod\limits_{j = 1}^n \Gamma (g + \imath \gamma_j ) } = \frac{(2\beta)^{-\imath \bgg_n}}{\prod\limits_{j = 1}^n \Gamma \bigl( \frac{1}{2} + \frac{\alpha}{\beta} + \imath \gamma_j \bigr) }  \sim \biggl[ \frac{ (e\beta/\alpha)^{\alpha/\beta} }{\sqrt{2\pi}} \biggr]^n \, (2\alpha)^{-\imath \bgg_n}, \qquad \beta \to 0,
\end{align}
due to Stirling formula. 

\begin{example} \label{ex:MB1-1}
	In the case $n = 1$ we have
	\begin{align}
		\Psi_{\lambda}(x) = \frac{e^{\frac{x}{2}}}{\sqrt{2\beta}}  \, W_{\frac{1}{2} - g, - \imath \lambda}(2\beta e^{-x}) = e^{\beta e^{-x}} \int_{\mathbb{R} - \imath \epsilon} \frac{d\gamma}{2\pi} \; \frac{(2\beta)^{-\imath \gamma} \, \Gamma(\imath \gamma + \imath \lambda)  \, \Gamma(\imath \gamma - \imath \lambda)}{\Gamma(g + \imath \gamma)} \; e^{\imath \gamma x},
	\end{align}
	which is well-known formula for Whittaker function~\cite[\href{http://dlmf.nist.gov/13.16.E12}{(13.16.12)}]{DLMF}.
\end{example}

The obtained Mellin--Barnes representation for $BC$ wave functions is somewhat different from the one for $GL$ wave functions. The latter is defined by recursive formula~\cite[Theorem~4.1]{KL2}
\begin{align} \label{Phi-MB-0}
	\Phi_{\bm{\lambda}_n}(\bm{x}_n) = \int_{(\mathbb{R} + \imath \epsilon)^{n - 1}} d\bm{\gamma}_{n - 1} \; \hat{\mu}(\bm{\gamma}_{n - 1}) \, e^{\imath x_n(\underline{\bm{\lambda}}_n - \underline{\bm{\gamma}}_{n - 1})} \, \prod_{j = 1}^n \prod_{k = 1}^{n - 1} \Gamma(\imath \lambda_j - \imath \gamma_k) \, \Phi_{\bm{\gamma}_{n - 1}}(\bm{x}_{n - 1}),
\end{align}
which starts from the one-particle function $\Phi_{\lambda_1}(x_1) = e^{\imath \lambda_1 x_1}$. The recursive Mellin--Barnes formula also exists for $B$ Toda chain~\cite[(4.11)]{GaKL}, while for $BC$ model its construction (as well as the construction of dual $Q$-operators) remains an interesting open problem. 

Let us remark that the equivalence of Mellin--Barnes and Gauss--Givental representations for $GL$ Toda chain was established in~\cite[Section~3]{GLO1},~\cite[Appendix~B]{Kozl}.

In Sections~\ref{sec:orth} and~\ref{sec:compl} we prove orthogonality and completeness relations for $BC$ wave functions
\begin{align} \label{ort1} 
	& \int_{\R^n}\ d\bx_n \; \overline{\Psi_{\bm{\lambda}_n}(\bx_n)} \, \Psi_{\bm{\rho}_n}(\bx_n) = \frac{1}{\bcdmu(\bm{\lambda}_n) } \, \bdelta_{\mathrm{sym}}(\bm{\lambda}_n, \bm{\rho}_n), \\[6pt]
 \label{ort2}
	& \int_{\mathbb{R}^n} d\bm{\lambda}_n \; \bcdmu(\bm{\lambda}_n) \, \overline{\Psi_{\bm{\lambda}_n}(\bm{x}_n)} \, \Psi_{\bm{\lambda}_n}(\bm{y}_n) = \delta(x_1 - y_1) \cdots \delta(x_n - y_n),
\end{align}
using Gauss--Givental~\eqref{Psi-GG-0} and Mellin--Barnes~\eqref{Psi-MB-0} representations correspondingly. Here $\bcdmu(\bm{\lambda}_n)$ is $BC$ spectral measure
\begin{align}
	\bcdmu(\bm{\lambda}_n) = \frac{1}{n! \, (4\pi)^n } \prod_{1 \leq j < k \leq n} \frac{1}{\Gamma(\pm \imath \lambda_j \pm \imath \lambda_k) } \, \prod_{j = 1}^n \frac{ \Gamma(g \pm \imath \lambda_j ) }{ \Gamma(\pm2 \imath \lambda_j) }
\end{align}
and in the orthogonality relation we have delta function symmetric with respect to signed permutations
\begin{align*}
	\bdelta_{\mathrm{sym}} (\bm{\lambda}_n, \bm{\rho}_n) = \frac{1}{n! \, 2^n} \sum_{\substack{\sigma \in S_n \\[2pt] \bm{\epsilon}_n \in \{1, -1\}^n}}\delta \bigl( \lambda_1 - \epsilon_1 \rho_{\sigma(1)} \bigr) \cdots \delta \bigl( \lambda_n - \epsilon_n \rho_{\sigma(n)} \bigr).
\end{align*}
Our derivaion of \rf{ort1} and \rf{ort2} is rather heuristic. However, we expect that it can be made rigorous by arguments similar to~\cite{Kozl},~\cite{DKM2},~\cite{BK}.

Notice that poles of the integrand in Mellin--Barnes formula~\eqref{Psi-MB-0}
\begin{align}
	\gamma_j = \pm \lambda_k + \imath m, \qquad m \in \mathbb{N}_0
\end{align}
are determined by spectral parameters $\lambda_k$. By shifting integration contours one can show that $BC$ wave function is entire in spectral parameters, see Appendix~\ref{sec:MB-bounds}, where we also prove absolute convergence of all Mellin--Barnes integrals.

\begin{customcor}{\ref{cor:Psi-an}}
	The function $\Psi_{\bm{\lambda}_n}(\bm{x}_n)$ can be analytically continued to $\bm{\lambda}_n \in \mathbb{C}^n$.
\end{customcor}

The dependence of wave functions over spectral parameters can be described as eigenvalue problem with respect to dual Hamiltonians, obtained by van Diejen and Emsiz~\cite[Theorem~3]{DE}. The simplest of them is 
\begin{align}
	\HBCdual_1 = 2\beta\sum_{ j = 1}^n \,  \prod_{ \substack{k = 1 \\ k\not=j} }^n \frac{1}{\lambda_j^2 - \lambda_k^2}  \biggl( \, \frac{\imath \lambda_j +g }{2 \lambda_j (2 \lambda_j - \imath)} \,  \bigl( e^{-\imath \partial_{\lambda_j}} - 1 \bigr) 
	+  \frac{-\imath \lambda_j + g }{2 \lambda_j (2 \lambda_j + \imath)}\,  \bigl( e^{\imath \partial_{\lambda_j}} - 1 \bigr) \biggr),
\end{align}
the others are given by the formula \rf{HBC}. In Section~\ref{sec:MB-int} we prove that our wave functions are indeed the eigenfunctions of van Diejen--Emsiz Hamiltonians.

\begin{customthm}{\ref{theoremMB4}}
	The wave function $	\Psi_{\bm{\lambda}_n}(\bm{x}_n)$ satisfies van Diejen--Emsiz equations
	\begin{align*}
		\HBCdual_s \,	\Psi_{\bm{\lambda}_n}(\bm{x}_n) = e^{x_{n - s + 1} + \ldots + x_n} \, \Psi_{\bm{\lambda}_n}(\bm{x}_n)
	\end{align*}   
	for $s = 1, \dots, n$.
\end{customthm}

This result is essentially equivalent to the statement that the kernel of Mellin--Barnes representation~\eqref{Psi-MB-0}
\begin{align} \label{K1}
	K_1( \bm{\lambda}_n,  \bm{\gamma}_n)  = \frac{(2\beta)^{-\imath \bgg_n} \prod\limits_{j, k=1}^n \Gamma(\imath \gamma_j \pm \imath \lambda_k) }{ \prod\limits_{1\leq j < k\leq n} \Gamma(\imath \gamma_j + \imath \gamma_k) \, \prod\limits_{j = 1}^n \Gamma (g + \imath \gamma_j ) }
\end{align} 
intertwines dual Hamiltonians of $BC$ and $GL$ models
\begin{align}
	\HBCdual_s(\bm{\lambda}_n) \, K_1( \bm{\lambda}_n,  \bm{\gamma}_n) = \Hdual_s(-\bm{\gamma}_n) \, K_1( \bm{\lambda}_n, \bm{\gamma}_n),\qquad s=1,\ldots,n,
\end{align}
where in the arguments of operators we emphasized on which variables they act. Here $\Hdual_s$ are dual $GL$ Toda Hamiltonians derived by Babelon~\cite{B}, the first one is
\begin{align}
	\Hdual_1 = \sum_{j = 1}^n \, \prod_{ \substack{k = 1 \\ k \not= j} }^n \frac{1}{\imath \lambda_j - \imath \lambda_k} \, e^{- \imath \partial_{\lambda_j} },
\end{align}
and the others are given by the formula~\eqref{Hdual}. In addition, we also found another kernel function $K_2(\bm{\lambda}_n, \bm{\gamma}_n)$ \rf{MB1} that leads to the second Mellin--Barnes representation of $BC$ wave functions. 

\begin{customprop}{\ref{propK2}}
	For $\bm{\lambda}_n \in \mathbb{R}^n$ and $\epsilon \in (0, g)$
	\begin{align*} 
		\Psi_{\bm{\lambda}_n}(\bm{x}_n)  = e^{-\beta e^{-x_1}}  \int_{(\mathbb{R} - \imath \epsilon)^n}  d\bm{\gamma}_n \; \hat{\mu}(\bm{\gamma}_n) \, \frac{(2\beta)^{-\imath \bgg_n} \prod\limits_{j, k=1}^n \Gamma(\imath \gamma_j \pm \imath \lambda_k) \prod\limits_{j = 1}^n \Gamma (g - \imath \gamma_j )  }{ \prod\limits_{1\leq j < k\leq n} \Gamma(\imath \gamma_j + \imath \gamma_k) \, \prod\limits_{j = 1}^n \Gamma (g \pm \imath \lambda_j ) } \, \Phi_{\bm{\gamma}_n}(\bm{x}_n).
	\end{align*}
\end{customprop}

In the limit $\beta \to 0$, which corresponds to $B$ Toda chain, two Mellin--Barnes representations coincide. The equality of these representations in the general case is established in Section~\ref{sec:MB-int} with the help of Gustafson type integral calculated in~\cite{DM}. 

\begin{example} \label{ex:MB2-1}
	In the case $n = 1$ the last formula equals
	\begin{align*}
		\Psi_{\lambda}(x) = \frac{e^{\frac{x}{2}}}{\sqrt{2\beta}}  \, W_{\frac{1}{2} - g, - \imath \lambda}(2\beta e^{-x}) = e^{-\beta e^{-x}} \int_{\mathbb{R} - \imath \epsilon} \frac{d\gamma}{2\pi} \; \frac{(2\beta)^{-\imath \gamma} \, \Gamma(\imath \gamma + \imath \lambda)  \, \Gamma(\imath \gamma - \imath \lambda) \, \Gamma(g - \imath \gamma)}{\Gamma(g + \imath \lambda) \, \Gamma(g - \imath \lambda)} \; e^{\imath \gamma x},
	\end{align*}
	which coincides with the known expression for Whittaker function~\cite[\href{http://dlmf.nist.gov/13.16.E11}{(13.16.11)}]{DLMF}.
\end{example}

In Section~\ref{sec:QISM} following the approach of~\cite{KL, KL2, IS} we independently check that the wave function in Mellin--Barnes representation~\eqref{Psi-MB-0} diagonalizes Toda chain Hamiltonians~\eqref{B-eigen}
\begin{align}
	\B_n(u) \, \Psi_{\bm{\lambda}_n}(\bm{x}_n) = (-1)^n \biggl(u - \frac{\imath}{2} \biggr) \prod_{j = 1}^n (u^2 - \lambda_j^2) \, \Psi_{\bm{\lambda}_n}(\bm{x}_n),
\end{align}
and also calculate the action of the operator $\D_n(u)$ from monodromy matrix~\eqref{QI4} 
\begin{align} \label{D-Psi}
	\D_n(\pm \lambda_j) \, \Psi_{\bm{\lambda}_n}(\bm{x}_n) = -\beta (g \pm \imath \lambda_j) \, \Psi_{\lambda_1, \dots, \lambda_j \mp \imath, \dots, \lambda_n}(\bm{x}_n), \qquad j = 1, \dots ,n.
\end{align}
As follows from definition~\eqref{QI4}, $\D_n(\imath/2) = -\alpha$ and $\D_n(u)$ is polynomial of degree $2n$. Thus, we have exactly $2n + 1$ points to interpolate its action at arbitrary point.
Notice that the formula~\eqref{D-Psi} is in consistency with the symmetry~\eqref{Psi-sym-0} and relation
\begin{align}
	\D_n(\lambda) \, \LLambda_n(\lambda) = - \beta (g + \imath \lambda) \, \LLambda_n(\lambda - \imath),
\end{align}
derived in our previous paper~\cite{BDK}.

To compare our wave function with hyperoctahedral Whittaker function from~\cite{DE} analysed through series representation, in Section~\ref{sec:asymp} we calculate the asymptotics of $\Psi_{\bm{\lambda}_n}(\bm{x}_n)$ in the domain
$$0 \ll x_1\ll x_2\ll\ldots\ll x_n$$
using Mellin--Barnes representation. It is given by~\eqref{AS17}
\begin{align}
	\Psi^{as}_{\bl_n}(\bx_n) = \sum_{w\in W_n}\prod_{k = 1}^n\frac{\Gamma(2\imath w(\l_k))}{\Gamma(\imath w(\l_k)+g)}\prod_{1 \leq i<j \leq n}\Gamma (\imath w(\l_{j})\pm \imath w(\l_{i}) ) \, e^{\imath\sum_{k = 1}^n w(\l_{k})\tilde{x}_k},
\end{align} 
where $\tilde{x}_k = x_k - \ln(2\beta)$ and $W_n = S_n \times Z_2^n$ is Weyl group of $BC_n$ root system. This asymptotic formula coincides with the known asymptotics of the hyperoctahedral Whittaker function~\cite[(3.5b)]{DE}, so we conclude that functions coincide totally since the latter is unique (to compare with~\cite{DE} put $2\beta = 1$).

Finally, in Section~\ref{sec:asymp} we perform residue calculation of the Mellin--Barnes integral~\eqref{Psi-MB-0}. This gives us convergent Harish-Chandra series~\rf{AS15},~\rf{AS16}, which are multivariable generalizations of classical series representation of Whittaker function $W_{\kappa,\mu}$~\cite[\href{http://dlmf.nist.gov/13.14.E33}{(13.14.33)}, \href{http://dlmf.nist.gov/13.14.E6}{(13.14.6)}]{DLMF}.

Let us remark that many of the above results are in parallel to those obtained recently for $BC$ open $SL(2,\mathbb{C})$ spin chain~\cite{ADV2}. In particular, notice that in the Mellin--Barnes representation~\eqref{Psi-MB-0} we expand $BC$ wave functions in the basis of the functions $e^{\beta e^{-x_1}} \, \Phi_{\bm{\lambda}_n}(\bm{x}_n)$. The latter are eigenfunctions of the linear combination $A_n(u) + \imath \beta B_n(u)$ (see Section~\ref{sec:QISM} and particularly formulas~\eqref{AB0}). Similar expansion for the spin chain is performed in~\cite[Section 8]{ADV2}.

\newpage

\section{Gauss--Givental representation}\label{section2}

\subsection{Diagrams} \label{sec:diagrams}
Basic calculations for Gauss--Givental representation reduce to certain integral identities. It is convenient to visualize the corresponding integrals and identities using diagrams. In this section we introduce graphic notations and fundamental identities in graphic language. 

\begin{figure}[h] \centering \vspace{0.2cm}
	\begin{tikzpicture}[thick, line cap = round]
		\def\l{1.4}
		\def\r{1.5pt}
		\def\h{0.9}
		\def\s{6}
		\def\d{-2.7}
		\def\a{2.5}
		\def\t{0.03}
		
		\draw[fill=black] (0.5*\l, -0.43*\h) circle (0.4*\r) node[xshift = -0.4cm] {$x$} node[yshift = 0.4cm] {\color{spec} $\lambda$} node[xshift = 1.2cm] {$ = \;\; e^{\imath \lambda x}$};
		\draw (\s, -\h) node[xshift = -0.25cm, yshift = -0.15cm] {$y$} -- (\s + \l, 0) node[xshift = 0.25cm, yshift = 0.15cm] {$x$};
		\draw (\s + 2.3*\l, 0) node[xshift = -0.25cm, yshift = 0.15cm] {$x$} -- node[xshift = 1.15cm, yshift = 0cm] {$=$ \hspace{2.7cm} $= \;\; \exp(-e^{x - y})$} (\s + 3.3*\l, -\h) node[xshift = 0.25cm, yshift = -0.15cm] {$y$};
		
		\draw[line width = 0.6pt] (\a + \t, \d) -- node[xshift = 5.1cm] {\color{spec} $\lambda$ \color{black} \hspace{0.6cm} $= \; \; \dfrac{(2\beta)^{\imath \lambda}}{\Gamma(g - \imath \lambda)} \;
			(1 + \beta e^{-y})^{- \imath \lambda - g} \; (1 - \beta e^{-y})^{- \imath \lambda + g -1}			 
			\;  \theta(y - \ln \beta)$} (\a + \t, \d-\l)
		node[xshift = -0.05cm, yshift = -0.3cm] {$y$};
		\draw[line width = 0.6pt] (\a-\t, \d) --  (\a-\t, \d-\l);
	\end{tikzpicture}
	\vspace{0.2cm}	\caption{Elements of diagrams} \label{fig:lines}
\end{figure}

Kernels of all operators related to Gauss--Givental representation are built from three types of functions, which we depict by vertices and lines with labels, see Figure~\ref{fig:lines}. We distinguish top and bottom of ordinary line, since the function $\exp(-e^{x - y})$ is not symmetric in~$x, y$. For the proofs of various integral identities we also need functions depicted by dashed lines, see Figure~\ref{fig:dashed}. They act like ``ghosts'', since they appear only during calculations, but don't enter the final answer. For clarity, we use different colour for the parameters of dashed lines.

\begin{figure}[h] \centering
	\begin{tikzpicture}[thick, line cap = round]
		\def\l{1.4}
		\def\r{1.5pt}
		\def\d{7.3}		
		\draw[dashed] (0, -1.2) node[left] {$x$} -- node[yshift = 0.3cm] {\color{spec2} $\lambda$} (\l, -1.2) node[right] {$y \;\; = \; (e^x + e^y)^{- \imath \lambda}$};
		\draw[dashed] (\d, -0.3*\l) node[xshift=0.15cm, yshift=0.2cm] {$x$} -- node[xshift = 2.3cm] {\color{spec2} $\lambda$ \color{black} $\qquad  = \; (1 - e^{x - y})^{\imath \lambda - 1} \; \theta(y - x)$} (\d, -1.3*\l) node[xshift=0.15cm, yshift=-0.2cm] {$y$};
	\end{tikzpicture}
	\vspace{0.2cm}
	\caption{``Ghosts''} \label{fig:dashed}
\end{figure}

Finally, bold vertices depict integration over $\mathbb{R}$ with respect to corresponding variable, see example in Figure~\ref{fig:diag-ex}. Here and in what follows we omit labels of integrated variables.

\begin{figure}[h] \centering
	\vspace{0.7cm}
	\begin{tikzpicture}[thick, line cap = round]
		\def\lv{1.3}
		\def\l{1.1}
		\def\r{1.5pt}
		\def\h{0.7}
		\def\d{6.1}
		\def\s{4.2}
		
		\draw (0, -\lv) -- (-\l, -\lv - \h) node[xshift = -0.15cm, yshift = -0.3cm] {$x$};
		\draw (0, -\lv) -- (\l, -\lv - \h) node[xshift = 0.15cm, yshift = -0.3cm] {$y$};
		\draw[fill = black] (0, -\lv) circle (\r) node[xshift = 0.25cm, yshift = 0.25cm] {\color{spec} $\lambda$} node[xshift = 6cm, yshift = -0.4cm] {$= \qquad \displaystyle \int_{\mathbb{R}} dz \; \exp( \imath \lambda z - e^{z - x} - e^{z - y} )$};
	\end{tikzpicture}
	\vspace{0.2cm}
	\caption{Bold vertices depict integration} \label{fig:diag-ex}
\end{figure}

Using the above rules we can picture many useful integral identities as transformations of diagrams, see figures below. Notice that apart from diagrams some relations contain coefficients with gamma functions. From the whole list of identities the basic ones are \textit{star-triangle} and \textit{flip} relations (Figures~\ref{fig:star-triangle} and~\ref{fig:flip}), since all others can be obtained from them. Alternatively, one can prove all relations independently.

The first star-triangle relation from Figure~\ref{fig:star-triangle} represents the following integral identity
\begin{multline}\label{st-tr}
	\int_{z}^\infty dw \; (1 - e^{z - w})^{\imath \lambda - 1} \, \exp ( \imath \lambda w - e^{w - x} - e^{w - y} ) \\[6pt]
	= \Gamma(\imath \lambda) \; (e^x + e^y)^{-\imath \lambda} \, \exp\bigl( \imath \lambda (x + y) - e^{z - x} - e^{z - y} \bigr).
\end{multline}
The proof is straightforward: after the change of integration variable $w \to t$
\begin{align}
	t = (e^{-x} + e^{-y})(e^w - e^z), \qquad dt = (e^{-x} + e^{-y}) \, e^w \,  dw,
\end{align}
the left hand side acquires the form 
\beqq
 (e^{-x} + e^{-y})^{-\imath \lambda} \, \exp( - e^{z - x} - e^{z - y} ) \, \int_{0}^\infty dt \; t^{\imath \lambda - 1} \, e^{-t}.
\eeqq
The remaining integral gives gamma function $\Gamma(\imath \lambda)$. Hence, we arrived at the right hand side of star-triangle identity~\eqref{st-tr}. 
The second star-triangle relation can be proven in analogous way. Note that gamma function integral and, as a consequence, the ``star'' integrals are absolutely convergent under assumption $\Re(\imath \lambda) > 0$. 
\begin{figure}[t] \centering 
	\begin{tikzpicture}[thick, line cap = round]
		\def\lv{1.15}
		\def\l{1.1}
		\def\r{1.5pt}
		\def\h{0.7}
		\def\d{6.1}
		\def\s{4.2}
		
		\draw[dashed] (0, 0) node[yshift = 0.3cm] {$z$} -- node[xshift = -0.3cm, yshift = 0.1cm] {\color{spec2} $\lambda$} (0, -\lv) node[xshift = 0.25cm, yshift = 0.2cm] {\color{spec} $\lambda$};
		\draw (0, -\lv) -- (-\l, -\lv - \h) node[xshift = -0.15cm, yshift = -0.3cm] {$x$};
		\draw (0, -\lv) -- (\l, -\lv - \h) node[xshift = 0.15cm, yshift = -0.3cm] {$y$};
		\draw[fill = black] (0, -\lv) circle (\r) node[xshift = 3.1cm] {$= \qquad \Gamma(\imath \lambda)$};
		
		\draw (\d, 0) node[yshift = 0.3cm] {$z$} -- (\d - \l, -\lv - \h) node[xshift = -0.1cm, yshift = -0.3cm] {$x$} node[xshift = -0.25cm, yshift = 0.25cm] {\color{spec} $\lambda$};
		\draw (\d, 0) -- (\d + \l, -\lv - \h) node[xshift = 0.1cm, yshift = -0.3cm] {$y$} node[xshift = 0.22cm, yshift = 0.25cm] {\color{spec} $\lambda$};
		\draw[dashed] (\d - \l, -\lv - \h) -- node[yshift = -0.3cm] {\color{spec2} $\lambda$} (\d + \l, -\lv - \h);
		
		\draw (-\l, -\s) node[xshift = -0.2cm, yshift = 0.2cm] {$x$} -- (0, -\s - \h);
		\draw (\l, -\s) node[xshift = 0.2cm, yshift = 0.2cm] {$y$} -- (0, -\s - \h);
		\draw[dashed] (0, -\s - \h) -- node[xshift = -0.3cm] {\color{spec2} $\lambda$} (0, -\s - \h - \lv) node[yshift = -0.3cm] {$z$};
		\draw[fill = black] (0, -\s - \h) circle (\r) node[xshift = 0.4cm, yshift = -0.15cm] {\color{spec} $-\lambda$} node[xshift = 3.1cm] {$= \qquad \Gamma(\imath \lambda)$};
		
		\draw (\d - \l, -\s) node[xshift = -0.2cm, yshift = 0.2cm] {$x$} -- (\d, -\s - \h - \lv) node[yshift = -0.3cm] {$z$};
		\draw (\d + \l, -\s) node[xshift = 0.2cm, yshift = 0.2cm] {$y$} -- (\d, -\s - \h - \lv);
		\draw[dashed] (\d - \l, -\s) -- node[yshift = 0.3cm] {\color{spec2} $\lambda$} (\d + \l, -\s);
	\end{tikzpicture}
	\vspace{0.2cm}
	\caption{Star-triangle relations} \label{fig:star-triangle}
\end{figure}
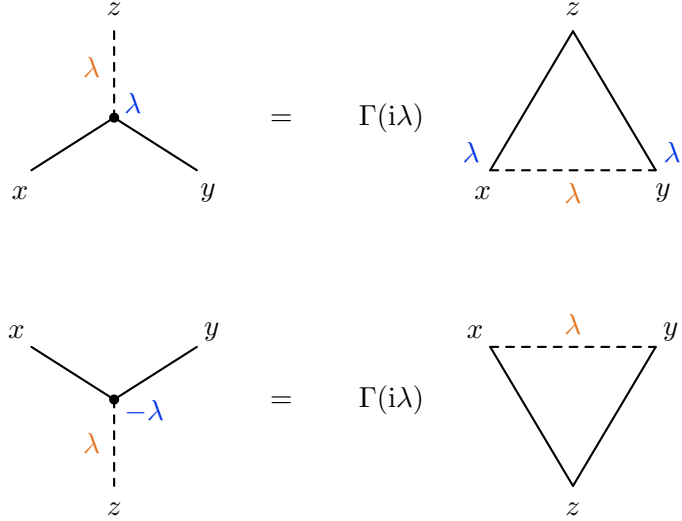

The chain relations from Figure~\ref{fig:chain} represent reductions of star-triangle relations as~\mbox{$z \to \pm \infty$}, when some lines disappear. For example, taking limit $z \to -\infty$ of the first star-triangle relation with the help of formulas
\begin{align}
	\lim_{z \to -\infty} (1 - e^{z - w})^{\imath \lambda - 1} \, \theta(w - z) = 1, \qquad \lim_{z \to -\infty} \exp(-e^{z - x}) = 1
\end{align}
we arrive at the first chain relation.
	\begin{figure}[t] \centering 
	
	\begin{tikzpicture}[thick, line cap = round]
		\def\lv{1.3}
		\def\l{1.1}
		\def\r{1.5pt}
		\def\h{0.7}
		\def\d{6.1}
		\def\s{4.2}
		
		\draw (0, -\lv) -- (-\l, -\lv - \h) node[xshift = -0.15cm, yshift = -0.3cm] {$x$};
		\draw (0, -\lv) -- (\l, -\lv - \h) node[xshift = 0.15cm, yshift = -0.3cm] {$y$};
		\draw[fill = black] (0, -\lv) circle (\r) node[xshift = 0.25cm, yshift = 0.2cm] {\color{spec} $\lambda$} node[xshift = 3.1cm, yshift = -0.4cm] {$= \qquad \Gamma(\imath \lambda)$};
		
		\draw[dashed] (\d - \l, -\lv - 0.5*\h) node[xshift = -0.1cm, yshift = -0.3cm] {$x$} node[xshift = -0.25cm, yshift = 0.25cm] {\color{spec} $\lambda$} -- node[yshift = -0.3cm] {\color{spec2} $\lambda$} (\d + \l, -\lv - 0.5*\h) node[xshift = 0.1cm, yshift = -0.3cm] {$y$} node[xshift = 0.22cm, yshift = 0.25cm] {\color{spec} $\lambda$};
		
		\draw(-\l, -\s) node[xshift = -0.2cm, yshift = 0.2cm] {$x$} -- (0, -\s - \h);
		\draw (\l, -\s) node[xshift = 0.2cm, yshift = 0.2cm] {$y$} -- (0, -\s - \h);
		\draw[fill = black] (0, -\s - \h) circle (\r) node[xshift = 0.4cm, yshift = -0.15cm] {\color{spec} $-\lambda$} node[xshift = 3.1cm, yshift = 0.4cm] {$= \qquad \Gamma(\imath \lambda)$};
		
		\draw[dashed] (\d - \l, -\s - 0.5*\h) node[xshift = -0.2cm, yshift = 0.2cm] {$x$} -- node[yshift = 0.3cm] {\color{spec2} $\lambda$} (\d + \l, -\s - 0.5*\h)  node[xshift = 0.2cm, yshift = 0.2cm] {$y$};
	\end{tikzpicture}
	\vspace{0.2cm}
	\caption{Chain relations} \label{fig:chain}
\end{figure}

 To justify interchange of limit and integration note that ``chain'' integral is analytic in $\lambda$ in the domain $\Re(\imath \lambda)>0$, see Figure~\ref{fig:diag-ex}. So, one can, at first, assume $\Re(\imath \lambda) \geq 1$, then apply dominated convergence theorem using the bound
\begin{align}
	\bigl| (1 - e^{z - w})^{\imath \lambda - 1} \, \theta(w - z) \bigr| \leq 1,
\end{align}
and analytically continue the answer at the end.

\begin{figure}[H] \centering 
	
	\vspace{0.2cm}
	\begin{tikzpicture}[thick, line cap = round]
		\def\l{1.1}
		\def\r{1.5pt}
		\def\h{1}
		\def\d{4.5}
		\def\s{-3.7}
		\def\t{-7}
		
		\draw (-\l, \h) -- (0, 0); 
		\draw (\l, \h) -- (0, 0);
		\draw (0, 0) -- (-\l, -\h);
		\draw (0, 0) -- (\l, -\h);
		\draw[dashed] (-\l, \h) -- node[yshift = 0.3cm] {\color{spec2} $\lambda$} (\l, \h);
		\draw[fill = black] (0, 0) circle (\r) node[xshift = 0.4cm] {\color{spec} $\lambda$} node[xshift = 2.2cm] {$=$};
		
		\draw (\d - \l, \h) -- (\d, 0); 
		\draw (\d + \l, \h) -- (\d, 0);
		\draw (\d, 0) -- (\d - \l, -\h);
		\draw (\d, 0) -- (\d + \l, -\h);
		\draw[dashed] (\d - \l, -\h) node[xshift = -0.25 cm, yshift = -0.1cm] {\color{spec} $\lambda$} -- node[yshift = -0.3cm] {\color{spec2} $\lambda$} (\d + \l, -\h) node[xshift = 0.2 cm, yshift = -0.1cm] {\color{spec} $\lambda$} ;
		\draw[fill = black] (\d, 0) circle (\r) node[xshift = 0.5cm] {\color{spec} $- \lambda$};
		
		\draw (-\l, \s + \h) -- (0, \s); 
		\draw (\l, \s + \h) -- (0, \s);
		\draw (0, \s) -- (\l, \s - \h);
		\draw[dashed] (-\l, \s + \h) -- node[yshift = 0.3cm] {\color{spec2} $\lambda$} (\l, \s + \h);
		\draw[fill = black] (0, \s) circle (\r) node[xshift = 0.4cm] {\color{spec} $\lambda$} node[xshift = 2.2cm] {$=$};
		
		\draw (\d - \l, \s + \h) -- (\d, \s); 
		\draw (\d + \l, \s + \h) -- (\d, \s);
		\draw (\d, \s) -- (\d + \l, \s - \h) node[xshift = 0.2 cm, yshift = -0.1cm] {\color{spec} $\lambda$};
		\draw[fill = black] (\d, \s) circle (\r) node[xshift = 0.5cm] {\color{spec} $- \lambda$};
		
		\draw (-\l, \t + \h) node[xshift = -0.1 cm, yshift = 0.25cm] {\color{spec} $-\lambda$} -- (0, \t) ; 
		\draw (0, \t) -- (-\l, \t - \h);
		\draw (0, \t) -- (\l, \t - \h);
		\draw[fill = black] (0, \t) circle (\r) node[xshift = 0.4cm] {\color{spec} $\lambda$} node[xshift = 2.2cm] {$=$};
		
		\draw (\d - \l, \t + \h) -- (\d, \t); 
		\draw (\d, \t) -- (\d - \l, \t-\h);
		\draw (\d, \t) -- (\d + \l, \t-\h);
		\draw[dashed] (\d - \l, \t-\h) node[xshift = -0.25 cm, yshift = -0.1cm] {\color{spec} $\lambda$} -- node[yshift = -0.3cm] {\color{spec2} $\lambda$} (\d + \l, \t-\h) node[xshift = 0.2 cm, yshift = -0.1cm] {\color{spec} $\lambda$} ;
		\draw[fill = black] (\d, \t) circle (\r) node[xshift = 0.5cm] {\color{spec} $- \lambda$};
		
	\end{tikzpicture}
	\vspace{0.1cm}
	\caption{Cross and reduced cross relations} \label{fig:cross} \vspace{0.2cm}
\end{figure}

The cross relation from Figure~\ref{fig:cross} can be proved in two steps using star-triangle transformations, as shown in Figure~\ref{fig:cross-proof}. Let us remark that in these pictures and in what follows we frequently omit labels of vertices.

\begin{figure}[h] \centering \vspace{0.2cm}
	\begin{tikzpicture}[thick, line cap = round]
		\def\l{1.1}
		\def\r{1.5pt}
		\def\h{1}
		\def\d{6}
		\def\a{4.8}
		\def\s{-4}
		\def\t{-7.5}
		
		\draw (-\l, \h) -- (0, 0); 
		\draw (\l, \h) -- (0, 0);
		\draw (0, 0) -- (-\l, -\h);
		\draw (0, 0) -- (\l, -\h);
		\draw[dashed] (-\l, \h) -- node[yshift = 0.3cm] {\color{spec2} $\lambda$} (\l, \h);
		\draw[fill = black] (0, 0) circle (\r) node[xshift = 0.4cm] {\color{spec} $\lambda$} node[xshift = 2.2cm] {$=$};
		
		\draw (\d - \l, 1.65*\h) -- (\d, 0.65*\h); 
		\draw (\d + \l, 1.65*\h) -- (\d, 0.65*\h);
		\draw (\d, -0.65*\h) -- (\d - \l, -1.65*\h);
		\draw (\d, -0.65*\h) -- (\d + \l, -1.65*\h);
		\draw[dashed] (\d, 0.65*\h) -- node[xshift = -1.4cm] {$\dfrac{1}{\Gamma(\imath \lambda)}$ \hspace{0.6cm} \color{spec2} $\lambda$} node[xshift = 2.2cm] {$=$} (\d, -0.65*\h);
		\draw[fill = black] (\d, 0.65*\h) circle (\r) node[xshift = 0.5cm] {\color{spec} $- \lambda$};
		\draw[fill = black] (\d, -0.65*\h) circle (\r) node[xshift = 0.35cm, yshift = 0.1cm] {\color{spec} $\lambda$};
		
		\draw (\a + \d - \l, \h) -- (\a + \d, 0); 
		\draw (\a + \d + \l, \h) -- (\a + \d, 0);
		\draw (\a + \d, 0) -- (\a + \d - \l, -\h);
		\draw (\a + \d, 0) -- (\a + \d + \l, -\h);
		\draw[dashed] (\a + \d - \l, -\h) node[xshift = -0.25 cm, yshift = -0.1cm] {\color{spec} $\lambda$} -- node[yshift = -0.3cm] {\color{spec2} $\lambda$} (\a + \d + \l, -\h) node[xshift = 0.2 cm, yshift = -0.1cm] {\color{spec} $\lambda$} ;
		\draw[fill = black] (\a + \d, 0) circle (\r) node[xshift = 0.5cm] {\color{spec} $- \lambda$};
		
	\end{tikzpicture}
	\vspace{0.2cm}
	\caption{Proof of cross relation using star-triangle relations} \label{fig:cross-proof} \vspace{0.2cm}
\end{figure}
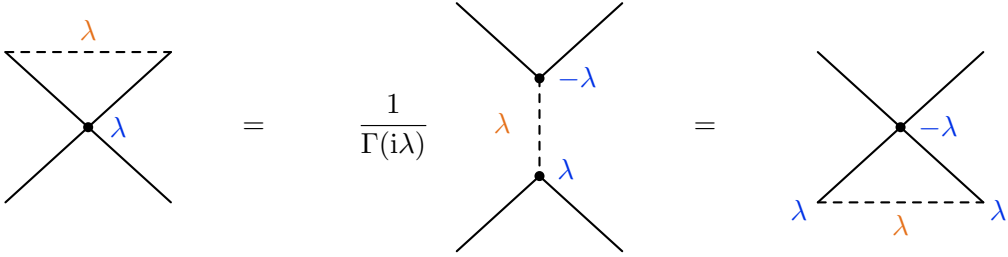

Also note that both sides of cross relation are analytic in $\lambda \in \mathbb{C}$. Thus, to prove it using star-triangle relations we assume $\Re (\imath \lambda) > 0$ and analytically continue the answer at the end.

The two other relations from Figure~\ref{fig:cross} can be obtained from the first one in the limit when one of non-integrated variables tends to $\pm \infty$. 

\begin{figure}[H] \centering \vspace{0.4cm}
	
	\begin{tikzpicture}[thick, line cap = round]
		\def\l{1.4}
		\def\r{1.5pt}
		\def\d{7}
		\def\hh{1.45}
		\def\t{0.03}
		
		\draw[line width = 0.6pt] (\t, \hh) -- node[xshift = -0.35cm, yshift = 0.15cm] {\color{spec} \footnotesize $\lambda$} (\t, 0);
		\draw[line width = 0.6pt] (-\t, \hh) --  (-\t, 0);
		
		\draw[line width = 0.6pt] (\t + 2*\l, \hh) -- node[xshift = 0.25cm, yshift = 0.1cm] {\color{spec} \footnotesize $\rho$} (\t + 2*\l, 0);
		\draw[line width = 0.6pt] (-\t + 2*\l, \hh) --  (-\t + 2*\l, 0);
		
		\draw[dashed] (0, 0) node[xshift = -0.5cm, yshift = -0.1cm] {\color{spec} \footnotesize  $-2\lambda$} -- node[yshift = 0.3cm] {\color{spec2} \footnotesize  $-\lambda - \rho$} (2*\l, 0) node[xshift = 0.5cm, yshift = -0.1cm] {\color{spec} \footnotesize  $-2\rho$} node[xshift = -0.1cm, yshift = -0.3cm] {$y$} node[xshift = 2.1cm] {$=$};
		\draw[dashed] (0, 0) -- node[xshift = 0.6cm] {\color{spec2} \footnotesize  $\lambda - \rho$} (0, -\hh) node[xshift = 0.25cm, yshift = -0.2cm] {$x$} node[xshift = -0.6cm, yshift = -0.15cm] {\color{spec} \footnotesize  $\lambda - \rho$};
		
		\draw[fill = black] (0, 0) circle (\r);

		\draw[line width = 0.6pt] (\d + \t, \hh) -- node[xshift = -0.3cm, yshift = 0.1cm] {\color{spec} \footnotesize $\rho$} (\d + \t, 0);
		\draw[line width = 0.6pt] (\d -\t, \hh) --  (\d -\t, 0);
		
		\draw[line width = 0.6pt] (\d + \t + 2*\l, \hh) -- node[xshift = 0.25cm, yshift = 0.15cm] {\color{spec} \footnotesize $\lambda$} (\d + \t + 2*\l, 0);
		\draw[line width = 0.6pt] (\d -\t + 2*\l, \hh) --  (\d -\t + 2*\l, 0);
		
		\draw[dashed] (\d, 0) node[xshift = -0.5cm, yshift = -0.1cm] {\color{spec} \footnotesize $-2\rho$} node[xshift = 0.1cm, yshift = -0.3cm] {$x$} -- node[yshift = 0.3cm] {\color{spec2} \footnotesize $-\lambda - \rho$} (\d + 2*\l, 0) node[xshift = 0.5cm, yshift = -0.1cm] {\color{spec} \footnotesize $-2\lambda$} ;
		\draw[dashed] (2*\l + \d, 0) -- node[xshift = -0.6cm] {\color{spec2} \footnotesize $\lambda - \rho$} (2*\l + \d, -\hh) node[xshift = -0.25cm, yshift = -0.2cm] {$y$} node[xshift = 0.6cm, yshift = -0.15cm] {\color{spec} \footnotesize $\lambda - \rho$};
		
		\draw[fill = black] (2*\l + \d, 0) circle (\r);
		
	\end{tikzpicture}
	\vspace{-0.2cm}
	\caption{Flip relation} \label{fig:flip} 	
	\vspace{0.2cm}
\end{figure}
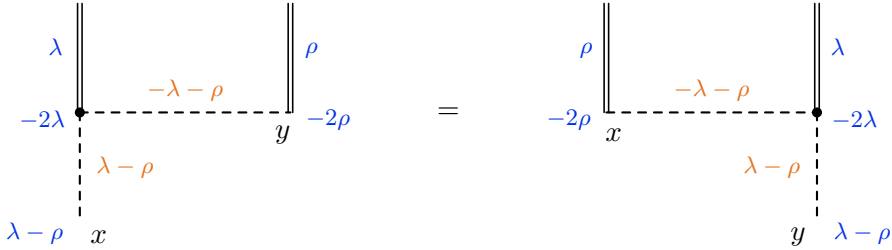

Now consider flip relation from Figure~\ref{fig:flip}. Vertical lines come with step functions, hence, diagram from the left represents expression
\begin{align}
	\frac{(2\beta)^{\imath (\lambda + \rho)}}{\Gamma(g - \imath \lambda) \, \Gamma(g - \imath \rho)} \, F(x, y) \, \theta(x - \ln \beta) \, \theta(y - \ln \beta),
\end{align}
where we denoted
\begin{align} \label{Fflip}
	\begin{aligned}
		& F(x, y) = e^{\imath (\lambda - \rho) x -2\imath \rho y} \, (1 + \beta e^{-y})^{- \imath \rho - g} \; (1 - \beta e^{-y})^{- \imath \rho + g - 1} \\[6pt]
		& \times \int_{\ln \beta}^x dz \;\, e^{-2\imath \lambda z} \, (1 + \beta e^{-z})^{- \imath \lambda - g} \; (1 - \beta e^{-z})^{- \imath \lambda +g - 1} \, (e^z + e^y)^{\imath(\lambda + \rho)} \; (1 - e^{z - x})^{\imath (\lambda - \rho) - 1}.
	\end{aligned}
\end{align}
The flip relation just states that this function is symmetric $F(x, y) = F(y, x)$, the proof is given in Appendix~\ref{app:flip}. Let us note that the above integral is absolutely convergent under condition~$\Re(\imath \rho) < \Re (\imath \lambda) < g$.

For calculations we actually need not flip relation itself, but the transformation from Figure~\ref{fig:flip-tr}. The latter follows from the sequence of flip and star-triangle relations pictured in Figure~\ref{fig:flip-tr-proof}. On the last step of this sequence two dashed lines cancel each other.

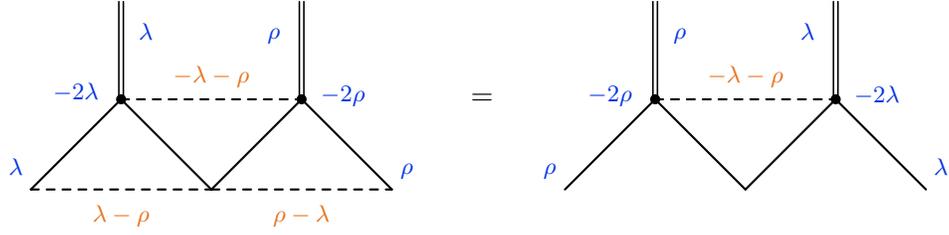
\begin{figure}[h] \centering 
	\begin{tikzpicture}[thick, line cap = round]
		\def\l{1.2}
		\def\r{1.5pt}
		\def\h{1.2}
		\def\d{7.1}
		\def\hh{1.3}
		\def\t{0.03}
		
		\draw[line width = 0.6pt] (\t, \hh) -- node[xshift = 0.3cm, yshift = 0.25cm] {\color{spec} \footnotesize $\lambda$} (\t, 0);
		\draw[line width = 0.6pt] (-\t, \hh) --  (-\t, 0);
		
		\draw[line width = 0.6pt] (\t + 2*\l, \hh) -- node[xshift = -0.4cm, yshift = 0.2cm] {\color{spec} \footnotesize $\rho$} (\t + 2*\l, 0);
		\draw[line width = 0.6pt] (-\t + 2*\l, \hh) --  (-\t + 2*\l, 0);
		
		\draw (0, 0) -- (\l, -\h);
		\draw (0, 0) -- (-\l, -\h) node[xshift = -0.2cm, yshift = 0.3cm] {\color{spec} \footnotesize  $\lambda$};
		\draw (2*\l, 0) -- (\l, -\h);
		\draw (2*\l, 0) -- (3*\l, -\h) node[xshift = 0.2cm, yshift = 0.25cm] {\color{spec} \footnotesize  $\rho$};
		\draw[dashed] (0, 0) -- node[yshift = 0.3cm] {\color{spec2} \footnotesize  $-\lambda - \rho$} (2*\l, 0);
		\draw[dashed] (-\l, -\h) -- node[yshift = -0.35cm] {\color{spec2} \footnotesize  $\lambda - \rho$} (\l, -\h);
		\draw[dashed] (\l, -\h) -- node[yshift = -0.35cm] {\color{spec2} \footnotesize  $\rho - \lambda$} (3*\l, -\h);
		
		\draw[fill = black] (0, 0) circle (\r) node[xshift = -0.6cm, yshift = 0.1cm] {\color{spec} \footnotesize  $- 2 \lambda$};
		\draw[fill = black] (2*\l, 0) circle (\r) node[xshift = 0.55cm, yshift = 0.05cm] {\color{spec} \footnotesize  $-2\rho$} node[xshift = 2.4cm] {$=$};

		\draw[line width = 0.6pt] (\d + \t, \hh) -- node[xshift = 0.3cm, yshift = 0.2cm] {\color{spec} \footnotesize  $\rho$} (\d + \t, 0);
		\draw[line width = 0.6pt] (\d-\t, \hh) --  (\d-\t, 0);
		
		\draw[line width = 0.6pt] (\d + \t + 2*\l, \hh) -- node[xshift = -0.4cm, yshift = 0.25cm] {\color{spec} \footnotesize  $\lambda$} (\d + \t + 2*\l, 0);
		\draw[line width = 0.6pt] (\d-\t + 2*\l, \hh) --  (\d-\t + 2*\l, 0);
		
		\draw (\d, 0) -- (\d + \l, -\h);
		\draw (\d, 0) -- (\d-\l, -\h) node[xshift = -0.2cm, yshift = 0.25cm] {\color{spec} \footnotesize  $\rho$};
		\draw (\d + 2*\l, 0) -- (\d + \l, -\h);
		\draw (\d + 2*\l, 0) -- (\d + 3*\l, -\h) node[xshift = 0.2cm, yshift = 0.3cm] {\color{spec} \footnotesize  $\lambda$};
		\draw[dashed] (\d, 0) -- node[yshift = 0.3cm] {\color{spec2} \footnotesize $- \lambda - \rho$} (\d + 2*\l, 0);
		
		\draw[fill = black] (\d, 0) circle (\r) node[xshift = -0.6cm, yshift = 0.05cm] {\color{spec} \footnotesize  $- 2 \rho$};
		\draw[fill = black] (\d + 2*\l, 0) circle (\r) node[xshift = 0.55cm, yshift = 0.05cm] {\color{spec} \footnotesize  $-2\lambda$};
	\end{tikzpicture}
	\vspace{0cm}
	\caption{Twin peaks relation} \label{fig:flip-tr} 
\end{figure}

Diagrams from Figure~\ref{fig:flip-tr} represent absolutely convergent integrals under condition $\Re(\imath \lambda)$, $\Re(\imath \rho) < g$, which we impose because of double lines. To perform all steps from Figure~\ref{fig:flip-tr-proof} we should also assume $\Re(\imath \rho) < \Re(\imath \lambda)$. After calculations this assumption can be removed by analytic continuation.

\begin{figure}[H] \centering \vspace{0.2cm}
	\begin{tikzpicture}[thick, line cap = round]
		\def\l{1.3}
		\def\r{1.5pt}
		\def\h{2.1}
		\def\a{-6}
		\def\d{9}
		\def\hh{1.4}
		\def\t{0.03}
		\def\b{1.18}
		
		\draw[line width = 0.6pt] (\t, \hh) -- node[xshift = 0.3cm, yshift = 0.25cm] {\color{spec} \footnotesize $\lambda$} (\t, 0);
		\draw[line width = 0.6pt] (-\t, \hh) --  (-\t, 0);
		
		\draw[line width = 0.6pt] (\t + 2*\l, \hh) -- node[xshift = -0.4cm, yshift = 0.2cm] {\color{spec} \footnotesize $\rho$} (\t + 2*\l, 0);
		\draw[line width = 0.6pt] (-\t + 2*\l, \hh) --  (-\t + 2*\l, 0);
		
		\draw (0, 0) -- (\l, -\h);
		\draw (0, 0) -- (-\l, -\h) node[xshift = -0.25cm, yshift = -0.1cm] {\color{spec} \footnotesize $\lambda$};
		\draw (2*\l, 0) -- (\l, -\h);
		\draw (2*\l, 0) -- (3*\l, -\h) node[xshift = 0.25cm, yshift = -0.15cm] {\color{spec} \footnotesize $\rho$};
		\draw[dashed] (0, 0) -- node[yshift = 0.3cm] {\color{spec2} \footnotesize $-\lambda - \rho$} (2*\l, 0);
		\draw[dashed] (-\l, -\h) -- node[yshift = -0.35cm] {\color{spec2} \footnotesize $\lambda - \rho$} (\l, -\h);
		\draw[dashed] (\l, -\h) -- node[yshift = -0.35cm] {\color{spec2} \footnotesize $\rho - \lambda$} (3*\l, -\h);
		
		\draw[fill = black] (0, 0) circle (\r) node[xshift = -0.6cm, yshift = 0.1cm] {\color{spec} \footnotesize $- 2 \lambda$};
		\draw[fill = black] (2*\l, 0) circle (\r) node[xshift = 0.55cm, yshift = 0.05cm] {\color{spec} \footnotesize $-2\rho$};

		\draw[arrow, gray!60] (\l, -\b-\h) -- (\l, -1.45*\b-\h);

		\draw[line width = 0.6pt] (\t, \a + \hh) -- node[xshift = 0.3cm, yshift = 0.25cm] {\color{spec} \footnotesize $\lambda$} (\t, \a);
		\draw[line width = 0.6pt] (-\t, \a +  \hh) --  (-\t, \a);
		
		\draw[line width = 0.6pt] (\t + 2*\l, \a + \hh) -- node[xshift = -0.4cm, yshift = 0.2cm] {\color{spec} \footnotesize $\rho$} (\t + 2*\l, \a);
		\draw[line width = 0.6pt] (-\t + 2*\l, \a + \hh) --  (-\t + 2*\l, \a);
		
		\draw (0, \a - 0.6*\h) -- (\l, \a-\h) node[xshift = -0.45cm, yshift = -0.25cm] {\color{spec} \footnotesize $\rho - \lambda$};
		\draw (0, \a - 0.6*\h) -- (-\l, \a-\h) node[xshift = -0.25cm, yshift = -0.15cm] {\color{spec} \footnotesize $\rho$};
		\draw[dashed] (0, \a) -- node[right] {\color{spec2} \footnotesize $\lambda - \rho$} (0, \a - 0.6*\h) node[xshift = -0.6cm, yshift = 0.15cm] {\color{spec} \footnotesize $\lambda - \rho$};
		\draw (2*\l, \a) -- (\l, \a-\h);
		\draw (2*\l, \a) -- (3*\l, \a-\h) node[xshift = 0.25cm, yshift = -0.15cm] {\color{spec} \footnotesize $\rho$};
		\draw[dashed] (0, \a) -- node[yshift = 0.3cm] {\color{spec2} \footnotesize $-\lambda - \rho$} (2*\l, \a);
		\draw[dashed] (\l, \a-\h) -- node[yshift = -0.35cm] {\color{spec2} \footnotesize $\rho - \lambda$} (3*\l, \a-\h);
		
		\draw[fill = black] (0, \a) circle (\r) node[xshift = -0.6cm, yshift = 0.1cm] {\color{spec} \footnotesize $- 2 \lambda$};
		\draw[fill = black] (2*\l, \a) circle (\r) node[xshift = 0.55cm, yshift = 0.05cm] {\color{spec} \footnotesize $-2\rho$};
		\draw[fill = black] (0, \a - 0.6*\h) circle (\r);

		\draw[arrow, gray!60] (0.62*\d, 1.05*\a) -- (0.68*\d, 1.05*\a);

		\draw[line width = 0.6pt] (\d + \t, \a + \hh) -- node[xshift = 0.3cm, yshift = 0.2cm] {\color{spec} \footnotesize $\rho$} (\d + \t, \a);
		\draw[line width = 0.6pt] (\d -\t, \a +  \hh) --  (\d -\t, \a);
		
		\draw[line width = 0.6pt] (\d + \t + 2*\l, \a + \hh) -- node[xshift = -0.4cm, yshift = 0.25cm] {\color{spec} \footnotesize $\lambda$} (\d + \t + 2*\l, \a);
		\draw[line width = 0.6pt] (\d -\t + 2*\l, \a + \hh) --  (\d -\t + 2*\l, \a);
		
		\draw (\d, \a) -- (\d + \l, \a-\h) node[xshift = -0.45cm, yshift = -0.25cm] {\color{spec} \footnotesize $\rho - \lambda$};
		\draw (\d, \a) -- (\d -\l, \a-\h) node[xshift = -0.25cm, yshift = -0.15cm] {\color{spec} \footnotesize $\rho$};
		\draw (\d + 2*\l, \a - 0.6*\h) -- (\d + \l, \a-\h);
		\draw (\d + 2*\l, \a - 0.6*\h) -- (\d + 3*\l, \a-\h) node[xshift = 0.25cm, yshift = -0.15cm] {\color{spec} \footnotesize $\rho$};
		\draw[dashed] (\d + 2*\l, \a) -- node[xshift = -0.6cm] {\color{spec2} \footnotesize $\lambda - \rho$} (\d + 2*\l, \a - 0.6*\h) node[xshift = 0.6cm, yshift = 0.15cm] {\color{spec} \footnotesize $\lambda - \rho$};
		\draw[dashed] (\d, \a) -- node[yshift = 0.3cm] {\color{spec2} \footnotesize $-\lambda - \rho$} (\d + 2*\l, \a);
		\draw[dashed] (\d + \l, \a-\h) -- node[yshift = -0.35cm] {\color{spec2} \footnotesize $\rho - \lambda$} (\d + 3*\l, \a-\h);
		
		\draw[fill = black] (\d, \a) circle (\r) node[xshift = -0.6cm, yshift = 0.05cm] {\color{spec} \footnotesize $- 2 \rho$};
		\draw[fill = black] (\d + 2*\l, \a) circle (\r) node[xshift = 0.55cm, yshift = 0.05cm] {\color{spec} \footnotesize $-2\lambda$};
		\draw[fill = black] (\d + 2*\l, \a - 0.6*\h) circle (\r);

		\draw[arrow, gray!60] (\d + \l, -1.45*\b-\h) -- (\d + \l, -\b-\h);

		\draw[line width = 0.6pt] (\d + \t, \hh) -- node[xshift = 0.3cm, yshift = 0.2cm] {\color{spec} \footnotesize $\rho$} (\d + \t, 0);
		\draw[line width = 0.6pt] (\d-\t, \hh) --  (\d-\t, 0);
		
		\draw[line width = 0.6pt] (\d + \t + 2*\l, \hh) -- node[xshift = -0.4cm, yshift = 0.25cm] {\color{spec} \footnotesize $\lambda$} (\d + \t + 2*\l, 0);
		\draw[line width = 0.6pt] (\d-\t + 2*\l, \hh) --  (\d-\t + 2*\l, 0);
		
		\draw (\d, 0) -- (\d + \l, -\h);
		\draw (\d, 0) -- (\d-\l, -\h) node[xshift = -0.25cm, yshift = -0.15cm] {\color{spec} \footnotesize $\rho$};
		\draw (\d + 2*\l, 0) -- (\d + \l, -\h);
		\draw (\d + 2*\l, 0) -- (\d + 3*\l, -\h) node[xshift = 0.2cm, yshift = -0.1cm] {\color{spec} \footnotesize $\lambda$};
		\draw[dashed] (\d, 0) -- node[yshift = 0.3cm] {\color{spec2} \footnotesize $- \lambda - \rho$} (\d + 2*\l, 0);
		
		\draw[fill = black] (\d, 0) circle (\r) node[xshift = -0.6cm, yshift = 0.05cm] {\color{spec} \footnotesize $- 2 \rho$};
		\draw[fill = black] (\d + 2*\l, 0) circle (\r) node[xshift = 0.55cm, yshift = 0.05cm] {\color{spec} \footnotesize $-2\lambda$};
	\end{tikzpicture}
	\vspace{0.2cm}
	\caption{Sequence of flip and star-triangle relations} \label{fig:flip-tr-proof}  \vspace{0.2cm}
\end{figure}
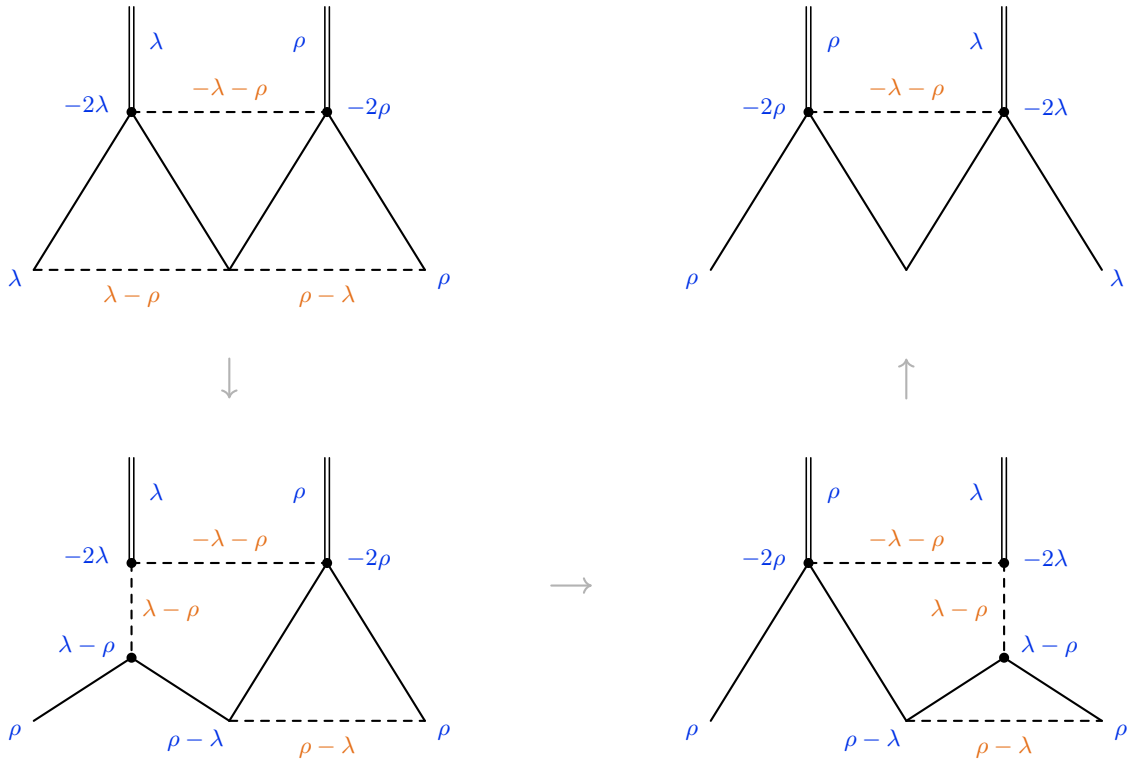

\begin{figure}[H] \centering
	\begin{tikzpicture}[thick, line cap = round]
		\def\l{1.4}
		\def\r{1.5pt}
		\def\d{8}
		
		\draw[dashed] (0, 0) node[left] {$x$} -- node[yshift = 0.3cm] {\color{spec2} $\lambda$} (\l, 0) node[right] {$y \;\; = \; (e^x + e^y)^{- \imath \lambda}$};
		\draw[dashed] (\d + \l, 0) node[right] {$\;\;= \;\; 1$} arc (75:105:3.85) node[xshift = 1cm, yshift = 0.5cm] {\color{spec2} $\lambda$};
		\draw[dashed] (\d + \l, 0) arc (-75:-105:3.85) node[xshift = 1cm, yshift = -0.5cm] {\color{spec2} $-\lambda$};
	\end{tikzpicture}
	\caption{Cancellation of lines} \label{fig:anih}
\end{figure}

\begin{figure}[H] \centering \vspace{0.2cm}
	\begin{tikzpicture}[thick, line cap = round]
		\def\l{1.4}
		\def\r{1.5pt}
		\def\d{7.6}
		\def\hh{1.4}
		\def\t{0.03}
		
		\draw[line width = 0.6pt] (\t, \hh) -- node[xshift = -0.35cm, yshift = 0.2cm] {\color{spec} \footnotesize \footnotesize $\lambda$} (\t, 0);
		\draw[line width = 0.6pt] (-\t, \hh) --  (-\t, 0);
		
		\draw[dashed] (0, 0) node[xshift = -0.5cm, yshift = -0.1cm] {\color{spec} \footnotesize $-2\lambda$} -- node[yshift = 0.3cm] {\color{spec2} \footnotesize $-\lambda - \rho$} (2*\l, 0) node[xshift = 0.5cm, yshift = -0.2cm] {\small $\ln \beta$} node[xshift = 2.5cm] {$=$ \small $\quad\; \Gamma(\imath \lambda - \imath \rho)$};
		\draw[dashed] (0, 0) -- node[xshift = 0.6cm] {\color{spec2} \footnotesize $\lambda - \rho$} (0, -\hh) node[xshift = 0.25cm, yshift = -0.2cm] {$x$} node[xshift = -0.6cm, yshift = -0.15cm] {\color{spec} \footnotesize $\lambda - \rho$};
		
		\draw[fill = black] (0, 0) circle (\r);

		\draw[line width = 0.6pt] (\d + \t, \hh) -- node[xshift = -0.35cm, yshift = 0.15cm] {\color{spec} \footnotesize $\rho$} (\d + \t, 0);
		\draw[line width = 0.6pt] (\d -\t, \hh) --  (\d -\t, 0);
		
		\draw[dashed] (\d, 0) node[xshift = -0.5cm, yshift = -0.1cm] {\color{spec} \footnotesize $-2\rho$} node[xshift = 0.1cm, yshift = -0.3cm] {$x$} -- node[yshift = 0.3cm] {\color{spec2} \footnotesize $-\lambda - \rho$} (\d + 2*\l, 0) node[xshift = 0.5cm, yshift = -0.2cm] {\small $\ln \beta$} ;
		
	\end{tikzpicture}
	\vspace{0cm}
	\caption{Reduced flip relation} \label{fig:flip-lim} 	 
\end{figure}

In the limit $y \to \ln \beta^+$ flip relation reduces to the integral identity pictured in Figure~\ref{fig:flip-lim}. This limiting identity is proven in Appendix~\ref{app:flip-lim}. The further limit $x \to +\infty$ is straightforward and gives the relation depicted in Figure~\ref{fig:flip-lim2}.

\begin{figure}[H] \centering \vspace{0.2cm}
	\begin{tikzpicture}[thick, line cap = round]
		\def\l{1.4}
		\def\r{1.5pt}
		\def\d{7.6}
		\def\hh{1.4}
		\def\t{0.03}
		
		\draw[line width = 0.6pt] (\t, \hh) -- node[xshift = -0.35cm, yshift = 0.2cm] {\color{spec} \footnotesize \footnotesize $\lambda$} (\t, 0);
		\draw[line width = 0.6pt] (-\t, \hh) --  (-\t, 0);
		
		\draw[dashed] (0, 0) node[xshift = -0.5cm, yshift = -0.1cm] {\color{spec} \footnotesize $-2\lambda$} -- node[yshift = 0.3cm] {\color{spec2} \footnotesize $-\lambda - \rho$} (2*\l, 0) node[xshift = 0.5cm, yshift = -0.2cm] {\small $\ln \beta$} node[xshift = 3cm, yshift = 0.5cm] {$=$ \small $\quad\; (2\beta)^{\imath \rho} \, \dfrac{\Gamma(\imath \lambda - \imath \rho)}{\Gamma(g - \imath \rho)}$};
		
		\draw[fill = black] (0, 0) circle (\r);
		
	\end{tikzpicture}
	\vspace{0cm}
	\caption{Relation from Figure~\ref{fig:flip-lim} in the limit $x \to +\infty$} \label{fig:flip-lim2} 	 
\end{figure}

A combination of reduced flip and star-triangle relations gives another useful transformation shown in Figure~\ref{fig:flip-lim-tr}. Due to presence of double lines it holds under assumption $\Re(\imath \lambda), \Re(\imath \rho) < g$. Moreover, in the case $\rho = - \lambda$ one of the dashed lines disappears and the identity simplifies, see Figure~\ref{fig:flip-lim-tr2}.

\begin{figure}[H] \centering \vspace{0.2cm}
	\begin{tikzpicture}[thick, line cap = round]
		\def\l{1.2}
		\def\r{1.5pt}
		\def\h{1.2}
		\def\d{7}
		\def\hh{1.3}
		\def\t{0.03}
		
		\draw[line width = 0.6pt] (\t, \hh) -- node[xshift = 0.3cm, yshift = 0.25cm] {\color{spec} \footnotesize $\lambda$} (\t, 0);
		\draw[line width = 0.6pt] (-\t, \hh) --  (-\t, 0);

		\draw (0, 0) -- (\l, -\h) node[xshift = 0.2cm, yshift = 0.25cm] {\color{spec} \footnotesize  $-\rho$};
		\draw (0, 0) -- (-\l, -\h) node[xshift = -0.2cm, yshift = 0.3cm] {\color{spec} \footnotesize  $\lambda$};
		\draw[dashed] (0, 0) -- node[yshift = 0.3cm] {\color{spec2} \footnotesize  $-\lambda - \rho$} (2*\l, 0) node[xshift = 0.4cm, yshift = -0.25cm] {\small $\ln \beta$} node[xshift = 2cm] {$=$};
		\draw[dashed] (-\l, -\h) -- node[yshift = -0.35cm] {\color{spec2} \footnotesize  $\lambda - \rho$} (\l, -\h);
		
		\draw[fill = black] (0, 0) circle (\r) node[xshift = -0.6cm, yshift = 0.1cm] {\color{spec} \footnotesize  $- 2 \lambda$};

		\draw[line width = 0.6pt] (\d + \t, \hh) -- node[xshift = 0.3cm, yshift = 0.2cm] {\color{spec} \footnotesize  $\rho$} (\d + \t, 0);
		\draw[line width = 0.6pt] (\d-\t, \hh) --  (\d-\t, 0);
		
		\draw (\d, 0) -- (\d + \l, -\h) node[xshift = 0.2cm, yshift = 0.3cm] {\color{spec} \footnotesize  $-\lambda$};
		\draw (\d, 0) -- (\d-\l, -\h) node[xshift = -0.2cm, yshift = 0.25cm] {\color{spec} \footnotesize  $\rho$};
		\draw[dashed] (\d, 0) -- node[yshift = 0.3cm] {\color{spec2} \footnotesize $- \lambda - \rho$} (\d + 2*\l, 0) node[xshift = 0.4cm, yshift = -0.25cm] {\small $\ln \beta$};
		
		\draw[fill = black] (\d, 0) circle (\r) node[xshift = -0.6cm, yshift = 0.05cm] {\color{spec} \footnotesize  $- 2 \rho$};
	\end{tikzpicture}
	\vspace{0cm}
	\caption{Combination of reduced flip and star-triangle relations} \label{fig:flip-lim-tr} 	 
\end{figure}

\begin{figure}[H] \centering \vspace{0.2cm}
	\begin{tikzpicture}[thick, line cap = round]
		\def\l{1.2}
		\def\r{1.5pt}
		\def\h{1.2}
		\def\d{5}
		\def\hh{1.3}
		\def\t{0.03}
		
		\draw[line width = 0.6pt] (\t, \hh) -- node[xshift = 0.3cm, yshift = 0.2cm] {\color{spec} \footnotesize $\lambda$} (\t, 0);
		\draw[line width = 0.6pt] (-\t, \hh) --  (-\t, 0);

		\draw (0, 0) -- (\l, -\h) node[xshift = 0.2cm, yshift = 0.3cm] {\color{spec} \footnotesize  $\lambda$};
		\draw (0, 0) -- (-\l, -\h) node[xshift = -0.2cm, yshift = 0.3cm] {\color{spec} \footnotesize  $\lambda$};
		\draw[dashed] (-\l, -\h) -- node[yshift = -0.35cm] {\color{spec2} \footnotesize  $2\lambda$} (\l, -\h);
		
		\draw[fill = black] (0, 0) circle (\r) node[xshift = -0.6cm, yshift = 0.1cm] {\color{spec} \footnotesize  $- 2 \lambda$}  node[xshift = 2.4cm] {$ = $};

		\draw[line width = 0.6pt] (\d + \t, \hh) -- node[xshift = 0.4cm, yshift = 0.2cm] {\color{spec} \footnotesize  $-\lambda$} (\d + \t, 0);
		\draw[line width = 0.6pt] (\d-\t, \hh) --  (\d-\t, 0);
		
		\draw (\d, 0) -- (\d + \l, -\h) node[xshift = 0.2cm, yshift = 0.3cm] {\color{spec} \footnotesize  $-\lambda$};
		\draw (\d, 0) -- (\d-\l, -\h) node[xshift = -0.35cm, yshift = 0.3cm] {\color{spec} \footnotesize  $-\lambda$};
		
		\draw[fill = black] (\d, 0) circle (\r) node[xshift = -0.5cm, yshift = 0.1cm] {\color{spec} \footnotesize  $2 \lambda$};
	\end{tikzpicture}
	\vspace{0cm}
	\caption{The case $\rho = - \lambda$} \label{fig:flip-lim-tr2} 	 
\end{figure}

\subsection{Raising operator and reflection symmetry} \label{sec:rais-op}
The raising operator $\LLambda_n(\lambda)$ for $BC$ Toda chain is an integral operator acting on functions~$\phi(\bm{x}_{n - 1})$
\begin{align}
	\bigl[ \LLambda_n(\lambda) \, \phi \bigr] (\bm{x}_n) = \int_{\mathbb{R}^{n - 1}} d\bm{y}_{n - 1} \; \LLambda_\lambda(\bm{x}_n | \bm{y}_{n - 1}) \, \phi(\bm{y}_{n - 1})
\end{align}
with the kernel
\begin{align} \label{Lker}
	\begin{aligned}
		& \LLambda_\lambda(\bm{x}_n | \bm{y}_{n - 1}) = \frac{(2\beta)^{\imath \lambda}}{\Gamma(g - \imath \lambda)} \int_{\mathbb{R}^n} d\bm{z}_n \;
		(1 + \beta e^{-z_1})^{- \imath \lambda - g} \; (1 - \beta e^{-z_1})^{- \imath \lambda + g - 1}\, \theta(z_1 - \ln \beta) \\[6pt]
		& \quad \times  \exp \biggl( \imath \lambda \bigl( \underline{\bm{x}}_n + \underline{\bm{y}}_{n - 1} - 2 \underline{\bm{z}}_n \bigr) - \sum_{j = 1}^{n - 1} (e^{z_j - x_j} + e^{z_j - y_j} + e^{x_j - z_{j + 1}} + e^{y_j - z_{j + 1}} )  - e^{z_n - x_n} \biggr).
	\end{aligned}
\end{align}
It is easy to understand the structure of this integral by looking at the corresponding diagrams in Figure~\ref{fig:Lker}. For the parameter of the raising operator we always assume $\lambda \in \mathbb{R}$, and the kernel is clearly absolutely convergent in this case (since $g > 0$). 

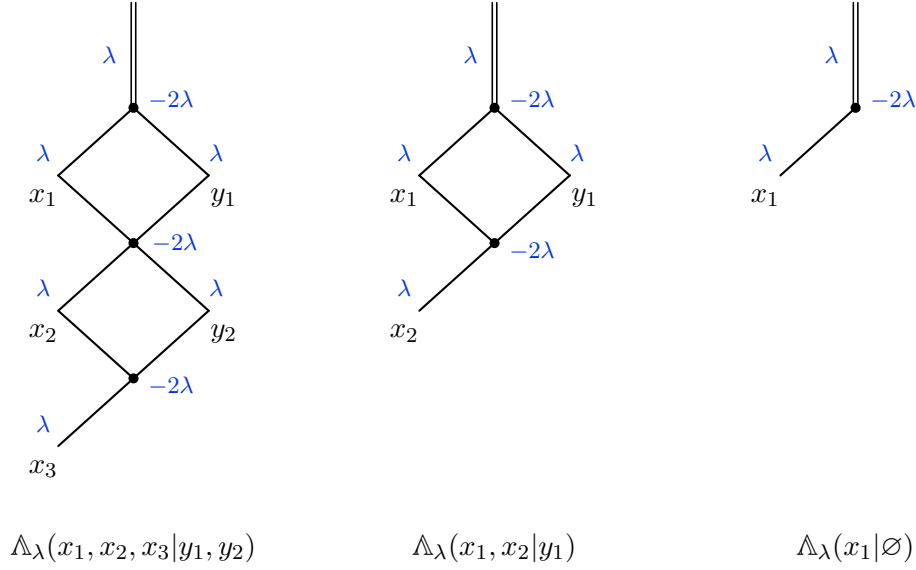
\begin{figure}[t] \centering
	\begin{tikzpicture}[thick, line cap = round]
		\def\l{1}
		\def\r{1.5pt}
		\def\h{0.9}
		\def\d{4.8}
		\def\dd{9.6}
		\def\hh{1.4}
		\def\t{0.03}
		
		\draw[line width = 0.6pt] (\t, \hh) -- node[xshift = -0.35cm] {\color{spec} \footnotesize $\lambda$} (\t, 0);
		\draw[line width = 0.6pt] (-\t, \hh) --  (-\t, 0);
		
		\draw (0, 0) -- (\l, -\h) node[xshift = 0.2cm,yshift = -0.3cm] {$y_1$} node[xshift = 0.1cm, yshift = 0.3cm] {\color{spec} \footnotesize $\lambda$};
		\draw (\l, -\h) -- (0, -2*\h);
		\draw (0, 0) -- (-\l, -\h) node[xshift = -0.2cm, yshift = -0.3cm] {$x_1$} node[xshift = -0.2cm, yshift = 0.3cm] {\color{spec} \footnotesize $\lambda$};
		\draw (-\l, -\h) -- (0, -2*\h);
		\draw (0, -2*\h) -- (\l, -3*\h) node[xshift = 0.2cm,yshift = -0.3cm] {$y_2$} node[xshift = 0.1cm, yshift = 0.3cm] {\color{spec} \footnotesize $\lambda$};
		\draw (0, -2*\h) -- (-\l, -3*\h) node[xshift = -0.2cm, yshift = -0.3cm] {$x_2$} node[xshift = -0.2cm, yshift = 0.3cm] {\color{spec} \footnotesize $\lambda$};
		\draw (-\l, -3*\h) -- (0, -4*\h);
		\draw (\l, -3*\h) -- (0, -4*\h);
		\draw (0, -4*\h) -- (-\l, -5*\h) node[xshift = -0.2cm, yshift = -0.3cm] {$x_3$} node[xshift = -0.2cm, yshift = 0.3cm] {\color{spec} \footnotesize $\lambda$};
		
		\draw[fill = black] (0, 0) circle (\r) node[xshift = 0.5cm, yshift = 0.1cm] {\color{spec} \footnotesize $- 2 \lambda$};
		\draw[fill = black] (0, -2*\h) node[xshift = 0.55cm] {\color{spec} \footnotesize $-2\lambda$} circle (\r);
		\draw[fill = black] (0, -4*\h) node[xshift = 0.5cm, yshift = -0.1cm] {\color{spec} \footnotesize $- 2 \lambda$} circle (\r);
		
		\node (L3) at (0, -6.5*\h) {$\LLambda_\lambda(x_1, x_2, x_3 | y_1, y_2)$};

		
		\draw[line width = 0.6pt] (\d + \t, \hh) -- node[xshift = -0.35cm] {\color{spec} \footnotesize $\lambda$} (\d + \t, 0);
		\draw[line width = 0.6pt] (\d-\t, \hh) --  (\d-\t, 0);
		
		\draw (\d, 0) -- (\d + \l, -\h) node[xshift = 0.2cm,yshift = -0.3cm] {$y_1$} 	node[xshift = 0.1cm, yshift = 0.3cm] {\color{spec} \footnotesize $\lambda$};
		\draw (\d + \l, -\h) -- (\d, -2*\h);
		\draw (\d, 0) -- (\d-\l, -\h) node[xshift = -0.2cm, yshift = -0.3cm] {$x_1$} 	node[xshift = -0.2cm, yshift = 0.3cm] {\color{spec} \footnotesize $\lambda$};
		\draw (\d-\l, -\h) -- (\d, -2*\h);
		\draw (\d, -2*\h) -- (\d-\l, -3*\h) node[xshift = -0.2cm, yshift = -0.3cm] {$x_2$} 	node[xshift = -0.2cm, yshift = 0.3cm] {\color{spec} \footnotesize $\lambda$};
		
		\draw[fill = black] (\d, 0) circle (\r) node[xshift = 0.5cm, yshift = 0.1cm] {\color{spec} \footnotesize $- 2 \lambda$};
		\draw[fill = black] (\d, -2*\h) node[xshift = 0.5cm, yshift = -0.1cm] {\color{spec} \footnotesize $-2\lambda$} circle (\r);
		
		\node (L2) at (\d, -6.5*\h) {$\LLambda_\lambda(x_1, x_2 | y_1)$};

		
		\draw[line width = 0.6pt] (\dd + \t, \hh) -- node[xshift = -0.35cm] {\color{spec} \footnotesize $\lambda$} (\dd + \t, 0);
		\draw[line width = 0.6pt] (\dd-\t, \hh) --  (\dd-\t, 0);
		
		\draw (\dd, 0) -- (\dd-\l, -\h) node[xshift = -0.2cm, yshift = -0.3cm] {$x_1$} 	node[xshift = -0.2cm, yshift = 0.3cm] {\color{spec} \footnotesize $\lambda$};
		
		\draw[fill = black] (\dd, 0) circle (\r) node[xshift = 0.5cm, yshift = 0.1cm] {\color{spec} \footnotesize $- 2 \lambda$};
		
		\node (L1) at (\dd, -6.5*\h) {$\LLambda_\lambda(x_1 | \varnothing)$};
	\end{tikzpicture}
	\vspace{0.1cm}
	\caption{Kernels of raising operators} \label{fig:Lker} \vspace{0.1cm}
\end{figure}

Recall the space of continuous polynomially bounded functions
\begin{align}\label{Pn}
	\mathcal{P}_n = \bigl\{ \phi \in C(\mathbb{R}^n) \colon \quad | \phi(\bm{x}_n) | \leq P(|x_1|, \dots, |x_n|), \quad P \text{ --- polynomial} \bigr\}.
\end{align}
By Proposition~\ref{prop:LQQr-space} the raising operator respects this space
\begin{align}
	\LLambda_n(\lambda) \colon \; \mathcal{P}_{n - 1} \; \to \; \mathcal{P}_n, \qquad \lambda \in \mathbb{R}.
\end{align}
In this section we prove the reflection symmetry.

\begin{proposition} \label{prop:Lrefl}
	For $\lambda \in \mathbb{R}$ the relation
	\begin{align}\label{L-refl}
		\LLambda_n(\lambda) = \LLambda_n(-\lambda)
	\end{align}
	holds on the space $\mathcal{P}_{n - 1}$.
\end{proposition}

\begin{proof}
	Since both sides are well defined on $\mathcal{P}_{n - 1}$, it is sufficient to establish the equality of kernels
	\begin{align}
		\LLambda_\lambda(\bm{x}_n | \bm{y}_{n - 1}) = \LLambda_{-\lambda}(\bm{x}_n | \bm{y}_{n - 1}).
	\end{align}
	This can be done using three types of diagram transformations, as shown in Figure~\ref{fig:Lsym}. 
	
	The first step is to use \textit{reduced} cross relation from Figure~\ref{fig:cross}. As a result, the dashed line appears at the bottom, and spectral parameters below it change signs. After that we apply \textit{full} cross relation from Figure~\ref{fig:cross} in order to move this dashed line to the top rhombus. Again, spectral parameters on its way change signs.
	
	Finally, we use identity depicted in Figure~\ref{fig:flip-lim-tr2}. This removes dashed line from the picture, while changing the signs of spectral parameters above it. The resulting diagram represents kernel of the operator from the right hand side of identity~\eqref{L-refl}.
	
	Although Figure~\ref{fig:Lsym} corresponds to the case $n = 3$, generalization to arbitrary $n$ is simple: with many rhombuses in the middle one just uses cross relation many times to move dashed line from bottom to top.
\end{proof}

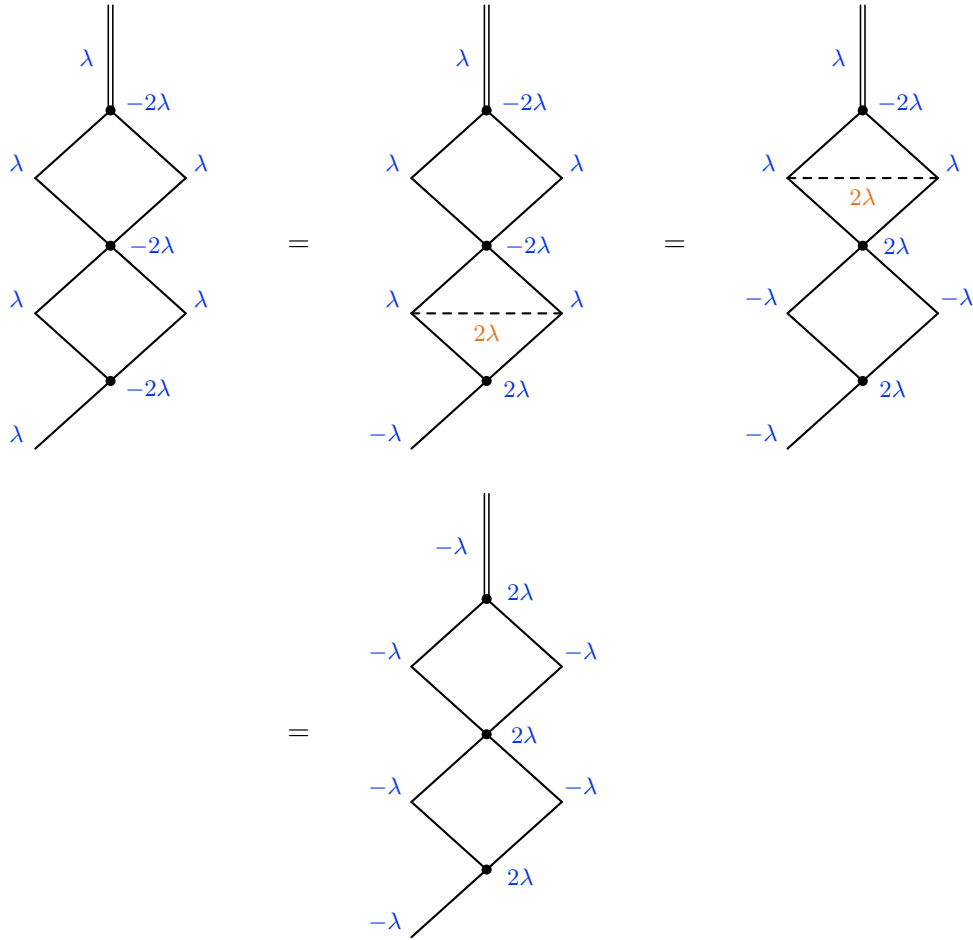
\begin{figure}[h] \centering
	\begin{tikzpicture}[thick, line cap = round]
		\def\l{1}
		\def\r{1.5pt}
		\def\R{13pt}
		\def\h{0.9}
		\def\d{5}
		\def\dd{10}
		\def\ddd{5}
		\def\hh{1.4}
		\def\t{0.03}
		\def\a{0}
		\def\b{-6.5}
		
		
		\draw[line width = 0.6pt] (\t, \hh) -- node[xshift = -0.35cm] {\color{spec} \footnotesize $\lambda$} (\t, 0);
		\draw[line width = 0.6pt] (-\t, \hh) --  (-\t, 0);
		
		\draw (0, 0) -- (\l, -\h) node[xshift = 0.2cm, yshift = 0.2cm] {\color{spec} \footnotesize $\lambda$};
		\draw (\l, -\h) -- (0, -2*\h);
		\draw (0, 0) -- (-\l, -\h) node[xshift = -0.25cm, yshift = 0.2cm] {\color{spec} \footnotesize $\lambda$};
		\draw (-\l, -\h) -- (0, -2*\h);
		\draw (0, -2*\h) -- (\l, -3*\h) node[xshift = 0.2cm, yshift = 0.2cm] {\color{spec} \footnotesize $\lambda$};
		\draw (0, -2*\h) -- (-\l, -3*\h) node[xshift = -0.25cm, yshift = 0.2cm] {\color{spec} \footnotesize $\lambda$};
		\draw (-\l, -3*\h) -- (0, -4*\h);
		\draw (\l, -3*\h) -- (0, -4*\h);
		\draw (0, -4*\h) -- (-\l, -5*\h) node[xshift = -0.25cm, yshift = 0.2cm] {\color{spec} \footnotesize $\lambda$};
		
		\draw[fill = black] (0, 0) circle (\r) node[xshift = 0.5cm, yshift = 0.1cm] {\color{spec} \footnotesize $- 2 \lambda$};
		\draw[fill = black] (0, -2*\h) node[xshift = 0.55cm] {\color{spec} \footnotesize $-2\lambda$} circle (\r) node[xshift = 2.5cm] {$=$};
		\draw[fill = black] (0, -4*\h) node[xshift = 0.5cm, yshift = -0.1cm] {\color{spec} \footnotesize $- 2 \lambda$} circle (\r);
		
		
		\draw[line width = 0.6pt] (\d + \t, \a + \hh) -- node[xshift = -0.35cm] {\color{spec} \footnotesize $\lambda$} (\d + \t, \a);
		\draw[line width = 0.6pt] (\d - \t, \a + \hh) --  (\d - \t, \a);
		
		\draw (\d, \a) -- (\d + \l, \a - \h) node[xshift = 0.2cm, yshift = 0.2cm] {\color{spec} \footnotesize $\lambda$};
		\draw (\d + \l, \a - \h) -- (\d, \a - 2*\h);
		\draw (\d, \a) -- (\d - \l, \a - \h) node[xshift = -0.25cm, yshift = 0.2cm] {\color{spec} \footnotesize $\lambda$};
		\draw (\d - \l, \a - \h) -- (\d, \a - 2*\h);
		\draw (\d, \a - 2*\h) -- (\d + \l, \a - 3*\h) node[xshift = 0.2cm, yshift = 0.2cm] {\color{spec} \footnotesize $\lambda$};
		\draw (\d, \a - 2*\h) -- (\d - \l, \a - 3*\h) node[xshift = -0.25cm, yshift = 0.2cm] {\color{spec} \footnotesize $\lambda$};
		\draw (\d - \l, \a - 3*\h) -- (\d, \a - 4*\h);
		\draw (\d + \l, \a - 3*\h) -- (\d, \a - 4*\h);
		\draw (\d, \a - 4*\h) -- (\d - \l, \a - 5*\h) node[xshift = -0.35cm, yshift = 0.2cm] {\color{spec} \footnotesize $-\lambda$};
		\draw[dashed] (\d - \l, \a - 3*\h) -- node[below] {\color{spec2} \footnotesize $2\lambda$}  (\d + \l, \a - 3*\h);
		
		\draw[fill = black] (\d, \a) circle (\r) node[xshift = 0.5cm, yshift = 0.1cm] {\color{spec} \footnotesize $- 2 \lambda$};
		\draw[fill = black] (\d, \a - 2*\h) node[xshift = 0.55cm] {\color{spec} \footnotesize $-2\lambda$} circle (\r) node[xshift = 2.5cm] {$=$};
		\draw[fill = black] (\d, \a - 4*\h) node[xshift = 0.4cm, yshift = -0.1cm] {\color{spec} \footnotesize $2 \lambda$} circle (\r);
		
		
		\draw[line width = 0.6pt] (\dd + \t, \a + \hh) -- node[xshift = -0.35cm] {\color{spec} \footnotesize $\lambda$} (\dd + \t, \a);
		\draw[line width = 0.6pt] (\dd - \t, \a + \hh) --  (\dd - \t, \a);
		
		
		\draw (\dd, \a) -- (\dd + \l, \a - \h) node[xshift = 0.2cm, yshift = 0.2cm] {\color{spec} \footnotesize $\lambda$};
		\draw (\dd + \l, \a - \h) -- (\dd, \a - 2*\h);
		\draw (\dd, \a) -- (\dd - \l, \a - \h) node[xshift = -0.25cm, yshift = 0.2cm] {\color{spec} \footnotesize $\lambda$};
		\draw (\dd - \l, \a - \h) -- (\dd, \a - 2*\h);
		\draw (\dd, \a - 2*\h) -- (\dd + \l, \a - 3*\h) node[xshift = 0.25cm, yshift = 0.2cm] {\color{spec} \footnotesize $-\lambda$};
		\draw (\dd, \a - 2*\h) -- (\dd - \l, \a - 3*\h) node[xshift = -0.35cm, yshift = 0.2cm] {\color{spec} \footnotesize $-\lambda$};
		\draw (\dd - \l, \a - 3*\h) -- (\dd, \a - 4*\h);
		\draw (\dd + \l, \a - 3*\h) -- (\dd, \a - 4*\h);
		\draw (\dd, \a - 4*\h) -- (\dd - \l, \a - 5*\h) node[xshift = -0.35cm, yshift = 0.2cm] {\color{spec} \footnotesize $-\lambda$};
		\draw[dashed] (\dd - \l, \a - \h) -- node[below] {\color{spec2} \footnotesize $2\lambda$}  (\dd + \l, \a - \h);
		
		\draw[fill = black] (\dd, \a) circle (\r) node[xshift = 0.5cm, yshift = 0.1cm] {\color{spec} \footnotesize $- 2 \lambda$};
		\draw[fill = black] (\dd, \a - 2*\h) circle (\r) node[xshift = 0.45cm] {\color{spec} \footnotesize $2\lambda$};
		\draw[fill = black] (\dd, \a - 4*\h) node[xshift = 0.4cm, yshift = -0.1cm] {\color{spec} \footnotesize $2 \lambda$} circle (\r);
		
		
		\draw[line width = 0.6pt] (\ddd+ \t, \b + \hh) -- node[xshift = -0.5cm] {\color{spec} \footnotesize $-\lambda$} (\ddd+ \t, \b);
		\draw[line width = 0.6pt] (\ddd- \t, \b + \hh) --  (\ddd- \t, \b);
		
		
		\draw (\ddd, \b ) -- (\ddd+ \l, \b - \h) node[xshift = 0.25cm, yshift = 0.2cm] {\color{spec} \footnotesize $-\lambda$};
		\draw (\ddd+ \l, \b - \h) -- (\ddd, \b - 2*\h);
		\draw (\ddd, \b) -- (\ddd- \l, \b - \h) node[xshift = -0.35cm, yshift = 0.2cm] {\color{spec} \footnotesize $-\lambda$};
		\draw (\ddd- \l, \b - \h) -- (\ddd, \b - 2*\h);
		\draw (\ddd, \b - 2*\h) -- (\ddd+ \l, \b - 3*\h) node[xshift = 0.25cm, yshift = 0.2cm] {\color{spec} \footnotesize $-\lambda$};
		\draw (\ddd, \b - 2*\h) -- (\ddd- \l, \b - 3*\h) node[xshift = -0.35cm, yshift = 0.2cm] {\color{spec} \footnotesize $-\lambda$};
		\draw (\ddd- \l, \b - 3*\h) -- (\ddd, \b - 4*\h);
		\draw (\ddd+ \l, \b - 3*\h) -- (\ddd, \b - 4*\h);
		\draw (\ddd, \b - 4*\h) -- (\ddd- \l, \b -5*\h) node[xshift = -0.35cm, yshift = 0.2cm] {\color{spec} \footnotesize $-\lambda$};
		
		\draw[fill = black] (\ddd, \b) circle (\r) node[xshift = 0.45cm, yshift = 0.1cm] {\color{spec} \footnotesize $2 \lambda$};
		\draw[fill = black] (\ddd, \b - 2*\h) node[xshift = 0.5cm] {\color{spec} \footnotesize $2\lambda$} node[xshift = -2.5cm] {$= $} circle (\r);
		\draw[fill = black] (\ddd, \b - 4*\h) node[xshift = 0.45cm, yshift = -0.1cm] {\color{spec} \footnotesize $2 \lambda$} circle (\r);
	\end{tikzpicture}
	\vspace{0.2cm}
	\caption{Reflection symmetry of raising operator} \label{fig:Lsym}
\end{figure}

\subsection{Baxter operator and its reduction} \label{sec:baxt-op}

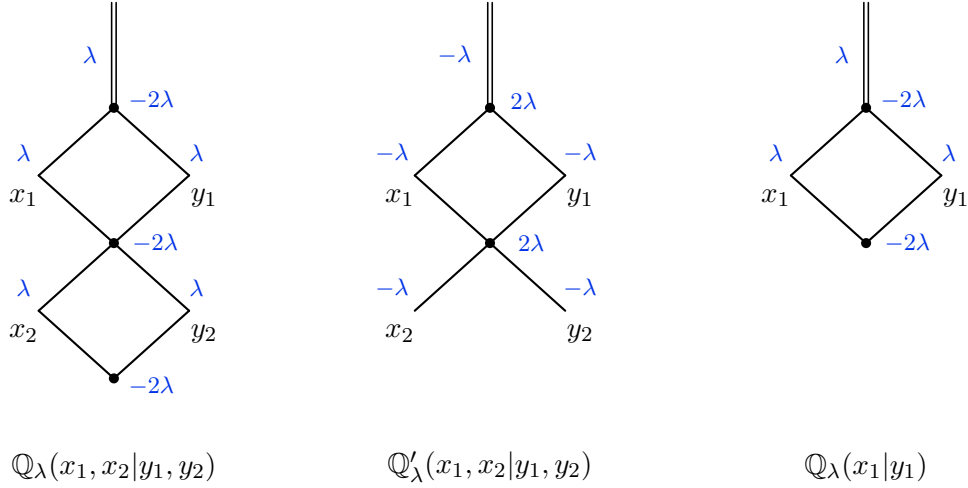
\begin{figure}[t] \centering
	\begin{tikzpicture}[thick, line cap = round]
		\def\l{1}
		\def\r{1.5pt}
		\def\h{0.9}
		\def\d{5}
		\def\a{-5}
		\def\dd{10}
		\def\hh{1.4}
		\def\t{0.03}
		\def\s{5.3}

		\draw[line width = 0.6pt] (\t, \hh) -- node[xshift = -0.35cm] {\color{spec} \footnotesize $\lambda$} (\t, 0);
		\draw[line width = 0.6pt] (-\t, \hh) --  (-\t, 0);
		
		\draw (0, 0) -- (\l, -\h) node[xshift = 0.2cm,yshift = -0.3cm] {$y_1$} node[xshift = 0.1cm, yshift = 0.3cm] {\color{spec} \footnotesize $\lambda$};
		\draw (\l, -\h) -- (0, -2*\h);
		\draw (0, 0) -- (-\l, -\h) node[xshift = -0.2cm, yshift = -0.3cm] {$x_1$} node[xshift = -0.2cm, yshift = 0.3cm] {\color{spec} \footnotesize $\lambda$};
		\draw (-\l, -\h) -- (0, -2*\h);
		\draw (0, -2*\h) -- (\l, -3*\h) node[xshift = 0.2cm,yshift = -0.3cm] {$y_2$} node[xshift = 0.1cm, yshift = 0.3cm] {\color{spec} \footnotesize $\lambda$};
		\draw (0, -2*\h) -- (-\l, -3*\h) node[xshift = -0.2cm, yshift = -0.3cm] {$x_2$} node[xshift = -0.2cm, yshift = 0.3cm] {\color{spec} \footnotesize $\lambda$};
		\draw (-\l, -3*\h) -- (0, -4*\h);
		\draw (\l, -3*\h) -- (0, -4*\h);
		
		\draw[fill = black] (0, 0) circle (\r) node[xshift = 0.5cm, yshift = 0.1cm] {\color{spec} \footnotesize $- 2 \lambda$};
		\draw[fill = black] (0, -2*\h) node[xshift = 0.55cm] {\color{spec} \footnotesize $-2\lambda$} circle (\r);
		\draw[fill = black] (0, -4*\h) node[xshift = 0.5cm, yshift = -0.1cm] {\color{spec} \footnotesize $- 2 \lambda$} circle (\r);
		
		\node (Q2) at (0, -\s*\h) {$\QQ_\lambda(x_1, x_2 | y_1, y_2)$};

		\draw[line width = 0.6pt] (\d + \t, \hh) -- node[xshift = -0.5cm] {\color{spec} \footnotesize $-\lambda$} (\d + \t, 0);
		\draw[line width = 0.6pt] (\d-\t, \hh) --  (\d-\t, 0);
		
		\draw (\d, 0) -- (\d + \l, -\h) node[xshift = 0.2cm,yshift = -0.3cm] {$y_1$} node[xshift = 0.2cm, yshift = 0.3cm] {\color{spec} \footnotesize $-\lambda$};
		\draw (\d + \l, -\h) -- (\d, -2*\h);
		\draw (\d, 0) -- (\d-\l, -\h) node[xshift = -0.2cm, yshift = -0.3cm] {$x_1$} node[xshift = -0.3cm, yshift = 0.3cm] {\color{spec} \footnotesize $-\lambda$};
		\draw (\d-\l, -\h) -- (\d, -2*\h);
		\draw (\d, -2*\h) -- (\d + \l, -3*\h) node[xshift = 0.2cm,yshift = -0.3cm] {$y_2$} node[xshift = 0.2cm, yshift = 0.3cm] {\color{spec} \footnotesize $-\lambda$};
		\draw (\d, -2*\h) -- (\d-\l, -3*\h) node[xshift = -0.2cm, yshift = -0.3cm] {$x_2$} node[xshift = -0.3cm, yshift = 0.3cm] {\color{spec} \footnotesize $-\lambda$};
		
		\draw[fill = black] (\d, 0) circle (\r) node[xshift = 0.45cm, yshift = 0.1cm] {\color{spec} \footnotesize $2 \lambda$};
		\draw[fill = black] (\d, -2*\h) node[xshift = 0.55cm] {\color{spec} \footnotesize $2\lambda$} circle (\r);
		
		\node (Qr2) at (\d, -\s*\h) {$\QQr_\lambda(x_1, x_2 | y_1, y_2)$};

		\draw[line width = 0.6pt] (\dd + \t, \hh) -- node[xshift = -0.35cm] {\color{spec} \footnotesize $\lambda$} (\dd + \t, 0);
		\draw[line width = 0.6pt] (\dd -\t, \hh) --  (\dd - \t, 0);
		
		\draw (\dd, 0) -- (\dd + \l, -\h) node[xshift = 0.2cm,yshift = -0.3cm] {$y_1$} node[xshift = 0.1cm, yshift = 0.3cm] {\color{spec} \footnotesize $\lambda$};
		\draw (\dd + \l, -\h) -- (\dd, -2*\h);
		\draw (\dd, 0) -- (\dd -\l, -\h) node[xshift = -0.2cm, yshift = -0.3cm] {$x_1$} node[xshift = -0.2cm, yshift = 0.3cm] {\color{spec} \footnotesize $\lambda$};
		\draw (\dd -\l, -\h) -- (\dd, -2*\h);
		
		\draw[fill = black] (\dd, 0) circle (\r) node[xshift = 0.5cm, yshift = 0.1cm] {\color{spec} \footnotesize $- 2 \lambda$};
		\draw[fill = black] (\dd, -2*\h) node[xshift = 0.55cm] {\color{spec} \footnotesize $-2\lambda$} circle (\r);
		
		\node (Q1) at (\dd, -\s*\h) {$\QQ_\lambda(x_1 | y_1)$};
		
	\end{tikzpicture}
	\vspace{0.3cm}
	\caption{Kernels of Baxter operators} \label{fig:Qker}
\end{figure}

Baxter operator $\QQ_n(\lambda)$ is an integral operator acting on functions $\phi(\bm{x}_n)$ 
\begin{align}\label{GG13}
	\bigl[ \QQ_n(\lambda) \, \phi \bigr] (\bm{x}_n) = \int_{\mathbb{R}^{n}} d\bm{y}_{n} \; \QQ_\lambda(\bm{x}_n | \bm{y}_{n}) \, \phi(\bm{y}_{n}),
\end{align}
where the kernel
\begin{align}
	\begin{aligned}
		\QQ_\lambda(\bm{x}_n | \bm{y}_{n}) & = \frac{(2\beta)^{\imath \lambda}}{\Gamma(g - \imath \lambda)} \int_{\mathbb{R}^{n + 1}} d\bm{z}_{n + 1}
		\; (1 + \beta e^{-z_1})^{- \imath \lambda - g} \; (1 - \beta e^{-z_1})^{- \imath \lambda + g - 1}\,	\theta(z_1 - \ln \beta)
		\\[6pt]
		&  \times \exp \biggl( \imath \lambda \bigl( \underline{\bm{x}}_n + \underline{\bm{y}}_{n} - 2 \underline{\bm{z}}_{n + 1} \bigr) - \sum_{j = 1}^{n} (e^{z_j - x_j} + e^{z_j - y_j} + e^{x_j - z_{j + 1}} + e^{y_j - z_{j + 1}} ) \biggr).
	\end{aligned}
\end{align}
Let us define the closely related \textit{reduced} Baxter operator $\QQr_n(\lambda)$ as an integral operator 
\begin{align}\label{GG14}
	\bigl[ \QQr_n(\lambda) \, \phi \bigr] (\bm{x}_n) = \int_{\mathbb{R}^{n}} d\bm{y}_{n} \; \QQr_\lambda(\bm{x}_n | \bm{y}_{n}) \, \phi(\bm{y}_{n}),
\end{align}
whose kernel is given by a slightly smaller integral
\begin{align}  \label{Qr-ker}
		& \QQr_\lambda(\bm{x}_n | \bm{y}_{n}) = \frac{(2\beta)^{-\imath \lambda}}{\Gamma(g + \imath \lambda)}\int_{\mathbb{R}^{n}}  d\bm{z}_{n} \;
		\; \bigl(1 + \beta e^{-z_1}\bigr)^{\imath \lambda - g } \, \bigl(1 - \beta e^{-z_1}\bigr)^{\imath \lambda + g - 1 } \, \theta(z_1 - \ln \beta) \\[6pt] \nonumber
		& \times \exp \biggl( \imath \lambda \bigl( 2 \underline{\bm{z}}_{n}  - \underline{\bm{x}}_n - \underline{\bm{y}}_{n} \bigr) - \sum_{j = 1}^{n - 1} (e^{z_j - x_j} + e^{z_j - y_j} + e^{x_j - z_{j + 1}} + e^{y_j - z_{j + 1}} ) - e^{z_n - x_n} - e^{z_n - y_n}\biggr).
\end{align}
The examples of these kernels are depicted in Figure~\ref{fig:Qker}. Notice that spectral parameters in the kernels of these operators have opposite signs.

By Proposition~\ref{prop:LQQr-space} the above Baxter operators act invariantly on the space of continuous polynomially bounded functions~\eqref{Pn}
\begin{align}
	& \QQ_n(\lambda) \colon \; \mathcal{P}_{n} \; \to \; \mathcal{P}_n, && \hspace{-2cm} \Im \lambda \in (-g, 0), \\[6pt]
	& \QQr_n(\lambda) \colon \; \mathcal{P}_{n} \; \to \; \mathcal{P}_n, && \hspace{-2cm} \Im \lambda < 0.
\end{align}
With the above assumptions on parameters their kernels are absolutely convergent. 

\begin{proposition} \label{prop:Qred}
	For $\Im \lambda \in(-g, 0)$ the relation
	\begin{align}\label{QQrel}
		\QQ_n(\lambda) = \Gamma(2\imath \lambda) \; \QQr_n(\lambda)
	\end{align}
	holds on $\mathcal{P}_n$.
\end{proposition}

\begin{proof}
	Since both sides are well defined on $\mathcal{P}_n$, it is sufficient to show the equality of kernels, which can be done with the help of diagrams, see Figure~\ref{fig:Qred}. On the first step one applies chain relation from Figure~\ref{fig:chain} to obtain dashed line at the bottom, while the rest of transformations coincide with the ones in Figure~\ref{fig:Lsym}.
\end{proof}

In the case $n = 0$, in accordance with~\eqref{QQrel} and the above definitions, we denote 
\begin{align} \label{QQr0}
	\QQr_0(\lambda) = \frac{(2\beta)^{-\imath \lambda}}{\Gamma(g + \imath \lambda)} \, \mathrm{Id}, \qquad \QQ_0(\lambda) = \frac{(2\beta)^{-\imath \lambda} \, \Gamma(2\imath \lambda)}{\Gamma(g + \imath \lambda)} \,  \mathrm{Id}.
\end{align}

\begin{figure}[H] \centering \vspace{0.5cm}
	\begin{tikzpicture}[thick, line cap = round]
		\def\l{1}
		\def\r{1.5pt}
		\def\h{0.9}
		\def\d{6}
		\def\dd{12}
		\def\ddd{6}
		\def\hh{1.4}
		\def\t{0.03}
		\def\a{0}
		\def\b{-5.4}
		
		\draw[line width = 0.6pt] (\t, \hh) -- node[xshift = -0.35cm] {\color{spec} \footnotesize $\lambda$} (\t, 0);
		\draw[line width = 0.6pt] (-\t, \hh) --  (-\t, 0);
		
		\draw (0, 0) -- (\l, -\h) node[xshift = 0.2cm, yshift = 0.2cm] {\color{spec} \footnotesize $\lambda$};
		\draw (\l, -\h) -- (0, -2*\h);
		\draw (0, 0) -- (-\l, -\h) node[xshift = -0.25cm, yshift = 0.2cm] {\color{spec} \footnotesize $\lambda$};
		\draw (-\l, -\h) -- (0, -2*\h);
		\draw (0, -2*\h) -- (\l, -3*\h) node[xshift = 0.2cm, yshift = 0.2cm] {\color{spec} \footnotesize $\lambda$};
		\draw (0, -2*\h) -- (-\l, -3*\h) node[xshift = -0.25cm, yshift = 0.2cm] {\color{spec} \footnotesize $\lambda$};
		\draw (-\l, -3*\h) -- (0, -4*\h);
		\draw (\l, -3*\h) -- (0, -4*\h);
		
		\draw[fill = black] (0, 0) circle (\r) node[xshift = 0.5cm, yshift = 0.1cm] {\color{spec} \footnotesize $- 2 \lambda$};
		\draw[fill = black] (0, -2*\h) node[xshift = 0.55cm] {\color{spec} \footnotesize $-2\lambda$} circle (\r) node[xshift = 3.2cm, yshift = 0.5cm] {$=$ \small $\quad \Gamma(2\imath \lambda)$};
		\draw[fill = black] (0, -4*\h) node[xshift = 0.5cm, yshift = -0.1cm] {\color{spec} \footnotesize $- 2 \lambda$} circle (\r);

		\draw[line width = 0.6pt] (\d + \t, \a + \hh) -- node[xshift = -0.35cm] {\color{spec} \footnotesize $\lambda$} (\d + \t, \a);
		\draw[line width = 0.6pt] (\d - \t, \a + \hh) --  (\d - \t, \a);
		
		\draw (\d, \a) -- (\d + \l, \a - \h) node[xshift = 0.2cm, yshift = 0.2cm] {\color{spec} \footnotesize $\lambda$};
		\draw (\d + \l, \a - \h) -- (\d, \a - 2*\h);
		\draw (\d, \a) -- (\d - \l, \a - \h) node[xshift = -0.25cm, yshift = 0.2cm] {\color{spec} \footnotesize $\lambda$};
		\draw (\d - \l, \a - \h) -- (\d, \a - 2*\h);
		\draw (\d, \a - 2*\h) -- (\d + \l, \a - 3*\h) node[xshift = 0.2cm, yshift = 0.2cm] {\color{spec} \footnotesize $\lambda$};
		\draw (\d, \a - 2*\h) -- (\d - \l, \a - 3*\h) node[xshift = -0.25cm, yshift = 0.2cm] {\color{spec} \footnotesize $\lambda$};
		\draw[dashed] (\d - \l, \a - 3*\h) -- node[below] {\color{spec2} \footnotesize $2\lambda$}  (\d + \l, \a - 3*\h);
		
		\draw[fill = black] (\d, \a) circle (\r) node[xshift = 0.5cm, yshift = 0.1cm] {\color{spec} \footnotesize $- 2 \lambda$};
		\draw[fill = black] (\d, \a - 2*\h) node[xshift = 0.55cm] {\color{spec} \footnotesize $-2\lambda$} circle (\r) node[xshift = 3.2cm, yshift = 0.5cm] {$=$ \small $\quad \Gamma(2\imath \lambda)$};

		\draw[line width = 0.6pt] (\dd + \t, \a + \hh) -- node[xshift = -0.35cm] {\color{spec} \footnotesize $\lambda$} (\dd + \t, \a);
		\draw[line width = 0.6pt] (\dd - \t, \a + \hh) --  (\dd - \t, \a);
		
		\draw (\dd, \a) -- (\dd + \l, \a - \h) node[xshift = 0.2cm, yshift = 0.2cm] {\color{spec} \footnotesize $\lambda$};
		\draw (\dd + \l, \a - \h) -- (\dd, \a - 2*\h);
		\draw (\dd, \a) -- (\dd - \l, \a - \h) node[xshift = -0.25cm, yshift = 0.2cm] {\color{spec} \footnotesize $\lambda$};
		\draw (\dd - \l, \a - \h) -- (\dd, \a - 2*\h);
		\draw (\dd, \a - 2*\h) -- (\dd + \l, \a - 3*\h) node[xshift = 0.25cm, yshift = 0.2cm] {\color{spec} \footnotesize $-\lambda$};
		\draw (\dd, \a - 2*\h) -- (\dd - \l, \a - 3*\h) node[xshift = -0.35cm, yshift = 0.2cm] {\color{spec} \footnotesize $-\lambda$};
		\draw[dashed] (\dd - \l, \a - \h) -- node[below] {\color{spec2} \footnotesize $2\lambda$}  (\dd + \l, \a - \h);
		
		\draw[fill = black] (\dd, \a) circle (\r) node[xshift = 0.5cm, yshift = 0.1cm] {\color{spec} \footnotesize $- 2 \lambda$};
		\draw[fill = black] (\dd, \a - 2*\h) node[xshift = 0.45cm] {\color{spec} \footnotesize $2\lambda$} circle (\r);

		\draw[line width = 0.6pt] (\ddd+ \t, \b + \hh) -- node[xshift = -0.5cm] {\color{spec} \footnotesize $-\lambda$} (\ddd+ \t, \b);
		\draw[line width = 0.6pt] (\ddd- \t, \b + \hh) --  (\ddd- \t, \b);
		
		\draw (\ddd, \b ) -- (\ddd+ \l, \b - \h) node[xshift = 0.25cm, yshift = 0.2cm] {\color{spec} \footnotesize $-\lambda$};
		\draw (\ddd+ \l, \b - \h) -- (\ddd, \b - 2*\h);
		\draw (\ddd, \b) -- (\ddd- \l, \b - \h) node[xshift = -0.35cm, yshift = 0.2cm] {\color{spec} \footnotesize $-\lambda$};
		\draw (\ddd- \l, \b - \h) -- (\ddd, \b - 2*\h);
		\draw (\ddd, \b - 2*\h) -- (\ddd+ \l, \b - 3*\h) node[xshift = 0.25cm, yshift = 0.2cm] {\color{spec} \footnotesize $-\lambda$};
		\draw (\ddd, \b - 2*\h) -- (\ddd- \l, \b - 3*\h) node[xshift = -0.35cm, yshift = 0.2cm] {\color{spec} \footnotesize $-\lambda$};
		
		\draw[fill = black] (\ddd, \b) circle (\r) node[xshift = 0.45cm, yshift = 0.1cm] {\color{spec} \footnotesize $2 \lambda$};
		\draw[fill = black] (\ddd, \b - 2*\h) node[xshift = 0.5cm] {\color{spec} \footnotesize $2\lambda$} node[xshift = -2.8cm, yshift = 0.5cm] {$=$ \small $\quad \Gamma(2\imath \lambda)$} circle (\r);
	\end{tikzpicture}
	\vspace{0.2cm}
	\caption{Reduction of Baxter operator} \label{fig:Qred}
\end{figure}
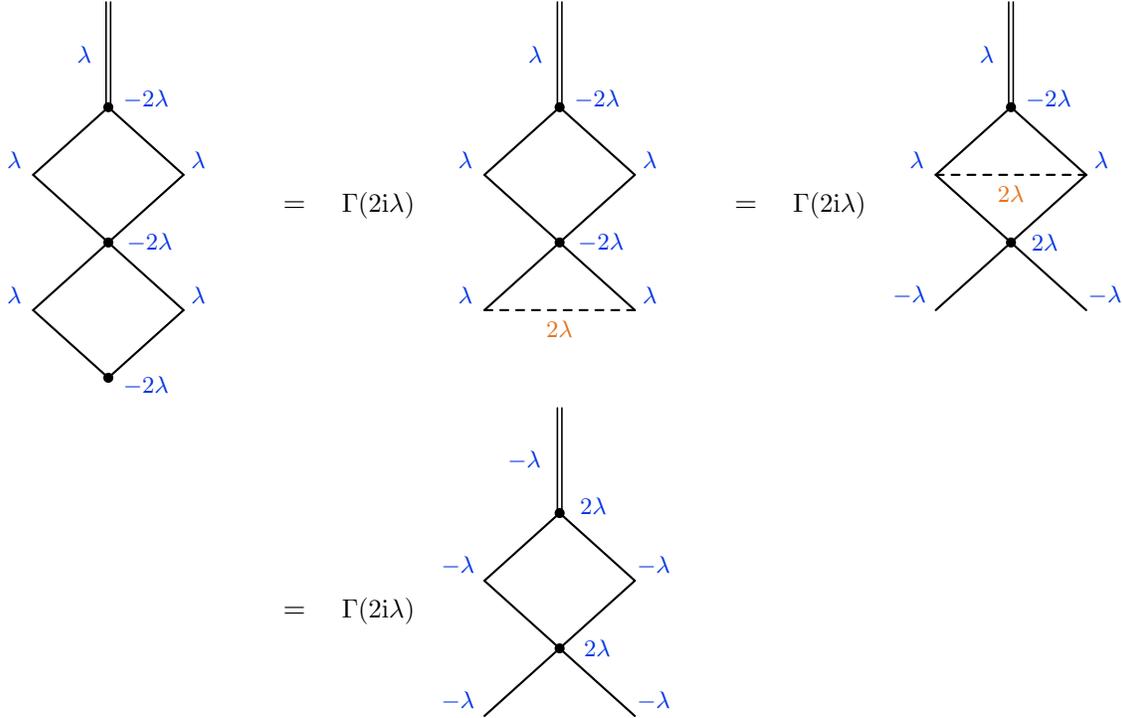

\subsection{Commutativity and exchange relations} \label{sec:op-rel}

\begin{proposition} \label{prop:QL-rel}
	The relations
	\begin{align}
		& \QQ_n(\lambda)\, \QQ_n(\rho)= \QQ_n(\rho)\, \QQ_n(\lambda), &&  \Im \lambda, \, \Im \rho \in(-g, 0), \\[6pt] \label{QQrcomm}
		& \QQr_n(\lambda)\, \QQr_n(\rho)= \QQr_n(\rho)\, \QQr_n(\lambda), &&  \Im \lambda, \, \Im \rho < 0, \\[6pt]
		& \QQ_n(\lambda) \, \LLambda_n(\rho) = \Gamma(\imath \lambda \pm \imath \rho) \, \LLambda_n(\rho) \, \QQ_{n - 1}(\lambda), && \Im \lambda \in (-g, 0), \; \rho \in \mathbb{R}, \\[6pt] \label{QrL}
		& \QQr_n(\lambda) \, \LLambda_n(\rho) =  \Gamma(\imath \lambda \pm \imath \rho) \, \LLambda_n(\rho) \, \QQr_{n - 1}(\lambda), && \Im \lambda < 0, \; \rho \in \mathbb{R}, \\[6pt] \label{LL}
		& \LLambda_{n + 1}(\lambda) \, \LLambda_n(\rho) = \LLambda_{n + 1}(\rho) \, \LLambda_n(\lambda), && \lambda, \rho \in \mathbb{R},
	\end{align}
	hold on the spaces $\mathcal{P}_n$ and $\mathcal{P}_{n - 1}$ correspondingly.
\end{proposition}

\begin{proof}
By Proposition~\ref{prop:Qred}, $\QQ_n(\lambda) = \Gamma(2\imath \lambda) \, \QQr_n(\lambda)$, so it is enough to prove the relations with reduced Baxter operators~$\QQr_n(\lambda)$. 

\paragraph{Commutativity~\eqref{QQrcomm}.} In Appendix~\ref{sec:ker-op-eq} we show that it is sufficient to establish the equality of the corresponding kernels
\begin{align} \label{QrQr-prod}
	\int_{\mathbb{R}^n} d\bm{y}_n \; \QQr_\lambda(\bm{x}_n | \bm{y}_n) \, \QQr_\rho(\bm{y}_n | \bm{z}_n) = \int_{\mathbb{R}^n} d\bm{y}_n \; \QQr_\rho(\bm{x}_n | \bm{y}_n) \, \QQr_\lambda(\bm{y}_n | \bm{z}_n) .
\end{align}
This can be done using transformations of diagrams, as shown in Figure~\ref{fig:Qcomm}. Let us postpone the convergence questions to the end of the proof. As before, the pictures correspond to the case $n = 3$, however, generalization to arbitrary~$n$ is straightforward. 

The first diagram in Figure~\ref{fig:Qcomm} depicts kernel of the product $\QQr_3(\lambda) \, \QQr_3(\rho)$. Commutativity is equivalent to the symmetry of this kernel with respect to $\lambda, \rho$. 

The first step is to use chain relation (Figure~\ref{fig:chain}) for the bold vertex at the bottom to obtain the second diagram with dashed line. For brevity, we omit coefficient $\Gamma(\imath \lambda + \imath \rho)$ appearing in this transformation, since it is symmetric in $\lambda, \rho$.

Second, we apply cross relation (Figure~\ref{fig:cross}) two times in order to move dashed line to the top. Notice that some spectral parameters below change. 

Next, we use identity from Figure~\ref{fig:flip-tr} that interchanges parameters $\lambda, \rho$ in double lines. Due to this identity two more dashed lines appear. These new dashed lines then can be brought to the bottom of diagram with the help of cross relations (Figure~\ref{fig:cross}).

Finally, dashed lines at the bottom of the fifth diagram can be removed using reduced cross relations (Figure~\ref{fig:cross}). As a result we arrive at the \textit{sixth} diagram, which differs from the \textit{third} one only by interchange of $\lambda$ and $\rho$. Therefore, the initial integral is symmetric in~$\lambda, \rho$, which proves commutativity of Baxter operators.

It is left to argue that each diagram on the way corresponds to the absolutely convergent integral. Consider any diagram and denote the corresponding integral as $D_{\lambda, \rho}(\bm{x}_{n}| \bm{y}_{n})$. From definitions of the lines (Figures~\ref{fig:lines},~\ref{fig:dashed}) we have
\begin{align}
	\bigl|  D_{\lambda, \rho}(\bm{x}_{n}| \bm{y}_{n}) \bigr| \leq C(\lambda, \rho) \, D_{\imath \Im \lambda, \imath \Im \rho}(\bm{x}_{n}| \bm{y}_{n}).
\end{align}
Hence, it is enough to consider the case $\lambda, \rho \in \imath \mathbb{R}$. In Appendix~\ref{sec:ker-prod} we prove that the kernel of the Baxter operators product~\eqref{QrQr-prod}, which corresponds to the first diagram, is convergent. Since the next diagram differs only by one transformation, the corresponding integral is convergent at least in one order (such that the transformed integral is taken in the first place). By Fubini--Tonelli theorem it is therefore convergent in any order. The third diagram differs by two transformations, and for both of them we can inductively apply the same arguments. This induction can be continued until the last diagram, so that all diagrams represent convergent integrals.

\paragraph{Relation~\eqref{QrL}.} Again, as argued in Section~\ref{sec:ker-op-eq}, it is sufficient to establish the equality of the corresponding kernels
\begin{align}
	\int_{\mathbb{R}^n} d\bm{y}_n \; \QQr_\lambda(\bm{x}_n | \bm{y}_n) \, \LLambda_\rho(\bm{y}_n | \bm{z}_{n - 1}) = \Gamma(\imath \lambda \pm \imath \rho) \, \int_{\mathbb{R}^{n - 1}} d\bm{y}_{n - 1} \; \LLambda_\rho(\bm{x}_n | \bm{y}_{n - 1}) \, \QQr_\lambda(\bm{y}_{n - 1} | \bm{z}_{n - 1}) .
\end{align}
With the help of diagrams it can be done almost in the same way, as for commutativity of Baxter operators. The first diagram from Figure~\ref{fig:QrL} depicts the kernel of operator $\QQr_n(\lambda) \, \LLambda_n(-\rho)$. Here we reflect the parameter of raising operator $\rho$ to ease comparison with Figure~\ref{fig:Qcomm}. Namely, notice that the first diagrams of Figures~\ref{fig:Qcomm} and~\ref{fig:QrL} coincide up to one line at the bottom. 

To arrive at the second diagram of Figure~\ref{fig:QrL} one repeats the same four steps, as in Figure~\ref{fig:Qcomm}. The rest of transformations are slightly different. Passing to the third diagram we use reduced cross relation (Figure~\ref{fig:cross}) for the bottom vertex from the left and chain relation (Figure~\ref{fig:chain}) for the bottom vertex from the right. The resulting diagram contains only one dashed line, which we move to the bottom using cross relation (Figure~\ref{fig:cross}). Finally, to arrive at the last diagram, which depicts kernel of operator $\LLambda_n(-\rho) \, \QQr_{n - 1}(\lambda)$, we again use reduced cross relation (Figure~\ref{fig:cross}). 

Note that gamma functions, which enter the right hand side of identity~\eqref{QrL}, appear during the first two transitions of Figure~\ref{fig:QrL}. This finishes the proof of desired relation. By the same arguments, as before, all diagrams correspond to absolutely convergent integrals.

\paragraph{Relation~\eqref{LL}.} The proof is analogous to the previous ones.
\end{proof}

\begin{remark}
	Since the first diagrams of Figures~\ref{fig:Qcomm} and~\ref{fig:QrL} coincide up to one line at the bottom, the latter can be obtained as the limit of the former
	\begin{align}
		\mathrm{Kernel} \, \bigl[ \QQr_n(\lambda) \, \LLambda_n(-\rho) \bigr] (\bm{x}_n | \bm{y}_{n - 1}) = \lim_{y_n \to \infty} \, e^{\imath \rho y_n} \, \mathrm{Kernel} \, \bigl[ \QQr_n(\lambda) \, \QQr_n(\rho) \bigr] (\bm{x}_n | \bm{y}_{n}).
	\end{align}
	Thus, once the commutativity of Baxter operators is established, one can deduce the relation between Baxter and raising operators~\eqref{QrL} from it by taking the above limit. This approach has been applied for Calogero--Ruijsenaars types of models \cite{BDKK2, BCDK}, which lack local diagrams transformations.
\end{remark}


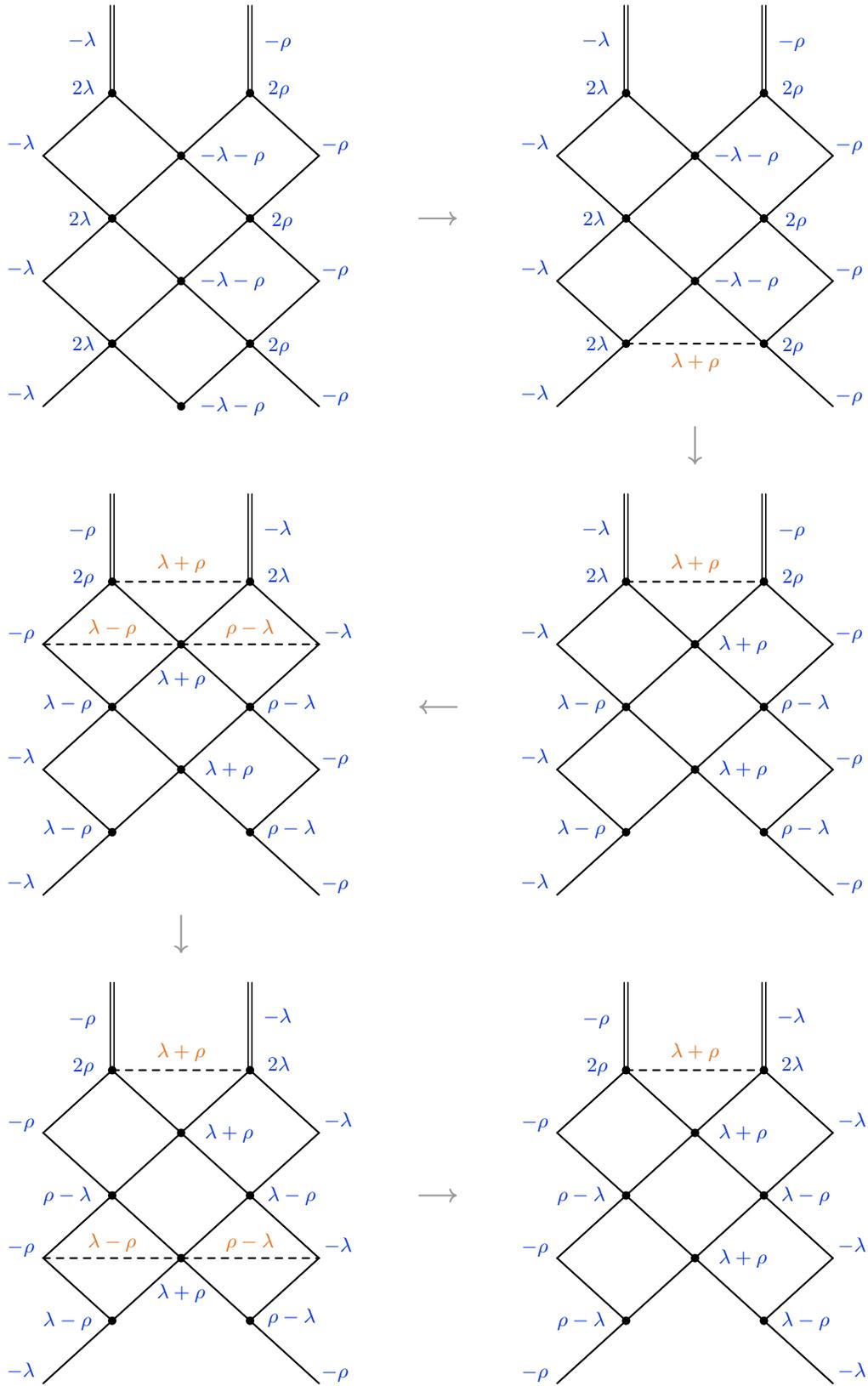
\begin{figure}[H] \centering \vspace{-1cm}
	\begin{tikzpicture}[thick, line cap = round, scale = 0.98]
		\def\l{1.1}
		\def\r{1.5pt}
		\def\h{1}
		\def\d{8.2}
		\def\dd{0.3}
		\def\hh{1.4}
		\def\t{0.03}
		\def\a{-7.8}
		\def\b{-15.6}

		
		\draw[line width = 0.6pt] ( \t, \hh) -- node[xshift = -0.5cm, yshift = 0.15cm] {\color{spec} \footnotesize $-\lambda$} ( \t, 0);
		\draw[line width = 0.6pt] (- \t, \hh) --  (- \t, 0);
		
		\draw (0, 0 ) -- ( \l, - \h);
		\draw ( \l, - \h) -- (0, - 2*\h);
		\draw (0, 0) -- (- \l, - \h) node[xshift = -0.35cm, yshift = 0.2cm] {\color{spec} \footnotesize $-\lambda$};
		\draw (- \l, - \h) -- (0, - 2*\h);
		\draw (0, - 2*\h) -- ( \l, - 3*\h) ;
		\draw (0, - 2*\h) -- (- \l, - 3*\h) node[xshift = -0.35cm, yshift = 0.2cm] {\color{spec} \footnotesize $-\lambda$};
		\draw (- \l, - 3*\h) -- (0, - 4*\h);
		\draw ( \l, - 3*\h) -- (0, - 4*\h);
		\draw (0, - 4*\h) -- (- \l, -5*\h) node[xshift = -0.35cm, yshift = 0.2cm] {\color{spec} \footnotesize $-\lambda$};
		\draw (0, - 4*\h) -- (\l, -5*\h);
		
		\draw[fill = black] (0, 0) circle (\r) node[xshift = -0.45cm, yshift = 0.1cm] {\color{spec} \footnotesize $2 \lambda$};
		\draw[fill = black] (0, - 2*\h) node[xshift = -0.5cm] {\color{spec} \footnotesize $2\lambda$} circle (\r);
		\draw[fill = black] (0, - 4*\h) node[xshift = -0.45cm] {\color{spec} \footnotesize $2 \lambda$} circle (\r);

		\draw[line width = 0.6pt] (2*\l + \t, \hh) -- node[xshift = 0.4cm, yshift = 0.1cm] {\color{spec} \footnotesize $-\rho$} (2*\l + \t, 0);
		\draw[line width = 0.6pt] (2*\l - \t,  \hh) --  (2*\l - \t, 0);
		
		\draw (2*\l, 0) -- (2*\l + \l, - \h) node[xshift = 0.25cm, yshift = 0.15cm] {\color{spec} \footnotesize $-\rho$};
		\draw (2*\l + \l, - \h) -- (2*\l, - 2*\h);
		\draw (2*\l, 0) -- (2*\l - \l, - \h);
		\draw (2*\l- \l, - \h) -- (2*\l, - 2*\h);
		\draw (2*\l, - 2*\h) -- (2*\l + \l, - 3*\h) node[xshift = 0.25cm, yshift = 0.15cm] {\color{spec} \footnotesize $-\rho$};
		\draw (2*\l, - 2*\h) -- (2*\l- \l, - 3*\h);
		\draw (2*\l - \l, - 3*\h) -- (2*\l, - 4*\h);
		\draw (2*\l + \l, - 3*\h) -- (2*\l, - 4*\h);
		\draw (2*\l, - 4*\h) -- (2*\l - \l, -5*\h);
		\draw (2*\l, - 4*\h) -- (2*\l + \l, -5*\h) node[xshift = 0.25cm, yshift = 0.15cm] {\color{spec} \footnotesize $-\rho$};
		
		\draw[fill = black] (2*\l, 0) circle (\r) node[xshift = 0.45cm, yshift = 0.05cm] {\color{spec} \footnotesize $2 \rho$};
		\draw[fill = black] (2*\l, - 2*\h) node[xshift = 0.5cm, yshift = -0.05cm] {\color{spec} \footnotesize $2\rho$} circle (\r);
		\draw[fill = black] (2*\l, - 4*\h) node[xshift = 0.45cm, yshift = -0.05cm] {\color{spec} \footnotesize $2 \rho$} circle (\r);
		
		\draw[fill = black] (\l, -\h) circle (\r) node[xshift = 0.8cm] {\color{spec} \footnotesize $-\lambda-\rho$};
		\draw[fill = black] (\l, -3*\h) circle (\r) node[xshift = 0.8cm] {\color{spec} \footnotesize $-\lambda-\rho$};
		\draw[fill = black] (\l, -5*\h) circle (\r) node[xshift = 0.8cm] {\color{spec} \footnotesize $-\lambda-\rho$};

		
		\draw[line width = 0.6pt] (\d + \t, \hh) -- node[xshift = -0.5cm, yshift = 0.15cm] {\color{spec} \footnotesize $-\lambda$} (\d + \t, 0);
		\draw[line width = 0.6pt] (\d - \t, \hh) --  (\d - \t, 0);
		
		\draw (\d, 0 ) -- (\d + \l, - \h);
		\draw (\d + \l, - \h) -- (\d, - 2*\h);
		\draw (\d, 0) -- (\d - \l, - \h) node[xshift = -0.35cm, yshift = 0.2cm] {\color{spec} \footnotesize $-\lambda$};
		\draw (\d - \l, - \h) -- (\d, - 2*\h);
		\draw (\d, - 2*\h) -- (\d + \l, - 3*\h) ;
		\draw (\d, - 2*\h) -- (\d - \l, - 3*\h) node[xshift = -0.35cm, yshift = 0.2cm] {\color{spec} \footnotesize $-\lambda$};
		\draw (\d - \l, - 3*\h) -- (\d, - 4*\h);
		\draw (\d + \l, - 3*\h) -- (\d, - 4*\h);
		\draw (\d, - 4*\h) -- (\d - \l, -5*\h) node[xshift = -0.35cm, yshift = 0.2cm] {\color{spec} \footnotesize $-\lambda$};
		
		\draw[fill = black] (\d, 0) circle (\r) node[xshift = -0.45cm, yshift = 0.1cm] {\color{spec} \footnotesize $2 \lambda$};
		\draw[fill = black] (\d, - 2*\h) node[xshift = -0.5cm] {\color{spec} \footnotesize $2\lambda$} circle (\r);
		\draw[fill = black] (\d, - 4*\h) node[xshift = -0.45cm] {\color{spec} \footnotesize $2 \lambda$} circle (\r);

		\draw[line width = 0.6pt] (\d + 2*\l + \t, \hh) -- node[xshift = 0.4cm, yshift = 0.1cm] {\color{spec} \footnotesize $-\rho$} (\d + 2*\l + \t, 0);
		\draw[line width = 0.6pt] (\d + 2*\l - \t,  \hh) --  (\d + 2*\l - \t, 0);
		
		\draw (\d + 2*\l, 0) -- (\d + 2*\l + \l, - \h) node[xshift = 0.25cm, yshift = 0.15cm] {\color{spec} \footnotesize $-\rho$};
		\draw (\d + 2*\l + \l, - \h) -- (\d + 2*\l, - 2*\h);
		\draw (\d + 2*\l, 0) -- (\d + 2*\l - \l, - \h);
		\draw (\d + 2*\l- \l, - \h) -- (\d + 2*\l, - 2*\h);
		\draw (\d + 2*\l, - 2*\h) -- (\d + 2*\l + \l, - 3*\h) node[xshift = 0.25cm, yshift = 0.15cm] {\color{spec} \footnotesize $-\rho$};
		\draw (\d + 2*\l, - 2*\h) -- (\d + 2*\l- \l, - 3*\h);
		\draw (\d + 2*\l - \l, - 3*\h) -- (\d + 2*\l, - 4*\h);
		\draw (\d + 2*\l + \l, - 3*\h) -- (\d + 2*\l, - 4*\h);
		\draw (\d + 2*\l, - 4*\h) -- (\d + 2*\l + \l, -5*\h) node[xshift = 0.25cm, yshift = 0.15cm] {\color{spec} \footnotesize $-\rho$};
		
		\draw[fill = black] (\d + 2*\l, 0) circle (\r) node[xshift = 0.45cm, yshift = 0.05cm] {\color{spec} \footnotesize $2 \rho$};
		\draw[fill = black] (\d + 2*\l, - 2*\h) node[xshift = 0.5cm, yshift = -0.05cm] {\color{spec} \footnotesize $2\rho$} circle (\r);
		\draw[fill = black] (\d + 2*\l, - 4*\h) node[xshift = 0.45cm, yshift = -0.05cm] {\color{spec} \footnotesize $2 \rho$} circle (\r);
		
		\draw[fill = black] (\d + \l, -\h) circle (\r) node[xshift = 0.8cm] {\color{spec} \footnotesize $-\lambda-\rho$};
		\draw[fill = black] (\d + \l, -3*\h) circle (\r) node[xshift = 0.8cm] {\color{spec} \footnotesize $-\lambda-\rho$};
		\draw[dashed] (\d, - 4*\h) -- node[below] {\color{spec2} \footnotesize $\lambda + \rho$} (\d + 2*\l, - 4*\h);

		
		\draw[line width = 0.6pt] (\d + \t, \a + \hh) -- node[xshift = -0.5cm, yshift = 0.15cm] {\color{spec} \footnotesize $-\lambda$} (\d + \t, \a);
		\draw[line width = 0.6pt] (\d - \t, \a + \hh) --  (\d - \t, \a);
		
		\draw (\d, \a) -- (\d + \l, \a - \h);
		\draw (\d + \l, \a - \h) -- (\d, \a - 2*\h);
		\draw (\d, \a) -- (\d - \l, \a - \h) node[xshift = -0.35cm, yshift = 0.2cm] {\color{spec} \footnotesize $-\lambda$};
		\draw (\d - \l, \a - \h) -- (\d, \a - 2*\h);
		\draw (\d, \a - 2*\h) -- (\d + \l, \a - 3*\h) ;
		\draw (\d, \a - 2*\h) -- (\d - \l, \a - 3*\h) node[xshift = -0.35cm, yshift = 0.2cm] {\color{spec} \footnotesize $-\lambda$};
		\draw (\d - \l, \a - 3*\h) -- (\d, \a - 4*\h);
		\draw (\d + \l, \a - 3*\h) -- (\d, \a - 4*\h);
		\draw (\d, \a - 4*\h) -- (\d - \l, \a - 5*\h) node[xshift = -0.35cm, yshift = 0.2cm] {\color{spec} \footnotesize $-\lambda$};
		
		\draw[fill = black] (\d, \a) circle (\r) node[xshift = -0.45cm, yshift = 0.1cm] {\color{spec} \footnotesize $2 \lambda$};
		\draw[fill = black] (\d, \a - 2*\h) node[xshift = -0.7cm, yshift = 0.05cm] {\color{spec} \footnotesize $\lambda - \rho$} circle (\r);
		\draw[fill = black] (\d, \a - 4*\h) node[xshift = -0.7cm, yshift = 0.05cm] {\color{spec} \footnotesize $\lambda - \rho$} circle (\r);

		\draw[line width = 0.6pt] (\d + 2*\l + \t, \a + \hh) -- node[xshift = 0.4cm, yshift = 0.1cm] {\color{spec} \footnotesize $-\rho$} (\d + 2*\l + \t, \a);
		\draw[line width = 0.6pt] (\d + 2*\l - \t,  \a + \hh) --  (\d + 2*\l - \t, \a);
		
		\draw (\d + 2*\l, \a) -- (\d + 2*\l + \l, \a - \h) node[xshift = 0.25cm, yshift = 0.15cm] {\color{spec} \footnotesize $-\rho$};
		\draw (\d + 2*\l + \l, \a - \h) -- (\d + 2*\l, \a - 2*\h);
		\draw (\d + 2*\l, \a) -- (\d + 2*\l - \l, \a - \h);
		\draw (\d + 2*\l- \l, \a - \h) -- (\d + 2*\l, \a - 2*\h);
		\draw (\d + 2*\l, \a - 2*\h) -- (\d + 2*\l + \l, \a - 3*\h) node[xshift = 0.25cm, yshift = 0.15cm] {\color{spec} \footnotesize $-\rho$};
		\draw (\d + 2*\l, \a - 2*\h) -- (\d + 2*\l- \l, \a - 3*\h);
		\draw (\d + 2*\l - \l, \a - 3*\h) -- (\d + 2*\l, \a - 4*\h);
		\draw (\d + 2*\l + \l, \a - 3*\h) -- (\d + 2*\l, \a - 4*\h);
		\draw (\d + 2*\l, \a - 4*\h) -- (\d + 2*\l + \l, \a - 5*\h) node[xshift = 0.25cm, yshift = 0.15cm] {\color{spec} \footnotesize $-\rho$};
		
		\draw[fill = black] (\d + 2*\l, \a) circle (\r) node[xshift = 0.45cm, yshift = 0.05cm] {\color{spec} \footnotesize $2 \rho$};
		\draw[fill = black] (\d + 2*\l, \a - 2*\h) node[xshift = 0.65cm, yshift = 0.05cm] {\color{spec} \footnotesize $\rho - \lambda$} circle (\r);
		\draw[fill = black] (\d + 2*\l, \a - 4*\h) node[xshift = 0.65cm, yshift = 0.05cm] {\color{spec} \footnotesize $\rho - \lambda$} circle (\r);
		
		\draw[fill = black] (\d + \l, \a -\h) circle (\r) node[xshift = 0.75cm] {\color{spec} \footnotesize $\lambda + \rho$};
		\draw[fill = black] (\d + \l, \a - 3*\h) circle (\r) node[xshift = 0.75cm] {\color{spec} \footnotesize $\lambda + \rho$};
		\draw[dashed] (\d, \a) -- node[above] {\color{spec2} \footnotesize $\lambda + \rho$} (\d + 2*\l, \a);

		
		\draw[line width = 0.6pt] (\t, \a + \hh) -- node[xshift = -0.5cm, yshift = 0.1cm] {\color{spec} \footnotesize $-\rho$} (\t, \a);
		\draw[line width = 0.6pt] (- \t, \a + \hh) --  (- \t, \a);
		
		\draw (0, \a) -- ( \l, \a - \h);
		\draw ( \l, \a - \h) -- (0, \a - 2*\h);
		\draw (0, \a) -- (- \l, \a - \h) node[xshift = -0.35cm, yshift = 0.15cm] {\color{spec} \footnotesize $-\rho$};
		\draw (- \l, \a - \h) -- (0, \a - 2*\h);
		\draw (0, \a - 2*\h) -- ( \l, \a - 3*\h) ;
		\draw (0, \a - 2*\h) -- (- \l, \a - 3*\h) node[xshift = -0.35cm, yshift = 0.2cm] {\color{spec} \footnotesize $-\lambda$};
		\draw (- \l, \a - 3*\h) -- (0, \a - 4*\h);
		\draw ( \l, \a - 3*\h) -- (0, \a - 4*\h);
		\draw (0, \a - 4*\h) -- (- \l, \a - 5*\h) node[xshift = -0.35cm, yshift = 0.2cm] {\color{spec} \footnotesize $-\lambda$};
		
		\draw[fill = black] (0, \a) circle (\r) node[xshift = -0.45cm, yshift = 0.05cm] {\color{spec} \footnotesize $2 \rho$};
		\draw[fill = black] (0, \a - 2*\h) node[xshift = -0.7cm, yshift = 0.05cm] {\color{spec} \footnotesize $\lambda - \rho$} circle (\r);
		\draw[fill = black] (0, \a - 4*\h) node[xshift = -0.7cm, yshift = 0.05cm] {\color{spec} \footnotesize $\lambda - \rho$} circle (\r);
		\draw[dashed] ( -\l, \a - \h) -- node[yshift = 0.25cm] {\color{spec2} \footnotesize $\lambda - \rho$} ( \l, \a - \h);

		\draw[line width = 0.6pt] ( 2*\l + \t, \a + \hh) -- node[xshift = 0.4cm, yshift = 0.15cm] {\color{spec} \footnotesize $-\lambda$} ( 2*\l + \t, \a);
		\draw[line width = 0.6pt] ( 2*\l - \t,  \a + \hh) --  ( 2*\l - \t, \a);
		
		\draw ( 2*\l, \a) -- ( 2*\l + \l, \a - \h) node[xshift = 0.3cm, yshift = 0.2cm] {\color{spec} \footnotesize $-\lambda$};
		\draw ( 2*\l + \l, \a - \h) -- ( 2*\l, \a - 2*\h);
		\draw ( 2*\l, \a) -- ( 2*\l - \l, \a - \h);
		\draw ( 2*\l- \l, \a - \h) -- ( 2*\l, \a - 2*\h);
		\draw ( 2*\l, \a - 2*\h) -- ( 2*\l + \l, \a - 3*\h) node[xshift = 0.25cm, yshift = 0.15cm] {\color{spec} \footnotesize $-\rho$};
		\draw ( 2*\l, \a - 2*\h) -- ( 2*\l- \l, \a - 3*\h);
		\draw ( 2*\l - \l, \a - 3*\h) -- ( 2*\l, \a - 4*\h);
		\draw ( 2*\l + \l, \a - 3*\h) -- ( 2*\l, \a - 4*\h);
		\draw ( 2*\l, \a - 4*\h) -- ( 2*\l + \l, \a - 5*\h) node[xshift = 0.25cm, yshift = 0.15cm] {\color{spec} \footnotesize $-\rho$};
		
		\draw[fill = black] ( 2*\l, \a) circle (\r) node[xshift = 0.45cm, yshift = 0.1cm] {\color{spec} \footnotesize $2 \lambda$};
		\draw[fill = black] ( 2*\l, \a - 2*\h) node[xshift = 0.65cm, yshift = 0.05cm] {\color{spec} \footnotesize $\rho - \lambda$} circle (\r);
		\draw[fill = black] ( 2*\l, \a - 4*\h) node[xshift = 0.65cm, yshift = 0.05cm] {\color{spec} \footnotesize $\rho - \lambda$} circle (\r);
		
		\draw[fill = black] ( \l, \a -\h) circle (\r) node[yshift = -0.55cm] {\color{spec} \footnotesize $\lambda + \rho$};
		\draw[fill = black] ( \l, \a - 3*\h) circle (\r) node[xshift = 0.75cm] {\color{spec} \footnotesize $\lambda + \rho$};
		\draw[dashed] ( 2*\l - \l, \a - \h) -- node[yshift = 0.25cm] {\color{spec2} \footnotesize $\rho - \lambda$} ( 2*\l + \l, \a - \h);
		\draw[dashed] (0, \a) -- node[above] {\color{spec2} \footnotesize $\lambda + \rho$} ( 2*\l, \a);

		
		\draw[line width = 0.6pt] (\t, \b + \hh) -- node[xshift = -0.5cm, yshift = 0.1cm] {\color{spec} \footnotesize $-\rho$} (\t, \b);
		\draw[line width = 0.6pt] (- \t, \b + \hh) --  (- \t, \b);
		
		\draw (0, \b) -- ( \l, \b - \h);
		\draw ( \l, \b - \h) -- (0, \b - 2*\h);
		\draw (0, \b) -- (- \l, \b - \h) node[xshift = -0.35cm, yshift = 0.15cm] {\color{spec} \footnotesize $-\rho$};
		\draw (- \l, \b - \h) -- (0, \b - 2*\h);
		\draw (0, \b - 2*\h) -- ( \l, \b - 3*\h) ;
		\draw (0, \b - 2*\h) -- (- \l, \b - 3*\h) node[xshift = -0.35cm, yshift = 0.15cm] {\color{spec} \footnotesize $-\rho$};
		\draw (- \l, \b - 3*\h) -- (0, \b - 4*\h);
		\draw ( \l, \b - 3*\h) -- (0, \b - 4*\h);
		\draw (0, \b - 4*\h) -- (- \l, \b - 5*\h) node[xshift = -0.35cm, yshift = 0.2cm] {\color{spec} \footnotesize $-\lambda$};
		
		\draw[fill = black] (0, \b) circle (\r) node[xshift = -0.45cm, yshift = 0.05cm] {\color{spec} \footnotesize $2 \rho$};
		\draw[fill = black] (0, \b - 2*\h) node[xshift = -0.7cm] {\color{spec} \footnotesize $\rho - \lambda$} circle (\r);
		\draw[fill = black] (0, \b - 4*\h) node[xshift = -0.7cm] {\color{spec} \footnotesize $\lambda - \rho$} circle (\r);
		\draw[dashed] ( - \l, \b - 3*\h) -- node[yshift = 0.25cm] {\color{spec2} \footnotesize $\lambda - \rho$} ( \l, \b - 3*\h);

		\draw[line width = 0.6pt] ( 2*\l + \t, \b + \hh) -- node[xshift = 0.4cm, yshift = 0.15cm] {\color{spec} \footnotesize $-\lambda$} ( 2*\l + \t, \b);
		\draw[line width = 0.6pt] ( 2*\l - \t,  \b + \hh) --  ( 2*\l - \t, \b);
		
		\draw ( 2*\l, \b) -- ( 2*\l + \l, \b - \h) node[xshift = 0.3cm, yshift = 0.2cm] {\color{spec} \footnotesize $-\lambda$};
		\draw ( 2*\l + \l, \b - \h) -- ( 2*\l, \b - 2*\h);
		\draw ( 2*\l, \b) -- ( 2*\l - \l, \b - \h);
		\draw ( 2*\l- \l, \b - \h) -- ( 2*\l, \b - 2*\h);
		\draw ( 2*\l, \b - 2*\h) -- ( 2*\l + \l, \b - 3*\h) node[xshift = 0.3cm, yshift = 0.2cm] {\color{spec} \footnotesize $-\lambda$};
		\draw ( 2*\l, \b - 2*\h) -- ( 2*\l- \l, \b - 3*\h);
		\draw ( 2*\l - \l, \b - 3*\h) -- ( 2*\l, \b - 4*\h);
		\draw ( 2*\l + \l, \b - 3*\h) -- ( 2*\l, \b - 4*\h);
		\draw ( 2*\l, \b - 4*\h) -- ( 2*\l + \l, \b - 5*\h) node[xshift = 0.25cm, yshift = 0.15cm] {\color{spec} \footnotesize $-\rho$};
		
		\draw[fill = black] ( 2*\l, \b) circle (\r) node[xshift = 0.45cm, yshift = 0.1cm] {\color{spec} \footnotesize $2 \lambda$};
		\draw[fill = black] ( 2*\l, \b - 2*\h) node[xshift = 0.65cm] {\color{spec} \footnotesize $\lambda - \rho$} circle (\r);
		\draw[fill = black] ( 2*\l, \b - 4*\h) node[xshift = 0.65cm, yshift = 0.05cm] {\color{spec} \footnotesize $\rho - \lambda$} circle (\r);
		\draw[dashed] ( 2*\l - \l, \b - 3*\h) -- node[yshift = 0.25cm] {\color{spec2} \footnotesize $\rho - \lambda$} ( 2*\l + \l, \b - 3*\h);
		
		\draw[fill = black] ( \l, \b -\h) circle (\r) node[xshift = 0.75cm] {\color{spec} \footnotesize $\lambda + \rho$};
		\draw[fill = black] ( \l, \b - 3*\h) circle (\r) node[yshift = -0.55cm] {\color{spec} \footnotesize $\lambda + \rho$};
		\draw[dashed] (0, \b) -- node[above] {\color{spec2} \footnotesize $\lambda + \rho$} ( 2*\l, \b);

		\draw[line width = 0.6pt] (\d + \t, \b + \hh) -- node[xshift = -0.5cm, yshift = 0.1cm] {\color{spec} \footnotesize $-\rho$} (\d + \t, \b);
		\draw[line width = 0.6pt] (\d - \t, \b + \hh) --  (\d - \t, \b);
		
		\draw (\d, \b) -- (\d + \l, \b - \h);
		\draw (\d + \l, \b - \h) -- (\d, \b - 2*\h);
		\draw (\d, \b) -- (\d - \l, \b - \h) node[xshift = -0.35cm, yshift = 0.15cm] {\color{spec} \footnotesize $-\rho$};
		\draw (\d - \l, \b - \h) -- (\d, \b - 2*\h);
		\draw (\d, \b - 2*\h) -- (\d + \l, \b - 3*\h) ;
		\draw (\d, \b - 2*\h) -- (\d - \l, \b - 3*\h) node[xshift = -0.35cm, yshift = 0.15cm] {\color{spec} \footnotesize $-\rho$};
		\draw (\d - \l, \b - 3*\h) -- (\d, \b - 4*\h);
		\draw (\d + \l, \b - 3*\h) -- (\d, \b - 4*\h);
		\draw (\d, \b - 4*\h) -- (\d - \l, \b - 5*\h) node[xshift = -0.35cm, yshift = 0.15cm] {\color{spec} \footnotesize $-\rho$};
		
		\draw[fill = black] (\d, \b) circle (\r) node[xshift = -0.45cm, yshift = 0.05cm] {\color{spec} \footnotesize $2 \rho$};
		\draw[fill = black] (\d, \b - 2*\h) node[xshift = -0.7cm] {\color{spec} \footnotesize $\rho - \lambda$} circle (\r);
		\draw[fill = black] (\d, \b - 4*\h) node[xshift = -0.7cm] {\color{spec} \footnotesize $\rho - \lambda$} circle (\r);

		\draw[line width = 0.6pt] (\d + 2*\l + \t, \b + \hh) -- node[xshift = 0.4cm, yshift = 0.15cm] {\color{spec} \footnotesize $-\lambda$} (\d + 2*\l + \t, \b);
		\draw[line width = 0.6pt] (\d + 2*\l - \t,  \b + \hh) --  (\d + 2*\l - \t, \b);
		
		\draw (\d + 2*\l, \b) -- (\d + 2*\l + \l, \b - \h) node[xshift = 0.3cm, yshift = 0.2cm] {\color{spec} \footnotesize $-\lambda$};
		\draw (\d + 2*\l + \l, \b - \h) -- (\d + 2*\l, \b - 2*\h);
		\draw (\d + 2*\l, \b) -- (\d + 2*\l - \l, \b - \h);
		\draw (\d + 2*\l- \l, \b - \h) -- (\d + 2*\l, \b - 2*\h);
		\draw (\d + 2*\l, \b - 2*\h) -- (\d + 2*\l + \l, \b - 3*\h) node[xshift = 0.3cm, yshift = 0.2cm] {\color{spec} \footnotesize $-\lambda$};
		\draw (\d + 2*\l, \b - 2*\h) -- (\d + 2*\l- \l, \b - 3*\h);
		\draw (\d + 2*\l - \l, \b - 3*\h) -- (\d + 2*\l, \b - 4*\h);
		\draw (\d + 2*\l + \l, \b - 3*\h) -- (\d + 2*\l, \b - 4*\h);
		\draw (\d + 2*\l, \b - 4*\h) -- (\d + 2*\l + \l, \b - 5*\h) node[xshift = 0.3cm, yshift = 0.2cm] {\color{spec} \footnotesize $-\lambda$};
		
		\draw[fill = black] (\d + 2*\l, \b) circle (\r) node[xshift = 0.45cm, yshift = 0.1cm] {\color{spec} \footnotesize $2 \lambda$};
		\draw[fill = black] (\d + 2*\l, \b - 2*\h) node[xshift = 0.65cm] {\color{spec} \footnotesize $\lambda - \rho$} circle (\r);
		\draw[fill = black] (\d + 2*\l, \b - 4*\h) node[xshift = 0.65cm] {\color{spec} \footnotesize $\lambda - \rho$} circle (\r);
		
		\draw[fill = black] (\d + \l, \b -\h) circle (\r) node[xshift = 0.75cm] {\color{spec} \footnotesize $\lambda + \rho$};
		\draw[fill = black] (\d + \l, \b - 3*\h) circle (\r) node[xshift = 0.75cm] {\color{spec} \footnotesize $\lambda + \rho$};
		\draw[dashed] (\d, \b) -- node[above] {\color{spec2} \footnotesize $\lambda + \rho$} (\d + 2*\l, \b);

		\draw[arrow, gray!80] (\l + 0.5*\d - \dd, - 2*\h) -- (\l + 0.5*\d + \dd, - 2*\h);
		\draw[arrow, gray!80] (\l + 0.5*\d + \dd, \a - 2*\h) -- (\l + 0.5*\d - \dd, \a - 2*\h);
		\draw[arrow, gray!80] (\l + 0.5*\d - \dd, \b - 2*\h) -- (\l + 0.5*\d + \dd, \b - 2*\h);
		
		\draw[arrow, gray!80] (\d + \l, \a + 1.55*\hh + \dd) -- (\d + \l, \a + 1.55*\hh - \dd);
		\draw[arrow, gray!80] (\l, \b + 1.55*\hh + \dd) -- (\l, \b + 1.55*\hh - \dd);
	\end{tikzpicture}
	\vspace{0.2cm}
	\caption{Commutativity of Baxter operators} \label{fig:Qcomm}
\end{figure}

\newpage

\begin{figure}[H] \centering \vspace{-1cm}
	\begin{tikzpicture}[thick, line cap = round, scale = 0.98]
		\def\l{1.1}
		\def\r{1.5pt}
		\def\h{1}
		\def\d{8.2}
		\def\dd{0.3}
		\def\hh{1.4}
		\def\t{0.03}
		\def\a{-7.8}
		\def\b{-15.6}

		
		\draw[line width = 0.6pt] ( \t, \hh) -- node[xshift = -0.5cm, yshift = 0.15cm] {\color{spec} \footnotesize $-\lambda$} ( \t, 0);
		\draw[line width = 0.6pt] (- \t, \hh) --  (- \t, 0);
		
		\draw (0, 0 ) -- ( \l, - \h);
		\draw ( \l, - \h) -- (0, - 2*\h);
		\draw (0, 0) -- (- \l, - \h) node[xshift = -0.35cm, yshift = 0.2cm] {\color{spec} \footnotesize $-\lambda$};
		\draw (- \l, - \h) -- (0, - 2*\h);
		\draw (0, - 2*\h) -- ( \l, - 3*\h) ;
		\draw (0, - 2*\h) -- (- \l, - 3*\h) node[xshift = -0.35cm, yshift = 0.2cm] {\color{spec} \footnotesize $-\lambda$};
		\draw (- \l, - 3*\h) -- (0, - 4*\h);
		\draw ( \l, - 3*\h) -- (0, - 4*\h);
		\draw (0, - 4*\h) -- (- \l, -5*\h) node[xshift = -0.35cm, yshift = 0.2cm] {\color{spec} \footnotesize $-\lambda$};
		\draw (0, - 4*\h) -- (\l, -5*\h);
		
		\draw[fill = black] (0, 0) circle (\r) node[xshift = -0.45cm, yshift = 0.1cm] {\color{spec} \footnotesize $2 \lambda$};
		\draw[fill = black] (0, - 2*\h) node[xshift = -0.5cm] {\color{spec} \footnotesize $2\lambda$} circle (\r);
		\draw[fill = black] (0, - 4*\h) node[xshift = -0.45cm] {\color{spec} \footnotesize $2 \lambda$} circle (\r);

		\draw[line width = 0.6pt] (2*\l + \t, \hh) -- node[xshift = 0.4cm, yshift = 0.1cm] {\color{spec} \footnotesize $-\rho$} (2*\l + \t, 0);
		\draw[line width = 0.6pt] (2*\l - \t,  \hh) --  (2*\l - \t, 0);
		
		\draw (2*\l, 0) -- (2*\l + \l, - \h) node[xshift = 0.25cm, yshift = 0.15cm] {\color{spec} \footnotesize $-\rho$};
		\draw (2*\l + \l, - \h) -- (2*\l, - 2*\h);
		\draw (2*\l, 0) -- (2*\l - \l, - \h);
		\draw (2*\l- \l, - \h) -- (2*\l, - 2*\h);
		\draw (2*\l, - 2*\h) -- (2*\l + \l, - 3*\h) node[xshift = 0.25cm, yshift = 0.15cm] {\color{spec} \footnotesize $-\rho$};
		\draw (2*\l, - 2*\h) -- (2*\l- \l, - 3*\h);
		\draw (2*\l - \l, - 3*\h) -- (2*\l, - 4*\h);
		\draw (2*\l + \l, - 3*\h) -- (2*\l, - 4*\h);
		\draw (2*\l, - 4*\h) -- (2*\l - \l, -5*\h);
		
		\draw[fill = black] (2*\l, 0) circle (\r) node[xshift = 0.45cm, yshift = 0.05cm] {\color{spec} \footnotesize $2 \rho$};
		\draw[fill = black] (2*\l, - 2*\h) node[xshift = 0.5cm, yshift = -0.05cm] {\color{spec} \footnotesize $2\rho$} circle (\r);
		\draw[fill = black] (2*\l, - 4*\h) node[xshift = 0.45cm, yshift = -0.05cm] {\color{spec} \footnotesize $2 \rho$} circle (\r);
		
		\draw[fill = black] (\l, -\h) circle (\r) node[xshift = 0.8cm] {\color{spec} \footnotesize $-\lambda-\rho$};
		\draw[fill = black] (\l, -3*\h) circle (\r) node[xshift = 0.8cm] {\color{spec} \footnotesize $-\lambda-\rho$};
		\draw[fill = black] (\l, -5*\h) circle (\r) node[xshift = 0.8cm] {\color{spec} \footnotesize $-\lambda-\rho$};

		
		\draw[line width = 0.6pt] (\d + \t,  \hh) -- node[xshift = -0.5cm, yshift = 0.1cm] {\color{spec} \footnotesize $-\rho$} (\d + \t, 0);
		\draw[line width = 0.6pt] (\d - \t, \hh) --  (\d - \t, 0);
		
		\draw (\d , 0) -- (\d +  \l,  - \h);
		\draw (\d +  \l,  - \h) -- (\d,  - 2*\h);
		\draw (\d, 0) -- (\d - \l,  - \h) node[xshift = -0.35cm, yshift = 0.15cm] {\color{spec} \footnotesize $-\rho$};
		\draw (\d - \l,  - \h) -- (\d,  - 2*\h);
		\draw (\d,  - 2*\h) -- (\d +  \l,  - 3*\h) ;
		\draw (\d,  - 2*\h) -- (\d - \l,  - 3*\h) node[xshift = -0.35cm, yshift = 0.15cm] {\color{spec} \footnotesize $-\rho$};
		\draw (\d - \l,  - 3*\h) -- (\d,  - 4*\h);
		\draw (\d +  \l,  - 3*\h) -- (\d,  - 4*\h);
		\draw (\d,  - 4*\h) -- (\d - \l,  - 5*\h) node[xshift = -0.35cm, yshift = 0.2cm] {\color{spec} \footnotesize $-\lambda$};
		
		\draw[fill = black] (\d, 0) circle (\r) node[xshift = -0.45cm, yshift = 0.05cm] {\color{spec} \footnotesize $2 \rho$};
		\draw[fill = black] (\d,  - 2*\h) node[xshift = -0.7cm] {\color{spec} \footnotesize $\rho - \lambda$} circle (\r);
		\draw[fill = black] (\d,  - 4*\h) node[xshift = -0.7cm] {\color{spec} \footnotesize $\lambda - \rho$} circle (\r);
		\draw[dashed] (\d - \l,  - 3*\h) -- node[yshift = 0.25cm] {\color{spec2} \footnotesize $\lambda - \rho$} (\d +  \l,  - 3*\h);

		\draw[line width = 0.6pt] (\d + 2*\l + \t,  + \hh) -- node[xshift = 0.4cm, yshift = 0.15cm] {\color{spec} \footnotesize $-\lambda$} (\d +  2*\l + \t, 0);
		\draw[line width = 0.6pt] (\d +  2*\l - \t,   + \hh) --  (\d + 2*\l - \t, 0);
		
		\draw (\d + 2*\l, 0) -- ( \d + 2*\l + \l,  - \h) node[xshift = 0.3cm, yshift = 0.2cm] {\color{spec} \footnotesize $-\lambda$};
		\draw (\d +  2*\l + \l,  - \h) -- (\d +  2*\l,  - 2*\h);
		\draw (\d +  2*\l, 0) -- ( \d + 2*\l - \l,  - \h);
		\draw ( \d + 2*\l- \l,  - \h) -- ( \d + 2*\l,  - 2*\h);
		\draw ( \d + 2*\l,  - 2*\h) -- ( \d + 2*\l + \l,  - 3*\h) node[xshift = 0.3cm, yshift = 0.2cm] {\color{spec} \footnotesize $-\lambda$};
		\draw ( \d + 2*\l,  - 2*\h) -- ( \d + 2*\l- \l,  - 3*\h);
		\draw ( \d + 2*\l - \l,  - 3*\h) -- ( \d + 2*\l,  - 4*\h);
		\draw ( \d + 2*\l + \l,  - 3*\h) -- ( \d + 2*\l,  - 4*\h);
		
		\draw[fill = black] ( \d + 2*\l, 0) circle (\r) node[xshift = 0.45cm, yshift = 0.1cm] {\color{spec} \footnotesize $2 \lambda$};
		\draw[fill = black] ( \d + 2*\l,  - 2*\h) node[xshift = 0.65cm] {\color{spec} \footnotesize $\lambda - \rho$} circle (\r);
		\draw[fill = black] ( \d + 2*\l,  - 4*\h) node[xshift = 0.65cm, yshift = 0.05cm] {\color{spec} \footnotesize $\rho - \lambda$} circle (\r);
		\draw[dashed] ( \d + 2*\l - \l,  - 3*\h) -- node[yshift = 0.25cm] {\color{spec2} \footnotesize $\rho - \lambda$} ( \d + 2*\l + \l,  - 3*\h);
		
		\draw[fill = black] ( \d + \l,  -\h) circle (\r) node[xshift = 0.75cm] {\color{spec} \footnotesize $\lambda + \rho$};
		\draw[fill = black] ( \d + \l,  - 3*\h) circle (\r) node[yshift = -0.55cm] {\color{spec} \footnotesize $\lambda + \rho$};
		\draw[dashed] (\d, 0) -- node[above] {\color{spec2} \footnotesize $\lambda + \rho$} ( \d + 2*\l, 0);

		
		\draw[line width = 0.6pt] (\d + \t, \a + \hh) -- node[xshift = -0.5cm, yshift = 0.1cm] {\color{spec} \footnotesize $-\rho$} (\d + \t, \a);
		\draw[line width = 0.6pt] (\d - \t, \a + \hh) --  (\d - \t, \a);
		
		\draw (\d, \a) -- (\d + \l, \a - \h);
		\draw (\d + \l, \a - \h) -- (\d, \a - 2*\h);
		\draw (\d, \a) -- (\d - \l, \a - \h) node[xshift = -0.35cm, yshift = 0.15cm] {\color{spec} \footnotesize $-\rho$};
		\draw (\d - \l, \a - \h) -- (\d, \a - 2*\h);
		\draw (\d, \a - 2*\h) -- (\d + \l, \a - 3*\h) ;
		\draw (\d, \a - 2*\h) -- (\d - \l, \a - 3*\h) node[xshift = -0.35cm, yshift = 0.15cm] {\color{spec} \footnotesize $-\rho$};
		\draw (\d - \l, \a - 3*\h) -- (\d, \a - 4*\h);
		\draw (\d + \l, \a - 3*\h) -- (\d, \a - 4*\h);
		\draw (\d, \a - 4*\h) -- (\d - \l, \a - 5*\h) node[xshift = -0.35cm, yshift = 0.15cm] {\color{spec} \footnotesize $-\rho$};
		
		\draw[fill = black] (\d, \a) circle (\r) node[xshift = -0.45cm, yshift = 0.05cm] {\color{spec} \footnotesize $2 \rho$};
		\draw[fill = black] (\d, \a - 2*\h) node[xshift = -0.7cm] {\color{spec} \footnotesize $\rho - \lambda$} circle (\r);
		\draw[fill = black] (\d, \a - 4*\h) node[xshift = -0.7cm] {\color{spec} \footnotesize $\rho - \lambda$} circle (\r);

		\draw[line width = 0.6pt] (\d + 2*\l + \t, \a + \hh) -- node[xshift = 0.4cm, yshift = 0.15cm] {\color{spec} \footnotesize $-\lambda$} (\d + 2*\l + \t, \a);
		\draw[line width = 0.6pt] (\d + 2*\l - \t,  \a + \hh) --  (\d + 2*\l - \t, \a);
		
		\draw (\d + 2*\l, \a) -- (\d + 2*\l + \l, \a - \h) node[xshift = 0.3cm, yshift = 0.2cm] {\color{spec} \footnotesize $-\lambda$};
		\draw (\d + 2*\l + \l, \a - \h) -- (\d + 2*\l, \a - 2*\h);
		\draw (\d + 2*\l, \a) -- (\d + 2*\l - \l, \a - \h);
		\draw (\d + 2*\l- \l, \a - \h) -- (\d + 2*\l, \a - 2*\h);
		\draw (\d + 2*\l, \a - 2*\h) -- (\d + 2*\l + \l, \a - 3*\h) node[xshift = 0.3cm, yshift = 0.2cm] {\color{spec} \footnotesize $-\lambda$};
		\draw (\d + 2*\l, \a - 2*\h) -- (\d + 2*\l- \l, \a - 3*\h);
		
		\draw[fill = black] (\d + 2*\l, \a) circle (\r) node[xshift = 0.45cm, yshift = 0.1cm] {\color{spec} \footnotesize $2 \lambda$};
		\draw[fill = black] (\d + 2*\l, \a - 2*\h) node[xshift = 0.65cm] {\color{spec} \footnotesize $\lambda - \rho$} circle (\r);
		
		\draw[fill = black] (\d + \l, \a -\h) circle (\r) node[xshift = 0.75cm] {\color{spec} \footnotesize $\lambda + \rho$};
		\draw[fill = black] (\d + \l, \a - 3*\h) circle (\r) node[xshift = 0.75cm] {\color{spec} \footnotesize $\lambda + \rho$};
		\draw[dashed] (\d, \a) -- node[above] {\color{spec2} \footnotesize $\lambda + \rho$} (\d + 2*\l, \a);

		
		\draw[line width = 0.6pt] ( \t, \a + \hh) -- node[xshift = -0.5cm, yshift = 0.1cm] {\color{spec} \footnotesize $-\rho$} ( \t, \a);
		\draw[line width = 0.6pt] (- \t, \a + \hh) --  (- \t, \a);
		
		\draw (0, \a) -- ( \l, \a - \h);
		\draw ( \l, \a - \h) -- (0, \a - 2*\h);
		\draw (0, \a) -- (- \l, \a - \h) node[xshift = -0.35cm, yshift = 0.15cm] {\color{spec} \footnotesize $-\rho$};
		\draw (- \l, \a - \h) -- (0, \a - 2*\h);
		\draw (0, \a - 2*\h) -- ( \l, \a - 3*\h) ;
		\draw (0, \a - 2*\h) -- (- \l, \a - 3*\h) node[xshift = -0.35cm, yshift = 0.15cm] {\color{spec} \footnotesize $-\rho$};
		\draw (- \l, \a - 3*\h) -- (0, \a - 4*\h);
		\draw ( \l, \a - 3*\h) -- (0, \a - 4*\h);
		\draw (0, \a - 4*\h) -- (- \l, \a - 5*\h) node[xshift = -0.35cm, yshift = 0.15cm] {\color{spec} \footnotesize $-\rho$};
		
		\draw[fill = black] (0, \a) circle (\r) node[xshift = -0.45cm, yshift = 0.05cm] {\color{spec} \footnotesize $2 \rho$};
		\draw[fill = black] (0, \a - 2*\h) node[xshift = -0.45cm] {\color{spec} \footnotesize $2\rho$} circle (\r);
		\draw[fill = black] (0, \a - 4*\h) node[xshift = -0.7cm] {\color{spec} \footnotesize $\rho - \lambda$} circle (\r);

		\draw[line width = 0.6pt] ( 2*\l + \t, \a + \hh) -- node[xshift = 0.4cm, yshift = 0.15cm] {\color{spec} \footnotesize $-\lambda$} ( 2*\l + \t, \a);
		\draw[line width = 0.6pt] ( 2*\l - \t,  \a + \hh) --  ( 2*\l - \t, \a);
		
		\draw ( 2*\l, \a) -- ( 2*\l + \l, \a - \h) node[xshift = 0.3cm, yshift = 0.2cm] {\color{spec} \footnotesize $-\lambda$};
		\draw ( 2*\l + \l, \a - \h) -- ( 2*\l, \a - 2*\h);
		\draw ( 2*\l, \a) -- ( 2*\l - \l, \a - \h);
		\draw ( 2*\l- \l, \a - \h) -- ( 2*\l, \a - 2*\h);
		\draw ( 2*\l, \a - 2*\h) -- ( 2*\l + \l, \a - 3*\h) node[xshift = 0.3cm, yshift = 0.2cm] {\color{spec} \footnotesize $-\lambda$};
		\draw ( 2*\l, \a - 2*\h) -- ( 2*\l- \l, \a - 3*\h);
		
		\draw[fill = black] ( 2*\l, \a) circle (\r) node[xshift = 0.45cm, yshift = 0.1cm] {\color{spec} \footnotesize $2 \lambda$};
		\draw[fill = black] ( 2*\l, \a - 2*\h) node[xshift = 0.45cm] {\color{spec} \footnotesize $2\lambda$} circle (\r);
		
		\draw[fill = black] ( \l, \a -\h) circle (\r) node[xshift = 0.75cm] {\color{spec} \footnotesize $- \lambda - \rho$};
		\draw[fill = black] ( \l, \a - 3*\h) circle (\r) node[xshift = 0.75cm] {\color{spec} \footnotesize $\lambda + \rho$};
		\draw[dashed] (0, \a - 2*\h) -- node[above] {\color{spec2} \footnotesize $\lambda + \rho$} ( 2*\l, \a - 2*\h);

		
		\draw[line width = 0.6pt] ( \t, \b + \hh) -- node[xshift = -0.5cm, yshift = 0.1cm] {\color{spec} \footnotesize $-\rho$} ( \t, \b);
		\draw[line width = 0.6pt] (- \t, \b + \hh) --  (- \t, \b);
		
		\draw (0, \b) -- ( \l, \b - \h);
		\draw ( \l, \b - \h) -- (0, \b - 2*\h);
		\draw (0, \b) -- (- \l, \b - \h) node[xshift = -0.35cm, yshift = 0.15cm] {\color{spec} \footnotesize $-\rho$};
		\draw (- \l, \b - \h) -- (0, \b - 2*\h);
		\draw (0, \b - 2*\h) -- ( \l, \b - 3*\h) ;
		\draw (0, \b - 2*\h) -- (- \l, \b - 3*\h) node[xshift = -0.35cm, yshift = 0.15cm] {\color{spec} \footnotesize $-\rho$};
		\draw (- \l, \b - 3*\h) -- (0, \b - 4*\h);
		\draw ( \l, \b - 3*\h) -- (0, \b - 4*\h);
		\draw (0, \b - 4*\h) -- (- \l, \b - 5*\h) node[xshift = -0.35cm, yshift = 0.15cm] {\color{spec} \footnotesize $-\rho$};
		
		\draw[fill = black] (0, \b) circle (\r) node[xshift = -0.45cm, yshift = 0.05cm] {\color{spec} \footnotesize $2 \rho$};
		\draw[fill = black] (0, \b - 2*\h) node[xshift = -0.45cm] {\color{spec} \footnotesize $2\rho$} circle (\r);
		\draw[fill = black] (0, \b - 4*\h) node[xshift = -0.45cm] {\color{spec} \footnotesize $2\rho$} circle (\r);

		\draw[line width = 0.6pt] ( 2*\l + \t, \b + \hh) -- node[xshift = 0.4cm, yshift = 0.15cm] {\color{spec} \footnotesize $-\lambda$} ( 2*\l + \t, \b);
		\draw[line width = 0.6pt] ( 2*\l - \t,  \b + \hh) --  ( 2*\l - \t, \b);
		
		\draw ( 2*\l, \b) -- ( 2*\l + \l, \b - \h) node[xshift = 0.3cm, yshift = 0.2cm] {\color{spec} \footnotesize $-\lambda$};
		\draw ( 2*\l + \l, \b - \h) -- ( 2*\l, \b - 2*\h);
		\draw ( 2*\l, \b) -- ( 2*\l - \l, \b - \h);
		\draw ( 2*\l- \l, \b - \h) -- ( 2*\l, \b - 2*\h);
		\draw ( 2*\l, \b - 2*\h) -- ( 2*\l + \l, \b - 3*\h) node[xshift = 0.3cm, yshift = 0.2cm] {\color{spec} \footnotesize $-\lambda$};
		\draw ( 2*\l, \b - 2*\h) -- ( 2*\l- \l, \b - 3*\h);
		
		\draw[fill = black] ( 2*\l, \b) circle (\r) node[xshift = 0.45cm, yshift = 0.1cm] {\color{spec} \footnotesize $2 \lambda$};
		\draw[fill = black] ( 2*\l, \b - 2*\h) node[xshift = 0.45cm] {\color{spec} \footnotesize $2\lambda$} circle (\r);
		
		\draw[fill = black] ( \l, \b -\h) circle (\r) node[xshift = 0.75cm] {\color{spec} \footnotesize $- \lambda - \rho$};
		\draw[fill = black] ( \l, \b - 3*\h) circle (\r) node[xshift = 0.75cm] {\color{spec} \footnotesize $- \lambda - \rho$};

		
		\hypersetup{linkcolor=gray!80}
		\draw[arrow, gray!80] (\l + 0.5*\d - \dd, - 2*\h) -- node[yshift = 0.6cm, align=center] {\footnotesize see \\[-3pt] \footnotesize Fig.~\ref{fig:Qcomm}} (\l + 0.5*\d + \dd, - 2*\h);
		\hypersetup{linkcolor=black}
		\draw[arrow, gray!80] (\l + 0.5*\d + \dd, \a - 2*\h) -- (\l + 0.5*\d - \dd, \a - 2*\h);
		
		\draw[arrow, gray!80] (\d + \l, \a + 1.55*\hh + \dd) -- (\d + \l, \a + 1.55*\hh - \dd);
		\draw[arrow, gray!80] (\l, \b + 1.55*\hh + \dd) -- (\l, \b + 1.55*\hh - \dd);
	\end{tikzpicture}
	\vspace{0.2cm}
	\caption{Local relation between Baxter and raising operators} \label{fig:QrL}
\end{figure}
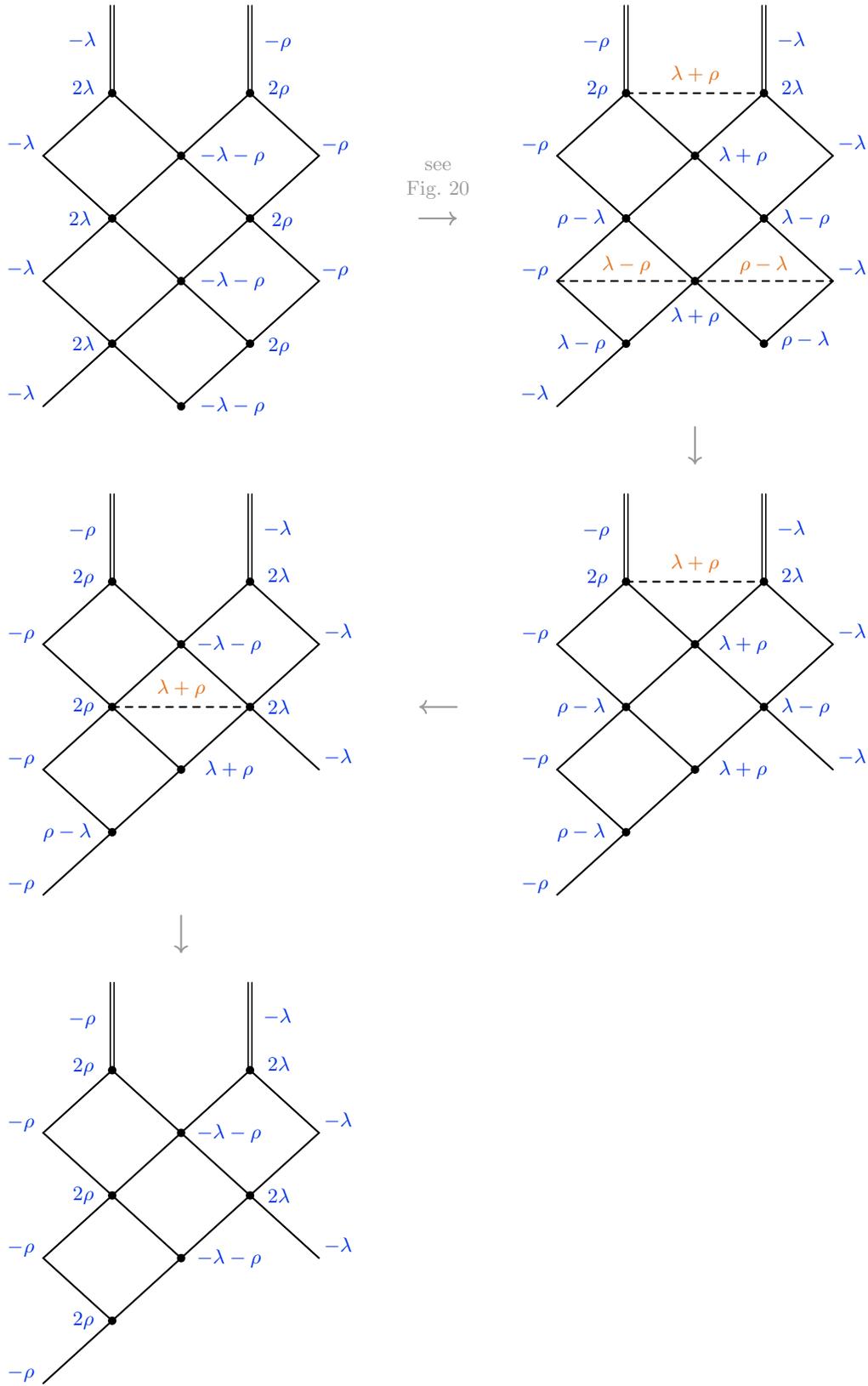

\subsection{Properties of wave functions}

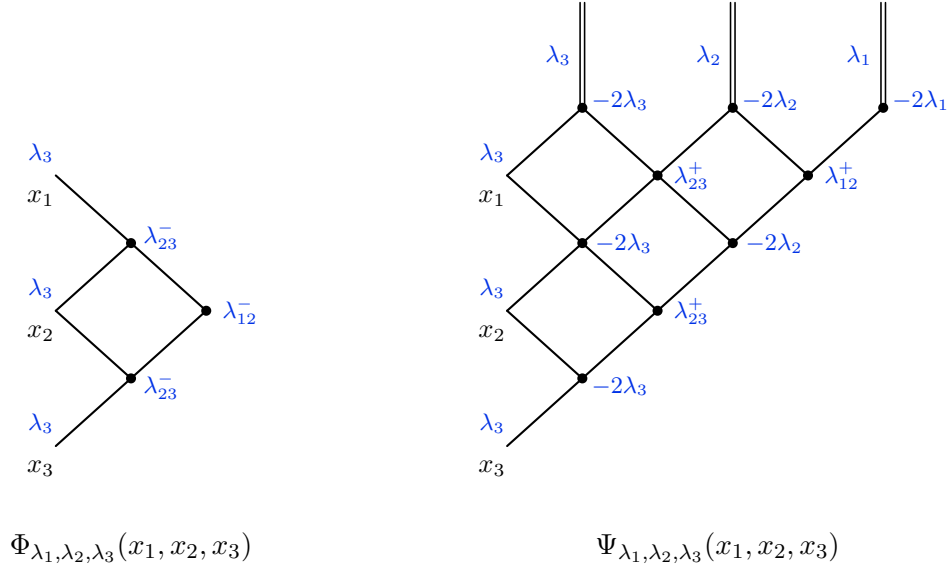
\begin{figure}[t] \centering
	\begin{tikzpicture}[thick, line cap = round]
		\def\l{1}
		\def\r{1.5pt}
		\def\h{0.9}
		\def\s{-0.9}
		\def\d{-7}
		\def\hh{1.4}
		\def\t{0.03}
		
		\draw[line width = 0.6pt] (\t, \hh) -- node[xshift = -0.35cm] {\color{spec} \footnotesize $\lambda_3$} (\t, 0);
		\draw[line width = 0.6pt] (-\t, \hh) --  (-\t, 0);
		
		\draw (0, 0) -- (\l, -\h)  node[xshift = 0.45cm, yshift = 0cm] {\color{spec} \footnotesize $\lambda_{23}^+$};
		\draw (\l, -\h) -- (0, -2*\h);
		\draw (0, 0) -- (-\l, -\h) node[xshift = -0.2cm, yshift = -0.3cm] {\small $x_1$} node[xshift = -0.2cm, yshift = 0.3cm] {\color{spec} \footnotesize $\lambda_3$};
		\draw (-\l, -\h) -- (0, -2*\h);
		\draw (0, -2*\h) -- (\l, -3*\h) node[xshift = 0.45cm, yshift = 0cm] {\color{spec} \footnotesize $\lambda_{23}^+$};
		\draw (0, -2*\h) -- (-\l, -3*\h) node[xshift = -0.2cm, yshift = -0.3cm] {\small $x_2$} node[xshift = -0.2cm, yshift = 0.3cm] {\color{spec} \footnotesize $\lambda_3$};
		\draw (-\l, -3*\h) -- (0, -4*\h);
		\draw (\l, -3*\h) -- (0, -4*\h);
		\draw (0, -4*\h) -- (-\l, -5*\h) node[xshift = -0.2cm, yshift = -0.3cm] {\small $x_3$} node[xshift = -0.2cm, yshift = 0.3cm] {\color{spec} \footnotesize $\lambda_3$};
		
		\draw[line width = 0.6pt] (2*\l + \t, \hh) -- node[xshift = -0.35cm] {\color{spec} \footnotesize $\lambda_2$} (2*\l + \t, 0);
		\draw[line width = 0.6pt] (2*\l -\t, \hh) --  (2*\l-\t, 0);
		
		\draw (2*\l, 0) -- (\l, -\h);
		\draw (2*\l, 0) -- (3*\l, -\h) node[xshift = 0.45cm, yshift = 0cm] {\color{spec} \footnotesize $\lambda_{12}^+$};
		\draw (\l, -\h) -- (2*\l, -2*\h);
		\draw (3*\l, -\h) -- (2*\l, -2*\h);
		\draw (2*\l, -2*\h) -- (\l, -3*\h);
		
		\draw[line width = 0.6pt] (4*\l + \t, \hh) -- node[xshift = -0.35cm] {\color{spec} \footnotesize $\lambda_1$} (4*\l + \t, 0);
		\draw[line width = 0.6pt] (4*\l -\t, \hh) --  (4*\l-\t, 0);
		
		\draw (4*\l, 0) -- (3*\l, -\h);
		
		\draw[fill = black] (0, 0) circle (\r) node[xshift = 0.5cm, yshift = 0.1cm] {\color{spec} \footnotesize $- 2 \lambda_3$};
		\draw[fill = black] (0, -2*\h) node[xshift = 0.55cm] {\color{spec} \footnotesize $-2\lambda_3$} circle (\r);
		\draw[fill = black] (0, -4*\h) node[xshift = 0.5cm, yshift = -0.1cm] {\color{spec} \footnotesize $- 2 \lambda_3$} circle (\r);
		
		\draw[fill = black] (\l, -\h) circle (\r);
		\draw[fill = black] (\l, -3*\h) circle (\r);
		
		\draw[fill = black] (2*\l, 0) circle (\r) node[xshift = 0.5cm, yshift = 0.1cm] {\color{spec} \footnotesize $- 2 \lambda_2$};
		\draw[fill = black] (2*\l, -2*\h) circle (\r) node[xshift = 0.55cm, yshift = 0cm] {\color{spec} \footnotesize $- 2 \lambda_2$};
		
		\draw[fill = black] (3*\l, -\h) circle (\r);
		
		\draw[fill = black] (4*\l, 0) circle (\r) node[xshift = 0.5cm, yshift = 0.1cm] {\color{spec} \footnotesize $- 2 \lambda_1$};
		
		\node (Psi) at (1.8*\l, -6.5*\h) {$\Psi_{\lambda_1, \lambda_2, \lambda_3}(x_1, x_2, x_3)$};

		
		\draw (\d, \s) node[xshift = -0.2cm, yshift = -0.3cm] {\small $x_1$} node[xshift = -0.2cm, yshift = 0.3cm] {\color{spec} \footnotesize $\lambda_3$} -- (\d + \l, -\h + \s) node[xshift = 0.4cm, yshift = 0.1cm] {\color{spec} \footnotesize $\lambda_{23}^-$};
		\draw (\d + \l, -\h + \s) -- (\d, -2*\h + \s);
		\draw (\d, -2*\h + \s) node[xshift = -0.2cm, yshift = -0.3cm] {\small $x_2$} node[xshift = -0.2cm, yshift = 0.3cm] {\color{spec} \footnotesize $\lambda_3$} -- (\d + \l, -3*\h + \s) node[xshift = 0.4cm, yshift = -0.1cm] {\color{spec} \footnotesize $\lambda_{23}^-$};
		\draw (\d + \l, -3*\h + \s) -- (\d, -4*\h + \s) node[xshift = -0.2cm, yshift = -0.3cm] {\small $x_3$} node[xshift = -0.2cm, yshift = 0.3cm] {\color{spec} \footnotesize $\lambda_3$};
		
		\draw (\d + \l, -\h + \s) -- (\d + 2*\l, -2*\h + \s) node[xshift = 0.45cm, yshift = 0cm] {\color{spec} \footnotesize $\lambda_{12}^-$};
		\draw (\d + 2*\l, -2*\h + \s) -- (\d + \l, -3*\h + \s);
		
		\draw[fill = black] (\d + \l, -\h + \s) circle (\r);
		\draw[fill = black] (\d + \l, -3*\h + \s) circle (\r);
		\draw[fill = black] (\d + 2*\l, -2*\h + \s) circle (\r);
		
		\node (Phi) at (\l + \d, -6.5*\h) {$\Phi_{\lambda_1, \lambda_2, \lambda_3}(x_1, x_2, x_3)$};
		
	\end{tikzpicture}
	\vspace{0.2cm}
	\caption{Wave functions of $GL$ and $BC$ systems ($\lambda_{jk}^{\pm} = \lambda_j \pm \lambda_k$)} \label{fig:eigen} \vspace{0.1cm}
\end{figure}

The wave functions of $BC$ Toda chain are given by iterative integrals
\begin{align}\label{Psi-GG}
	\Psi_{\bm{\lambda}_n}(\bm{x}_n) =  \LLambda_n(\lambda_n) \cdots \LLambda_1(\lambda_1) \cdot 1,
\end{align} 
where we assume $\bm{\lambda}_n \in \mathbb{R}^n$. This multiple integral is easy to visualize using diagrams, see Figure~\ref{fig:eigen}, where for comparison we also depict wave functions of $GL$ Toda chain defined by~\eqref{Phi-GG-0}.

\begin{theorem} \label{thm:Psi-sym}
	Let $\sigma \in S_n$, $\bm{\epsilon}_n \in \{1, -1\}^n$. Then
	\begin{align}
		\Psi_{\l_1,\ldots, \l_n}(\bm{x}_n) = \Psi_{\ve_1\l_{\sigma(1)},\ldots, \ve_n\l_{\sigma(n)}}(\bm{x}_n).
	\end{align}
\end{theorem}  

\begin{proof}
	By Proposition~\ref{prop:LQQr-space}, the raising operators respect the spaces of polynomially bounded continuous functions
	\begin{align}
		\LLambda_k(\lambda_k)  \colon \; \mathcal{P}_{k - 1} \; \to \; \mathcal{P}_k, \qquad \lambda_k \in \mathbb{R},
	\end{align}
	which implies the convergence of the integral~\eqref{Psi-GG}. Besides, by Propositions~\ref{prop:Lrefl},~\ref{prop:QL-rel}, the raising operators satisfy the identities
	\begin{align} \label{Lprop}
		\LLambda_k(\lambda) = \LLambda_k(-\lambda), \qquad \LLambda_{k + 1}(\lambda) \, \LLambda_{k}(\rho) = \LLambda_{k + 1}(\rho) \, \LLambda_{k}(\lambda)
	\end{align}
	that hold on $\mathcal{P}_{k - 1}$. This leads to the claim.
\end{proof}

\begin{corollary} \label{cor:Psi-real}
	Let $\bm{\lambda}_n \in \mathbb{R}^n$. Then $\Psi_{\bm{\lambda}_n}(\bm{x}_n) \in \mathbb{R}$.
\end{corollary}

\begin{proof}
	By definition~\eqref{Lker}, for $\lambda \in \mathbb{R}$
	\begin{align}
		\overline{\LLambda_\lambda(\bm{x}_n| \bm{y}_{n - 1}) }= \LLambda_{-\lambda}(\bm{x}_n| \bm{y}_{n - 1}) ,
	\end{align}
	and consequently, for $\bm{\lambda}_n \in \mathbb{R}^n$
	\begin{align}
		\overline{\Psi_{\bm{\lambda}_n}(\bm{x}_n)} = \Psi_{-\bm{\lambda}_n}(\bm{x}_n).
	\end{align}
	Hence, by Theorem~\ref{thm:Psi-sym} the wave function is real.
\end{proof}

\begin{theorem} \label{theorem2.2}
	Let $\bm{\lambda}_n \in \mathbb{R}^n$. Then
	\begin{align}
		& \QQ_n(\lambda) \, \Psi_{\bm{\lambda}_n}(\bm{x}_n) = \frac{(2\beta)^{-\imath \lambda} \, \Gamma(2\imath \lambda) }{\Gamma(g + \imath \lambda)} \, \prod_{j = 1}^n \Gamma(\imath \lambda \pm \imath \lambda_j) \, \Psi_{\bm{\lambda}_n}(\bm{x}_n),  &&  \Im \lambda \in (-g, 0), \\[6pt] \label{QrPsi}
		& \QQr_n(\lambda) \, \Psi_{\bm{\lambda}_n}(\bm{x}_n) = \frac{(2\beta)^{-\imath \lambda}}{\Gamma(g + \imath \lambda)} \,  \prod_{j = 1}^n  \Gamma(\imath \lambda \pm \imath \lambda_j) \, \Psi_{\bm{\lambda}_n}(\bm{x}_n), && \Im \lambda < 0.
	\end{align}
\end{theorem}

\begin{proof}
	By Proposition~\ref{prop:QL-rel}, Baxter operator and its reduced version satisfy the relations (which hold on $\mathcal{P}_{n - 1}$)
	\begin{align} \label{QL-rel}
		& \QQ_n(\lambda) \, \LLambda_n(\rho) = \Gamma(\imath \lambda \pm \imath \rho) \, \LLambda_n(\rho) \, \QQ_{n - 1}(\lambda), && \hspace{-1cm} \Im \lambda \in (-g, 0), \quad \rho \in \mathbb{R},  \\[6pt]
		& \QQr_n(\lambda) \, \LLambda_n(\rho) = \Gamma(\imath \lambda \pm \imath \rho) \, \LLambda_n(\rho) \, \QQr_{n - 1}(\lambda), && \hspace{-1cm} \Im \lambda < 0, \quad \rho \in \mathbb{R},
	\end{align}
	where in particular
	\begin{align}
		\QQ_0(\lambda) = \frac{(2\beta)^{-\imath \lambda} \, \Gamma(2\imath \lambda) }{\Gamma(g + \imath \lambda)} \, \mathrm{Id}, \qquad \QQr_0(\lambda) = \frac{(2\beta)^{-\imath \lambda}}{\Gamma(g + \imath \lambda)} \, \mathrm{Id}.
	\end{align}
	Thus, from iterative representation~\eqref{Psi-GG} we obtain the claimed formulas.
\end{proof}

\subsection{Local relation between $GL$ and $BC$ operators} \label{sec:gl-bc-rel}

For any integral operator $K$
\begin{align}
	[K \, \phi](\bm{x}_n) = \int_{\mathbb{R}^m} d\bm{y}_m \; K(\bm{x}_n|\bm{y}_m) \, \phi(\bm{y}_m)
\end{align}
denote its transpose $K^t$ by the formula
\begin{align} \label{t-op-def}
	[K^t \, \phi](\bm{x}_m) = \int_{\mathbb{R}^n} d\bm{y}_n \; K(\bm{y}_n|\bm{x}_m) \, \phi(\bm{y}_n).
\end{align}
The following identity is the key ingredient in the calculation of the scalar product between $GL$ and $BC$ wave functions, performed in the next section.

\begin{proposition} \label{prop:GL-BC-rel}
	Let $n \geq 2$, $\Im \lambda< 0$ and $\rho \in \mathbb{R}$. Then the relation
	\begin{multline}\label{QL-ABC}
		Q^t_{n - 1}(\lambda) \, \Lambda^t_{n}(- \lambda) \, \exp(- \beta e^{-x_1}) \, \LLambda_n(\rho) \\[6pt]
		= \Gamma(\imath \lambda \pm \imath \rho) \; Q^{t}_{n - 1}(\rho) \, Q^t_{n - 1}(-\rho) \, \exp(-\beta e^{-x_1}) \, \QQr_{n - 1}(\lambda)
	\end{multline}
	holds as the equality between corresponding kernels.
\end{proposition}

\begin{proof}
	The proof of the identity~\eqref{QL-ABC} in terms of diagrams is shown in Figures~\ref{fig:ABCrel},~\ref{fig:ABCrel2}. In Section~\ref{sec:ker-prod} we prove that the kernel of the left hand side~\eqref{QL-ABC}, which corresponds to the first diagram, is absolutely convergent. By the same arguments, as in the proof of Proposition~\ref{prop:QL-rel}, all of the rest diagrams also correspond to absolutely convergent integrals. 
	
	As before, the pictures are for the particular case $n = 3$, while generalization to arbitrary $n$ is straightforward. In all diagrams we indicate coordinate $\ln \beta$, which comes from the function $ \exp(- \beta e^{-x_1}) $ in~\eqref{QL-ABC}. 
	
	The first step is to use chain relation (Figure~\ref{fig:chain}) for the vertex at the bottom and obtain one dashed line. Then one moves this dashed line to the top using cross relations (Figure~\ref{fig:cross}). After that we repeat the same procedure for the vertex at the bottom from the right.
	
	Doing so we arrive at the \textit{fifth} diagram with two dashed lines at the top. The next step is to use identity from Figure~\ref{fig:flip-lim-tr}, which removes one dashed line and changes some spectral parameters. In particular, the parameter of double line is changed $\rho \to \lambda$. Notice that for this identity it is crucial that the top dashed line connects double line with the coordinate~$\ln\beta$.
	
	On the next step, shown in Figure~\ref{fig:ABCrel2}, we use the identity pictured in Figure~\ref{fig:QQcomm}. It represents commutativity of $GL$ Toda $Q$-operators (see~\eqref{Q-ker},~\eqref{QQ-gl}). This step enables us to again use cross relation (Figure~\ref{fig:cross}) to move dashed line to the bottom. 
	
	Passing to the last diagram we use reduced cross relation (Figure~\ref{fig:cross}), so that the dashed line disappears, and we arrive at the kernel of operator from the right hand side of identity~\eqref{QL-ABC}. Collecting factors appearing in all above transformations we obtain the coefficient $\Gamma(\imath \lambda \pm \imath \rho )$.
\end{proof}

\begin{remark}
	With the help of \cite[Corollary 2]{BDK} together with bounds from Appendix~\ref{sec:GG-bounds} one can show that the identity~\eqref{QL-ABC} holds on $\mathcal{P}_{n - 1}$.
\end{remark}

\begin{figure}[H] \centering \vspace{-1cm}
	\begin{tikzpicture}[thick, line cap = round, scale = 0.98]
		\def\l{1.1}
		\def\r{1.5pt}
		\def\h{1}
		\def\d{8.2}
		\def\dd{0.3}
		\def\hh{1.4}
		\def\t{0.03}
		\def\a{-7.8}
		\def\b{-15.6}

		
		\draw (0, 0 ) node[xshift = -0.4cm, yshift = 0.2cm] {\footnotesize $\ln \beta$} -- ( \l, - \h);
		\draw ( \l, - \h) -- (0, - 2*\h);
		\draw (0, - 2*\h) -- ( \l, - 3*\h) ;
		\draw (0, - 2*\h) -- (- \l, - 3*\h) node[xshift = -0.4cm, yshift = 0.1cm] {\color{spec} \footnotesize $-\lambda$};
		\draw (- \l, - 3*\h) -- (0, - 4*\h);
		\draw ( \l, - 3*\h) -- (0, - 4*\h);
		\draw (0, - 4*\h) -- (- \l, -5*\h) node[xshift = -0.4cm, yshift = 0.1cm] {\color{spec} \footnotesize $-\lambda$};
		\draw (0, - 4*\h) -- (\l, -5*\h);
		
		\draw[fill = black] (0, - 2*\h) node[xshift = -0.45cm, yshift = 0.1cm] {\color{spec} \footnotesize $2\lambda$} circle (\r);
		\draw[fill = black] (0, - 4*\h) node[xshift = -0.5cm] {\color{spec} \footnotesize $2 \lambda$} circle (\r);

		\draw[line width = 0.6pt] (2*\l + \t, \hh) -- node[xshift = 0.3cm, yshift = 0.1cm] {\color{spec} \footnotesize $\rho$} (2*\l + \t, 0);
		\draw[line width = 0.6pt] (2*\l - \t,  \hh) --  (2*\l - \t, 0);
		
		\draw (2*\l, 0) -- (2*\l + \l, - \h) node[xshift = 0.25cm, yshift = 0.05cm] {\color{spec} \footnotesize $\rho$};
		\draw (2*\l + \l, - \h) -- (2*\l, - 2*\h);
		\draw (2*\l, 0) -- (2*\l - \l, - \h);
		\draw (2*\l- \l, - \h) -- (2*\l, - 2*\h);
		\draw (2*\l, - 2*\h) -- (2*\l + \l, - 3*\h) node[xshift = 0.25cm, yshift = 0.05cm] {\color{spec} \footnotesize $\rho$};
		\draw (2*\l, - 2*\h) -- (2*\l- \l, - 3*\h);
		\draw (2*\l - \l, - 3*\h) -- (2*\l, - 4*\h);
		\draw (2*\l + \l, - 3*\h) -- (2*\l, - 4*\h);
		\draw (2*\l, - 4*\h) -- (2*\l - \l, -5*\h);
		
		\draw[fill = black] (2*\l, 0) circle (\r) node[xshift = 0.5cm, yshift = 0.05cm] {\color{spec} \footnotesize $-2 \rho$};
		\draw[fill = black] (2*\l, - 2*\h) node[xshift = 0.5cm] {\color{spec} \footnotesize $-2\rho$} circle (\r);
		\draw[fill = black] (2*\l, - 4*\h) node[xshift = 0.5cm, yshift = -0.05cm] {\color{spec} \footnotesize $-2 \rho$} circle (\r);
		
		\draw[fill = black] (\l, -\h) circle (\r) node[xshift = 0.7cm] {\color{spec} \footnotesize $\rho - \lambda$};
		\draw[fill = black] (\l, -3*\h) circle (\r) node[xshift = 0.7cm] {\color{spec} \footnotesize $\rho - \lambda$};
		\draw[fill = black] (\l, -5*\h) circle (\r) node[xshift = 0.55cm, yshift = -0.15cm] {\color{spec} \footnotesize $\rho - \lambda$};

		
		\draw (\d, 0) node[xshift = -0.4cm, yshift = 0.2cm] {\footnotesize $\ln \beta$} -- (\d + \l, - \h);
		\draw (\d + \l, - \h) -- (\d, - 2*\h);
		\draw (\d, - 2*\h) -- (\d + \l, - 3*\h) ;
		\draw (\d, - 2*\h) -- (\d - \l, - 3*\h) node[xshift = -0.4cm, yshift = 0.1cm] {\color{spec} \footnotesize $-\lambda$};
		\draw (\d - \l, - 3*\h) -- (\d, - 4*\h);
		\draw (\d + \l, - 3*\h) -- (\d, - 4*\h);
		\draw (\d, - 4*\h) -- (\d - \l, -5*\h) node[xshift = -0.4cm, yshift = 0.1cm] {\color{spec} \footnotesize $-\lambda$};
		
		\draw[fill = black] (\d, - 2*\h) node[xshift = -0.45cm, yshift = 0.1cm] {\color{spec} \footnotesize $2\lambda$} circle (\r);
		\draw[fill = black] (\d, - 4*\h) node[xshift = -0.5cm] {\color{spec} \footnotesize $2 \lambda$} circle (\r);

		\draw[line width = 0.6pt] (\d + 2*\l + \t, \hh) -- node[xshift = 0.3cm, yshift = 0.1cm] {\color{spec} \footnotesize $\rho$} (\d + 2*\l + \t, 0);
		\draw[line width = 0.6pt] (\d + 2*\l - \t,  \hh) --  (\d + 2*\l - \t, 0);
		
		\draw (\d + 2*\l, 0) -- (\d + 2*\l + \l, - \h) node[xshift = 0.25cm, yshift = 0.05cm] {\color{spec} \footnotesize $\rho$};
		\draw (\d + 2*\l + \l, - \h) -- (\d + 2*\l, - 2*\h);
		\draw (\d + 2*\l, 0) -- (\d + 2*\l - \l, - \h);
		\draw (\d + 2*\l- \l, - \h) -- (\d + 2*\l, - 2*\h);
		\draw (\d + 2*\l, - 2*\h) -- (\d + 2*\l + \l, - 3*\h) node[xshift = 0.25cm, yshift = 0.05cm] {\color{spec} \footnotesize $\rho$};
		\draw (\d + 2*\l, - 2*\h) -- (\d + 2*\l- \l, - 3*\h);
		\draw (\d + 2*\l - \l, - 3*\h) -- (\d + 2*\l, - 4*\h);
		\draw (\d + 2*\l + \l, - 3*\h) -- (\d + 2*\l, - 4*\h);
		
		\draw[fill = black] (\d + 2*\l, 0) circle (\r) node[xshift = 0.5cm, yshift = 0.05cm] {\color{spec} \footnotesize $-2 \rho$};
		\draw[fill = black] (\d + 2*\l, - 2*\h) node[xshift = 0.5cm] {\color{spec} \footnotesize $-2\rho$} circle (\r);
		\draw[fill = black] (\d + 2*\l, - 4*\h) node[xshift = 0.5cm, yshift = -0.05cm] {\color{spec} \footnotesize $-2 \rho$} circle (\r);
		
		\draw[fill = black] (\d + \l, -\h) circle (\r) node[xshift = 0.7cm] {\color{spec} \footnotesize $\rho - \lambda$};
		\draw[fill = black] (\d + \l, -3*\h) circle (\r) node[xshift = 0.7cm] {\color{spec} \footnotesize $\rho - \lambda$};
		\draw[dashed] (\d, - 4*\h) -- node[below] {\color{spec2} \footnotesize $\lambda - \rho$} (\d + 2*\l, - 4*\h);

		
		\draw (\d, \a) node[xshift = -0.4cm, yshift = 0.2cm] {\footnotesize $\ln \beta$} -- (\d + \l, \a - \h);
		\draw (\d + \l, \a - \h) -- (\d, \a - 2*\h);
		\draw (\d, \a - 2*\h) -- (\d + \l, \a - 3*\h) ;
		\draw (\d, \a - 2*\h) -- (\d - \l, \a - 3*\h) node[xshift = -0.4cm, yshift = 0.1cm] {\color{spec} \footnotesize $-\lambda$};
		\draw (\d - \l, \a - 3*\h) -- (\d, \a - 4*\h);
		\draw (\d + \l, \a - 3*\h) -- (\d, \a - 4*\h);
		\draw (\d, \a - 4*\h) -- (\d - \l, \a - 5*\h) node[xshift = -0.4cm, yshift = 0.1cm] {\color{spec} \footnotesize $-\lambda$};
		
		\draw[fill = black] (\d, \a - 2*\h) node[xshift = -0.65cm, yshift = 0.1cm] {\color{spec} \footnotesize $\lambda + \rho$} circle (\r);
		\draw[fill = black] (\d, \a - 4*\h) node[xshift = -0.7cm] {\color{spec} \footnotesize $ \lambda + \rho$} circle (\r);

		\draw[line width = 0.6pt] (\d + 2*\l + \t, \a + \hh) -- node[xshift = 0.3cm, yshift = 0.1cm] {\color{spec} \footnotesize $\rho$} (\d + 2*\l + \t, \a);
		\draw[line width = 0.6pt] (\d + 2*\l - \t,  \a + \hh) --  (\d + 2*\l - \t, \a);
		
		\draw (\d + 2*\l, \a) -- (\d + 2*\l + \l, \a - \h) node[xshift = 0.25cm, yshift = 0.05cm] {\color{spec} \footnotesize $\rho$};
		\draw (\d + 2*\l + \l, \a - \h) -- (\d + 2*\l, \a - 2*\h);
		\draw (\d + 2*\l, \a) -- (\d + 2*\l - \l, \a - \h);
		\draw (\d + 2*\l- \l, \a - \h) -- (\d + 2*\l, \a - 2*\h);
		\draw (\d + 2*\l, \a - 2*\h) -- (\d + 2*\l + \l, \a - 3*\h) node[xshift = 0.25cm, yshift = 0.05cm] {\color{spec} \footnotesize $\rho$};
		\draw (\d + 2*\l, \a - 2*\h) -- (\d + 2*\l- \l, \a - 3*\h);
		\draw (\d + 2*\l - \l, \a - 3*\h) -- (\d + 2*\l, \a - 4*\h);
		\draw (\d + 2*\l + \l, \a - 3*\h) -- (\d + 2*\l, \a - 4*\h);
		
		\draw[fill = black] (\d + 2*\l, \a) circle (\r) node[xshift = 0.5cm, yshift = 0.05cm] {\color{spec} \footnotesize $-2\rho$};
		\draw[fill = black] (\d + 2*\l, \a - 2*\h) node[xshift = 0.75cm] {\color{spec} \footnotesize $-\lambda - \rho$} circle (\r);
		\draw[fill = black] (\d + 2*\l, \a - 4*\h) node[xshift = 0.7cm, yshift = -0.1cm] {\color{spec} \footnotesize $-\lambda - \rho$} circle (\r);
		
		\draw[fill = black] (\d + \l, \a - \h) circle (\r) node[xshift = 0.75cm] {\color{spec} \footnotesize $ \lambda - \rho$};
		\draw[fill = black] (\d + \l, \a - 3*\h) circle (\r) node[xshift = 0.75cm] {\color{spec} \footnotesize $\lambda - \rho$};
		\draw[dashed] (\d, \a) -- node[above] {\color{spec2} \footnotesize $\lambda - \rho$} (\d + 2*\l, \a);

		
		\draw (0, \a) node[xshift = -0.4cm, yshift = 0.2cm] {\footnotesize $\ln \beta$} -- (\l, \a - \h);
		\draw ( \l, \a - \h) -- (0, \a - 2*\h);
		\draw (0, \a - 2*\h) -- ( \l, \a - 3*\h) ;
		\draw (0, \a - 2*\h) -- (- \l, \a - 3*\h) node[xshift = -0.4cm, yshift = 0.1cm] {\color{spec} \footnotesize $- \lambda$};
		\draw (- \l, \a - 3*\h) -- (0, \a - 4*\h);
		\draw ( \l, \a - 3*\h) -- (0, \a - 4*\h);
		\draw (0, \a - 4*\h) -- (- \l, \a - 5*\h) node[xshift = -0.4cm, yshift = 0.1cm] {\color{spec} \footnotesize $- \lambda$};
		
		\draw[fill = black] (0, \a - 2*\h) node[xshift = -0.65cm, yshift = 0.1cm] {\color{spec} \footnotesize $\lambda + \rho$} circle (\r);
		\draw[fill = black] (0, \a - 4*\h) node[xshift = -0.7cm] {\color{spec} \footnotesize $\lambda + \rho$} circle (\r);

		\draw[line width = 0.6pt] ( 2*\l + \t, \a + \hh) -- node[xshift = 0.3cm, yshift = 0.1cm] {\color{spec} \footnotesize $\rho$} ( 2*\l + \t, \a);
		\draw[line width = 0.6pt] ( 2*\l - \t,  \a + \hh) --  ( 2*\l - \t, \a);
		
		\draw ( 2*\l, \a) -- ( 2*\l + \l, \a - \h) node[xshift = 0.25cm, yshift = 0.05cm] {\color{spec} \footnotesize $\rho$};
		\draw ( 2*\l + \l, \a - \h) -- ( 2*\l, \a - 2*\h);
		\draw ( 2*\l, \a) -- ( 2*\l - \l, \a - \h);
		\draw ( 2*\l- \l, \a - \h) -- ( 2*\l, \a - 2*\h);
		\draw ( 2*\l, \a - 2*\h) -- ( 2*\l + \l, \a - 3*\h) node[xshift = 0.25cm, yshift = 0.05cm] {\color{spec} \footnotesize $\rho$};
		\draw ( 2*\l, \a - 2*\h) -- ( 2*\l- \l, \a - 3*\h);
		
		\draw[fill = black] ( 2*\l, \a) circle (\r) node[xshift = 0.5cm, yshift = 0.05cm] {\color{spec} \footnotesize $-2\rho$};
		\draw[fill = black] ( 2*\l, \a - 2*\h) node[xshift = 0.75cm] {\color{spec} \footnotesize $- \lambda - \rho$} circle (\r);
		\draw[dashed] ( 2*\l - \l, \a - 3*\h) -- node[below] {\color{spec2} \footnotesize $\lambda + \rho$} ( 2*\l + \l, \a - 3*\h);		
		
		\draw[fill = black] ( \l, \a - \h) circle (\r) node[xshift = -0.7cm] {\color{spec} \footnotesize $\lambda - \rho$};
		\draw[fill = black] ( \l, \a - 3*\h) circle (\r) node[xshift = -0.7cm] {\color{spec} \footnotesize $\lambda - \rho$};
		\draw[dashed] (0, \a) -- node[above] {\color{spec2} \footnotesize $\lambda - \rho$} ( 2*\l, \a);

		
		\draw (0, \b) node[xshift = -0.4cm, yshift = 0.2cm] {\footnotesize $\ln \beta$} -- (\l, \b - \h);
		\draw ( \l, \b - \h) -- (0, \b - 2*\h);
		\draw (0, \b - 2*\h) -- ( \l, \b - 3*\h) ;
		\draw (0, \b - 2*\h) -- (- \l, \b - 3*\h) node[xshift = -0.4cm, yshift = 0.1cm] {\color{spec} \footnotesize $-\lambda$};
		\draw (- \l, \b - 3*\h) -- (0, \b - 4*\h);
		\draw ( \l, \b - 3*\h) -- (0, \b - 4*\h);
		\draw (0, \b - 4*\h) -- (- \l, \b - 5*\h) node[xshift = -0.4cm, yshift = 0.1cm] {\color{spec} \footnotesize $-\lambda$};
		
		\draw[fill = black] (0, \b - 2*\h) node[xshift = -0.65cm, yshift = 0.1cm] {\color{spec} \footnotesize $\lambda + \rho$} circle (\r);
		\draw[fill = black] (0, \b - 4*\h) node[xshift = -0.7cm] {\color{spec} \footnotesize $\lambda + \rho$} circle (\r);

		\draw[line width = 0.6pt] ( 2*\l + \t, \b + \hh) -- node[xshift = 0.3cm, yshift = 0.1cm] {\color{spec} \footnotesize $\rho$} ( 2*\l + \t, \b);
		\draw[line width = 0.6pt] ( 2*\l - \t,  \b + \hh) --  ( 2*\l - \t, \b);
		
		\draw ( 2*\l, \b) -- ( 2*\l + \l, \b - \h) node[xshift = 0.25cm, yshift = 0.05cm] {\color{spec} \footnotesize $\rho$};
		\draw ( 2*\l + \l, \b - \h) -- ( 2*\l, \b - 2*\h);
		\draw ( 2*\l, \b) -- ( 2*\l - \l, \b - \h);
		\draw ( 2*\l- \l, \b - \h) -- ( 2*\l, \b - 2*\h);
		\draw ( 2*\l, \b - 2*\h) -- ( 2*\l + \l, \b - 3*\h) node[xshift = 0.3cm, yshift = 0.1cm] {\color{spec} \footnotesize $-\lambda$};
		\draw ( 2*\l, \b - 2*\h) -- ( 2*\l- \l, \b - 3*\h);
		
		\draw[fill = black] ( 2*\l, \b) circle (\r) node[xshift = 0.5cm, yshift = 0.05cm] {\color{spec} \footnotesize $-2\rho$};
		\draw[fill = black] ( 2*\l, \b - 2*\h) node[xshift = 0.7cm] {\color{spec} \footnotesize $\lambda + \rho$} circle (\r);
		\draw[dashed] ( 2*\l - \l, \b - 1*\h) -- node[below] {\color{spec2} \footnotesize $\lambda + \rho$} ( 2*\l + \l, \b - 1*\h);		
		
		\draw[fill = black] ( \l, \b - \h) circle (\r) node[xshift = -0.7cm] {\color{spec} \footnotesize $ \lambda - \rho$};
		\draw[fill = black] ( \l, \b - 3*\h) circle (\r) node[xshift = -0.6cm] {\color{spec} \footnotesize $- 2\rho$};
		\draw[dashed] (0, \b) -- node[above] {\color{spec2} \footnotesize $\lambda - \rho$} ( 2*\l, \b);

		
		\draw (\d, \b) node[xshift = -0.4cm, yshift = 0.2cm] {\footnotesize $\ln \beta$} -- (\d + \l, \b - \h);
		\draw ( \d + \l, \b - \h) -- (\d, \b - 2*\h);
		\draw (\d, \b - 2*\h) -- ( \d + \l, \b - 3*\h) ;
		\draw (\d, \b - 2*\h) -- (\d - \l, \b - 3*\h) node[xshift = -0.4cm, yshift = 0.1cm] {\color{spec} \footnotesize $-\lambda$};
		\draw (\d - \l, \b - 3*\h) -- (\d, \b - 4*\h);
		\draw ( \d + \l, \b - 3*\h) -- (\d, \b - 4*\h);
		\draw (\d, \b - 4*\h) -- (\d - \l, \b - 5*\h) node[xshift = -0.4cm, yshift = 0.1cm] {\color{spec} \footnotesize $-\lambda$};
		
		\draw[fill = black] (\d, \b - 2*\h) node[xshift = -0.65cm, yshift = 0.1cm] {\color{spec} \footnotesize $ \lambda + \rho$} circle (\r);
		\draw[fill = black] (\d, \b - 4*\h) node[xshift = -0.7cm] {\color{spec} \footnotesize $\lambda + \rho$} circle (\r);

		\draw[line width = 0.6pt] ( \d + 2*\l + \t, \b + \hh) -- node[xshift = 0.4cm, yshift = 0.1cm] {\color{spec} \footnotesize $-\lambda$} ( \d + 2*\l + \t, \b);
		\draw[line width = 0.6pt] ( \d + 2*\l - \t,  \b + \hh) --  ( \d + 2*\l - \t, \b);
		
		\draw ( \d + 2*\l, \b) -- ( \d + 2*\l + \l, \b - \h) node[xshift = 0.3cm, yshift = 0.1cm] {\color{spec} \footnotesize $-\lambda$};
		\draw ( \d + 2*\l + \l, \b - \h) -- ( \d + 2*\l, \b - 2*\h);
		\draw ( \d + 2*\l, \b) -- ( \d + 2*\l - \l, \b - \h);
		\draw ( \d + 2*\l- \l, \b - \h) -- ( \d + 2*\l, \b - 2*\h);
		\draw ( \d + 2*\l, \b - 2*\h) -- ( \d + 2*\l + \l, \b - 3*\h) node[xshift = 0.3cm, yshift = 0.1cm] {\color{spec} \footnotesize $-\lambda$};
		\draw ( \d + 2*\l, \b - 2*\h) -- ( \d + 2*\l- \l, \b - 3*\h);
		
		\draw[fill = black] ( \d + 2*\l, \b) circle (\r) node[xshift = 0.45cm, yshift = 0.05cm] {\color{spec} \footnotesize $2\lambda$};
		\draw[fill = black] ( \d + 2*\l, \b - 2*\h) node[xshift = 0.7cm] {\color{spec} \footnotesize $\lambda + \rho$} circle (\r);
		
		\draw[fill = black] (\d + \l, \b - \h) circle (\r) node[xshift = -0.6cm] {\color{spec} \footnotesize $- 2\rho$};
		\draw[fill = black] (\d + \l, \b - 3*\h) circle (\r) node[xshift = -0.6cm] {\color{spec} \footnotesize $- 2\rho$};
		\draw[dashed] (\d, \b) -- node[above] {\color{spec2} \footnotesize $\lambda - \rho$} ( \d + 2*\l, \b);

		
		\draw[arrow, gray!80] (\l + 0.5*\d - \dd, - 2*\h) -- (\l + 0.5*\d + \dd, - 2*\h);
		\draw[arrow, gray!80] (\l + 0.5*\d + \dd, \a - 2*\h) -- (\l + 0.5*\d - \dd, \a - 2*\h);
		\draw[arrow, gray!80] (\l + 0.5*\d - \dd, \b - 2*\h) -- (\l + 0.5*\d + \dd, \b - 2*\h);
		
		\draw[arrow, gray!80] (\d + \l, \a + 1.55*\hh + \dd) -- (\d + \l, \a + 1.55*\hh - \dd);
		\draw[arrow, gray!80] (\l, \b + 1.55*\hh + \dd) -- (\l, \b + 1.55*\hh - \dd);
	\end{tikzpicture}
	\vspace{0.2cm}
	\caption{Proof of identity~\eqref{QL-ABC} (continues in Figure~\ref{fig:ABCrel2})} \label{fig:ABCrel}
\end{figure}
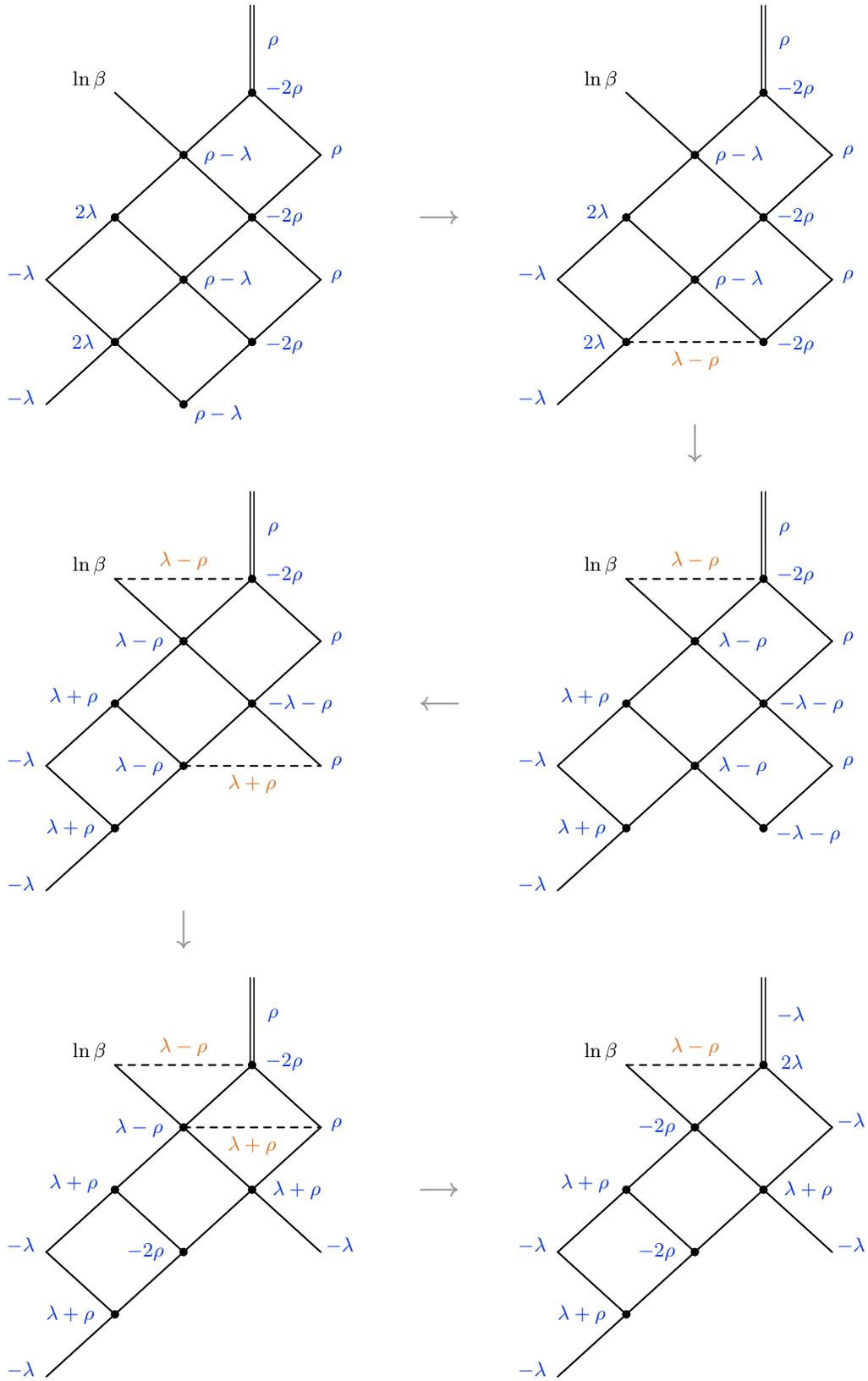

\newpage

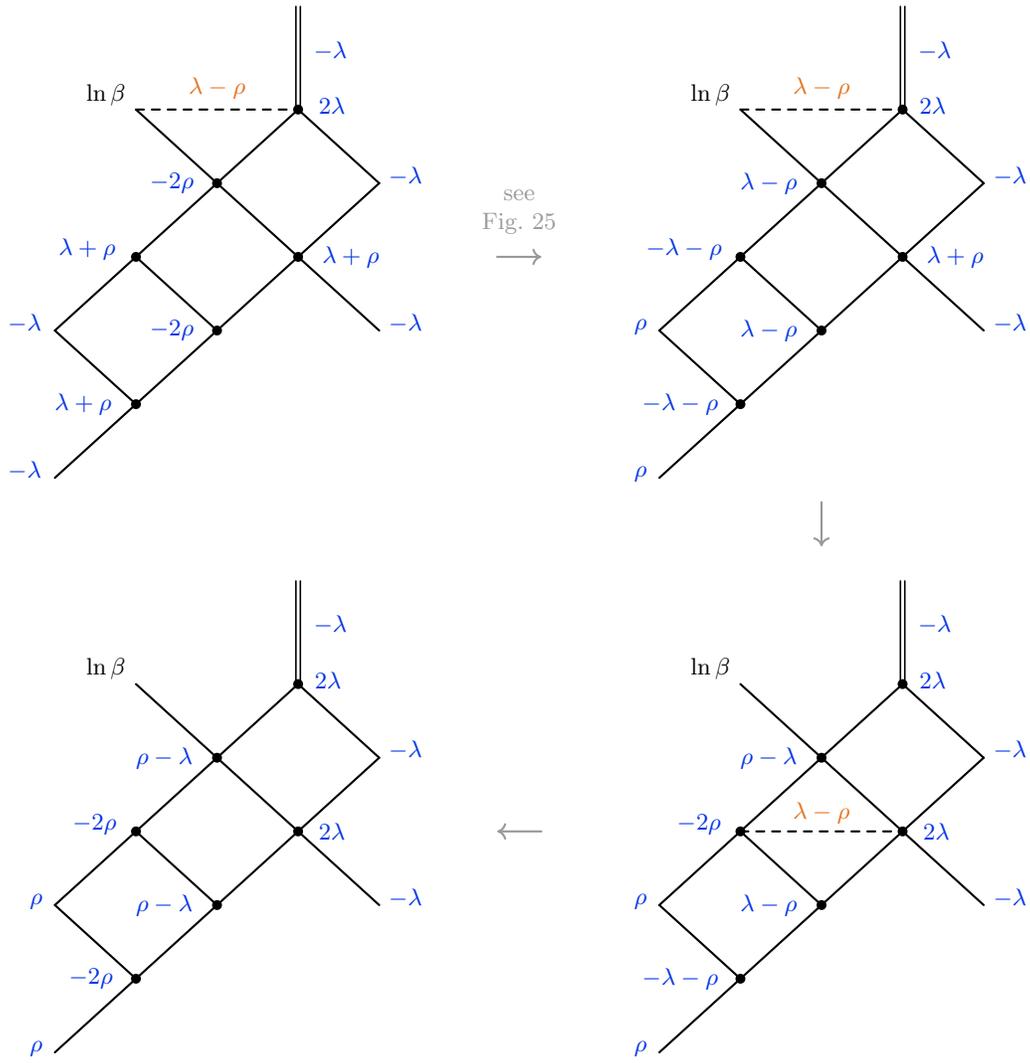
\begin{figure}[H] \centering \vspace{-1cm}
	\begin{tikzpicture}[thick, line cap = round, scale = 0.98]
		\def\l{1.1}
		\def\r{1.5pt}
		\def\h{1}
		\def\d{8.2}
		\def\dd{0.3}
		\def\hh{1.4}
		\def\t{0.03}
		\def\a{-7.8}
		\def\b{-15.6}

		
		\draw (0, 0) node[xshift = -0.4cm, yshift = 0.2cm] {\footnotesize $\ln \beta$} -- (\l, - \h);
		\draw (\l, - \h) -- (0, - 2*\h);
		\draw (0, - 2*\h) -- (  \l, - 3*\h) ;
		\draw (0, - 2*\h) -- (- \l, - 3*\h) node[xshift = -0.4cm, yshift = 0.1cm] {\color{spec} \footnotesize $-\lambda$};
		\draw (- \l, - 3*\h) -- (0, - 4*\h);
		\draw (  \l, - 3*\h) -- (0, - 4*\h);
		\draw (0, - 4*\h) -- (- \l, - 5*\h) node[xshift = -0.4cm, yshift = 0.1cm] {\color{spec} \footnotesize $-\lambda$};
		
		\draw[fill = black] (0, - 2*\h) node[xshift = -0.65cm, yshift = 0.1cm] {\color{spec} \footnotesize $\lambda + \rho$} circle (\r);
		\draw[fill = black] (0, - 4*\h) node[xshift = -0.7cm] {\color{spec} \footnotesize $\lambda + \rho$} circle (\r);

		\draw[line width = 0.6pt] (  2*\l + \t, + \hh) -- node[xshift = 0.4cm, yshift = 0.1cm] {\color{spec} \footnotesize $-\lambda$} (  2*\l + \t, 0);
		\draw[line width = 0.6pt] (  2*\l - \t,  + \hh) --  (  2*\l - \t, 0);
		
		\draw (  2*\l, 0) -- (  2*\l + \l, - \h) node[xshift = 0.35cm, yshift = 0.1cm] {\color{spec} \footnotesize $-\lambda$};
		\draw (  2*\l + \l, - \h) -- (  2*\l, - 2*\h);
		\draw (  2*\l, 0) -- (  2*\l - \l, - \h);
		\draw (  2*\l- \l, - \h) -- (  2*\l, - 2*\h);
		\draw (  2*\l, - 2*\h) -- (  2*\l + \l, - 3*\h) node[xshift = 0.35cm, yshift = 0.1cm] {\color{spec} \footnotesize $-\lambda$};
		\draw (  2*\l, - 2*\h) -- (  2*\l- \l, - 3*\h);
		
		\draw[fill = black] ( 2*\l, 0) circle (\r) node[xshift = 0.45cm, yshift = 0.05cm] {\color{spec} \footnotesize $2\lambda$};
		\draw[fill = black] (2*\l, - 2*\h) node[xshift = 0.7cm] {\color{spec} \footnotesize $\lambda + \rho$} circle (\r);
		
		\draw[fill = black] (\l, - \h) circle (\r) node[xshift = -0.6cm] {\color{spec} \footnotesize $- 2\rho$};
		\draw[fill = black] (\l, - 3*\h) circle (\r) node[xshift = -0.6cm] {\color{spec} \footnotesize $- 2\rho$};
		\draw[dashed] (0, 0) -- node[above] {\color{spec2} \footnotesize $\lambda - \rho$} ( 2*\l, 0);

		
		\draw (\d, 0) node[xshift = -0.4cm, yshift = 0.2cm] {\footnotesize $\ln \beta$} -- (\d + \l, - \h);
		\draw ( \d + \l, - \h) -- (\d, - 2*\h);
		\draw (\d, - 2*\h) -- ( \d + \l, - 3*\h) ;
		\draw (\d, - 2*\h) -- (\d - \l, - 3*\h) node[xshift = -0.25cm, yshift = 0.05cm] {\color{spec} \footnotesize $\rho$};
		\draw (\d - \l, - 3*\h) -- (\d, - 4*\h);
		\draw ( \d + \l, - 3*\h) -- (\d, - 4*\h);
		\draw (\d, - 4*\h) -- (\d - \l, - 5*\h) node[xshift = -0.25cm, yshift = 0.05cm] {\color{spec} \footnotesize $\rho$};
		
		\draw[fill = black] (\d, - 2*\h) node[xshift = -0.75cm, yshift = 0.1cm] {\color{spec} \footnotesize $-\lambda - \rho$} circle (\r);
		\draw[fill = black] (\d, - 4*\h) node[xshift = -0.8cm] {\color{spec} \footnotesize $-\lambda - \rho$} circle (\r);

		\draw[line width = 0.6pt] ( \d + 2*\l + \t, + \hh) -- node[xshift = 0.4cm, yshift = 0.1cm] {\color{spec} \footnotesize $-\lambda$} ( \d + 2*\l + \t, 0);
		\draw[line width = 0.6pt] ( \d + 2*\l - \t,  + \hh) --  ( \d + 2*\l - \t, 0);
		
		\draw ( \d + 2*\l, 0) -- ( \d + 2*\l + \l, - \h) node[xshift = 0.35cm, yshift = 0.1cm] {\color{spec} \footnotesize $-\lambda$};
		\draw ( \d + 2*\l + \l, - \h) -- ( \d + 2*\l, - 2*\h);
		\draw ( \d + 2*\l, 0) -- ( \d + 2*\l - \l, - \h);
		\draw ( \d + 2*\l- \l, - \h) -- ( \d + 2*\l, - 2*\h);
		\draw ( \d + 2*\l, - 2*\h) -- ( \d + 2*\l + \l, - 3*\h) node[xshift = 0.35cm, yshift = 0.1cm] {\color{spec} \footnotesize $-\lambda$};
		\draw ( \d + 2*\l, - 2*\h) -- ( \d + 2*\l- \l, - 3*\h);
		
		\draw[fill = black] ( \d + 2*\l, 0) circle (\r) node[xshift = 0.4cm, yshift = 0.05cm] {\color{spec} \footnotesize $2\lambda$};
		\draw[fill = black] ( \d + 2*\l, - 2*\h) node[xshift = 0.7cm] {\color{spec} \footnotesize $\lambda + \rho$} circle (\r);
		
		\draw[fill = black] (\d + \l, - \h) circle (\r) node[xshift = -0.7cm] {\color{spec} \footnotesize $\lambda - \rho$};
		\draw[fill = black] (\d + \l, - 3*\h) circle (\r) node[xshift = -0.7cm] {\color{spec} \footnotesize $\lambda - \rho$};
		\draw[dashed] (\d, 0) -- node[above] {\color{spec2} \footnotesize $\lambda - \rho$} ( \d + 2*\l, 0);

		
		\draw (\d, \a) node[xshift = -0.4cm, yshift = 0.2cm] {\footnotesize $\ln \beta$} -- (\d + \l, \a - \h);
		\draw ( \d + \l, \a - \h) -- (\d, \a - 2*\h);
		\draw (\d, \a - 2*\h) -- ( \d + \l, \a - 3*\h) ;
		\draw (\d, \a - 2*\h) -- (\d - \l, \a - 3*\h) node[xshift = -0.25cm, yshift = 0.05cm] {\color{spec} \footnotesize $\rho$};
		\draw (\d - \l, \a - 3*\h) -- (\d, \a - 4*\h);
		\draw ( \d + \l, \a - 3*\h) -- (\d, \a - 4*\h);
		\draw (\d, \a - 4*\h) -- (\d - \l, \a - 5*\h) node[xshift = -0.25cm, yshift = 0.05cm] {\color{spec} \footnotesize $\rho$};
		
		\draw[fill = black] (\d, \a - 2*\h) node[xshift = -0.55cm, yshift = 0.1cm] {\color{spec} \footnotesize $-2 \rho$} circle (\r);
		\draw[fill = black] (\d, \a - 4*\h) node[xshift = -0.8cm] {\color{spec} \footnotesize $-\lambda - \rho$} circle (\r);

		\draw[line width = 0.6pt] ( \d + 2*\l + \t, \a + \hh) -- node[xshift = 0.4cm, yshift = 0.1cm] {\color{spec} \footnotesize $-\lambda$} ( \d + 2*\l + \t, \a);
		\draw[line width = 0.6pt] ( \d + 2*\l - \t,  \a + \hh) --  ( \d + 2*\l - \t, \a);
		
		\draw ( \d + 2*\l, \a) -- ( \d + 2*\l + \l, \a - \h) node[xshift = 0.35cm, yshift = 0.1cm] {\color{spec} \footnotesize $-\lambda$};
		\draw ( \d + 2*\l + \l, \a - \h) -- ( \d + 2*\l, \a - 2*\h);
		\draw ( \d + 2*\l, \a) -- ( \d + 2*\l - \l, \a - \h);
		\draw ( \d + 2*\l- \l, \a - \h) -- ( \d + 2*\l, \a - 2*\h);
		\draw ( \d + 2*\l, \a - 2*\h) -- ( \d + 2*\l + \l, \a - 3*\h) node[xshift = 0.35cm, yshift = 0.1cm] {\color{spec} \footnotesize $-\lambda$};
		\draw ( \d + 2*\l, \a - 2*\h) -- ( \d + 2*\l- \l, \a - 3*\h);
		
		\draw[fill = black] ( \d + 2*\l, \a) circle (\r) node[xshift = 0.4cm, yshift = 0.05cm] {\color{spec} \footnotesize $2\lambda$};
		\draw[fill = black] ( \d + 2*\l, \a - 2*\h) node[xshift = 0.45cm] {\color{spec} \footnotesize $2\lambda$} circle (\r);
		
		\draw[fill = black] (\d + \l, \a - \h) circle (\r) node[xshift = -0.7cm] {\color{spec} \footnotesize $\rho - \lambda$};
		\draw[fill = black] (\d + \l, \a - 3*\h) circle (\r) node[xshift = -0.7cm] {\color{spec} \footnotesize $\lambda - \rho$};
		\draw[dashed] (\d, \a - 2*\h) -- node[yshift = 0.25cm] {\color{spec2} \footnotesize $\lambda - \rho$} ( \d + 2*\l, \a - 2*\h);

		
		\draw (0, \a) node[xshift = -0.4cm, yshift = 0.2cm] {\footnotesize $\ln \beta$} -- (+ \l, \a - \h);
		\draw ( + \l, \a - \h) -- (0, \a - 2*\h);
		\draw (0, \a - 2*\h) -- ( + \l, \a - 3*\h) ;
		\draw (0, \a - 2*\h) -- (- \l, \a - 3*\h) node[xshift = -0.25cm, yshift = 0.05cm] {\color{spec} \footnotesize $\rho$};
		\draw (- \l, \a - 3*\h) -- (0, \a - 4*\h);
		\draw ( + \l, \a - 3*\h) -- (0, \a - 4*\h);
		\draw (0, \a - 4*\h) -- (- \l, \a - 5*\h) node[xshift = -0.25cm, yshift = 0.05cm] {\color{spec} \footnotesize $\rho$};
		
		\draw[fill = black] (0, \a - 2*\h) node[xshift = -0.55cm, yshift = 0.1cm] {\color{spec} \footnotesize $-2 \rho$} circle (\r);
		\draw[fill = black] (0, \a - 4*\h) node[xshift = -0.6cm] {\color{spec} \footnotesize $-2 \rho$} circle (\r);

		\draw[line width = 0.6pt] ( + 2*\l + \t, \a + \hh) -- node[xshift = 0.4cm, yshift = 0.1cm] {\color{spec} \footnotesize $-\lambda$} ( + 2*\l + \t, \a);
		\draw[line width = 0.6pt] ( + 2*\l - \t,  \a + \hh) --  ( + 2*\l - \t, \a);
		
		\draw ( + 2*\l, \a) -- ( + 2*\l + \l, \a - \h) node[xshift = 0.35cm, yshift = 0.1cm] {\color{spec} \footnotesize $-\lambda$};
		\draw ( + 2*\l + \l, \a - \h) -- ( + 2*\l, \a - 2*\h);
		\draw ( + 2*\l, \a) -- ( + 2*\l - \l, \a - \h);
		\draw ( + 2*\l- \l, \a - \h) -- ( + 2*\l, \a - 2*\h);
		\draw ( + 2*\l, \a - 2*\h) -- ( + 2*\l + \l, \a - 3*\h) node[xshift = 0.35cm, yshift = 0.1cm] {\color{spec} \footnotesize $-\lambda$};
		\draw ( + 2*\l, \a - 2*\h) -- ( + 2*\l- \l, \a - 3*\h);
		
		\draw[fill = black] ( + 2*\l, \a) circle (\r) node[xshift = 0.4cm, yshift = 0.05cm] {\color{spec} \footnotesize $2\lambda$};
		\draw[fill = black] ( + 2*\l, \a - 2*\h) node[xshift = 0.45cm] {\color{spec} \footnotesize $2\lambda$} circle (\r);
		
		\draw[fill = black] (+ \l, \a - \h) circle (\r) node[xshift = -0.7cm] {\color{spec} \footnotesize $\rho - \lambda$};
		\draw[fill = black] (+ \l, \a - 3*\h) circle (\r) node[xshift = -0.7cm] {\color{spec} \footnotesize $\rho - \lambda$};

		
		\hypersetup{linkcolor=gray!80}
		\draw[arrow, gray!80] (\l + 0.5*\d - \dd, - 2*\h) -- node[yshift = 0.6cm, align=center] {\footnotesize see \\[-3pt] \footnotesize Fig.~\ref{fig:QQcomm}} (\l + 0.5*\d + \dd, - 2*\h);
		\hypersetup{linkcolor=black}
		\draw[arrow, gray!80] (\l + 0.5*\d + \dd, \a - 2*\h) -- (\l + 0.5*\d - \dd, \a - 2*\h);
		
		\draw[arrow, gray!80] (\d + \l, \a + 1.55*\hh + \dd) -- (\d + \l, \a + 1.55*\hh - \dd);
	\end{tikzpicture}
	\vspace{0.2cm}
	\caption{Proof of identity~\eqref{QL-ABC} (begins in Figure~\ref{fig:ABCrel})} \label{fig:ABCrel2}
\end{figure}

\begin{figure}[H] \centering \vspace{0.8cm}
	\begin{tikzpicture}[thick, line cap = round, scale = 0.98]
		\def\l{1.1}
		\def\r{1.5pt}
		\def\h{1}
		\def\d{6}
		\def\dd{0.3}
		\def\hh{1.4}
		\def\t{0.03}
		\def\a{-7.8}
		\def\b{-15.6}
		
		
		\draw (\l, - \h) node[xshift = -0.4cm, yshift = 0.1cm] {\color{spec} \footnotesize $-\rho$} -- (0, - 2*\h) ;
		\draw (0, - 2*\h) -- (  \l, - 3*\h) node[xshift = -0.45cm] {\color{spec} \footnotesize $-\rho$} node[xshift = 1.7cm] { \footnotesize $=$}  ;
		\draw (0, - 2*\h) -- (- \l, - 3*\h) node[xshift = -0.4cm, yshift = 0.1cm] {\color{spec} \footnotesize $-\lambda$};
		\draw (- \l, - 3*\h) -- (0, - 4*\h);
		\draw (  \l, - 3*\h) -- (0, - 4*\h);
		\draw (0, - 4*\h) -- (- \l, - 5*\h) node[xshift = -0.4cm, yshift = 0.1cm] {\color{spec} \footnotesize $-\lambda$};
		
		\draw[fill = black] (0, - 2*\h) node[xshift = -0.65cm, yshift = 0.1cm] {\color{spec} \footnotesize $\lambda + \rho$} circle (\r);
		\draw[fill = black] (0, - 4*\h) node[xshift = -0.7cm] {\color{spec} \footnotesize $ \lambda + \rho$} circle (\r);

		
		\draw (\d + \l, - \h) node[xshift = -0.35cm, yshift = 0.1cm] {\color{spec} \footnotesize $\lambda$} -- (\d, - 2*\h) ;
		\draw (\d, - 2*\h) -- (\d + \l, - 3*\h) node[xshift = -0.4cm] {\color{spec} \footnotesize $\lambda$} ;
		\draw (\d, - 2*\h) -- (\d - \l, - 3*\h) node[xshift = -0.25cm, yshift = 0.05cm] {\color{spec} \footnotesize $\rho$};
		\draw (\d - \l, - 3*\h) -- (\d, - 4*\h);
		\draw (\d + \l, - 3*\h) -- (\d, - 4*\h);
		\draw (\d, - 4*\h) -- (\d - \l, - 5*\h) node[xshift = -0.25cm, yshift = 0.05cm] {\color{spec} \footnotesize $\rho$};
		
		\draw[fill = black] (\d, - 2*\h) node[xshift = -0.75cm, yshift = 0.1cm] {\color{spec} \footnotesize $-\lambda - \rho$} circle (\r);
		\draw[fill = black] (\d, - 4*\h) node[xshift = -0.8cm] {\color{spec} \footnotesize $-\lambda - \rho$} circle (\r);

	\end{tikzpicture}
	\vspace{0.2cm}
	\caption{Commutativity of $GL$ Toda $Q$-operators} \label{fig:QQcomm}
\end{figure}
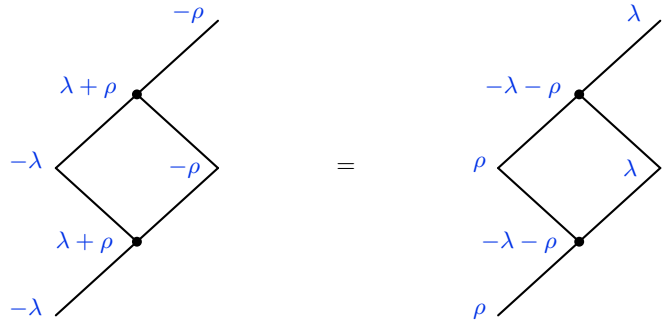

\subsection{Scalar product between $GL$ and $BC$ wave functions} \label{sec:gl-bc-prod}

The following proposition is needed to derive Mellin--Barnes representation of the $BC$ wave function.

\begin{proposition}\label{prop:gl-bc-prod}
	Let $\bm{\lambda}_n \in \mathbb{R}^n$, $\bm{\rho}_n \in \mathbb{C}^n$ such that $\Im \rho_1 = \ldots = \Im \rho_n < 0$. Then 
	\begin{align} \label{GG20}
		\int_{\mathbb{R}^n} d\bm{x}_n \; \Phi_{-\bm{\rho}_n}(\bm{x}_n)  \, \Psi_{\bm{\lambda}_n}(\bm{x}_n) \, \exp (- \beta e^{-x_1}) 
		=   \frac{(2\beta)^{-\imath \underline{\bm{\rho}}_n} \prod\limits_{j, k=1}^n \Gamma(\imath \rho_j \pm \imath \lambda_k) }{ \prod\limits_{1\leq j < k\leq n} \Gamma(\imath \rho_j + \imath \rho_k) \, \prod\limits_{j = 1}^n \Gamma (g + \imath \rho_j ) }.
	\end{align}
\end{proposition}

The proof is by induction over $n$. Denote the above integral as 
\begin{align} \label{Iint}
	I(\bm{\lambda}_n, \bm{\rho}_n) = \int_{\mathbb{R}^n} d\bm{x}_n \; \Phi_{-\bm{\rho}_n}(\bm{x}_n)  \, \Psi_{\bm{\lambda}_n}(\bm{x}_n) \, \exp (- \beta e^{-x_1}).
\end{align}
The following lemmas represent base case and induction step.

\begin{lemma}
	Let $\lambda \in \mathbb{R}$, $\Im \rho < 0$. Then
	\begin{align}
		I(\lambda, \rho) = \int_{\mathbb{R}} dx \; \Phi_{-\rho}(x) \, \Psi_{\lambda}(x) \, \exp (- \beta e^{-x}) = \frac{(2\beta)^{-\imath \rho} \, \Gamma(\imath \rho \pm \imath \lambda)}{\Gamma(g + \imath \rho)}.
	\end{align}
\end{lemma}

\begin{proof}
	Using Gauss--Givental representations for wave functions (see Example~\ref{ex:PhiPsi1}) we rewrite the above integral~as
	\begin{align} \label{sc-pr-1}
		\int_{\mathbb{R}} dx \; e^{- \imath \rho  x - \beta e^{-x}} \, \frac{(2\beta)^{\imath \lambda}}{\Gamma(g - \imath \lambda)} \, \int_{\ln \beta}^\infty dy \; (1 + \beta e^{-y})^{-\imath \lambda - g} \, (1 - \beta e^{-y})^{-\imath \lambda + g - 1 } \, e^{\imath \lambda (x - 2y) - e^{y - x}}.
	\end{align}
	Then it can be calculated with the help of diagrams, see Figure~\ref{fig:ABCrel1}. On the first step one uses chain relation (Figure~\ref{fig:chain}), and on the second step the identity from Figure~\ref{fig:flip-lim2}.
\end{proof}
	
	\begin{figure}[H] \centering \vspace{0cm}
		\begin{tikzpicture}[thick, line cap = round, scale = 0.98]
				\def\l{1.1}
				\def\r{1.5pt}
				\def\h{1}
				\def\d{7.6}
				\def\dd{0.3}
				\def\hh{1.4}
				\def\t{0.03}
				\def\a{-7.8}
				\def\b{-15.6}

				
				\draw (0, 0 ) node[xshift = -0.4cm, yshift = 0.2cm] {\footnotesize $\ln \beta$} -- ( \l, - \h);
				
				\draw[line width = 0.6pt] (2*\l + \t, \hh) -- node[xshift = 0.3cm, yshift = 0.1cm] {\color{spec} \footnotesize $\lambda$} (2*\l + \t, 0);
				\draw[line width = 0.6pt] (2*\l - \t,  \hh) --  (2*\l - \t, 0);
				
				\draw (2*\l, 0) -- (2*\l - \l, - \h);
				
				\draw[fill = black] (2*\l, 0) circle (\r) node[xshift = 0.5cm, yshift = 0.05cm] {\color{spec} \footnotesize $-2 \lambda$} node[xshift = 2.8cm] {$=$ \small $\quad \Gamma(\imath \rho - \imath \lambda)$};
				
				\draw[fill = black] (\l, -\h) circle (\r) node[xshift = 0.6cm, yshift = -0.15cm] {\color{spec} \footnotesize $\lambda - \rho$};

				
				\draw[dashed] (\d, 0) node[xshift = -0.4cm, yshift = 0.2cm] {\footnotesize $\ln \beta$} -- node[below] {\color{spec2} \footnotesize $\rho - \lambda$} (\d + 2*\l, 0);
				
				\draw[line width = 0.6pt] (\d + 2*\l + \t, \hh) -- node[xshift = 0.3cm, yshift = 0.1cm] {\color{spec} \footnotesize $\lambda$} (\d + 2*\l + \t, 0);
				\draw[line width = 0.6pt] (\d + 2*\l - \t,  \hh) --  (\d + 2*\l - \t, 0);
				
				\draw[fill = black] (\d + 2*\l, 0) circle (\r) node[xshift = 0.5cm, yshift = 0.05cm] {\color{spec} \footnotesize $-2 \lambda$} node[xshift = 3.2cm] {$=$ \small $\quad \Gamma(\imath \rho \pm \imath \lambda) \, \dfrac{ (2\beta)^{- \imath \rho}   }{ \Gamma (g + \imath \rho) }$};

			\end{tikzpicture}
		\vspace{0.1cm}
		\caption{Calculation of the integral~\eqref{sc-pr-1}} \label{fig:ABCrel1}
	\end{figure}
	
\begin{lemma}
	Let $\bm{\lambda}_n \in \mathbb{R}^n$, $\bm{\rho}_n \in \mathbb{C}^n$ such that $\Im \rho_1 = \ldots = \Im \rho_n < 0$. Then the integral~\eqref{Iint} satisfies the recurrence
	\begin{align} \label{Iint-rec}
		I(\bm{\lambda}_n, \bm{\rho}_n) = \frac{(2\beta)^{-\imath \rho_n} \, \Gamma(\imath \rho_n \pm \imath \lambda_n) \prod\limits_{j = 1}^{n - 1} \Gamma(\imath \rho_n \pm \imath \lambda_j) \, \Gamma(\imath \rho_j \pm \imath \lambda_n)}{\Gamma(g + \imath \rho_n) \prod\limits_{j = 1}^{n - 1} \Gamma(\imath \rho_j + \imath \rho_n)} \, I(\bm{\lambda}_{n - 1}, \bm{\rho}_{n - 1}).
	\end{align}
\end{lemma}

\begin{proof}
	Let us present main steps of the proof, postponing the justification of interchanging integrals and other convergence issues to the end. 
	
	First, rewrite the Gauss--Givental representation for $GL$ wave function~\eqref{Phi-GG-0}
	\begin{align}
		\Phi_{-\bm{\rho}_n}(\bm{x}_n) = \Lambda_n(-\rho_n) \, \Phi_{-\bm{\rho}_{n - 1}}(\bm{x}_{n - 1})
	\end{align}
	using the corresponding Baxter operator~\eqref{Q-Phi}
	\begin{align}\label{Phi-GG-Q}
		\begin{aligned}
			\Phi_{-\bm{\rho}_n}(\bm{x}_n)  &= \frac{1}{\prod\limits_{j = 1}^{n - 1} \Gamma(\imath \rho_n + \imath \rho_j)} \, \Lambda_n(- \rho_n) \, Q_{n - 1}(\rho_n) \, \Phi_{-\bm{\rho}_{n - 1}}(\bm{x}_{n - 1}) = \frac{1}{\prod\limits_{j = 1}^{n - 1} \Gamma(\imath \rho_n + \imath \rho_j)} \\[6pt]
			& \times \int_{\mathbb{R}^{n - 1}} d\bm{y}_{n - 1} \int_{\mathbb{R}^{n - 1}} d\bm{z}_{n - 1} \; \Lambda_{- \rho_n}(\bm{x}_n | \bm{y}_{n - 1}) \, Q_{\rho_n}(\bm{y}_{n - 1} | \bm{z}_{n - 1}) \, \Phi_{-\bm{\rho}_{n - 1}}(\bm{z}_{n - 1}).
		\end{aligned}
	\end{align}
	Besides, recall the Gauss--Givental representation of $BC$ wave function~\eqref{Psi-GG}
	\begin{align}
		\Psi_{\bm{\lambda}_n}(\bm{x}_n) = \LLambda_n(\lambda_n) \, \Psi_{\bm{\lambda}_{n - 1}}(\bm{x}_{n - 1}) = \int_{\mathbb{R}^{n - 1}} d\bm{s}_{n - 1} \; \LLambda_{\lambda_n}(\bm{x}_n | \bm{s}_{n - 1}) \, \Psi_{\bm{\lambda}_{n - 1}}(\bm{s}_{n - 1}) .
	\end{align}
	Inserting these formulas into the integral~\eqref{Iint} we obtain
	\begin{multline} \label{Iint2}
		I(\bm{\lambda}_n, \bm{\rho}_n) = \frac{1}{\prod\limits_{j = 1}^{n - 1} \Gamma(\imath \rho_n + \imath \rho_j)} \, \int_{\mathbb{R}^n } d\bm{x}_n  \int_{\mathbb{R}^{n - 1}} d\bm{y}_{n - 1} \int_{\mathbb{R}^{n - 1}} d\bm{z}_{n - 1} \int_{\mathbb{R}^{n - 1}} d\bm{s}_{n - 1} \; \exp(-\beta e^{-x_1}) \\[6pt]
		\times \Lambda_{- \rho_n}(\bm{x}_n | \bm{y}_{n - 1}) \, Q_{\rho_n}(\bm{y}_{n - 1} | \bm{z}_{n - 1})  \, \LLambda_{\lambda_n}(\bm{x}_n | \bm{s}_{n - 1}) \, \Phi_{-\bm{\rho}_{n - 1}}(\bm{z}_{n - 1}) \,\Psi_{\bm{\lambda}_{n - 1}}(\bm{s}_{n - 1}) .
	\end{multline}
	The integral over $\bm{x}_n$
	\begin{align}
		\int_{\mathbb{R}^n } d\bm{x}_n \; \exp(-\beta e^{-x_1}) \, \Lambda_{- \rho_n}(\bm{x}_n | \bm{y}_{n - 1}) \, Q_{\rho_n}(\bm{y}_{n - 1} | \bm{z}_{n - 1})  \, \LLambda_{\lambda_n}(\bm{x}_n | \bm{s}_{n - 1})
	\end{align}
	represents kernel of the operator product
	\begin{align}
		Q^t_{n - 1}(\rho_n) \, \Lambda^t_n(-\rho_n) \, \exp(-\beta e^{-x_1}) \, \LLambda_n(\lambda_n).
	\end{align}
	The definition of transposed operators is given in~\eqref{t-op-def}.
	
	By Proposition~\ref{prop:GL-BC-rel}, the relation
	\begin{multline}
		Q^t_{n - 1}(\rho_n) \, \Lambda^t_{n}(- \rho_n) \, \exp(- \beta e^{-x_1}) \, \LLambda_n(\lambda_n) \\[6pt]
		= \Gamma(\imath \rho_n \pm \imath \lambda_n) \; Q^{t}_{n - 1}(\lambda_n) \, Q^t_{n - 1}(-\lambda_n) \, \exp(-\beta e^{-x_1}) \, \QQr_{n - 1}(\rho_n)
	\end{multline}
	holds as the equality between corresponding kernels. Applying it to the integral~\eqref{Iint2} we arrive at
	\begin{multline}  \label{Iint3}
		I(\bm{\lambda}_n, \bm{\rho}_n) = \frac{\Gamma(\imath \rho_n \pm \imath \lambda_n)}{\prod\limits_{j = 1}^{n - 1} \Gamma(\imath \rho_n + \imath \rho_j)} \, \int_{\mathbb{R}^n } d\bm{x}_{n - 1}  \int_{\mathbb{R}^{n - 1}} d\bm{y}_{n - 1} \int_{\mathbb{R}^{n - 1}} d\bm{z}_{n - 1} \int_{\mathbb{R}^{n - 1}} d\bm{s}_{n - 1} \; \exp(-\beta e^{-x_1}) \\[6pt]
		\times  Q_{-\lambda_n}(\bm{x}_{n - 1} | \bm{y}_{n - 1}) \, Q_{\lambda_n}(\bm{y}_{n - 1} | \bm{z}_{n - 1}) \, \QQr_{\rho_n}(\bm{x}_{n - 1} | \bm{s}_{n - 1}) \, \Phi_{-\bm{\rho}_{n - 1}}(\bm{z}_{n - 1}) \,\Psi_{\bm{\lambda}_{n - 1}}(\bm{s}_{n - 1}) .
	\end{multline}
	Now the integral over $\bm{s}_{n - 1}$ can be calculated, since it represents action of $BC$ Baxter operator on its eigenfunction~\eqref{QrPsi}
	\begin{align} \label{QPsi-calc}
		\int_{\mathbb{R}^{n - 1}} d\bm{s}_{n - 1} \; \QQr_{\rho_n}(\bm{x}_{n - 1} | \bm{s}_{n - 1}) \, \Psi_{\bm{\lambda}_{n - 1}}(\bm{s}_{n - 1}) = \frac{(2\beta)^{- \imath \rho_n}}{\Gamma(g + \imath \rho_n)} \, \prod_{j = 1}^{n - 1} \Gamma(\imath \rho_n \pm \imath \lambda_j ) \, \Psi_{\bm{\lambda}_{n - 1}}(\bm{x}_{n - 1}).
	\end{align}
	Analogously for the integrals over $\bm{y}_{n - 1}$ and $\bm{z}_{n - 1}$ we use~\eqref{Q-Phi}
	\begin{multline} \label{QPhi-calc}
		\int_{\mathbb{R}^{n - 1}} d\bm{y}_{n - 1} \int_{\mathbb{R}^{n - 1}} d\bm{z}_{n - 1}  \; Q_{-\lambda_n}(\bm{x}_{n - 1} | \bm{y}_{n - 1}) \, Q_{\lambda_n}(\bm{y}_{n - 1} | \bm{z}_{n - 1}) \, \Phi_{-\bm{\rho}_{n - 1}}(\bm{z}_{n - 1}) \\
		= \prod_{j = 1}^{n - 1} \Gamma(\imath \rho_j \pm \imath \lambda_n) \, \Phi_{-\bm{\rho}_{n - 1}}(\bm{x}_{n - 1}).
	\end{multline}
	Inserting these results into~\eqref{Iint3} we are left with the integral over $\bm{x}_{n - 1}$, which gives the claimed relation~\eqref{Iint-rec}.
	
	It remain to justify all of the above manipulations. Denote 
	\begin{align}
		\Re \rho_j = \gamma_j \in \mathbb{R}, \qquad \Im \rho_j = - \epsilon < 0.
	\end{align}
	Then due to the property~\eqref{Phi-shift}
	\begin{align}
		\Phi_{- \bm{\rho}_n}(\bm{x}_n) = e^{- \epsilon \underline{\bm{x}}_n} \, \Phi_{-  \bm{\gamma}_{n}}(\bm{x}_n).
	\end{align}
	First, let us prove that the initial integral~\eqref{Iint2} is absolutely convergent. Combining Corollaries~\ref{cor:Phi-bound},~\ref{cor:Psi-Tspace} we have the estimate
	\begin{align}
		\bigl| e^{- \epsilon \underline{\bm{x}}_n - \beta e^{-x_1}} \, \Phi_{-\bm{\gamma}_n}(\bm{x}_n) \, \Psi_{\bm{\lambda}_n}(\bm{x}_n) \bigr| \leq C (1 + |\bm{x}_n|)^{-d}
	\end{align}
	for any $d \geq 0$, which implies convergence.
	
	Next, notice that the representation~\eqref{Phi-GG-Q} can be equivalently written as
	\begin{align}
		\Phi_{-\bm{\rho}_{n}}(\bm{x}_n) = \frac{1}{\prod\limits_{j = 1}^{n - 1} \Gamma(\imath \rho_n + \imath \rho_j)} \, e^{- \epsilon \underline{\bm{x}}_n} \, \Lambda_n(-\gamma_n) \, Q_{n - 1}(\gamma_n - 2 \imath \epsilon) \, \Phi_{-  \bm{\gamma}_{n - 1}}(\bm{x}_{n - 1}).
	\end{align}
	Since $\Phi_{-  \bm{\gamma}_{n - 1}}(\bm{x}_{n - 1}) \in \mathcal{P}_{n - 1}$ (Corollary~\ref{cor:Phi-bound}) and operators $Q_{n - 1}(\gamma_n - 2\imath \epsilon)$, $\Lambda_n(-\gamma_n)$ are well defined on $\mathcal{P}_{n - 1}$ (Proposition~\ref{prop:QL-A-space}), the above representation is convergent. 
	
	Now consider the multiple integral~\eqref{Iint2}. The integrand is estimated as
	\begin{multline}
		\bigl| \Lambda_{- \rho_n}(\bm{x}_n | \bm{y}_{n - 1}) \, Q_{\rho_n}(\bm{y}_{n - 1} | \bm{z}_{n - 1})  \, \LLambda_{\lambda_n}(\bm{x}_n | \bm{s}_{n - 1}) \, \Phi_{-\bm{\rho}_{n - 1}}(\bm{z}_{n - 1}) \,\Psi_{\bm{\lambda}_{n - 1}}(\bm{s}_{n - 1})  \bigr|  \\[6pt]
		\leq C(\bm{\lambda}_n) \, e^{- \epsilon \underline{\bm{x}}_n} \, \Lambda_{0}(\bm{x}_n | \bm{y}_{n - 1}) \, Q_{-2 \imath \epsilon}(\bm{y}_{n - 1} | \bm{z}_{n - 1})  \, \LLambda_{0}(\bm{x}_n | \bm{s}_{n - 1}) \, \Phi_{0, \dots, 0}(\bm{z}_{n - 1}) \,\Psi_{0, \dots, 0}(\bm{s}_{n - 1}).
	\end{multline}
	Hence, the integral is absolutely convergent at least in initial order, when it equals the original expression~\eqref{Iint} with $\bm{\lambda}_n = (0, \dots, 0)$ and $\bm{\rho}_n = -\imath \epsilon(1, \dots, 1)$, whose convergence is already established. Therefore, by Fubini--Tonelli theorem it converges in any order, and we can integrate over $\bm{x}_n$ first.
	
	Analogous arguments can be applied to the integral~\eqref{Iint3}: it is absolutely convergent in the order, when we evaluate integrals~\eqref{QPsi-calc},~\eqref{QPhi-calc} first, since in that case it boils down to the initial expression~\eqref{Iint} with decreased number of variables ($n \to n - 1$).
\end{proof}

\subsection{From Gauss--Givental to Mellin--Barnes} \label{sec:gg-mb-equiv}

In this section we prove that $BC$ Toda wave function defined by means of Gauss--Givental integral also admits Mellin--Barnes representation, which generalizes Iorgov--Shadura formula~\cite[(26)]{IS}. 

Recall $GL$ Toda spectral measure~\eqref{gl-measure-0}
\begin{align}
	\hat{\mu}(\bm{\lambda}_n) = \frac{1}{n! \, (2\pi)^n} \prod_{j \not= k} \frac{1}{\Gamma(\imath \lambda_j - \imath \lambda_k)}
\end{align}
and define the function 
\begin{align} \label{MB0}
	K_1( \bm{\lambda}_n,  \bm{\gamma}_n)  = \frac{(2\beta)^{-\imath \bgg_n} \prod\limits_{j, k=1}^n \Gamma(\imath \gamma_j \pm \imath \lambda_k) }{ \prod\limits_{1\leq j < k\leq n} \Gamma(\imath \gamma_j + \imath \gamma_k) \, \prod\limits_{j = 1}^n \Gamma (g + \imath \gamma_j ) }.
\end{align}

\begin{theorem}\label{theorem2.3}
	For $\bm{\lambda}_n \in \mathbb{R}^n$ and $\epsilon > 0$
	\begin{align} \label{Psi-MB}
		\Psi_{\bm{\lambda}_n}(\bm{x}_n)  = e^{\beta e^{-x_1}}  \int_{(\mathbb{R} - \imath \epsilon)^n}  d\bm{\gamma}_n \; \hat{\mu}(\bm{\gamma}_n) \, K_1(\bl_n,\bg_n) \, \Phi_{\bm{\gamma}_n}(\bm{x}_n).
	\end{align}
\end{theorem}

\begin{proof}
	The proof relies on Proposition~\ref{prop:gl-bc-prod} and inversion formula for $GL$ Toda wave function. Namely, it is known \cite{W} that smooth and sufficiently fast decaying functions $\phi(\bm{x}_n)$, which belong to the so called \textit{Whittaker Schwartz space}, can be expanded in terms of $GL$ Toda wave functions
	\begin{align}\label{inv-f}
		\phi(\bm{x}_n) = \int_{\mathbb{R}^n} d\bm{\gamma}_n \; \hat{\mu}(\bm{\gamma}_n) \, \Phi_{\bm{\gamma}_n}(\bm{x}_n) \, \int_{\mathbb{R}^n} d\bm{y}_n \; \overline{\Phi_{\bm{\gamma}_n}(\bm{y}_n)} \, \phi(\bm{y}_n).
	\end{align}
	Equivalently, $GL$ wave functions satisfy completeness relation~\eqref{gl-compl}. The definition of Whittaker Schwartz space is given in Section~\ref{sec:bc-bound}.
	
	The formula~\eqref{Psi-MB} can be written in the same spirit. Namely, using the relation~\eqref{Phi-shift}
	\begin{align}
		\Phi_{\bm{\gamma}_n - \imath \epsilon \bm{e}_n} (\bm{x}_n)= e^{\epsilon \underline{\bm{x}}_n} \, \Phi_{\bm{\gamma}_n}(\bm{x}_n), \qquad \bm{e}_n = (1, \dots, 1), 
	\end{align}
	and shifting integration variables $\gamma_j \to \gamma_j -\imath \epsilon$ we arrive at
	\begin{align}\label{Psi-MB2}
		e^{-\epsilon \underline{\bm{x}}_n - \beta e^{-x_1}} \, \Psi_{\bm{\lambda}_n}(\bm{x}_n) = \int_{\mathbb{R}^n}  d\bm{\gamma}_n \; \hat{\mu}(\bm{\gamma}_n) \, K_1(\bm{\lambda}_n, \bm{\gamma}_n - \imath \epsilon \bm{e}_n) \, \Phi_{\bm{\gamma}_n}(\bm{x}_n).
	\end{align}
	By Corollary~\ref{cor:Psi-Tspace}, the function $e^{-\epsilon \underline{\bm{x}}_n - \beta e^{-x_1}} \, \Psi_{\bm{\lambda}_n}(\bm{x}_n)$ belongs to Whittaker Schwartz space, and hence we are allowed to apply inversion formula~\eqref{inv-f} to it. 
	
	It is left to recall Proposition~\ref{prop:gl-bc-prod}, which states that the scalar product of this function with $GL$ Toda wave function coincides with the kernel~\eqref{MB0}. Namely, for $\bm{\gamma}_n \in \mathbb{R}^n$ we have
	\begin{align}
		\overline{\Phi_{\bm{\gamma}_n}(\bm{x}_n)} = \Phi_{-\bm{\gamma}_n}(\bm{x}_n)
	\end{align}
	see~\eqref{Lker-expl}, so that
	\begin{align}
		\begin{aligned}
			\int_{\mathbb{R}^n} d\bm{x}_n \; \overline{\Phi_{\bm{\gamma}_n}(\bm{x}_n)} \, e^{-\epsilon \underline{\bm{x}}_n - \beta e^{-x_1}} \, \Psi_{\bm{\lambda}_n}(\bm{x}_n) & = \int_{\mathbb{R}^n} d\bm{x}_n \; \Phi_{-\bm{\gamma}_n + \imath \epsilon \bm{e}_n}(\bm{x}_n) \, e^{- \beta e^{-x_1}} \, \Psi_{\bm{\lambda}_n}(\bm{x}_n) \\[6pt]
			& = K_1(\bm{\lambda}_n, \bm{\gamma}_n - \imath \epsilon \bm{e}_n).
		\end{aligned}
	\end{align}
\end{proof}

\subsection{Orthogonality} \label{sec:orth}

The Hamiltonians of $BC$ Toda chain are formally self-adjoint with respect to the standard scalar product
\begin{align}
	\langle \psi | \phi \rangle = \int_{\mathbb{R}^n} d\bm{x}_n \; \overline{\psi(\bm{x}_n)} \, \phi(\bm{x}_n).
\end{align}
In this section we present heuristic calculation of the scalar product between $BC$ wave functions
\begin{align} \label{sc-pr-bc}
	\langle \Psi_{\bm{\lambda}_n} | \Psi_{\bm{\rho}_n} \rangle = \frac{1}{\bcdmu(\bm{\lambda}_n) } \, \bdelta_{\mathrm{sym}}(\bm{\lambda}_n, \bm{\rho}_n), \qquad \bm{\lambda}_n, \bm{\rho}_n \in \mathbb{R}^n,
\end{align}
where we denote the measure
\begin{align}\label{bc-measure}
	\bcdmu(\bm{\lambda}_n) = \frac{1}{n! \, (4\pi)^n } \prod_{1 \leq j < k \leq n} \frac{1}{\Gamma(\pm \imath \lambda_j \pm \imath \lambda_k) } \, \prod_{j = 1}^n \frac{ \Gamma(g \pm \imath \lambda_j ) }{ \Gamma(\pm2 \imath \lambda_j) }
\end{align}
and delta function symmetric with respect to signed permutations
\begin{align}
	\bdelta_{\mathrm{sym}} (\bm{\lambda}_n, \bm{\rho}_n) = \frac{1}{n! \, 2^n} \sum_{\substack{\sigma \in S_n \\[2pt] \bm{\epsilon}_n \in \{1, -1\}^n}}\delta \bigl( \lambda_1 - \epsilon_1 \rho_{\sigma(1)} \bigr) \cdots \delta \bigl( \lambda_n - \epsilon_n \rho_{\sigma(n)} \bigr).
\end{align}
The calculation is very similar to the one for $GL$ Toda chain, see \cite[Section 4]{Sil}. Mind that most of manipulations below are formal, since we work with divergent integrals. Yet, in the case of $GL$ model the analogous derivation can be made rigorous by regularizing the corresponding integrals and working with suitable space of test functions, see \cite[Sections 1, 2]{Kozl}. We expect that the same justification can be made for the calculations below.

\paragraph{One particle.}

In the simplest case of one particle the scalar product can be calculated in a straightforward way using Gauss--Givental representation
\begin{align}
	\begin{aligned}
		\Psi_{\lambda}(x) & = \frac{(2\beta)^{\imath \lambda}}{\Gamma(g - \imath \lambda)} \, \int_{\ln\beta}^\infty dy \; e^{\imath \lambda(x - 2y) - e^{y - x}} \, (1 + \beta e^{-y})^{-\imath \lambda - g} \,  (1 - \beta e^{-y})^{-\imath \lambda + g - 1} \\[6pt]
		& = \frac{(2\beta)^{-\imath \lambda}}{\Gamma(g - \imath \lambda)} \, \int_0^\infty dt \; e^{\imath \lambda x - \beta (2t + 1) e^{-x}} \, (t + 1)^{-\imath \lambda - g} \, t^{-\imath \lambda + g - 1},
	\end{aligned}
\end{align}
where passing to the second line we changed variable $e^y = \beta (2t + 1)$. Inserting the last expression into the scalar product we have
\begin{multline}
	\langle \Psi_{\lambda} | \Psi_{\rho} \rangle = \frac{(2\beta)^{\imath \lambda}}{\Gamma(g + \imath \lambda)} \, \frac{(2\beta)^{-\imath \rho}}{\Gamma(g - \imath \rho)} \\[6pt]
	\times  \int_{\mathbb{R}} dx \, \int_0^\infty dt \, \int_0^\infty ds \; e^{\imath (\rho - \lambda)x - 2\beta(t + s + 1) e^{-x}} \, (t + 1)^{\imath \lambda - g} \, t^{\imath \lambda + g - 1} \, (s + 1)^{-\imath \rho - g} \, s^{-\imath \rho + g - 1}.
\end{multline}
Changing the variable $x \to x + \ln (2\beta(t + s + 1))$ we calculate the integral over $x$ 
\begin{align}
	\int_{\mathbb{R}} dx \; e^{\imath (\rho - \lambda)x - 2\beta(t + s + 1) e^{-x}} = (2\beta(t + s + 1))^{\imath(\rho - \lambda)} \, \Gamma(\imath \lambda - \imath \rho),
\end{align}
where for the right hand side to make sense we assume $\lambda \not= \rho$. In the remaining double integral we change the integration variable $s \to (t + 1)s$
\begin{multline}
	\int_0^\infty dt \, \int_0^\infty ds \; (t + s + 1)^{\imath(\rho - \lambda)}  \, (t + 1)^{\imath \lambda - g} \, t^{\imath \lambda + g - 1} \, (s + 1)^{-\imath \rho - g} \, s^{-\imath \rho + g - 1} \\[6pt]
	= \int_0^\infty dt \, \int_0^\infty ds \; t^{\imath \lambda + g - 1} \, (s + 1)^{\imath (\rho - \lambda)} \, (st + s + 1)^{-\imath \rho - g} \, s^{-\imath \rho + g  - 1}.
\end{multline}
After another change of integration variable $t \to (s + 1) t/s$ the integrals separate and can be evaluated explicitly
\begin{multline}
	\int_0^\infty ds \, \int_0^\infty dt \; t^{\imath \lambda + g - 1} \, (s + 1)^{\imath (\rho - \lambda)} \, (st + s + 1)^{-\imath \rho - g} \, s^{-\imath \rho + g  - 1} \\[6pt]
	= 	\int_0^\infty ds \, \int_0^\infty dt \; t^{\imath \lambda + g - 1} \, (t + 1)^{-\imath \rho - g} \, s^{- \imath (\lambda + \rho) - 1} = \frac{\Gamma(g + \imath \lambda) \, \Gamma(\imath \rho - \imath \lambda)}{\Gamma(g + \imath \rho)} \, 2\pi \delta(\lambda + \rho),
\end{multline}
where we used 
\begin{align}
\int_0^\infty ds \, s^{- \imath (\lambda + \rho) - 1} = 
\int_{-\infty}^\infty dp \, e^{- \imath p(\lambda + \rho)} = 
2\pi \delta(\lambda + \rho).
\end{align}

Collecting everything together we arrive at
\begin{align}
	\langle \Psi_{\lambda} | \Psi_{\rho} \rangle = 2\pi \frac{\Gamma(\pm 2\imath \lambda)}{\Gamma(g \pm \imath \lambda)} \, \delta(\lambda + \rho),
\end{align}
where we assume $\lambda \not= \rho$. Since wave functions are even in spectral parameters, the general formula can be restored by symmetry
\begin{align}\label{sc-pr-bc-1}
	\langle \Psi_{\lambda} | \Psi_{\rho} \rangle = 4\pi \frac{\Gamma(\pm 2\imath \lambda)}{\Gamma(g \pm \imath \lambda)} \, \frac{\delta(\lambda + \rho) + \delta(\lambda - \rho)}{2}.
\end{align}
It coincides with the claimed expression~\eqref{sc-pr-bc}.

\paragraph{Many particles.}

Using Gauss--Givental recursive formula we can rewrite the scalar product in the form
\begin{align}\label{sc-pr-LL}
	\langle \Psi_{\bm{\lambda}_n} | \Psi_{\bm{\rho}_n} \rangle = \LLambda_1^t(\lambda_n) \cdots \LLambda_n^t(\lambda_1) \, \LLambda_n(\rho_n) \cdots \LLambda_1(\rho_1) \cdot 1,
\end{align}
the transposed operators are defined in~\eqref{t-op-def}. The calculation is based on the local relation 
\begin{align} \label{LtL}
	\LLambda_n^t(\lambda) \, \LLambda_n(\rho) = \Gamma( \pm \imath \lambda \pm \imath \rho) \, \LLambda_{n - 1}(\rho) \, \LLambda_{n - 1}^t(\lambda), \qquad n \geq 2.
\end{align}
Let us postpone its proof to the end of this section. Applying it many times we have
\begin{align}
	\LLambda_n^t(\lambda_1) \, \LLambda_n(\rho_n) \cdots \LLambda_1(\rho_1) \cdot 1 = \prod_{k = 2}^{n} \Gamma( \pm \imath \lambda_1 \pm \imath \rho_k) \, \LLambda_{n - 1}(\rho_n) \cdots \LLambda_1(\rho_2) \, \LLambda_1^t(\lambda_1) \, \LLambda_1(\rho_1) \cdot 1.
\end{align}
Here from the right we have the scalar product of one particle wave functions
\begin{align}
	\LLambda_1^t(\lambda_1) \, \LLambda_1(\rho_1) \cdot 1 = \langle \Psi_{\lambda_1} | \Psi_{\rho_1} \rangle,
\end{align}
while the rest part represents wave function for $n - 1$ particles
\begin{align}
	\LLambda_n^t(\lambda_1) \, \LLambda_n(\rho_n) \cdots \LLambda_1(\rho_1) \cdot 1 = \langle \Psi_{\lambda_1} | \Psi_{\rho_1} \rangle \prod_{k = 2}^{n} \Gamma( \pm \imath \lambda_1 \pm \imath \rho_k) \, \LLambda_{n - 1}(\rho_n) \cdots \LLambda_1(\rho_2) \cdot 1.
\end{align}
Therefore, the whole scalar product~\eqref{sc-pr-LL} breaks into the simplest ones
\begin{align}
	\langle \Psi_{\bm{\lambda}_n} | \Psi_{\bm{\rho}_n} \rangle = \prod_{j = 1}^n \biggl[ \langle \Psi_{\lambda_j} | \Psi_{\rho_j} \rangle  \prod_{k = j + 1}^n \Gamma( \pm \imath \lambda_j \pm \imath \rho_k ) \biggr].
\end{align}
For the right hand side to make sense we assume $\lambda_j \not= \pm \rho_k$ for $j \not= k$. Inserting one particle scalar products~\eqref{sc-pr-bc-1} we obtain
\begin{align}
	\begin{aligned}
		\langle \Psi_{\bm{\lambda}_n} | \Psi_{\bm{\rho}_n} \rangle & = (4\pi)^n \prod_{j = 1}^n \frac{\Gamma(\pm 2\imath \lambda_j)}{\Gamma(g \pm \imath \lambda_j)} \prod_{1 \leq j < k \leq n} \Gamma( \pm \imath \lambda_j \pm \imath \lambda_k ) \\[6pt]
		& \times \frac{1}{2^n} \sum_{\bm{\epsilon}_n \in  \{1, -1\}^n } \delta(\lambda_1 - \epsilon_1 \rho_1) \cdots \delta(\lambda_n - \epsilon_n \rho_n).
	\end{aligned}
\end{align}
The general formula~\eqref{sc-pr-bc} without restrictions on $\lambda_j, \rho_k$ is restored using symmetry of wave functions with respect to permutations of spectral parameters (Theorem~\ref{thm:Psi-sym}).

It is left to prove the local relation with raising operators~\eqref{LtL}. This can be done using diagrams, see Figures~\ref{fig:LtL-1},~\ref{fig:LtL-2}. In almost all transformations we use chain relation (Figure~\ref{fig:chain}) to obtain dashed line and cross relation (Figure~\ref{fig:cross}) to move it in vertical direction. Besides, to obtain the last diagram in Figure~\ref{fig:LtL-1} we use identity with double lines from Figure~\ref{fig:flip-tr}. The gamma functions in the final relation~\eqref{LtL} appear each time we use chain relation.

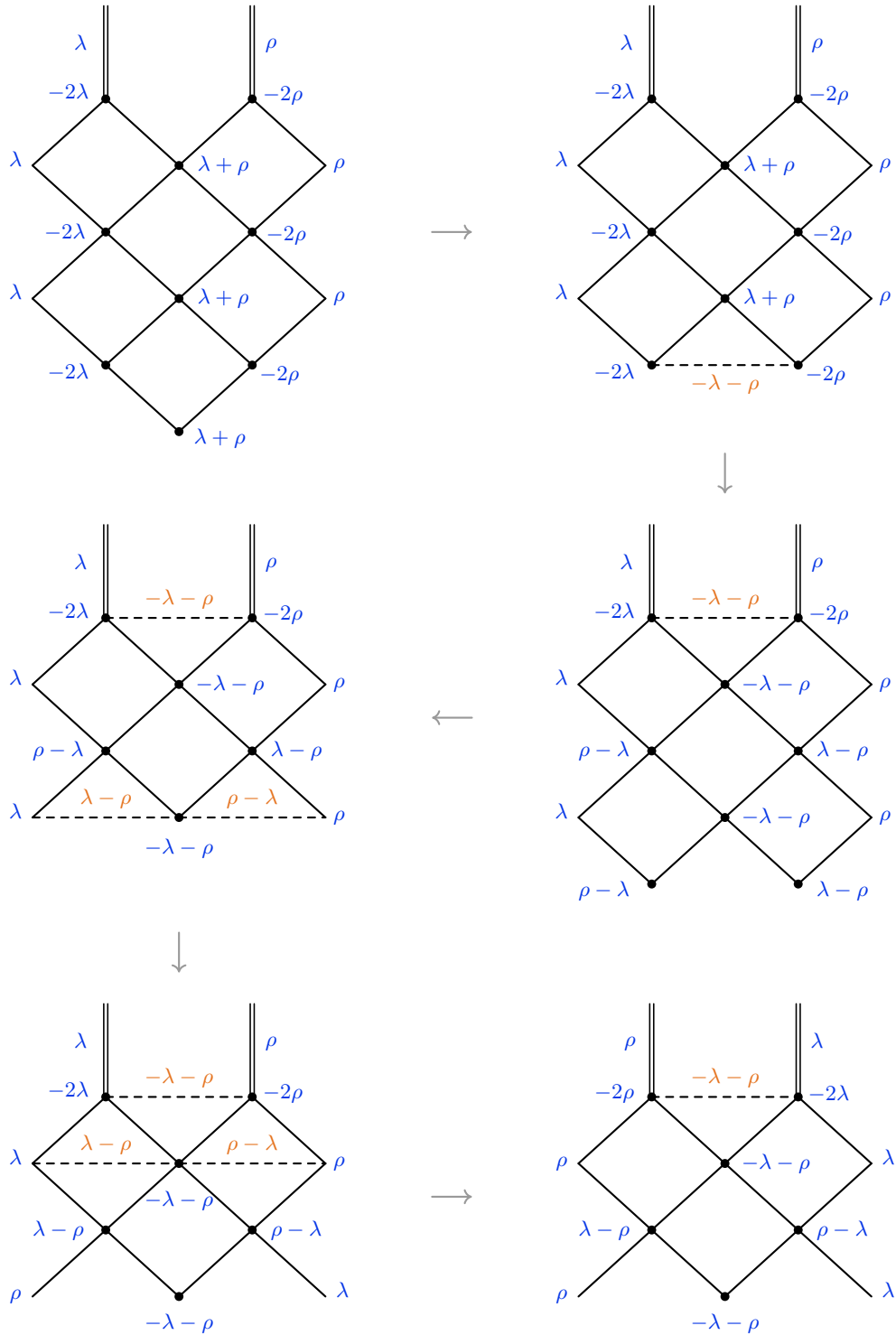
\begin{figure}[H] \centering \vspace{-1cm}
	\begin{tikzpicture}[thick, line cap = round, scale = 0.98]
		\def\l{1.1}
		\def\r{1.5pt}
		\def\h{1}
		\def\d{8.2}
		\def\dd{0.3}
		\def\hh{1.4}
		\def\t{0.03}
		\def\a{-7.8}
		\def\b{-15}

		
		\draw[line width = 0.6pt] ( \t, \hh) -- node[xshift = -0.4cm, yshift = 0.15cm] {\color{spec} \footnotesize $\lambda$} ( \t, 0);
		\draw[line width = 0.6pt] (- \t, \hh) --  (- \t, 0);
		
		\draw (0, 0 ) -- ( \l, - \h);
		\draw ( \l, - \h) -- (0, - 2*\h);
		\draw (0, 0) -- (- \l, - \h) node[xshift = -0.25cm, yshift = 0.1cm] {\color{spec} \footnotesize $\lambda$};
		\draw (- \l, - \h) -- (0, - 2*\h);
		\draw (0, - 2*\h) -- ( \l, - 3*\h) ;
		\draw (0, - 2*\h) -- (- \l, - 3*\h) node[xshift = -0.25cm, yshift = 0.1cm] {\color{spec} \footnotesize $\lambda$};
		\draw (- \l, - 3*\h) -- (0, - 4*\h);
		\draw ( \l, - 3*\h) -- (0, - 4*\h);
		\draw (0, - 4*\h) -- (\l, -5*\h);
		
		\draw[fill = black] (0, 0) circle (\r) node[xshift = -0.55cm, yshift = 0.1cm] {\color{spec} \footnotesize $-2 \lambda$};
		\draw[fill = black] (0, - 2*\h) node[xshift = -0.6cm] {\color{spec} \footnotesize $-2\lambda$} circle (\r);
		\draw[fill = black] (0, - 4*\h) node[xshift = -0.55cm, yshift = -0.1cm] {\color{spec} \footnotesize $-2 \lambda$} circle (\r);

		\draw[line width = 0.6pt] (2*\l + \t, \hh) -- node[xshift = 0.25cm, yshift = 0.1cm] {\color{spec} \footnotesize $\rho$} (2*\l + \t, 0);
		\draw[line width = 0.6pt] (2*\l - \t,  \hh) --  (2*\l - \t, 0);
		
		\draw (2*\l, 0) -- (2*\l + \l, - \h) node[xshift = 0.2cm, yshift = 0cm] {\color{spec} \footnotesize $\rho$};
		\draw (2*\l + \l, - \h) -- (2*\l, - 2*\h);
		\draw (2*\l, 0) -- (2*\l - \l, - \h);
		\draw (2*\l- \l, - \h) -- (2*\l, - 2*\h);
		\draw (2*\l, - 2*\h) -- (2*\l + \l, - 3*\h) node[xshift = 0.2cm, yshift = 0cm] {\color{spec} \footnotesize $\rho$};
		\draw (2*\l, - 2*\h) -- (2*\l- \l, - 3*\h);
		\draw (2*\l - \l, - 3*\h) -- (2*\l, - 4*\h);
		\draw (2*\l + \l, - 3*\h) -- (2*\l, - 4*\h);
		\draw (2*\l, - 4*\h) -- (2*\l - \l, -5*\h);
		
		\draw[fill = black] (2*\l, 0) circle (\r) node[xshift = 0.45cm, yshift = 0.05cm] {\color{spec} \footnotesize $-2 \rho$};
		\draw[fill = black] (2*\l, - 2*\h) node[xshift = 0.5cm, yshift = -0.05cm] {\color{spec} \footnotesize $-2\rho$} circle (\r);
		\draw[fill = black] (2*\l, - 4*\h) node[xshift = 0.4cm, yshift = -0.15cm] {\color{spec} \footnotesize $-2 \rho$} circle (\r);
		
		\draw[fill = black] (\l, -\h) circle (\r) node[xshift = 0.65cm] {\color{spec} \footnotesize $\lambda+\rho$};
		\draw[fill = black] (\l, -3*\h) circle (\r) node[xshift = 0.65cm] {\color{spec} \footnotesize $\lambda+\rho$};
		\draw[fill = black] (\l, -5*\h) circle (\r) node[xshift = 0.6cm, yshift = -0.1cm] {\color{spec} \footnotesize $\lambda+\rho$};

		
		\draw[line width = 0.6pt] (\d + \t, \hh) -- node[xshift = -0.4cm, yshift = 0.15cm] {\color{spec} \footnotesize $\lambda$} (\d + \t, 0);
		\draw[line width = 0.6pt] (\d - \t, \hh) --  (\d - \t, 0);
		
		\draw (\d, 0 ) -- (\d + \l, - \h);
		\draw (\d + \l, - \h) -- (\d, - 2*\h);
		\draw (\d, 0) -- (\d - \l, - \h) node[xshift = -0.25cm, yshift = 0.1cm] {\color{spec} \footnotesize $\lambda$};
		\draw (\d - \l, - \h) -- (\d, - 2*\h);
		\draw (\d, - 2*\h) -- (\d + \l, - 3*\h) ;
		\draw (\d, - 2*\h) -- (\d - \l, - 3*\h) node[xshift = -0.25cm, yshift = 0.1cm] {\color{spec} \footnotesize $\lambda$};
		\draw (\d - \l, - 3*\h) -- (\d, - 4*\h);
		\draw (\d + \l, - 3*\h) -- (\d, - 4*\h);
		
		\draw[fill = black] (\d, 0) circle (\r) node[xshift = -0.55cm, yshift = 0.1cm] {\color{spec} \footnotesize $-2 \lambda$};
		\draw[fill = black] (\d, - 2*\h) node[xshift = -0.6cm] {\color{spec} \footnotesize $-2\lambda$} circle (\r);
		\draw[fill = black] (\d, - 4*\h) node[xshift = -0.55cm, yshift = -0.1cm] {\color{spec} \footnotesize $-2 \lambda$} circle (\r);

		\draw[line width = 0.6pt] (\d + 2*\l + \t, \hh) -- node[xshift = 0.25cm, yshift = 0.1cm] {\color{spec} \footnotesize $\rho$} (\d + 2*\l + \t, 0);
		\draw[line width = 0.6pt] (\d + 2*\l - \t,  \hh) --  (\d + 2*\l - \t, 0);
		
		\draw (\d + 2*\l, 0) -- (\d + 2*\l + \l, - \h) node[xshift = 0.2cm, yshift = 0cm] {\color{spec} \footnotesize $\rho$};
		\draw (\d + 2*\l + \l, - \h) -- (\d + 2*\l, - 2*\h);
		\draw (\d + 2*\l, 0) -- (\d + 2*\l - \l, - \h);
		\draw (\d + 2*\l- \l, - \h) -- (\d + 2*\l, - 2*\h);
		\draw (\d + 2*\l, - 2*\h) -- (\d + 2*\l + \l, - 3*\h) node[xshift = 0.2cm, yshift = 0cm] {\color{spec} \footnotesize $\rho$};
		\draw (\d + 2*\l, - 2*\h) -- (\d + 2*\l- \l, - 3*\h);
		\draw (\d + 2*\l - \l, - 3*\h) -- (\d + 2*\l, - 4*\h);
		\draw (\d + 2*\l + \l, - 3*\h) -- (\d + 2*\l, - 4*\h);
		
		\draw[fill = black] (\d + 2*\l, 0) circle (\r) node[xshift = 0.45cm, yshift = 0.05cm] {\color{spec} \footnotesize $-2 \rho$};
		\draw[fill = black] (\d + 2*\l, - 2*\h) node[xshift = 0.5cm, yshift = -0.05cm] {\color{spec} \footnotesize $-2\rho$} circle (\r);
		\draw[fill = black] (\d + 2*\l, - 4*\h) node[xshift = 0.4cm, yshift = -0.15cm] {\color{spec} \footnotesize $-2 \rho$} circle (\r);
		
		\draw[fill = black] (\d + \l, -\h) circle (\r) node[xshift = 0.65cm] {\color{spec} \footnotesize $\lambda+\rho$};
		\draw[fill = black] (\d + \l, -3*\h) circle (\r) node[xshift = 0.65cm] {\color{spec} \footnotesize $\lambda+\rho$};
		\draw[dashed] (\d, - 4*\h) -- node[below] {\color{spec2} \footnotesize $-\lambda - \rho$} (\d + 2*\l, - 4*\h);

		
		\draw[line width = 0.6pt] (\d + \t, \a + \hh) -- node[xshift = -0.4cm, yshift = 0.15cm] {\color{spec} \footnotesize $\lambda$} (\d + \t, \a);
		\draw[line width = 0.6pt] (\d - \t, \a + \hh) --  (\d - \t, \a);
		
		\draw (\d, \a) -- (\d + \l, \a - \h);
		\draw (\d + \l, \a - \h) -- (\d, \a - 2*\h);
		\draw (\d, \a) -- (\d - \l, \a - \h) node[xshift = -0.25cm, yshift = 0.1cm] {\color{spec} \footnotesize $\lambda$};
		\draw (\d - \l, \a - \h) -- (\d, \a - 2*\h);
		\draw (\d, \a - 2*\h) -- (\d + \l, \a - 3*\h) ;
		\draw (\d, \a - 2*\h) -- (\d - \l, \a - 3*\h) node[xshift = -0.25cm, yshift = 0.1cm] {\color{spec} \footnotesize $\lambda$};
		\draw (\d - \l, \a - 3*\h) -- (\d, \a - 4*\h);
		\draw (\d + \l, \a - 3*\h) -- (\d, \a - 4*\h);
		
		\draw[fill = black] (\d, \a) circle (\r) node[xshift = -0.55cm, yshift = 0.1cm] {\color{spec} \footnotesize $-2 \lambda$};
		\draw[fill = black] (\d, \a - 2*\h) node[xshift = -0.7cm, yshift = 0cm] {\color{spec} \footnotesize $\rho - \lambda$} circle (\r);
		\draw[fill = black] (\d, \a - 4*\h) node[xshift = -0.7cm, yshift = -0.1cm] {\color{spec} \footnotesize $\rho - \lambda$} circle (\r);

		\draw[line width = 0.6pt] (\d + 2*\l + \t, \a + \hh) -- node[xshift = 0.25cm, yshift = 0.1cm] {\color{spec} \footnotesize $\rho$} (\d + 2*\l + \t, \a);
		\draw[line width = 0.6pt] (\d + 2*\l - \t,  \a + \hh) --  (\d + 2*\l - \t, \a);
		
		\draw (\d + 2*\l, \a) -- (\d + 2*\l + \l, \a - \h) node[xshift = 0.2cm, yshift = 0cm] {\color{spec} \footnotesize $\rho$};
		\draw (\d + 2*\l + \l, \a - \h) -- (\d + 2*\l, \a - 2*\h);
		\draw (\d + 2*\l, \a) -- (\d + 2*\l - \l, \a - \h);
		\draw (\d + 2*\l- \l, \a - \h) -- (\d + 2*\l, \a - 2*\h);
		\draw (\d + 2*\l, \a - 2*\h) -- (\d + 2*\l + \l, \a - 3*\h) node[xshift = 0.2cm, yshift = 0cm] {\color{spec} \footnotesize $\rho$};
		\draw (\d + 2*\l, \a - 2*\h) -- (\d + 2*\l- \l, \a - 3*\h);
		\draw (\d + 2*\l - \l, \a - 3*\h) -- (\d + 2*\l, \a - 4*\h);
		\draw (\d + 2*\l + \l, \a - 3*\h) -- (\d + 2*\l, \a - 4*\h);
		
		\draw[fill = black] (\d + 2*\l, \a) circle (\r) node[xshift = 0.45cm, yshift = 0.05cm] {\color{spec} \footnotesize $-2 \rho$};
		\draw[fill = black] (\d + 2*\l, \a - 2*\h) node[xshift = 0.65cm, yshift = 0cm] {\color{spec} \footnotesize $\lambda - \rho$} circle (\r);
		\draw[fill = black] (\d + 2*\l, \a - 4*\h) node[xshift = 0.65cm, yshift = -0.1cm] {\color{spec} \footnotesize $\lambda - \rho$} circle (\r);
		
		\draw[fill = black] (\d + \l, \a -\h) circle (\r) node[xshift = 0.75cm] {\color{spec} \footnotesize $-\lambda - \rho$};
		\draw[fill = black] (\d + \l, \a - 3*\h) circle (\r) node[xshift = 0.75cm] {\color{spec} \footnotesize $- \lambda - \rho$};
		\draw[dashed] (\d, \a) -- node[above] {\color{spec2} \footnotesize $- \lambda - \rho$} (\d + 2*\l, \a);

		
		\draw[line width = 0.6pt] (\t, \a + \hh) -- node[xshift = -0.4cm, yshift = 0.15cm] {\color{spec} \footnotesize $\lambda$} ( \t, \a);
		\draw[line width = 0.6pt] (- \t, \a + \hh) --  ( - \t, \a);
		
		\draw (0, \a) -- (\l, \a - \h);
		\draw (\l, \a - \h) -- (0, \a - 2*\h);
		\draw (0, \a) -- ( - \l, \a - \h) node[xshift = -0.25cm, yshift = 0.1cm] {\color{spec} \footnotesize $\lambda$};
		\draw ( - \l, \a - \h) -- (0, \a - 2*\h);
		\draw (0, \a - 2*\h) -- ( \l, \a - 3*\h) ;
		\draw (0, \a - 2*\h) -- (- \l, \a - 3*\h) node[xshift = -0.25cm, yshift = 0.1cm] {\color{spec} \footnotesize $\lambda$};
		
		\draw[fill = black] (0, \a) circle (\r) node[xshift = -0.55cm, yshift = 0.1cm] {\color{spec} \footnotesize $-2 \lambda$};
		\draw[fill = black] (0, \a - 2*\h) node[xshift = -0.7cm, yshift = 0cm] {\color{spec} \footnotesize $\rho - \lambda$} circle (\r);

		\draw[line width = 0.6pt] ( 2*\l + \t, \a + \hh) -- node[xshift = 0.25cm, yshift = 0.1cm] {\color{spec} \footnotesize $\rho$} ( 2*\l + \t, \a);
		\draw[line width = 0.6pt] ( 2*\l - \t,  \a + \hh) --  ( 2*\l - \t, \a);
		
		\draw ( 2*\l, \a) -- ( 2*\l + \l, \a - \h) node[xshift = 0.2cm, yshift = 0cm] {\color{spec} \footnotesize $\rho$};
		\draw ( 2*\l + \l, \a - \h) -- ( 2*\l, \a - 2*\h);
		\draw ( 2*\l, \a) -- ( 2*\l - \l, \a - \h);
		\draw ( 2*\l- \l, \a - \h) -- ( 2*\l, \a - 2*\h);
		\draw ( 2*\l, \a - 2*\h) -- ( 2*\l + \l, \a - 3*\h) node[xshift = 0.2cm, yshift = 0cm] {\color{spec} \footnotesize $\rho$};
		\draw ( 2*\l, \a - 2*\h) -- ( 2*\l- \l, \a - 3*\h);
		
		\draw[fill = black] ( 2*\l, \a) circle (\r) node[xshift = 0.45cm, yshift = 0.05cm] {\color{spec} \footnotesize $-2 \rho$};
		\draw[fill = black] ( 2*\l, \a - 2*\h) node[xshift = 0.65cm, yshift = 0cm] {\color{spec} \footnotesize $\lambda - \rho$} circle (\r);
		
		\draw[fill = black] ( \l, \a -\h) circle (\r) node[xshift = 0.75cm] {\color{spec} \footnotesize $- \lambda - \rho$};
		\draw[fill = black] ( \l, \a - 3*\h) circle (\r) node[yshift = -0.45cm] {\color{spec} \footnotesize $- \lambda - \rho$};
		\draw[dashed] (0, \a) -- node[above] {\color{spec2} \footnotesize $- \lambda - \rho$} ( 2*\l, \a);
		\draw[dashed] (- \l, \a - 3*\h) -- node[above] {\color{spec2} \footnotesize $\lambda - \rho$} ( \l, \a - 3*\h);
		\draw[dashed] ( 2*\l- \l, \a - 3*\h) -- node[above] {\color{spec2} \footnotesize $\rho - \lambda$} ( 2*\l + \l, \a - 3*\h);

		
		\draw[line width = 0.6pt] (\t, \b + \hh) -- node[xshift = -0.4cm, yshift = 0.15cm] {\color{spec} \footnotesize $\lambda$} ( \t, \b);
		\draw[line width = 0.6pt] (- \t, \b + \hh) --  ( - \t, \b);
		
		\draw (0, \b) -- (\l, \b - \h);
		\draw (\l, \b - \h) -- (0, \b - 2*\h);
		\draw (0, \b) -- ( - \l, \b - \h) node[xshift = -0.25cm, yshift = 0.1cm] {\color{spec} \footnotesize $\lambda$};
		\draw ( - \l, \b - \h) -- (0, \b - 2*\h);
		\draw (0, \b - 2*\h) -- ( \l, \b - 3*\h) ;
		\draw (0, \b - 2*\h) -- (- \l, \b - 3*\h) node[xshift = -0.25cm, yshift = 0cm] {\color{spec} \footnotesize $\rho$};
		
		\draw[fill = black] (0, \b) circle (\r) node[xshift = -0.55cm, yshift = 0.1cm] {\color{spec} \footnotesize $- 2 \lambda$};
		\draw[fill = black] (0, \b - 2*\h) node[xshift = -0.7cm, yshift = 0cm] {\color{spec} \footnotesize $\lambda - \rho$} circle (\r);

		\draw[line width = 0.6pt] ( 2*\l + \t, \b + \hh) -- node[xshift = 0.25cm, yshift = 0.1cm] {\color{spec} \footnotesize $\rho$} ( 2*\l + \t, \b);
		\draw[line width = 0.6pt] ( 2*\l - \t,  \b + \hh) --  ( 2*\l - \t, \b);
		
		\draw ( 2*\l, \b) -- ( 2*\l + \l, \b - \h) node[xshift = 0.2cm, yshift = 0cm] {\color{spec} \footnotesize $\rho$};
		\draw ( 2*\l + \l, \b - \h) -- ( 2*\l, \b - 2*\h);
		\draw ( 2*\l, \b) -- ( 2*\l - \l, \b - \h);
		\draw ( 2*\l- \l, \b - \h) -- ( 2*\l, \b - 2*\h);
		\draw ( 2*\l, \b - 2*\h) -- ( 2*\l + \l, \b - 3*\h) node[xshift = 0.25cm, yshift = 0.1cm] {\color{spec} \footnotesize $\lambda$};
		\draw ( 2*\l, \b - 2*\h) -- ( 2*\l- \l, \b - 3*\h);
		
		\draw[fill = black] ( 2*\l, \b) circle (\r) node[xshift = 0.45cm, yshift = 0.05cm] {\color{spec} \footnotesize $-2 \rho$};
		\draw[fill = black] ( 2*\l, \b - 2*\h) node[xshift = 0.65cm, yshift = 0cm] {\color{spec} \footnotesize $\rho - \lambda$} circle (\r);
		
		\draw[fill = black] ( \l, \b -\h) circle (\r) node[yshift = -0.55cm] {\color{spec} \footnotesize $- \lambda - \rho$};
		\draw[fill = black] ( \l, \b - 3*\h) circle (\r) node[yshift = -0.4cm] {\color{spec} \footnotesize $- \lambda - \rho$};
		\draw[dashed] (0, \b) -- node[above] {\color{spec2} \footnotesize $- \lambda - \rho$} ( 2*\l, \b);
		\draw[dashed] (- \l, \b - 1*\h) -- node[above] {\color{spec2} \footnotesize $\lambda - \rho$} ( \l, \b - 1*\h);
		\draw[dashed] ( 2*\l- \l, \b - 1*\h) -- node[above] {\color{spec2} \footnotesize $\rho - \lambda$} ( 2*\l + \l, \b - 1*\h);

		
		\draw[line width = 0.6pt] (\d + \t, \b + \hh) -- node[xshift = -0.35cm, yshift = 0.1cm] {\color{spec} \footnotesize $\rho$} (\d + \t, \b);
		\draw[line width = 0.6pt] (\d - \t, \b + \hh) --  (\d - \t, \b);
		
		\draw (\d, \b) -- (\d + \l, \b - \h);
		\draw (\d + \l, \b - \h) -- (\d, \b - 2*\h);
		\draw (\d, \b) -- (\d - \l, \b - \h) node[xshift = -0.25cm, yshift = 0cm] {\color{spec} \footnotesize $\rho$};
		\draw (\d - \l, \b - \h) -- (\d, \b - 2*\h);
		\draw (\d, \b - 2*\h) -- (\d +  \l, \b - 3*\h) ;
		\draw (\d, \b - 2*\h) -- (\d - \l, \b - 3*\h) node[xshift = -0.25cm, yshift = 0cm] {\color{spec} \footnotesize $\rho$};
		
		\draw[fill = black] (\d, \b) circle (\r) node[xshift = -0.55cm, yshift = 0.05cm] {\color{spec} \footnotesize $-2 \rho$};
		\draw[fill = black] (\d, \b - 2*\h) node[xshift = -0.7cm, yshift = 0cm] {\color{spec} \footnotesize $\lambda - \rho$} circle (\r);

		\draw[line width = 0.6pt] (\d +  2*\l + \t, \b + \hh) -- node[xshift = 0.25cm, yshift = 0.15cm] {\color{spec} \footnotesize $\lambda$} (\d +  2*\l + \t, \b);
		\draw[line width = 0.6pt] (\d +  2*\l - \t,  \b + \hh) --  (\d +  2*\l - \t, \b);
		
		\draw (\d +  2*\l, \b) -- (\d +  2*\l + \l, \b - \h) node[xshift = 0.25cm, yshift = 0.1cm] {\color{spec} \footnotesize $\lambda$};
		\draw ( \d + 2*\l + \l, \b - \h) -- ( \d + 2*\l, \b - 2*\h);
		\draw ( \d + 2*\l, \b) -- ( \d + 2*\l - \l, \b - \h);
		\draw ( \d + 2*\l- \l, \b - \h) -- ( \d + 2*\l, \b - 2*\h);
		\draw ( \d + 2*\l, \b - 2*\h) -- ( \d + 2*\l + \l, \b - 3*\h) node[xshift = 0.25cm, yshift = 0.1cm] {\color{spec} \footnotesize $\lambda$};
		\draw ( \d + 2*\l, \b - 2*\h) -- ( \d + 2*\l- \l, \b - 3*\h);
		
		\draw[fill = black] ( \d + 2*\l, \b) circle (\r) node[xshift = 0.45cm, yshift = 0.05cm] {\color{spec} \footnotesize $-2 \lambda$};
		\draw[fill = black] ( \d + 2*\l, \b - 2*\h) node[xshift = 0.65cm, yshift = 0cm] {\color{spec} \footnotesize $\rho - \lambda$} circle (\r);
		
		\draw[fill = black] (\d +  \l, \b -\h) circle (\r) node[xshift = 0.75cm] {\color{spec} \footnotesize $- \lambda - \rho$};
		\draw[fill = black] (\d +  \l, \b - 3*\h) circle (\r) node[yshift = -0.4cm] {\color{spec} \footnotesize $- \lambda - \rho$};
		\draw[dashed] (\d, \b) -- node[above] {\color{spec2} \footnotesize $- \lambda - \rho$} ( \d + 2*\l, \b);

		
		\draw[arrow, gray!80] (\l + 0.5*\d - \dd, - 2*\h) -- (\l + 0.5*\d + \dd, - 2*\h);
		\draw[arrow, gray!80] (\l + 0.5*\d + \dd, \a - 1.5*\h) -- (\l + 0.5*\d - \dd, \a - 1.5*\h);
		\draw[arrow, gray!80] (\l + 0.5*\d - \dd, \b - 1.5*\h) -- (\l + 0.5*\d + \dd, \b - 1.5*\h);
		
		\draw[arrow, gray!80] (\d + \l, \a + 1.55*\hh + \dd) -- (\d + \l, \a + 1.55*\hh - \dd);
		\draw[arrow, gray!80] (\l, \b + 1.55*\hh + \dd) -- (\l, \b + 1.55*\hh - \dd);
	\end{tikzpicture}
	\vspace{0.5cm}
	\caption{Proof of identity~\eqref{LtL} (continues in Figure~\ref{fig:LtL-2})} \label{fig:LtL-1}
\end{figure}

\begin{figure}[H] \centering \vspace{-1cm}
	\begin{tikzpicture}[thick, line cap = round, scale = 0.98]
		\def\l{1.1}
		\def\r{1.5pt}
		\def\h{1}
		\def\d{8.2}
		\def\dd{0.3}
		\def\hh{1.4}
		\def\t{0.03}
		\def\a{-6}

		
		\draw[line width = 0.6pt] (0+ \t, 0 + \hh) -- node[xshift = -0.35cm, yshift = 0.1cm] {\color{spec} \footnotesize $\rho$} (0+ \t, 0);
		\draw[line width = 0.6pt] (0- \t, 0 + \hh) --  (0- \t, 0);
		
		\draw (0, 0) -- (0+ \l, 0 - \h);
		\draw (0+ \l, 0 - \h) -- (0, 0 - 2*\h);
		\draw (0, 0) -- (0- \l, 0 - \h) node[xshift = -0.25cm, yshift = 0cm] {\color{spec} \footnotesize $\rho$};
		\draw (0- \l, 0 - \h) -- (0, 0 - 2*\h);
		\draw (0, 0 - 2*\h) -- (0+  \l, 0 - 3*\h) ;
		\draw (0, 0 - 2*\h) -- (0- \l, 0 - 3*\h) node[xshift = -0.25cm, yshift = 0cm] {\color{spec} \footnotesize $\rho$};
		
		\draw[fill = black] (0, 0) circle (\r) node[xshift = -0.55cm, yshift = 0.05cm] {\color{spec} \footnotesize $- 2 \rho$};
		\draw[fill = black] (0, 0 - 2*\h) node[xshift = -0.7cm, yshift = 0cm] {\color{spec} \footnotesize $\lambda - \rho$} circle (\r);

		\draw[line width = 0.6pt] (0+  2*\l + \t, 0 + \hh) -- node[xshift = 0.25cm, yshift = 0.15cm] {\color{spec} \footnotesize $\lambda$} (0+  2*\l + \t, 0);
		\draw[line width = 0.6pt] (0+  2*\l - \t,  0 + \hh) --  (0+  2*\l - \t, 0);
		
		\draw (0+  2*\l, 0) -- (0+  2*\l + \l, 0 - \h) node[xshift = 0.25cm, yshift = 0.1cm] {\color{spec} \footnotesize $\lambda$};
		\draw ( 0+ 2*\l + \l, 0 - \h) -- ( 0+ 2*\l, 0 - 2*\h);
		\draw ( 0+ 2*\l, 0) -- ( 0+ 2*\l - \l, 0 - \h);
		\draw ( 0+ 2*\l- \l, 0 - \h) -- ( 0+ 2*\l, 0 - 2*\h);
		\draw ( 0+ 2*\l, 0 - 2*\h) -- ( 0+ 2*\l + \l, 0 - 3*\h) node[xshift = 0.25cm, yshift = 0.1cm] {\color{spec} \footnotesize $\lambda$};
		\draw ( 0+ 2*\l, 0 - 2*\h) -- ( 0+ 2*\l- \l, 0 - 3*\h);
		
		\draw[fill = black] ( 0+ 2*\l, 0) circle (\r) node[xshift = 0.45cm, yshift = 0.05cm] {\color{spec} \footnotesize $-2 \lambda$};
		\draw[fill = black] ( 0+ 2*\l, 0 - 2*\h) node[xshift = 0.65cm, yshift = 0cm] {\color{spec} \footnotesize $\rho - \lambda$} circle (\r);
		
		\draw[fill = black] (0+  \l, 0 -\h) circle (\r) node[xshift = 0.75cm] {\color{spec} \footnotesize $- \lambda - \rho$};
		\draw[fill = black] (0+  \l, 0 - 3*\h) circle (\r) node[yshift = -0.4cm] {\color{spec} \footnotesize $- \lambda - \rho$};
		\draw[dashed] (0, 0) -- node[above] {\color{spec2} \footnotesize $- \lambda - \rho$} ( 0+ 2*\l, 0);


		\draw[line width = 0.6pt] (\d + \t, 0 + \hh) -- node[xshift = -0.35cm, yshift = 0.1cm] {\color{spec} \footnotesize $\rho$} (\d + \t, 0);
		\draw[line width = 0.6pt] (\d - \t, 0 + \hh) --  (\d - \t, 0);
		
		\draw (\d, 0) -- (\d + \l, 0 - \h);
		\draw (\d + \l, 0 - \h) -- (\d, 0 - 2*\h);
		\draw (\d, 0) -- (\d - \l, 0 - \h) node[xshift = -0.25cm, yshift = 0cm] {\color{spec} \footnotesize $\rho$};
		\draw (\d - \l, 0 - \h) -- (\d, 0 - 2*\h);
		\draw (\d, 0 - 2*\h) -- (\d - \l, 0 - 3*\h) node[xshift = -0.25cm, yshift = 0cm] {\color{spec} \footnotesize $\rho$};
		
		\draw[fill = black] (\d, 0) circle (\r) node[xshift = -0.55cm, yshift = 0.05cm] {\color{spec} \footnotesize $- 2 \rho$};
		\draw[fill = black] (\d, 0 - 2*\h) node[xshift = -0.7cm, yshift = 0cm] {\color{spec} \footnotesize $\lambda - \rho$} circle (\r);

		\draw[line width = 0.6pt] (\d +  2*\l + \t, 0 + \hh) -- node[xshift = 0.25cm, yshift = 0.1cm] {\color{spec} \footnotesize $\lambda$} (\d +  2*\l + \t, 0);
		\draw[line width = 0.6pt] (\d +  2*\l - \t,  0 + \hh) --  (\d +  2*\l - \t, 0);
		
		\draw (\d +  2*\l, 0) -- (\d +  2*\l + \l, 0 - \h) node[xshift = 0.25cm, yshift = 0.1cm] {\color{spec} \footnotesize $\lambda$};
		\draw ( \d + 2*\l + \l, 0 - \h) -- ( \d + 2*\l, 0 - 2*\h);
		\draw ( \d + 2*\l, 0) -- ( \d + 2*\l - \l, 0 - \h);
		\draw ( \d + 2*\l- \l, 0 - \h) -- ( \d + 2*\l, 0 - 2*\h);
		\draw ( \d + 2*\l, 0 - 2*\h) -- ( \d + 2*\l + \l, 0 - 3*\h) node[xshift = 0.25cm, yshift = 0.1cm] {\color{spec} \footnotesize $\lambda$};
		
		\draw[fill = black] ( \d + 2*\l, 0) circle (\r) node[xshift = 0.45cm, yshift = 0.05cm] {\color{spec} \footnotesize $- 2 \lambda$};
		\draw[fill = black] ( \d + 2*\l, 0 - 2*\h) node[xshift = 0.65cm, yshift = 0cm] {\color{spec} \footnotesize $\rho - \lambda$} circle (\r);
		
		\draw[fill = black] (\d +  \l, 0 -\h) circle (\r) node[xshift = 0.75cm] {\color{spec} \footnotesize $- \lambda - \rho$};
		\draw[dashed] (\d, 0) -- node[above] {\color{spec2} \footnotesize $- \lambda - \rho$} ( \d + 2*\l, 0);
		\draw[dashed] (\d, -2*\h) -- node[below] {\color{spec2} \footnotesize $\lambda + \rho$} ( \d + 2*\l, -2*\h);


		\draw[line width = 0.6pt] (\d + \t, \a + \hh) -- node[xshift = -0.35cm, yshift = 0.1cm] {\color{spec} \footnotesize $\rho$} (\d + \t, \a);
		\draw[line width = 0.6pt] (\d - \t, \a + \hh) --  (\d - \t, \a);
		
		\draw (\d, \a) -- (\d + \l, \a - \h);
		\draw (\d + \l, \a - \h) -- (\d, \a - 2*\h);
		\draw (\d, \a) -- (\d - \l, \a - \h) node[xshift = -0.25cm, yshift = 0cm] {\color{spec} \footnotesize $\rho$};
		\draw (\d - \l, \a - \h) -- (\d, \a - 2*\h);
		\draw (\d, \a - 2*\h) -- (\d - \l, \a - 3*\h) node[xshift = -0.25cm, yshift = 0cm] {\color{spec} \footnotesize $\rho$};
		
		\draw[fill = black] (\d, \a) circle (\r) node[xshift = -0.55cm, yshift = 0.05cm] {\color{spec} \footnotesize $- 2 \rho$};
		\draw[fill = black] (\d, \a - 2*\h) node[xshift = -0.55cm, yshift = -0cm] {\color{spec} \footnotesize $- 2\rho$} circle (\r);

		\draw[line width = 0.6pt] (\d +  2*\l + \t, \a + \hh) -- node[xshift = 0.25cm, yshift = 0.1cm] {\color{spec} \footnotesize $\lambda$} (\d +  2*\l + \t, \a);
		\draw[line width = 0.6pt] (\d +  2*\l - \t,  \a + \hh) --  (\d +  2*\l - \t, \a);
		
		\draw (\d +  2*\l, \a) -- (\d +  2*\l + \l, \a - \h) node[xshift = 0.25cm, yshift = 0.1cm] {\color{spec} \footnotesize $\lambda$};
		\draw ( \d + 2*\l + \l, \a - \h) -- ( \d + 2*\l, \a - 2*\h);
		\draw ( \d + 2*\l, \a) -- ( \d + 2*\l - \l, \a - \h);
		\draw ( \d + 2*\l- \l, \a - \h) -- ( \d + 2*\l, \a - 2*\h);
		\draw ( \d + 2*\l, \a - 2*\h) -- ( \d + 2*\l + \l, \a - 3*\h) node[xshift = 0.25cm, yshift = 0.1cm] {\color{spec} \footnotesize $\lambda$};
		
		\draw[fill = black] ( \d + 2*\l, \a) circle (\r) node[xshift = 0.45cm, yshift = 0.05cm] {\color{spec} \footnotesize $- 2 \lambda$};
		\draw[fill = black] ( \d + 2*\l, \a - 2*\h) node[xshift = 0.45cm, yshift = 0cm] {\color{spec} \footnotesize $- 2\lambda$} circle (\r);
		
		\draw[fill = black] (\d +  \l, \a -\h) circle (\r) node[xshift = 0.75cm] {\color{spec} \footnotesize $\lambda + \rho$};

		
		\draw[arrow, gray!80] (\l + 0.5*\d - \dd, - 1.5*\h) -- (\l + 0.5*\d + \dd, - 1.5*\h);
		
		\draw[arrow, gray!80] (\d + \l, \a + 1.55*\hh + \dd) -- (\d + \l, \a + 1.55*\hh - \dd);
	\end{tikzpicture}
	\vspace{0.5cm}
	\caption{Proof of identity~\eqref{LtL} (begins in Figure~\ref{fig:LtL-1})} \label{fig:LtL-2}
\end{figure}
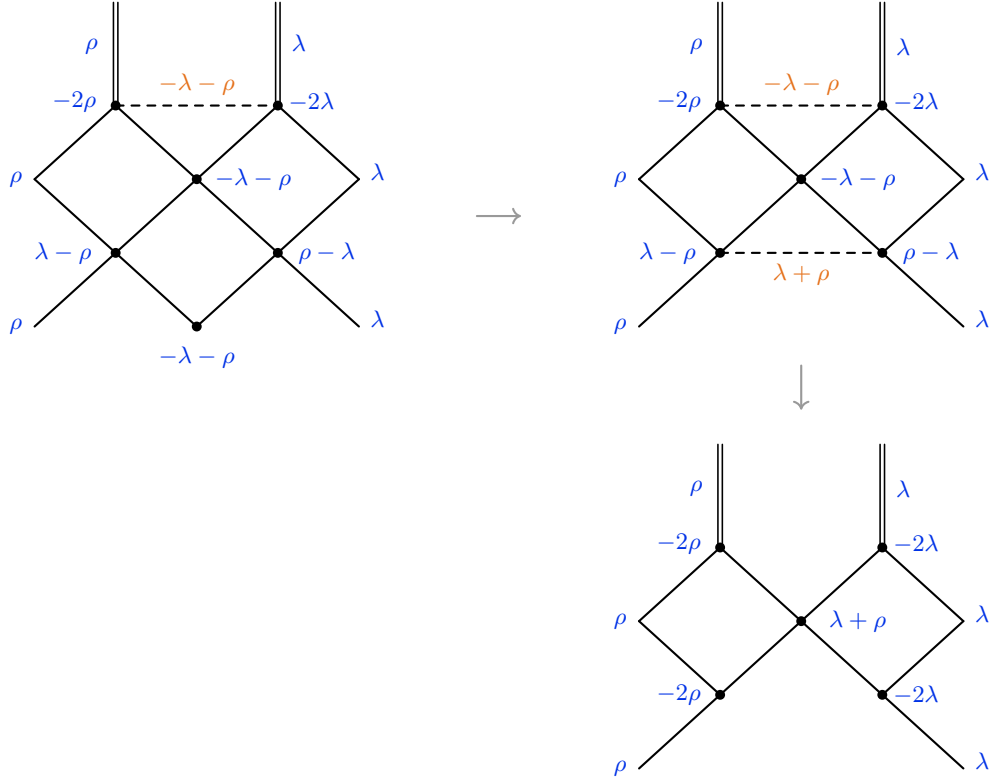

\section{Dual system}

In the previous sections we dealt with operators acting in the space of spatial variables $x_j$. On the other hand, $GL$ and $BC$ Toda wave functions can be also studied as functions of spectral variables $\lambda_j$, which involves analogous objects, such as Hamiltonians and integral operators, acting in the spectral space.

\subsection{$GL$ system} 

In this section we recall the known facts about $GL$ Toda chain: dual Hamiltonians, dual Baxter and raising operators, Mellin--Barnes representation of wave function. 

\subsubsection{Dual Hamiltonians, raising and Baxter operators}
Dual Hamiltonians acting on functions of spectral variables $\phi(\bm{\lambda}_{n})$ are given by~\cite{B}

\begin{align}\label{Hdual}
	\Hdual_s = \sum_{\substack{J \subset \{1, \dots, n\} \\[1pt] |J| = s}} \, \prod_{ \substack{j \in J \\ k \not\in J } } \frac{1}{\imath \lambda_j - \imath \lambda_k} \, \prod_{j \in J} e^{- \imath \partial_{\lambda_j}},
\end{align}
the simplest ones are
\begin{align}
	\Hdual_1 = \sum_{j = 1}^n \, \prod_{ \substack{k = 1 \\ k \not= j} }^n \frac{1}{\imath \lambda_j - \imath \lambda_k} \, e^{- \imath \partial_{\lambda_j} }, \qquad \Hdual_n = e^{- \imath (\partial_{\lambda_1} + \ldots + \partial_{\lambda_n})}.
\end{align}
Here $e^{-\imath \partial_{\lambda}} \phi(\lambda) = \phi(\lambda - \imath)$, so these operators are well defined on the functions analytic in $\lambda_j$ in the strips $\Im \lambda_j \in [-1, 0]$.
The Hamiltonians are formally self-adjoint with respect to the scalar product
\begin{align}\label{dual-sp}
	\langle \psi | \phi \rangle_{\hat{\mu}} = \int_{\mathbb{R}^n} d\bm{\lambda}_n \; \hat{\mu}(\bm{\lambda}_n) \, \overline{\psi(\bm{\lambda}_n)} \, \phi(\bm{\lambda}_n)
\end{align}
with the measure
\begin{align} \label{gl-measure}
	\hat{\mu}(\bm{\lambda}_n) = \frac{1}{n! \, (2\pi)^n} \prod_{1 \leq j \not= k \leq n} \frac{1}{\Gamma(\imath \lambda_j - \imath \lambda_k)}.
\end{align}
Note that this function appears in the orthogonality relation for $GL$ wave functions~\eqref{gl-orth}.

Dual raising operator $\hat{\Lambda}_n(x)$ acts on functions $\phi(\bm{\lambda}_{n - 1})$ by the formula
\begin{align}
	\bigl[ \hat{\Lambda}_n(x) \, \phi \bigr] (\bm{\lambda}_n) = \int_{(\mathbb{R} + \imath 0)^{n - 1}} \! d\bm{\gamma}_{n - 1} \; \hat{\mu}(\bm{\gamma}_{n - 1}) \, e^{ \imath (\underline{\bm{\lambda}}_n - \underline{\bm{\gamma}}_{n - 1} ) x} \, \prod_{j = 1}^n \prod_{k = 1}^{n - 1} \Gamma( \imath \lambda_j - \imath \gamma_k) \, \phi(\bm{\gamma}_{n - 1}).
\end{align}
Dual Baxter operator $\hat{Q}_n(x)$ acts on functions $\phi(\bm{\lambda}_n)$ by the formula
\begin{align} 
	\bigl[ \hat{Q}_n(x) \, \phi \bigr] (\bm{\lambda}_n) = \int_{(\mathbb{R} + \imath 0)^{n}} \! d\bm{\gamma}_{n} \; \hat{\mu}(\bm{\gamma}_{n}) \, e^{ \imath (\underline{\bm{\lambda}}_n - \underline{\bm{\gamma}}_{n} ) x} \, \prod_{j, k = 1}^n \Gamma( \imath \lambda_j - \imath \gamma_k) \, \phi(\bm{\gamma}_{n}).
\end{align}
These operators are defined in~\cite[Section 3]{GLO1}, although essentially the raising operator appears in the earlier work~\cite[(4.3)]{KL2}. 

Denote the reflection operator
\begin{align}
	\hat{\mathcal{I}}_n \colon \; \phi(\bm{\lambda}_n) \; \mapsto \; \phi(-\bm{\lambda}_n).
\end{align}
The family of Hamiltonians is preserved under its action
\begin{align}
	\hat{\mathcal{I}}_n \, \hat{H}_s \, \hat{\mathcal{I}}_n = \hat{H}_n^{-1} \, \hat{H}_{n - s}.
\end{align}
Define reflected raising and Baxter operators
\begin{align}
	\hat{\Lambda}'(x) = \hat{\mathcal{I}}_n \, \hat{\Lambda}_n(-x) \, \hat{\mathcal{I}}_n, \qquad \hat{Q}'_n(x) = \hat{\mathcal{I}}_n \, \hat{Q}_n(-x) \, \hat{\mathcal{I}}_n,
\end{align}
so that in particular,
\begin{align} \label{Qprime} 
	\bigl[ \hat{Q}'_n(x) \, \phi \bigr] (\bm{\lambda}_n) = \int_{(\mathbb{R} - \imath 0)^{n}} \! d\bm{\gamma}_{n} \; \hat{\mu}(\bm{\gamma}_{n}) \, e^{ \imath (\underline{\bm{\lambda}}_n - \underline{\bm{\gamma}}_{n} ) x} \, \prod_{j, k = 1}^n \Gamma( \imath \gamma_k - \imath \lambda_j) \, \phi(\bm{\gamma}_{n}).
\end{align}

\subsubsection{Local relations and eigenfunctions}

The above operators satisfy local relations
\begin{align}
	& \Hdual_s \, \hat{Q}_n(x) = \hat{Q}_n(x) \, \Hdual_s, \\[6pt]
	& \Hdual_s \, \hat{\Lambda}_n(x) = e^x \, \hat{\Lambda}_n(x) \, \Hdual_{s - 1}, \\[6pt] \label{QQh}
	& \hat{Q}_n(x) \, \hat{Q}_n(y) = \hat{Q}_n\bigl( - \ln(e^{-x} + e^{-y}) \bigr) = \hat{Q}_n(y) \, \hat{Q}_n(x) , \\[6pt]
	& \hat{Q}_n(x) \, \hat{\Lambda}_n(y) = \exp(-e^{y - x}) \, \hat{\Lambda}_n(y),
\end{align}
where in the second identity for $s = 1$ we denote $\Hdual_0 = \mathrm{Id}$.

The first two can be proved using symmetry of Hamiltonians with respect to the scalar product~\eqref{dual-sp} and simple kernel function type of identity~\cite{B}
\beq  \sum_{ \substack{ I \subset \{1, \dots, n\} \\ |I| = s }}\prod_{ \substack {i\in I \\ j\not\in I}}\frac{1}{x_i-x_j}
\prod_{ \substack {i\in I \\ a=1,\ldots,n}}(x_i-y_a)=\sum_{ \substack{ I \subset \{1, \dots, n\} \\ |I| = s }}\prod_{ \substack {a\in I \\ b\not\in I}}\frac{1}{y_b-y_a}
\prod_{ \substack {a\in I \\ i=1,\ldots,n}}(x_i-y_a),\eeq
or its degeneration as $y_n \to \infty$. The last two relations follow from degenerations of Gustafson identity~\eqref{Gust-int}.

In fact, in the present paper we won't need the above relations, but only some of their corollaries. Notice that the Mellin--Barnes representation of $GL$ Toda wave functions, mentioned in the introduction~\eqref{Phi-MB-0}, is given by iterative formula
\begin{align}\label{MBA}
	\Phi_{\bm{\lambda}_n}(\bm{x}_n) = \hat{\Lambda}_n(x_n) \cdots \hat{\Lambda}_1(x_1) \cdot 1.
\end{align}
Moreover, using Gauss--Givental representation~\eqref{Phi-GG-0} one can easily prove the symmetry
\begin{align}\label{GG24}
	\Phi_{\bm{\lambda}_n}(\bm{x}_n) =  \Phi_{-\bm{\lambda}_n}(-x_n, \dots, -x_1),
\end{align}
which in terms of dual operators reads
\begin{align}
	\Phi_{\bm{\lambda}_n}(\bm{x}_n) = \hat{\Lambda}'_n(x_1) \cdots \hat{\Lambda}_1'(x_n) \cdot 1.
\end{align}
Thus, from the above relations and definitions $GL$ Toda wave functions diagonalize dual Hamiltonians and Baxter operators
\begin{align} \label{Hdual-Phi}
	& \Hdual_s \, \Phi_{\bm{\lambda}_n}(\bm{x}_n) = e^{x_{n - s + 1} + \ldots + x_n} \, \Phi_{\bm{\lambda}_n}(\bm{x}_n), \\[6pt]
	& \hat{Q}_n(x) \, \Phi_{\bm{\lambda}_n}(\bm{x}_n) = \exp(- e^{x_n - x}) \, \Phi_{\bm{\lambda}_n}(\bm{x}_n), \\[6pt]
	& \hat{Q}'_n(x) \, \Phi_{\bm{\lambda}_n}(\bm{x}_n) = \exp(- e^{x - x_1}) \, \Phi_{\bm{\lambda}_n}(\bm{x}_n). \label{Qprimec}
\end{align}
Diagonalization of dual Hamiltonians is proven in~\cite[Section 4]{B}, while diagonalization of Baxter operators is shown in~\cite[Proposition 3.2]{GLO1}. Note that by Corollary~\ref{cor:Phi-an} the Mellin--Barnes representation~\eqref{MBA} can be analytically continued to $\bm{\lambda}_n \in \mathbb{C}^n$, so that the action of dual Hamiltonians~\eqref{Hdual} is well defined.

Orthogonality relation for wave functions in the space of spectral parameters is equivalent to completeness relation in the space of spatial parameters
\begin{align} \label{gl-compl-2}
	\int_{\mathbb{R}^n} d\bm{\lambda}_n \; \hat{\mu}(\bm{\lambda}_n) \, \overline{\Phi_{\bm{\lambda}_n}(\bm{x}_n)} \, \Phi_{\bm{\lambda}_n}(
	\bm{y}_n) = \delta(x_1 - y_1) \cdots \delta(x_n - y_n).
\end{align}
It is proved by different methods in~\cite{W},~\cite[Section 3]{Kozl},~\cite[Section 2]{DKM2}.

\subsection{$BC$ system}

\subsubsection{Dual Hamiltonians}

In this section we describe the dual Hamiltonians derived in~\cite[Theorem 3]{DE}. Introduce rational functions 
\begin{align}
	\begin{aligned}\label{MB01}
		& 
		B(u_1, u_2) = \frac{1}{u_1^2 - u_2^2},\qquad C(u_1, u_2) = \frac{1}{(u_1 + u_2)(u_1 + u_2 + 1)}
	\end{aligned}
\end{align} 
and their products for any subsets of indices $I, J \subset \{1,\ldots,n\}$ (such that $I \cap J = \varnothing$)
\begin{align}\label{MB02}
	\begin{aligned}
		& B_{I, J}(\bm{u}_n) = \prod_{ \substack{i \in I \\ j \in J} } B(u_i, u_j), \\
		& C_I(\bm{u}_n) = \prod_{i \in I}(u_i + g) \prod_{\substack{i,j\in I\\ i\leq j}}C(u_i,u_j),\\ & C'_I(\bm{u}_n) = \prod_{i \in I}(u_i + g) \prod_{\substack{i,j\in I\\ i\leq j}}(-C(u_i,u_j)).		
	\end{aligned}
\end{align}
Besides, for any subset $I$ define symmetrizer
\begin{equation}
	\mathrm{Refl}_I = \prod_{i \in I} \mathrm{Refl}_{i}, \qquad \mathrm{Refl}_i \bigl[ \phi(u_i) \bigr]= \phi(u_i) + \phi(-u_i).
\end{equation}
Dual van Diejen--Emsiz Hamiltonians are given by
\begin{align} \label{HBC}
	\HBCdual_s = (2\beta)^s\sum_{ \substack{ I,J,K:\ |K|=n-s, \\ I\sqcup J \sqcup K=\{1, \dots, n\}  }}\!\!\! \mathrm{Refl}_{I\sqcup J} \biggl[B_{I,J}( \imath \bm{\lambda}_n) B_{I,K}( \imath \bm{\lambda}_n)  B_{J,K}( \imath \bm{\lambda}_n)\,
	C_I( \imath \bm{\lambda}_n ) \, C'_J( \imath \bm{\lambda}_n ) \, \, \prod_{i \in I} e^{- \imath\partial_{\lambda_i}} \biggr].
\end{align}
In particular, for $s = 1$
\begin{align}
	\HBCdual_1 = 2\beta\sum_{ i = 1}^n \,\prod_{ \substack{j = 1 \\ j\not=i} }^n \frac{1}{\lambda_i^2 - \lambda_j^2} \Biggl[ \, \frac{\imath \lambda_i +g }{2 \lambda_i (2 \lambda_i - \imath)} \,  \bigl( e^{-\imath \partial_{\lambda_i}} - 1 \bigr) +  \frac{-\imath \lambda_i + g }{2 \lambda_i (2 \lambda_i + \imath)}\,  \bigl( e^{\imath \partial_{\lambda_i}} - 1 \bigr) \Biggr] .
\end{align}
Dual Hamiltonians are formally self-adjoint with respect to the scalar product
\begin{align}
	\langle \psi | \phi \rangle_{\bcdmu} = \int_{\mathbb{R}^n} d\bm{\lambda}_n \; \bcdmu(\bm{\lambda}_n) \, \overline{\psi(\bm{\lambda}_n)} \, \phi(\bm{\lambda}_n)
\end{align}
with the measure~\eqref{bc-measure}
\begin{align}
	\bcdmu(\bm{\lambda}_n) = \frac{1}{n! \, (4\pi)^n } \prod_{1 \leq j < k \leq n} \frac{1}{\Gamma(\pm \imath \lambda_j \pm \imath \lambda_k) } \, \prod_{j = 1}^n \frac{ \Gamma(g \pm \imath \lambda_j ) }{ \Gamma(\pm2 \imath \lambda_j) }.
\end{align}
Moreover, each summand of the operator \rf{HBC} is formally self-adjoint. Due to factorized form of coefficients, it is sufficient to check the identities
\begin{align}
	\frac{m(\l+\imath)}{m(\l)}=\frac{\overline{c(\l)}}{c(\l+\imath)},\qquad
	\frac{m(\l+\imath,\rho+\imath)}{m(\l,\rho)}=\frac{\overline{c(\l,\rho)}}{c(\l+\imath,\rho+\imath)}
\end{align}
where for all real $\l,\rho$ and $a$ 
\begin{align}
	& c(\l)=  \frac{\imath \lambda +g }{2 \lambda (2 \lambda - \imath)(\lambda^2 - a^2)}, && \hspace{-1cm} m(\l)= \frac{ \Gamma(g \pm \imath \lambda ) }{ \Gamma(\pm2 \imath \lambda) \Gamma(\pm \imath \l \pm \imath a)}, \\[6pt]
	& {c}(\l,\rho)=\frac{1}{(\l + \rho)(\l + \rho - \imath)}, &&  \hspace{-1cm} m(\l,\rho)=\frac{1}{
		\Gamma(\pm \imath \l \pm \imath \rho)}.
\end{align} 

\subsubsection{$GL$--$BC$ intertwiners and Mellin--Barnes integrals} \label{sec:MB-int}

In Section~\ref{sec:gg-mb-equiv} we derived Mellin--Barnes representation for the $BC$ wave functions. The kernel of this representation is given by the function~\eqref{MB0}
\begin{align}
	K_1( \bm{\lambda}_m,  \bm{\gamma}_n)  = \frac{(2\beta)^{-\imath \bgg_n} \prod\limits_{j=1}^n \prod\limits_{k = 1}^m \Gamma(\imath \gamma_j \pm \imath \lambda_k) }{ \prod\limits_{1\leq j < k\leq n} \Gamma(\imath \gamma_j + \imath \gamma_k) \, \prod\limits_{j = 1}^n \Gamma (g + \imath \gamma_j ) }
\end{align}
with $m = n$. Consider the similar kernel function
\begin{align} \label{MB1}
	K_2( \bm{\lambda}_m,  \bm{\gamma}_n)  = \frac{(2\beta)^{-\imath \bgg_n} \prod\limits_{j = 1}^n \prod\limits_{k = 1}^m \Gamma(\imath \gamma_j \pm \imath \lambda_k)\prod\limits_{j = 1}^n \Gamma (g - \imath \gamma_j ) }{ \prod\limits_{1 \leq j < k \leq n} \Gamma(\imath \gamma_j + \imath \gamma_k) \, \prod\limits_{j = 1}^m \Gamma (g \pm \imath \lambda_j ) }.
\end{align}
To prove that $BC$ wave functions diagonalize dual Hamiltonians~\eqref{HBC} we need the following identity.

\begin{proposition} \label{propK} For $n\leq m$
	kernel functions $K_1(\bl_m,\bg_n)$ and $K_2(\bl_m,\bg_n)$ intertwine dual Hamiltonians of $GL$ and $BC$ types
	\begin{align}\label{MB2}
		\HBCdual_s(\bm{\lambda}_m) \, K_i( \bm{\lambda}_m,  \bm{\gamma}_n) = \Hdual_s(-\bm{\gamma}_n) \, K_i( \bm{\lambda}_m, \bm{\gamma}_n),\qquad i=1,2,\ \ s=1,\ldots,n.
	\end{align}
\end{proposition}

\begin{proof}
By straightforward calculation, the relation \rf{MB2} for the function $K_1( \bm{\lambda}_m,  \bm{\gamma}_n)$ is equivalent to the following identity on rational functions
\begin{multline}\label{MB5}
		\sum_{ \substack{ I,J,K:\ |K|=m-s, \\ I\sqcup J \sqcup K=\{1, \dots, m\}  }}\!\!\! \mathrm{Refl}_{I\sqcup J} \biggl[B_{I,J}(  \bm{x}_m) B_{I,K}( \bm{x}_m)  B_{J,K}(\bm{x}_m)\,
		C_I( \bm{x}_m ) \, C'_J(  \bm{x}_m )D_I^n(  \bm{x}_m,\bm{y}_n)\biggr]\\ 
		=\sum_{\substack{L,M:\ |L|=s,\\ L\sqcup M =\{1, \dots, n\} }}\tilde{B}_{L,M}(\bm{y}_n)\tilde{C}_L(\bm{y}_n) \tilde{D}_L^m(\bm{x}_m,\bm{y}_n),
\end{multline}
where we denote
\begin{align} \label{MB6}
	\begin{aligned}
		& \tilde{B}_{L,M}(\by_n)= \prod_{\substack{a\in L\\ b\in M}}\frac{y_a+y_b+1}{y_b-y_a}, && \qquad 
		\tilde{C}_L(\by_n)=\prod_{a\in L}(y_a+g)\prod_{\substack{a,b\in L\\ a<b}}(y_a+y_b)(y_a+y_b+1), \\[6pt]
		& \tilde{D}^m_L(\bx_m,\by_n)=\prod_{\substack{a\in L,\\ i=1,\ldots, m}}\frac{1}{y_a^2-x_i^2}, && \qquad  D_I^n( \bm{x}_m,\bm{y}_n)=
		\prod_{\substack{i\in I,\\ a=1,\ldots ,n}} \frac{y_a+x_i+1}{y_a-x_i}.
	\end{aligned}
\end{align} 
Let us prove the identity \rf{MB5}. Its both sides are rational functions of $\bx_n$ and $\by_n$, symmetric with respect to permutations of the variables $\by_n$ and with respect to permutations and sign changes of variables $\bx_m$. Besides, they are of degree less then zero with respect to each variable~$x_i$. Thus, it is sufficient to prove that the difference between two sides has no poles with respect to each $x_i$. All the poles are simple. Symmetry implies that there are no poles of the form $x_i=x_j$ and $x_i=0$. The only poles to check are 
$$x_i=\pm\frac{1}{2}, \qquad x_i=\pm x_j\pm 1, \qquad x_i = \pm y_j.$$
Again, due to symmetry of our functions it is sufficient to check three possibilities
\beqq x_m= - \frac{1}{2},\qquad x_1=-x_2-1, \qquad x_m=y_n. \eeqq
Let us prove that residues at these points from different terms cancel each other.

\paragraph{Pole at $x_m = -1/2$.}
Clearly, the right hand side of \rf{MB5} does not have such a pole. Denote by $b_{I,J,K}$ the summand of the left hand side of \rf{MB5}, which corresponds to the decomposition $\{1,\ldots,m\}=I\sqcup J\sqcup K$, and assume that $I$ contains $m$. Denote $I'=I\setminus\{m\}$ and  $J'=J\sqcup \{m\}$. Then 
\begin{multline}
	\mathop{\mathrm{Res}}_{x_m = -1/2} b_{I,J,K}=- \mathop{\mathrm{Res}}_{x_m = -1/2} b_{I',J',K} \\[6pt]
	= -\frac{1}{2} \biggl(g-\frac{1}{2}\biggr) \prod_{i\in I'}\frac{1}{x_i^2-1/4}\prod_{j\in J}\frac{1}{1/4-x_j^2}\prod_{k\in K}\frac{1}{1/4-x_k^2}B_{I',J}B_{I',K} B_{J,K} C_{I'}C'_J D_{I'}^n.
\end{multline} 
Thus, two residues cancel each other, which verifies that the total residue of the left hand side vanishes.

\paragraph{Pole at $x_1 = -x_2 - 1$.}
Again, the right hand side of \rf{MB5} does note have such a pole. In the left hand side there are contributions from the summands $b_{I,J,K}$ when either $I$ or $J$ contains both $1$ and $2$. Let $1,2 \in I$. Denote  $I''=I\setminus \{1,2\}$, $J''=J\sqcup\{1,2\}$. Then
\begin{align}
	\begin{aligned}
		& \mathop{\mathrm{Res}}_{x_1=-x_2-1} b_{I,J,K} = - \mathop{\mathrm{Res}}_{x_1=-x_2-1} b_{I'',J'',K} \\[6pt]
		& \quad\quad= \prod_{i\in I''} \frac{1}{(x_i^2-x_1^2)(x_i^2-x_2^2)}\prod_{j\in J} \frac{1}{(x_1^2-x_j^2)(x_2^2-x_j^2)}\prod_{k\in K} \frac{1}{(x_1^2-x_k^2)(x_2^2-x_k^2)}\\[6pt]
		& \quad\quad  \times \frac{1}{x_1+x_2} \frac{x_1+g}{2x_1(2x_1+1)} \frac{x_2+g}{2x_2(2x_2+1)}B_{I'',J}B_{I'',K} B_{J,K} C_{I''}C'_J D_{I''}^n\Big|_{x_1=-x_2-1}.
	\end{aligned}
\end{align}
Again, two residues cancel each other.

\paragraph{Pole at $x_m = y_n$.}
In the right hand side of \rf{MB5} we count residues of the summands $\tilde{b}_{L,M}$ for $L$ that contain~$n$. Denote $L'=L\setminus\{n\}$. Then
\beq\label{MB7} \mathop{\mathrm{Res}}_{x_m=y_n}\tilde{b}_{L,M}=-\frac{y_n+g}{2y_n}\prod_{l = 1}^{m - 1}\frac{1}{y_n^2-x_l^2}\prod_{c = 1}^{n - 1}\frac{y_c+y_n+1}{y_c-y_n}\tilde{B}_{L',M} \tilde{C}_{L'}\tilde{D}_{L'}^{m-1}.
\eeq
Similarly, in the left hand side of \rf{MB5} we count residues of the summands $b_{I,J,K}$ for $I$ containing~$m$. Again denote $I'=I\setminus \{m\}$. Then 
\beq\label{MB8} \mathop{\mathrm{Res}}_{x_m=y_n}{b}_{I,J,K}=-\frac{y_n+g}{2y_n}\prod_{l = 1}^{m - 1}\frac{1}{y_n^2-x_l^2}\prod_{c = 1}^{n - 1}\frac{y_c+y_n+1}{y_c-y_n} {B}_{I',J}B_{I',K}B_{J.K}C_{I'}C'_J {D}_{I'}^{n-1}.
\eeq
Summing up \rf{MB8} over all partitions appearing in \rf{MB5} we conclude that the residue $\Res_{x_m=y_n}$ of the left hand side of \rf{MB5} equals to the left hand side of \rf{MB5} evaluated in variables $\bx_{m-1}$, $\by_{n-1}$ and multiplied by the function
\beqq -\frac{y_n+g}{2y_n}\prod_{l = 1}^{m - 1}\frac{1}{y_n^2-x_l^2}\prod_{c = 1}^{n - 1}\frac{y_c+y_n+1}{y_c-y_n}.\eeqq
The same holds for the right hand side of \rf{MB5}. This proves that the difference of two sides residues vanishes by induction over $n$. Note that in the base case $n=0$ the right hand side vanishes and the left hand side is a rational function without poles of degree less then zero, so it also equals zero. 
\smallskip

The relation \rf{MB2} for the kernel $	K_2( \bm{\lambda}_m,  \bm{\gamma}_n)$ can be also rewritten in a form of identity \rf{MB5}, once we substitute $g-1-u_i$ and $g-1-y_a$ instead of $g+u_i$ and $g+y_a$ in the definitions \rf{MB02} and \rf{MB6} of the functions $C_I(\bm{u}_m)$ and $\tilde{C}_L(\by_n)$:
\begin{align*}  
	& C_I(\bm{u}_m) = \prod_{i \in I}( g-1 - u_i) \prod_{\substack{i,j\in I\\ i\leq j}}C(u_i,u_j),\\[6pt]
	& \tilde{C}_L(\by_n)=\prod_{a\in L}(g-1-y_a)\prod_{\substack{a,b\in L\\ a<b}}(y_a+y_b)(y_a+y_b+1).
\end{align*} 
The proof of the corresponding identity is the same.
\end{proof}

By Theorem \ref{theorem2.3}, $BC$ wave functions admit Mellin--Barnes representation
\begin{align} \label{Psi-MB5}
	\Psi_{\bm{\lambda}_n}(\bm{x}_n)  = e^{\beta e^{-x_1}}  \int_{(\mathbb{R} - \imath \epsilon)^n}  d\bm{\gamma}_n \; \hat{\mu}(\bm{\gamma}_n) \, K_1(\bl_n,\bg_n) \, \Phi_{\bm{\gamma}_n}(\bm{x}_n),
\end{align}
where $\epsilon > 0$. It is absolutely convergent for $|\Im \lambda_j| < \epsilon$ and analytically continues to $\bm{\lambda}_n \in \mathbb{C}^n$ by shifting contours (i.e. increasing $\epsilon$), see Proposition~\ref{prop:bc-mb-bound} and Corollary~\ref{cor:Psi-an}. Thus, the action of dual Hamiltonians~\eqref{HBC} on these functions is well defined, and with the help of Proposition~\ref{propK} we can prove the following statement.

\begin{theorem}  \label{theoremMB4} 
	The wave function $	\Psi_{\bm{\lambda}_n}(\bm{x}_n)$ satisfies van Diejen--Emsiz equations
	\begin{align} \label{Hdual-Psi}
		\HBCdual_s \,	\Psi_{\bm{\lambda}_n}(\bm{x}_n) = e^{x_{n - s + 1} + \ldots + x_n} \, \Psi_{\bm{\lambda}_n}(\bm{x}_n)
	\end{align}   
	for $s = 1, \dots, n$.
\end{theorem}

\begin{proof}
	By Proposition~\ref{prop:bc-mb-bound}, the integral~\eqref{Psi-MB5} converges uniformly in $\lambda_j$ from compact subsets of the strips $|\Im \lambda_j| < \epsilon$. Hence, taking $\epsilon > 1$ we can interchange difference operators $\HBCdual_s$ with integration in~\eqref{Hdual-Psi} and then use Proposition~\ref{propK}
	\begin{align}
		\begin{aligned}
			\HBCdual_s(\bm{\lambda}_n) \,	\Psi_{\bm{\lambda}_n}(\bm{x}_n) & = e^{\beta e^{-x_1}}  \int_{(\mathbb{R} - \imath \epsilon)^n}  d\bm{\gamma}_n \; \hat{\mu}(\bm{\gamma}_n) \, \Bigl[ \HBCdual_s(\bm{\lambda}_n) K_1(\bl_n,\bg_n) \Bigr ]\, \Phi_{\bm{\gamma}_n}(\bm{x}_n) \\[6pt]
			& = e^{\beta e^{-x_1}}  \int_{(\mathbb{R} - \imath \epsilon)^n}  d\bm{\gamma}_n \; \hat{\mu}(\bm{\gamma}_n) \, \Bigl[ \Hdual_s(-\bm{\gamma}_n) K_1(\bl_n,\bg_n) \Bigr ] \, \Phi_{\bm{\gamma}_n}(\bm{x}_n).
		\end{aligned}
	\end{align}
	Next shift integration variables $\gamma_j = \rho_j - \imath \epsilon$, so that $\rho_j \in \mathbb{R}$ and
	\begin{align}
		\Phi_{\bm{\gamma}_n}(\bm{x}_n) = e^{\epsilon \underline{\bm{x}}_n} \, \Phi_{\bm{\rho}_n}(\bm{x}_n) = e^{\epsilon \underline{\bm{x}}_n} \, \overline{\Phi_{-\bm{\rho}_n}(\bm{x}_n)},
	\end{align}
	see~\eqref{Lker-expl} and~\eqref{Phi-GG-0}. As a result,
	\begin{align}
		\HBCdual_s(\bm{\lambda}_n) \,	\Psi_{\bm{\lambda}_n}(\bm{x}_n) = e^{\beta e^{-x_1} + \epsilon \underline{\bm{x}}_n }  \int_{\mathbb{R}^n}  d\bm{\rho}_n \; \hat{\mu}(\bm{\rho}_n) \, \Bigl[ \Hdual_s(-\bm{\rho}_n) K_1(\bl_n,\bm{\rho}_n - \imath \epsilon \bm{e}_n) \Bigr ] \, \overline{\Phi_{-\bm{\rho}_n}(\bm{x}_n)}.
	\end{align}
	Now recall that $GL$ dual Hamiltonians are symmetric with respect to the scalar product~\eqref{dual-sp} with the measure $\hat{\mu}(\bm{\rho}_n) = \hat{\mu}(-\bm{\rho}_n)$
	\begin{align}
		\HBCdual_s(\bm{\lambda}_n) \,	\Psi_{\bm{\lambda}_n}(\bm{x}_n) = e^{\beta e^{-x_1} + \epsilon \underline{\bm{x}}_n }  \int_{\mathbb{R}^n}  d\bm{\rho}_n \; \hat{\mu}(\bm{\rho}_n) \,  K_1(\bl_n,\bm{\rho}_n - \imath \epsilon \bm{e}_n) \, \overline{ \Bigl[ \Hdual_s(-\bm{\rho}_n) \, \Phi_{-\bm{\rho}_n}(\bm{x}_n) \Bigr] }.
	\end{align}
	It is left to use the fact that $\Phi_{-\bm{\rho}_n}(\bm{x}_n)$ diagonalize Hamiltonians $\Hdual_s$~\eqref{Hdual-Phi}.
\end{proof}

In a similar way Proposition~\ref{propK} implies that the function  
\beq\label{MB10}	\tilde{\Psi}_{\bm{\lambda}_n}(\bm{x}_n)= e^{-\beta e^{-x_1}} \hspace{-0.3cm} \int\limits_{(\mathbb{R} - \imath 0)^n} \!\! d\bm{\gamma}_n \; \mu(\bm{\gamma}_n)K_2(\bl_n,\bg_n)\Phi_{\bm{\gamma}_n}(\bm{x}_n)\eeq
also solves van Diejen--Emsiz equations 
\beqq \label{MB9} 	\HBCdual_s 	\, \tilde{\Psi}_{\bm{\lambda}_n}(\bm{x}_n) =
e^{x_{n - s + 1} + \ldots + x_n} \, \Psi_{\bm{\lambda}_n}(\bm{x}_n).\eeqq
By Proposition~\ref{prop:bc-mb2-bound} the integral~\eqref{MB10} is absolutely convergent for $| \Im \lambda_j | < \epsilon < g$, so that for the last property we in addition assume $g > 1$. 

As explained in Section~\ref{sec:results}, in the case $n = 1$ the Mellin--Barnes integrals~\eqref{Psi-MB5},~\eqref{MB10} represent two known expressions for Whittaker function, see Examples~\ref{ex:MB1-1},~\ref{ex:MB2-1}. Let us prove that they also coincide in the general case.

\begin{proposition}\label{propK2} 
	For $\bm{\lambda}_n \in \mathbb{R}^n$ and $\epsilon \in (0, g)$
	\begin{align} 
		\Psi_{\bm{\lambda}_n}(\bm{x}_n)  = e^{-\beta e^{-x_1}}  \int_{(\mathbb{R} - \imath \epsilon)^n}  d\bm{\gamma}_n \; \hat{\mu}(\bm{\gamma}_n) \, K_2(\bm{\lambda}_n, \bm{\gamma}_n) \, \Phi_{\bm{\gamma}_n}(\bm{x}_n).
	\end{align}
\end{proposition}

\begin{proof}
	We have to prove the equality
	\begin{align} \label{MB-1}
		& e^{\beta e^{-x_1}} \hspace{-0.3cm} \int\limits_{(\mathbb{R} - \imath \epsilon)^n} \!\! d\bm{\gamma}_n \; \hat{\mu}(\bm{\gamma}_n) \, \frac{ \prod\limits_{j,k = 1}^n \Gamma(\imath \gamma_j \pm \imath \lambda_k) }{ \prod\limits_{1 \leq j < k \leq n} \Gamma(\imath \gamma_j + \imath \gamma_k) \, \prod\limits_{j = 1}^n \Gamma (g + \imath \gamma_j ) } \, (2\beta)^{-\imath \underline{\bm{\gamma}}_n } \, \Phi_{\bm{\gamma}_n}(\bm{x}_n) \\[6pt] \label{MB-2}
		& = e^{-\beta e^{-x_1}} \hspace{-0.3cm} \int\limits_{(\mathbb{R} - \imath \epsilon)^n} \!\! d\bm{\gamma}_n \; \hat{\mu}(\bm{\gamma}_n) \, \frac{ \prod\limits_{j,k = 1}^n \Gamma(\imath \gamma_j \pm \imath \lambda_k) \, \prod\limits_{j = 1}^n \Gamma (g - \imath \gamma_j )}{ \prod\limits_{1 \leq j < k \leq n} \Gamma(\imath \gamma_j + \imath \gamma_k) \, \prod\limits_{j = 1}^n \Gamma (g \pm \imath \lambda_j ) } \, (2\beta)^{-\imath \underline{\bm{\gamma}}_n } \, \Phi_{\bm{\gamma}_n}(\bm{x}_n).
	\end{align}
	For that we need two ingredients. The first one is the action of $GL$~Toda dual (reflected) Baxter operator $\hat{Q}'_n(x)$  on $\Phi_{\bm{\gamma}_n}(\bm{x}_n)$, see \rf{Qprime} and \rf{Qprimec}, 
	which explicitly reads
	\begin{align} \label{Qhr-expl}
		\int\limits_{(\mathbb{R} + \imath r)^n} \!\! d\bm{\rho}_n \; \hat{\mu}(\bm{\rho}_n) \, e^{\imath (\underline{\bm{\gamma}}_n - \underline{\bm{\rho}}_n) x} \, \prod_{j, k = 1}^n \Gamma(\imath \rho_j - \imath \gamma_k) \, \Phi_{\bm{\rho}_n}(\bm{x}_n) =  e^{-e^{x - x_1}} \, \Phi_{\bm{\gamma}_n}(\bm{x}_n).
	\end{align}
	Here $r < \Im \gamma_k$ for all $k = 1, \dots, n$. 
	
	The second ingredient is the Gustafson type of integral 
	\begin{align}\label{Gust}
		\int d\bm{\gamma}_n \; \hat{\mu}(\bm{\gamma}_n) \; \frac{ \prod\limits_{j = 1}^{n + 1} \prod\limits_{k = 1}^{n} \Gamma(a_j- \imath \gamma_k) \;  \prod\limits_{j, k = 1}^n \Gamma(\imath \gamma_j \pm b_k) \,}{\prod\limits_{1 \leq j < k \leq n} \Gamma(\imath \gamma_j + \imath \gamma_k)} = \frac{\prod\limits_{j = 1}^{n + 1} \prod\limits_{k = 1}^{n} \Gamma(a_j \pm b_k)}{ \prod\limits_{1 \leq j < k \leq n + 1} \Gamma( a_j + a_k) },
	\end{align}
	where we assume that integration contours separate series of poles from gamma functions
	\begin{align}
		\gamma_j = - \imath a_k - \imath n, \qquad \gamma_j = \pm \imath b_k + \imath n, \qquad n \in \mathbb{N}_0,
	\end{align}
	which go downwards and upwards correspondingly. The above integral was calculated in \cite[(3)]{DM}. Let us remark that this calculation was based on the assumption of orthogonality and completeness of various spin chain eigenfunctions, which has been proven later in \cite{DKM}. 
	
	Now consider the integral~\eqref{MB-2} and rewrite the exponent behind the integral
	\begin{align}
		e^{- \beta e^{-x_1}} = e^{\beta e^{x_1}} \, e^{- e^{\ln (2\beta) - x_1}}.
	\end{align}
	Notice that the second factor coincides with the eigenvalue of Baxter operator~\eqref{Qhr-expl} with the parameter $x = \ln(2\beta)$. Therefore, we can rewrite the second representation~\eqref{MB-2} using the integral~\eqref{Qhr-expl}
	\begin{multline} \label{MB-2-2}
		e^{\beta e^{-x_1}} \hspace{-0.3cm} \int\limits_{(\mathbb{R} - \imath \epsilon)^n} \!\! d\bm{\gamma}_n \; \hat{\mu}(\bm{\gamma}_n) \, \frac{ \prod\limits_{j,k = 1}^n \Gamma(\imath \gamma_j \pm \imath \lambda_k) \, \prod\limits_{j = 1}^n \Gamma (g - \imath \gamma_j )}{ \prod\limits_{1 \leq j < k \leq n} \Gamma(\imath \gamma_j + \imath \gamma_k) \, \prod\limits_{j = 1}^n \Gamma (g \pm \imath \lambda_j ) } \, (2\beta)^{-\imath \underline{\bm{\gamma}}_n } \\[6pt]
		\times \int\limits_{(\mathbb{R} + \imath r)^n} \!\! d\bm{\rho}_n \; \hat{\mu}(\bm{\rho}_n) \, (2\beta)^{\imath \underline{\bm{\gamma}}_n - \imath \underline{\bm{\rho}}_n} \, \prod_{j, k = 1}^n \Gamma(\imath \rho_j - \imath \gamma_k) \, \Phi_{\bm{\rho}_n}(\bm{x}_n) .
	\end{multline}
	Since $\Im \gamma_j = - \epsilon$, we should take $r < - \epsilon$. Next let us change order of integrals and integrate over $\bm{\gamma}_n$ first. The corresponding integral coincides with Gustafson integral~\eqref{Gust}, so it can be calculated explicitly
	\begin{multline}
		\int\limits_{(\mathbb{R} - \imath \epsilon)^n} \!\! d\bm{\gamma}_n \; \hat{\mu}(\bm{\gamma}_n) \; \frac{ \prod\limits_{j,k = 1}^n \Gamma(\imath \gamma_j \pm \imath \lambda_k) \, \prod\limits_{j, k = 1}^n \Gamma(\imath \rho_j - \imath \gamma_k) \, \prod\limits_{j = 1}^n \Gamma (g - \imath \gamma_j )}{ \prod\limits_{1 \leq j < k \leq n} \Gamma(\imath \gamma_j + \imath \gamma_k) } \\[6pt]
		= \frac{ \prod\limits_{j,k = 1}^n \Gamma(\imath \rho_j \pm \imath \lambda_k) \, \prod\limits_{j = 1}^n \Gamma(g \pm \imath \lambda_j) }{ \prod\limits_{1 \leq j < k \leq n} \Gamma(\imath \rho_j + \imath \rho_k) \, \prod\limits_{j = 1}^n \Gamma(g + \imath \rho_j)}.
	\end{multline}
	Inserting the right hand side into expression~\eqref{MB-2-2} one obtains the first Mellin--Barnes representation~\eqref{MB-1}.
\end{proof}

\subsubsection{QISM technique} \label{sec:QISM}
Recall the monodromy matrices of $GL$ and $BC$ types~\eqref{QI3},~\rf{QI4}
\begin{align}  \label{QI3-2}
	& T_n(u) = L_n(u) \cdots L_1(u) = 
	\begin{pmatrix}
		A_n(u) & B_n(u) \\[4pt]
		C_n(u) & D_n(u)
	\end{pmatrix}, \\[6pt] \label{QI4-2}
	& \T_n(u) 
	= T_n(u) \, K(u) \, \sigma_2 \, T_n^t(-u) \, \sigma_2 =
	\begin{pmatrix}
		\A_n(u) & \B_n(u) \\[4pt]
		\C_n(u) & \D_n(u)
	\end{pmatrix},
\end{align}
where
\begin{align}
	 L_j(u) = 
 	\begin{pmatrix}
 		u + \imath \partial_{x_j} & e^{-x_j} \\[4pt]
 		-e^{x_j} & 0
 	\end{pmatrix}, \qquad 
	 K(u) =  
	 	\begin{pmatrix}
	 		-\alpha & u - \frac{\imath}{2} \\[6pt]
	 		- \beta^2 \bigl(u - \frac{\imath}{2} \bigr) & -\alpha
	 	\end{pmatrix}, \qquad 
	 	\sigma_2=\begin{pmatrix}
		0& -\imath \\[4pt]\imath&0\end{pmatrix}.
\end{align}
In this section we follow the approach of the works~\cite{KL, KL2, IS} to prove the relations 
\begin{align}
	\label{QI10}
	& \B_n(u) \, \Psi_{\bm{\lambda}_n}(\bm{x}_n) = (-1)^n \biggl(u - \frac{\imath}{2} \biggr) \prod_{j = 1}^n (u^2 - \lambda_j^2) \, \Psi_{\bm{\lambda}_n}(\bm{x}_n), \\[6pt]
	& \D(\pm \lambda_k) \, \Psi_{\bm{\lambda}_n}(\bx_n) = - \beta (g \pm \imath \lambda_k) \, \Psi_{\l_1,\ldots, \l_k\mp \imath,\ldots \l_n}(\bx_n), \qquad k = 1, \dots, n. \label{QI11}
\end{align}
These relations were derived in \cite{BDK} using Gauss--Givental representation. 
Nevertheless, here we present their derivation using Mellin--Barnes representation and classical QISM technique.

Note that the formulas \rf{QI11}, together with the relation
\begin{align}
	\D_n( \imath/2) = -\alpha, 
\end{align} 
which follows from definition, give $2n + 1$ points to interpolate the action of polynomial~$\D_n(u)$ of degree $2n$. The differential-difference system of equations \rf{QI10}, \rf{QI11} fixes the wave functions $\Psi_{\bm{\lambda}_n}(\bm{x}_n)$ up to multiplication by $\imath$-periodic function of $\lambda_k$. Hence, it is natural to we expect that one can derive dual system of van Diejen--Emsiz equations from them. 

\paragraph{Action of $\B_n(u)$.} Due to \rf{QI4-2} the generating function of $BC$ Hamiltonians can be expressed in terms of $GL$ monodromy matrix elements
 \begin{align}\label{QI5} 
 	\begin{aligned}
 		\B_n(u)& =\alpha\left( A_n(u)B_n(-u)-B_n(u)A_n(-u)\right)\\[6pt] 
 		& + \left(u-\frac{\imath}{2}\right)\left(A_n(u)A_n(-u)+\beta^2B_n(u)B_n(-u)\right) .
 	\end{aligned}
 \end{align}
By Theorem~\ref{theorem2.3},
\beq\label{Kn6} \Psi_{\bl_n}(\bx_n)=e^{\beta e^{-x_1}}\int_{(\R-\imath \epsilon)^n}d\bt_n \;  Q(\bl_n,\bt_n) \, \Phi_{\bt_n}(\bx_n) \eeq
where $\epsilon > 0$ and we denote
\beq \label{Kn7} Q(\bl_n,\bt_n)=\frac{(2\b)^{-\imath \underline{\bm{t}}_n} \, \prod\limits_{a,i=1}^n\Gamma\left(\imath(\pm\l_i+t_a)\right)}
{\prod\limits_{a\not= b} \Gamma\left(\imath(t_a-t_b)\right)\prod\limits_{a<b}\Gamma\left(\imath(t_a+t_b)\right) \, \prod\limits_{a=1}^n\Gamma\left(g +\imath t_a\right)}.
\eeq
So, to prove~\eqref{QI10} we need to calculate the action of $\B_n(u)$ on $e^{\beta e^{-x_1}} \, \Phi_{\bm{t}_n}(\bm{x}_n)$. The action of operators $A_n(u)$ and $B_n(u)$ on $GL$ wave function is known~\cite[(7), (8)]{IS}
\begin{align}\label{Kn4} 
	& A_n(u) \, \Phi_{\bt_n}(\bx_n)=\prod_{l=1}^n(u-t_l) \, \Phi_{\bt_n}(\bx_n),\\
	& B_n(u) \, \Phi_{\bt_n}(\bx_n)=\imath^{n-1}\sum_{p=1}^n\Biggl(\prod_{l\not=p}\frac{u-t_l}{t_p
		-t_l}\Biggr) \, \Phi_{t_1,\ldots, t_p+\imath,\ldots, t_n}(\bx_n). \label{Kn5}
\end{align}
From definition~\eqref{QI3-2} we have
$$A_n(u)=\tilde{A}_{n-1}(u) \, \imath \partial_{x_1} + a_n(u), \qquad B_n(u)=\tilde{A}_{n-1}(u) \, e^{-x_1},$$
where $a_n(u)$ does not contain $\partial_{x_1}$ and $\tilde{A}_{n-1}(u)$
is the $(1,1)$ element of the matrix
$$L_n(u)\ldots L_2(u).$$ 
Thus, 
\beq\label{Kn9} A_n(u)=B_n(u) \, e^{x_1} \, \imath \partial_{x_1}+a_n(u).\eeq 
The operators $A_n(u)$ and $B_n(u)$ when acting on the function~\eqref{Kn6} may act either on the function $\Phi_{\bt_n}(\bx_n)$ or on the prefactor $e^{\beta e^{-x_1}}$. Denote by $A^0_n(u)$ and $B^0_n(u)$ their parts which act only on the function $\Phi_{\bt_n}(\bx_n)$.
Since
\beq\label{Kn10} e^{x_1}\, \imath \partial_{x_1} \, e^{\beta e^{-x_1}}= -\imath\beta \, e^{\beta e^{-x_1}}\eeq
we have the relations
\begin{align} \label{AB0}
	A_n(u) = A_n^0(u) - \imath \beta B_n^0(u), \qquad B_n(u)=B^0_n(u).
\end{align}
For future reference, note also that similarly
\begin{align} \label{CD0}
	C_n(u) = C_n^0(u) - \imath \beta D_n^0(u), \qquad D_n(u) = D_n^0(u).
\end{align}
Inserting~\eqref{AB0} into~\eqref{QI5} we obtain
\beq\label{Kn12} \begin{split}\B_n(u)& =\a\left( A^0_n(u)B^0_n(-u)-B_n^0(u)A_n^0(-u)\right)+ \left(u-\frac{\imath}{2}\right)A^0_n(u)A^0_n(-u)\\ &- \imath\beta\left(u-\frac{\imath}{2}\right)\left( A^0_n(u)B^0_n(-u)+B_n^0(u)A_n^0(-u)\right). \end{split}
\eeq 
Calculate each term separately using \rf{Kn4} and \rf{Kn5}. We have
\begin{align*} 
	& A^0_n(u)A^0_n(-u) \, \Phi_{\bt_n}(\bx_n)=\prod_{l=1}^n(t_l^2-u^2) \, \Phi_{\bt_n}(\bx_n),\\[6pt]
	& A^0_n(u)B^0_n(-u) \, \Phi_{\bt_n}(\bx_n)=\imath^{n-1}\sum_{p=1}^n(u-t_p-\imath)\prod_{l\not=p}\frac{t_l^2-u^2}{t_p-t_l}
	\, \Phi_{t_1,\ldots, t_p+\imath,\ldots,t_n}(\bx_n),\\[6pt]
	& B^0_n(u)A^0_n(-u) \, \Phi_{\bt_n}(\bx_n)=\imath^{n-1}\sum_{p=1}^n(-u-t_p)\prod_{l\not=p}\frac{t_l^2-u^2}{t_p-t_l}
	\, \Phi_{t_1,\ldots, t_p+\imath,\ldots,t_n}(\bx_n).
\end{align*}
Consequently,
\beqq \begin{split}
	& \left(A^0_n(u)B^0_n(-u)-	B^0_n(u)A^0_n(-u)\right) \Phi_{\bt_n}(\bx_n) \\[6pt]
	& = 2\imath^{n-1}\left(u-\frac{\imath}{2}\right)\sum_{p=1}^n\prod_{l\not=p}\frac{t_l^2-u^2}{t_p-t_l}
	 \, \Phi_{t_1,\ldots, t_p+\imath,\ldots,t_n}(\bx_n) \end{split}
\eeqq
and
\beqq \begin{split}
	& \left(A^0_n(u)B^0_n(-u)+	B^0_n(u)A^0_n(-u)\right) \Phi_{\bt_n}(\bx_n)\\[6pt]
	& = -2\imath^{n-1}\left(t_p+\frac{\imath}{2}\right)\sum_{p=1}^n\prod_{l\not=p}\frac{t_l^2-u^2}{t_p-t_l}
	\, \Phi_{t_1,\ldots, t_p+\imath,\ldots,t_n}(\bx_n). \end{split}
\eeqq
Thus,
\begin{multline} \label{Kn14}
	\biggl(\a\left( A^0_n(u)B^0_n(-u)-B_n^0(u)A_n^0(-u)\right) - \imath\beta\left(u-\frac{\imath}{2}\right)\left( A^0_n(u)B^0_n(-u)+B_n^0(u)A_n^0(-u)\right)\biggr)\Phi_{\bt_n}(\bx_n)\\[6pt]
	= 2\beta \, \imath^{n-1}\left(u-\frac{\imath}{2}\right)\sum_{p=1}^n\left( \imath t_p + g - 1\right)\prod_{l\not=p}\frac{t_l^2-u^2}{t_p-t_l}
	\, \Phi_{t_1,\ldots, t_p+\imath,\ldots,t_n}(\bx_n),
\end{multline}
where, as before, $g = 1/2 + \alpha/\beta$. Substitute the expression \rf{Kn14} into the integral \rf{Kn6} and make a change of variable
in each summand $t_p\to t_p-\imath$. We then get precisely the same expression, as in \cite[Appendix B]{IS},
\begin{multline}\label{Kn15}
	\biggl(\a\left( A^0_n(u)B^0_n(-u)-B_n^0(u)A_n^0(-u)\right) - \imath \beta\left(u-\frac{\imath}{2}\right)\left( A^0_n(u)B^0_n(-u)+B_n^0(u)A_n^0(-u)\right)\biggr)\Psi_{\bt_n}(\bx_n)\\[6pt]
	= \left(u-\frac{\imath}{2}\right)\sum_{p=1}^ne^{\beta e^{-x_1}}\int_{(\R - \imath \epsilon)^n}d\bt_n \;  Q(\bl_n,\bt_n)\prod_{l\not=p}\frac{t_l^2-u^2}{t_l^2-t_p^2}\prod_{a=1}^n(\l_a^2-t_p^2) \, \Phi_{\bt_n}(\bx_n).
\end{multline}
Here we also shifted integration contour for $t_p$, which is justified by bounds from Appendix~\ref{sec:MB-bounds-bc}.
The rest is the same, as in \cite[Appendix B]{IS}: using the relation 
\beq\label{Kn16} A^0_n(u)A^0_n(-u) \, \Psi_{\bt_n}(\bx_n)=e^{\beta e^{-x_1}}\int_{(\R - \imath \epsilon)^n}d\bt_n \; Q(\bl_n,\bt_n)\prod_{a=1}^n(\l_a^2-u^2) \, \Phi_{\bt_n}(\bx_n) \eeq
and the interpolation type identity \cite[(46)]{IS}
\beq\label{Kn17}\sum_{p=1}^n\prod_{l\not=p}\frac{t_l^2-u^2}{t_l^2-t_p^2}\prod_{a=1}^n(\l_a^2-t_p^2)+
\prod_{l=1}^n(t_l^2-u^2)=\prod_{a=1}^n(\l_a^2-u^2)\eeq
we arrive at the desired formula \rf{QI10}. 

\paragraph{Action of $\D_n(u)$.} By definition~\eqref{QI4-2},
\begin{align}
	\begin{aligned}
		\D_n(u) & = \alpha \bigl(C_n(u) B_n(-u) - D_n(u) A_n(-u)\bigr) \\[6pt]
		& + \biggl( u - \frac{\imath}{2} \biggr) \bigl( C_n(u) A_n(-u) + \beta^2 D_n(u) B_n(-u) \bigr).
	\end{aligned}
\end{align}
Inserting the relations~\eqref{AB0},~\eqref{CD0} into this we obtain
\begin{align}
	\begin{aligned}
		{\D}_n(u) &= \alpha \bigl(C_n^0(u) B_n^0(-u) - D_n^0(u) A_n^0(-u)\bigr) + \biggl(u - \frac{\imath}{2} \biggr) C_n^0(u) A_n^0(-u)\\[4pt]
		& - \imath \beta  \biggl(u - \frac{\imath}{2} \biggr)  \bigl( C_n^0(u) B_n^0(-u) + D_n^0(u) A_n^0(-u) \bigr).
	\end{aligned}
\end{align}
For brevity, denote
\begin{align}
	\bm{t} \equiv \bm{t}_n, \qquad \bm{t}^{\pm p} = (t_1, \dots, t_p \pm \imath, \ldots, t_n),\qquad t_{pq}=t_p-t_q.
\end{align}
To act with the operator $\D_n(u)$ on the Mellin--Barnes representation~\eqref{Kn6} we need the following known relations~\cite[(7), (8), (9), (11)]{IS}
\begin{align}\label{ABCD}
	\begin{aligned}
		& A_n^0(u) \, \Phi_{\bm{t}}= \prod_{l = 1}^n (u - t_l) \, \Phi_{\bm{t}}, \\[6pt]
		& B_n^0(u) \, \Phi_{\bm{t}} = \imath^{n - 1} \sum_{p = 1}^n \Biggl( \prod_{l \not=p}\frac{u - t_l}{t_{pl}} \Biggr) \, \Phi_{\bm{t}^{+p}}, \\[6pt]
		& C_n^0(u) \, \Phi_{\bm{t}} = \imath^{-n - 1} \sum_{p = 1}^n \Biggl( \prod_{l \not=p}\frac{u - t_l}{t_{pl}} \Biggr) \, \Phi_{\bm{t}^{-p}}, \\[8pt]
		& D_n^0(u) \,  \Phi_{\bm{t}}= \imath \sum_{p = 1}^n \Biggl( \prod_{l \not=p}\frac{u - t_l}{t_{pl}} \Biggr) \Biggl( \frac{1}{\prod\limits_{l \not=p} (t_{pl} - \imath)} - \frac{1}{\prod\limits_{l \not=p} (t_{pl} + \imath)} \Biggr) \, \Phi_{\bm{t}}\\[6pt]
		& \hspace{1.5cm} - \sum_{ \substack{p, q \\ p \not= q}} \frac{1}{(t_{qp} - \imath) \prod\limits_{l \not= p } t_{pl} } \Biggl( \prod_{l \not=p,q}\frac{u - t_l}{t_{ql}} \Biggr) \, \Phi_{\bm{t}^{+p, -q}}.
	\end{aligned}
\end{align}
From them we deduce the formulas, which already appear in \cite{IS},
\begin{align}
	& \begin{aligned}
		C_n^0(u) B_n^0(-u) \, \Phi_{\bm{t}} & = - \sum_{q} \Biggl( \prod_{l \not= q} \frac{t_l^2 - u^2}{t_{ql} (t_{ql} + \imath)} \Biggr) \, \Phi_{\bm{t}} \\[6pt]
		& + \sum_{\substack{p,q \\ p\not=q}} \frac{(u - t_p - \imath)(u + t_q)}{(t_{qp} - \imath)t_{pq}} \Biggl(\prod_{l \not=p,q} \frac{t_l^2 - u^2}{ t_{pl} t_{ql} } \Biggr) \, \Phi_{\bm{t}^{+p, - q}},
	\end{aligned} \\[10pt]
	& \begin{aligned}
		D_n^0(u) A_n^0(-u) \, \Phi_{\bm{t}} &= \sum_q \Biggl( \prod_{l \not= q} \frac{t_l^2 - u^2}{t_{ql} } \Biggr) \Biggl[ \frac{\imath (u + t_q)}{ \prod\limits_{ l \not= q}(t_{ql} + \imath)} -  \frac{\imath (u + t_q)}{ \prod\limits_{ l \not= q}(t_{ql} - \imath)}  \Biggr] \, \Phi_{\bm{t}}\\[6pt]
		& - \sum_{\substack{p,q \\ p\not=q}} \frac{(u + t_p )(u + t_q)}{(t_{qp} - \imath)t_{pq}} \Biggl(\prod_{l \not=p,q} \frac{t_l^2 - u^2}{ t_{pl} t_{ql} } \Biggr) \, \Phi_{\bm{t}^{+p, - q}},
	\end{aligned} \\[10pt]
	& C_n^0(u) A_n^0(-u) \, \Phi_{\bm{t}} = - \imath^{-n - 1} \sum_{q} (u + t_q)  \Biggl( \prod_{l \not= q} \frac{t_l^2 - u^2}{t_{ql} } \Biggr) \, \Phi_{\bm{t}^{-q}}.
\end{align}
Acting with ${\D}(u)$ on the eigenfunction \eqref{Kn6}, as always, we change integration variables and shift contours to write the result in the form
\begin{align}
	{\D}(u) \, \Psi_{\bm{\lambda}} = e^{\beta e^{-x_1}} \int_{(\mathbb{R} - \imath \epsilon)^n} d\bm{t} \; Q (\bm{t}) \, \Phi_{\bm{t}} \, R(u),
\end{align}
where for brevity we denote $Q(\bm{t}) \equiv Q(\bm{\lambda}_n, \bm{t}_n)$. The functions appearing after all shifts are divided into several parts
\begin{align}
	R(u) = \sum_{q = 1}^n \bigl[ R^1_q(u) + R^2_q(u) + R^3_q(u) + R^4_q(u) \bigr],
\end{align}
namely,
\begin{align} \label{R1}
	& R^1_q(u)  = -\imath^{-n - 1} \biggl(u - \frac{\imath}{2} \biggr) (u + t_q + \imath) \Biggl( \prod_{l \not= q} \frac{t_l^2 - u^2}{t_{ql} + \imath} \Biggr) \, \frac{Q(\bm{t}^{+q})}{Q(\bm{t})}, \\[6pt]  \label{R2}
	& \begin{aligned}
		R^2_q(u) = \sum^n_{ \substack{p = 1 \\ p \not= q} } \, & 2\beta \biggl( u - \frac{\imath }{2} \biggr)  \frac{(u + t_q + \imath)\bigl( \imath t_p + \frac{\alpha}{\beta } + \frac{1}{2} \bigr)}{(t_{qp} + \imath ) ( t_{pq} - 2\imath )}  \\[2pt] 
		& \times \Biggl(  \prod_{l \not= p,q} \frac{t_l^2 - u^2}{(t_{pl} - \imath)(t_{ql} + \imath)}\Biggr) \frac{Q(\bm{t}^{-p, +q})}{Q(\bm{t})},
	\end{aligned} 
\end{align}
\begin{align}\label{R3}
	& R_q^3(u) = \biggl[ \beta  \biggl(u - \frac{\imath}{2} \biggr) (u + t_q + \imath) - \imath \alpha (u + t_q - \imath) \biggr] \Biggl( \prod_{l \not= q} \frac{t_l^2 - u^2}{t_{ql} (t_{ql} + \imath)} \Biggr) \\[6pt]  \label{R4}
	& R_q^4(u) = (u + t_q) \biggl[ - \beta  \biggl(u - \frac{\imath}{2} \biggr) + \imath \alpha  \biggr] \Biggl( \prod_{l \not= q} \frac{t_l^2 - u^2}{t_{ql} (t_{ql} - \imath)} \Biggr) .
\end{align}
Due to the explicit form of the kernel $Q(\bm{t})$ \eqref{Kn7} we have
\begin{align}
	& \frac{Q(\bm{t}^{+q})}{Q(\bm{t})} = 2\beta  \, \imath^{n + 1} \frac{ \bigl( \imath t_q + \frac{\alpha}{\beta } - \frac{1}{2} \bigr) }{ \prod\limits_{l = 1}^n \bigl[(t_q + \imath)^2 - \lambda_l^2 \bigr] } \, \Biggl( \prod_{l \not= q} \frac{(t_{ql} + \imath)(t_q + t_l + \imath)}{t_{ql}} \Biggr), \\[6pt]
	& \begin{aligned}
		\frac{Q(\bm{t}^{-p, +q})}{Q(\bm{t})} & = \frac{(t_{pq} - 2\imath) \bigl( \imath t_q + \frac{\alpha}{\beta } - \frac{1}{2} \bigr)}{t_{pq}  \bigl( \imath t_p + \frac{\alpha}{\beta } + \frac{1}{2} \bigr)} \Biggl( \prod_{l = 1}^n \frac{t_p^2 - \lambda_l^2}{ (t_q + \imath)^2 - \lambda_l^2 } \Biggr) \\[6pt]
		& \times \Biggl( \prod_{l \not= p,q} \frac{(t_{pl} - \imath)(t_{ql} + \imath)(t_q + t_l + \imath)}{t_{pl} t_{ql} (t_p + t_l)} \Biggr).
	\end{aligned}
\end{align}
Using the above formulas we want to show that
\begin{align}\label{R}
	R(\lambda_k) = - \beta  \Biggl( \imath \lambda_k + \frac{\alpha}{\beta } + \frac{1}{2} \Biggr) \, \prod_{l = 1}^n \frac{t_l + \lambda_k}{t_l - \lambda_k + \imath}
\end{align}
for any $k \in \{1, \dots, n\}$. Then it is easy to see that
\begin{align}
	Q(\bm{\lambda}, \bm{t}) \prod_{l = 1}^n \frac{t_l + \lambda_k}{t_l - \lambda_k + \imath} = Q (\bm{\lambda}^{-k}, \bm{t})
\end{align}
and thus we prove the desired relation \eqref{QI11} when $u = \lambda_k$. The case $u = -\lambda_k$ follows from the reflection symmetry $\lambda_k \to -\lambda_k$ of eigenfunctions $\Psi_{\bm{\lambda}}$.

Due to symmetry with respect to permutations of $\lambda_k$, it is sufficient to consider $R(\lambda_n)$. Consider the sum of the first two terms~\eqref{R1},~\eqref{R2}
\begin{align}
	\begin{aligned}
		R^1_q(u) + R^2_q(u) & = - 2\beta  \, \frac{ \bigl(u - \frac{\imath}{2}\bigr) (u + t_q + \imath) \bigl( \imath t_q + \frac{\alpha}{\beta } - \frac{1}{2} \bigr) }{ \prod\limits_{l = 1}^n \bigl[ (t_q + \imath)^2 - \lambda_l^2 \bigr]  } \, \Biggl( \prod_{l \not= q} \frac{(t_l^2 - u^2)(t_q + t_l + \imath)}{t_{ql}} \Biggr)  \\[6pt]
		& \times \left[ 1 + \sum_{p \not= q} \frac{ \prod\limits_{l  =1}^n (t_p^2 - \lambda_l^2) }{ (t_p^2 - u^2) \bigl[(t_q + \imath)^2 - t_p^2\bigr]  \prod\limits_{l \not= p,q} (t_p^2 - t_l^2) } \right].
	\end{aligned}
\end{align}
In the case $u = \lambda_n$ the expression in square brackets drastically simplifies
\begin{align}
	1 + \sum_{p \not= q} \frac{ \prod\limits_{l  =1}^{n - 1} (t_p^2 - \lambda_l^2) }{  \bigl[(t_q + \imath)^2 - t_p^2\bigr]  \prod\limits_{l \not= p,q} (t_p^2 - t_l^2) }  = \frac{ \prod\limits_{l  =1}^{n - 1} \bigl[ (t_q + \imath)^2 - \lambda_l^2 \bigr]  }{ \prod_{l \not= q} \bigl[ (t_q + \imath)^2 - t_l^2 \bigr] }.
\end{align}
As mentioned in \cite[p. 13]{IS}, this is a consequence of the interpolation identity formula
\begin{align} \label{QI13}
	\sum_{p = 1}^n \frac{ \prod_{l = 1}^{n - 1}(a_p - b_l) }{ \prod_{l \not= p}(a_p - a_l)} = 1.
\end{align}
Thus, we have
\begin{align}
	R^1_q(\lambda_n) + R^2_q(\lambda_n) = - 2\beta  \, \frac{ \bigl( \lambda_n - \frac{\imath}{2} \bigr) \bigl( \imath t_q + \frac{\alpha}{\beta } - \frac{1}{2} \bigr)}{t_q - \lambda_n + \imath} \Biggl( \prod_{l \not=q} \frac{t_l^2 - \lambda_n^2}{t_{ql} (t_{ql} + \imath)} \Biggr).
\end{align}
The next step is to sum up the third term \eqref{R3}
\begin{align}
	R_q^1(\lambda_n) + R_q^2(\lambda_n) + R_q^3(\lambda_n) = \beta  \, \frac{ (t_q^2 - \lambda_n^2) \bigl(\imath \lambda_n + \frac{\alpha}{\beta } + \frac{1}{2} \bigr) }{ \imath t_q - \imath \lambda_n - 1 } \Biggl( \prod_{l \not=q} \frac{t_l^2 - \lambda_n^2}{t_{ql} (t_{ql} + \imath)}  \Biggr).
\end{align}
Finally, if we add the forth term \eqref{R4} and sum over $q$ we recover the whole function
\begin{multline}
		R(\lambda_n) = \sum_{q = 1}^n \bigl[ R_q^1(\lambda_n) + R_q^2(\lambda_n) + R_q^3(\lambda_n) + R_q^4(\lambda_n) \bigr] \\[6pt]
		= \beta  \biggl( \imath \lambda_n + \frac{\alpha}{\beta } + \frac{1}{2} \biggr) \, \sum_{q = 1}^n \Biggl( \prod_{l \not= q} \frac{t_l^2 - \lambda_n^2}{ t_{ql} } \Biggr) \Biggl[ \frac{t_q^2 - \lambda_n^2}{ (\imath t_q - \imath \lambda_n - 1 ) \prod\limits_{l \not= q} (t_{ql} + \imath) } + \frac{\imath t_q + \imath \lambda_n}{ \prod\limits_{l \not= q} (t_{ql} - \imath) }\Biggr].
\end{multline}
Again, as it is proved in~\cite[p. 13]{IS}, the sum in the last line can be rewritten as
\begin{align}
	\sum_{q = 1}^n \Biggl( \prod_{l \not= q} \frac{t_l^2 - \lambda_n^2}{ t_{ql} } \Biggr) \Biggl[ \frac{t_q^2 - \lambda_n^2}{ (\imath t_q - \imath \lambda_n - 1 ) \prod\limits_{l \not= q} (t_{ql} + \imath) } + \frac{\imath t_q + \imath \lambda_n}{ \prod\limits_{l \not= q} (t_{ql} - \imath) }\Biggr] = - \prod_{l = 1}^n \frac{t_l + \lambda_n}{t_l - \lambda_n + \imath},
\end{align} 
and this gives us the result \eqref{R}. The last formula is again a consequence of interpolation identity \rf{QI13}.

\subsubsection{Asymptotics and hyperoctahedral Whittaker function} \label{sec:asymp}

In this section we calculate Mellin--Barnes integrals for $GL$ and $BC$ wave functions by residues. This gives us Harish-Chandra series representations and asymptotics of wave functions.

\paragraph{$GL$ system.}
It is known that the wave function $\Phi_{\bl_n}(\bx_n)$ fast decreases in all asymptotical zones, except
\beq\label{AS1} x_1<x_2<\ldots< x_n,\eeq
e.g. see the bound given in~\cite[Proposition 4.1.3]{BC}. In the latter zone it can be represented by a convergent series, which comes from residue calculation of Mellin--Barnes integral~\rf{MBA}. Note that for $x_1 \ll \ldots \ll x_n$ the interaction between particles disappears~\eqref{HA}.

Let $\mathcal{A}$ be the set of all lower triangular matrices with zeroes on diagonal, whose elements are nonnegative integers
\beq\label{AS2} \mathcal{A}=\{A = (a_{ij}) \in M_n(\mathbb{N}_0) \, |\, a_{ij} = 0 \text{ for all } i \leq j \}.\eeq
Each matrix $A\in \mathcal{A}$ determines another lower triangular matrix $P=(p_{ij})$ by the rule
\beq\label{AS3} p_{ij}=a_{ij}+a_{i+1,j}+\ldots+ a_{nj},\qquad 1\leq j<i\leq n.\eeq
For a given $A\in \mathcal{A}$ define the meromorphic function
 \beq\label{AS4} c_A(\bl_n)=\prod\limits_{k=2}^n\prod\limits_{i=1}^{k-1}\frac{(-1)^{a_{ki}}}{a_{ki}!}
 	\frac{\prod\limits_{j=1, j\not=i}^{k}\Gamma(\imath(\l_j-\l_i)- p_{ki}+p_{k+1,j})}{\prod\limits_{j=1,\, j\not=i}^{k-1}\Gamma(\imath(\l_j-\l_i)- p_{ki}+p_{kj})},\eeq
where we set $p_{n+1,i}=0.$
For example, for $n=3$
\begin{multline}
	c_A(\l_1,\l_2,\l_3) = \frac{(-1)^{a_{31}+a_{32}+a_{21}}}{a_{31}!a_{32}!a_{21}!}	\Gamma (\imath(\l_2-\l_1)-a_{21}-a_{31}+a_{32})\\[6pt]
	\times \frac{\Gamma\left(\imath(\l_3-\l_1)-a_{31}\right)\Gamma\left(\imath(\l_2-\l_1)-a_{31}\right)\Gamma\left(\imath(\l_3-\l_2)-a_{32}\right)\Gamma\left(\imath(\l_1-\l_2)-a_{32}\right)}{\Gamma\left(\imath(\l_2-\l_1)-a_{31}+a_{32}\right)\Gamma\left(\imath(\l_1-\l_2)-a_{32}+a_{31}\right)}.
\end{multline}
 Denote by $\vf_{\bl_n}(\bx_n)$ the series
 \beq\label{AS5} \vf_{\bl_n}(\bx_n)=e^{\imath\sum_{i=1}^n \l_ix_i}\sum_{A\in \mathcal{A}}c_A(\bl_n) \, e^{-\sum_{j<i}a_{ij}(x_i-x_j)}. \eeq
 Then the wave function $\Phi_{\bl_n}(\bx_n)$ can be represented (actually for any $\bx_n\in\C^n$~\cite[Proposition~1]{DE}) as the symmetrization of $\vf_{\bl_n}(\bx_n)$ over the spectral parameters
 	\beq\label{AS6}\Phi_{\bl_n}(\bx_n)=\sum_{\sigma\in S_n}\vf_{\l_{\sigma(1)},\ldots, \l_{\sigma(n)}}(\bx_n), \eeq 
 which follows from straightforward calculation of Mellin--Barnes integral~\eqref{MBA} by residues.
 
 \begin{remark}
 	For $n=2$ the decomposition \rf{AS6} expresses Macdonald function via modified Bessel functions \cite[\href{http://dlmf.nist.gov/10.27.E4}{(10.27.4)}]{DLMF},
 	\beqq K_\nu(z)=\frac{\pi}{2} \frac{I_{-\nu}(z)-I_
 		\nu(z)}{\sin \pi \nu}.\eeqq
 \end{remark}

The nearest poles ($a_{ij}=0$) determine the asymptotics of the Mellin--Barnes integral for $x_1 \ll \ldots \ll x_n$. Set
\beq \label{AS7}\Phi^{as}_{\bl_n}(\bx_n)=\sum_{\sigma\in S_n}\prod_{i<j}\Gamma\left(\imath(\l_{\sigma(j)}-\l_{\sigma(i)})\right)e^{\imath\sum_{k = 1}^n\l_{\sigma(k)}x_k}\eeq
Following the arguments presented in \cite{HR3} (see also \cite[Section 5.1]{BCDK}) one can show that in the region \rf{AS1} under condition  
 \beq \min_{i = 1, \dots, n-1} (x_{i+1}-x_i)=\ve>0 \eeq
 we have a bound
 \beq\label{AS8} |\Phi_{\bl_n}(\bx_n)-\Phi^{as}_{\bl_n}(\bx_n)|<C(\bl_n)
 \exp (-\ve/2)\eeq
 for any $\bl_n\in\R^n$, such that $\l_i\not=\l_j$.

The eigenfunction of $GL$ Toda Hamiltonians analytical in 
the region $x_1 < \ldots < x_n$ with plane wave asymptotic 
$e^{\imath\sum_k\l_k x_k}$ is usually called the Harish-Chandra function $\phi^{HC}_{\bl_n}(\bx_n)$~\cite{HC}.
It can be represented in a form of convergent series over 
variables $e^{-(x_k-x_{k+1})}$~\cite{DE0, DE,O}, and after 
the symmetrization over spectral parameters it gives the eigenfunction 
symmetric and analytical over spectral parameters.

Thus,
\beq\label{AS8a} \vf_{\bl_n}(\bx_n)= \prod_{i<j} \Gamma\left(\imath (\l_j-\l_i)\right)\phi^{HC}_{\bl_n}(\bx_n).\eeq
and the relation \rf{AS6} is precisely Harish-Chandra decomposition of the wave function $\Phi_{\bl_n}(\bx_n)$.

\paragraph{$BC$ system.}
Denote by $W_n=S_n \ltimes Z_2^n$ the Weyl group of $BC_n$ root system. Each element $w\in W_n$ is parametrized by permutation $\sigma(w)\in S_n$ and the collection of signs $\ve_k(w) \in \{1, -1\}$, so that its right action on $\R^n$ is given by the relation
\beqq w(x_k) =\ve_k(w)x_{\sigma(w)(k)}.\eeqq
Let $D_0$ be the asymptotical  zone
\beq \label{AS9}D_0=\{\bx_n\in\R^n\ |\ 0<x_1<x_2<\ldots <x_n\}.\eeq
Then the complement of $\R^n$ to the union of hyperplanes $x_i=\pm x_j$ and $x_k=0$ is a disjoint union of asymptotical domains
\beq\label{AS10} D_w =w(D_0),\qquad w\in W_n.\eeq
The wave function $\Psi_{\bl_n}(\bx_n)$ rapidly decays in all asymptotical domains except $D_0$, see Proposition~\ref{prop:bc-bound}. The Mellin--Barnes representation allows to represent it as a series analogous to \rf{AS5} and \rf{AS6}, which conjecturally converges for all $\bx_n\in\R^n$.

Let $\mathcal{B}$ be the set of all tuples of nonnegative integers (one can visualize them as a collection of a lower triangular matrix and a vector)
\beq\label{AS11} \mathcal{B}=\{B=(b_{ij},\ b_k) \in M_n(\mathbb{N}_0) \times \mathbb{N}_0^n \,|\, b_{ij} = 0 \text{ for all } i \leq j \}.\eeq
Each array $B\in \mathcal{B}$ determines a lower triangular matrix $Q=(q_{ij})$ by the rule
\beq\label{AS12} q_{ij}=b_{ij}+b_{i+1,j}+\ldots+ b_{nj}- b_j,\qquad 1\leq j<i\leq n, \qquad q_{n+1,j}=-b_j.\eeq
 For a given $B\in \mathcal{B}$ define the meromorphic functions $c_B^k(\bl_n)$ by the rules
\begin{align}\label{AS13} c_B^k(\bl_n)&=\prod\limits_{i=1}^{k-1}\frac{(-1)^{b_{ki}}}{b_{ki}!}
\frac{\prod\limits_{j=1, j\not=i}^{k}\Gamma(\imath(\l_j-\l_i)- q_{ki}+q_{k+1,j})}{\prod\limits_{j=1,\, j\not=i}^{k-1}\Gamma(\imath(\l_j-\l_i)- q_{ki}+q_{kj})},\qquad 2\leq k\leq n,\\
\notag c_B^1(\bl_n)&= \prod\limits_{k=1}^n\frac{(-1)^{b_k}\Gamma\left(2\imath\l_k-b_k\right)}{b_k!\Gamma\left(\imath\l_k+g-b_k\right)}\prod\limits_{1\leq j\not=i\leq n}\frac{\Gamma(\imath(\l_i\pm\l_j)-b_i) }{
\Gamma(\imath(\l_i-\l_j)-b_i+ b_j)}\prod\limits_{1\leq j<i\leq n}\frac{1}{\Gamma(\imath(\l_i+\l_j) - b_i - b_j)}
\end{align}
and set \beqq c_B(\bl_n)=\prod_{k=1}^nc_B^k(\bl_n). \eeqq 
Denote by $\psi_{\bl_n}(\bx_n)$ the series
\beq\label{AS15} \psi_{\bl_n}(\bx_n)=e^{\ \imath\sum\limits_{i=1}^n \l_i\tilde{x}_i + \beta e^{-x_1}}\sum_{B\in \mathcal{B}}c_B(\bl_n) \, e^{\ \sum\limits_{j<i}b_{ij}(x_j-x_i)-\sum\limits_{k=1}^nb_k \tilde{x}_k} \eeq
where
\beq  \tilde{x}_k=x_k-\ln 2\beta. \eeq
The straightforward calculation of the Mellin--Barnes integral~\eqref{Kn6} by residues together with the series representation of $GL$ wave function~\eqref{AS6} give the following statement.

\begin{theorem}\label{theoremAS} The wave function $\Psi_{\bl_n}(\bx_n)$ is the symmetrization of the series \rf{AS15}
	\beq\label{AS16} \Psi_{\bl_n}(\bx_n)=\sum_{w\in W_n}\psi_{w(\l_1),\ldots, w(\l_n)}(\bx_n).\eeq
	\end{theorem}

\begin{example}
	For $n=1$ the relation \rf{AS16} is the expression of Whittaker function
	$W_{k,\mu}(z)$ via Whittaker functions $M_{ k,\pm \mu}(z)$~\cite[\href{http://dlmf.nist.gov/13.14.E33}{(13.14.33)}]{DLMF}
	\beqq W_{k,\mu}(z)=\frac{\Gamma(-2\mu)}{\Gamma \bigl(\frac{1}{2} -\mu-k \bigr)}M_{k,\mu}(z)+
	\frac{\Gamma(2\mu)}{\Gamma \bigl(\frac{1}{2} +\mu-k \bigr)}M_{k,-\mu}(z).\eeqq
\end{example}

Again, the nearest poles $(b_{ij}=b_i=0)$ determine the asymptotics of the Mellin--Barnes integral in the domain $0 \ll x_1 \ll \ldots \ll x_n$. Set
\beq \label{AS17}\Psi^{as}_{\bl_n}(\bx_n)=\sum_{w\in W_n}\prod_{k = 1}^n\frac{\Gamma(2\imath w(\l_k))}{\Gamma(\imath w(\l_k)+g)}\prod_{1 \leq i<j\leq n}\Gamma(\imath w(\l_{j})\pm \imath w(\l_{i})) \, e^{\imath\sum_{k = 1}^n w(\l_{k})\tilde{x}_k}.\eeq
Following the approach of~\cite{HR3, BCDK} it is possible to show that in the region \rf{AS1} under condition  
\beq\label{AS18} \min_i \{(x_{i+1}-x_i), x_i\}=\ve>0 \eeq
we have the bound
\beq\label{AS19} |\Psi_{\bl_n}(\bx_n)-\Psi^{as}_{\bl_n}(\bx_n)|<C(\bl_n)
\exp (-\ve/2)\eeq
for any $\bl_n\in\R^n$, such that $\l_i\not=\pm\l_j$ and $\l_j\not=0$.
This is precisely the asymptotics of the hyperoctahedral Whittaker function~\cite[(3.5b)]{DE} (to compare set $2\beta = 1$). We thus conclude that $\Psi_{\bl_n}(\bx_n)$ coincides with hyperoctahedral Whittaker function due to uniqueness of the latter, and the relation \rf{AS16} is its Harish-Chandra decomposition. As shown in~\cite{DE}, hyperoctahedral Whittaker function is entire in spectral parameters, which is in accordance with our Corollary~\ref{cor:Psi-an}.

\subsubsection{Completeness} \label{sec:compl}

As we already discussed, the wave functions of $GL$ Toda chain satisfy completeness relation~\eqref{gl-compl-2}
\begin{align} \label{compl-gl}
	\int_{\mathbb{R}^n} d\bm{\lambda}_n \; \hat{\mu}(\bm{\lambda}_n) \, \overline{\Phi_{\bm{\lambda}_n}(\bm{x}_n)} \, \Phi_{\bm{\lambda}_n}(\bm{y}_n) = \delta(x_1 - y_1) \cdots \delta(x_n - y_n),
\end{align}
where $\bm{x}_n, \bm{y}_n \in \mathbb{R}^n$ and the measure is given by~\eqref{gl-measure}. In this section we present heuristic derivation of the completeness relation for $BC$ wave functions
\begin{align} \label{compl-bc}
	\int_{\mathbb{R}^n} d\bm{\lambda}_n \; \bcdmu(\bm{\lambda}_n) \, \overline{\Psi_{\bm{\lambda}_n}(\bm{x}_n)} \, \Psi_{\bm{\lambda}_n}(\bm{y}_n) = \delta(x_1 - y_1) \cdots \delta(x_n - y_n),
\end{align}
where the corresponding measure is determined by orthogonality relation
\begin{align}
	\bcdmu(\bm{\lambda}_n) = \frac{1}{n! \, (4\pi)^n } \prod_{1 \leq j < k \leq n} \frac{1}{\Gamma(\pm \imath \lambda_j \pm \imath \lambda_k) } \, \prod_{j = 1}^n \frac{ \Gamma(g \pm \imath \lambda_j ) }{ \Gamma(\pm2 \imath \lambda_j) },
\end{align}
see Section~\ref{sec:orth}. 

By Corollary~\ref{cor:Psi-real}, $\overline{\Psi_{\bm{\lambda}_n}(\bm{x}_n)} = \Psi_{\bm{\lambda}_n}(\bm{x}_n)$. Using it and the Mellin--Barnes representation~\eqref{Psi-MB} we rewrite the integral~\eqref{compl-bc} as
\begin{multline} \label{compl2}
	e^{\beta (e^{-x_1} + e^{-y_1})} \, \int\limits_{\mathbb{R}^n} d\bm{\lambda}_n \int\limits_{(\mathbb{R} - \imath \epsilon)^n} \!\! d\bm{\gamma}_n \int\limits_{(\mathbb{R} - \imath \epsilon)^n} \!\! d\bm{\rho}_n  \; \bcdmu(\bm{\lambda}_n) \, \hat{\mu}(\bm{\gamma}_n) \, \hat{\mu}(\bm{\rho}_n)  \\[6pt]
	\times \frac{ (2\beta)^{-\imath \underline{\bm{\gamma}}_n -\imath \underline{\bm{\rho}}_n} \prod\limits_{j,k = 1}^n \Gamma(\imath \gamma_j \pm \imath \lambda_k) \, \Gamma(\imath \rho_j \pm \imath \lambda_k) }{ \prod\limits_{1 \leq j < k \leq n} \Gamma(\imath \gamma_j + \imath \gamma_k) \, \Gamma(\imath \rho_j + \imath \rho_k) \, \prod\limits_{j = 1}^n \Gamma(g + \imath \gamma_j) \, \Gamma(g + \imath \rho_j)} \, \Phi_{\bm{\gamma}_n}(\bm{x}_n) \, \Phi_{\bm{\rho}_n}(\bm{y}_n) .
\end{multline}
First, let us calculate the integral over $\bm{\lambda}_n$. Note that it coincides with the reduction of Gustafson integral~\eqref{Gust-red} derived in Appendix~\ref{sec:gust-red}, so that
\begin{multline}
	\frac{1}{n! \, (4\pi)^n } \int\limits_{\mathbb{R}^n} d\bm{\lambda}_n \; \frac{ \prod\limits_{j,k = 1}^n \Gamma(\imath \gamma_j \pm \imath \lambda_k) \, \Gamma(\imath \rho_j \pm \imath \lambda_k) \prod\limits_{j = 1}^n \Gamma(g + \imath \gamma_j)}{ \prod\limits_{1 \leq j < k \leq n} \Gamma(\pm \imath \lambda_j \pm \imath \lambda_k)  \prod\limits_{j = 1}^n \Gamma(\pm 2\imath \lambda_j)} \\
	= \prod_{1 \leq j < k \leq n} \Gamma(\imath \gamma_j + \imath \gamma_k) \, \Gamma(\imath \rho_j + \imath \rho_k) \, \prod_{j = 1}^n \Gamma(g + \imath \gamma_j) \, \Gamma(g + \imath \rho_j) \prod_{j, k = 1}^n \Gamma(\imath \gamma_j + \imath \rho_k).
\end{multline}
Hence, the whole expression~\eqref{compl2} becomes
\begin{multline} \label{compl3}
	e^{\beta (e^{-x_1} + e^{-y_1})} \, \int\limits_{(\mathbb{R} - \imath \epsilon)^n} \!\! d\bm{\gamma}_n \int\limits_{(\mathbb{R} - \imath \epsilon)^n} \!\! d\bm{\rho}_n  \;  \hat{\mu}(\bm{\gamma}_n) \, \hat{\mu}(\bm{\rho}_n)  \\[-6pt]
	\times (2\beta)^{-\imath \underline{\bm{\gamma}}_n -\imath \underline{\bm{\rho}}_n} \prod_{j,k =1}^n \Gamma(\imath \gamma_j + \imath \rho_k) \, \Phi_{\bm{\gamma}_n}(\bm{x}_n) \, \Phi_{\bm{\rho}_n}(\bm{y}_n) .
\end{multline}
Now the integral over $\bm{\rho}_n$ coincides with the action of $GL$ dual Baxter operator $\hat{Q}'_n(\ln 2\beta)$ on its eigenfunction, see~\rf{Qprime} and~\rf{Qprimec},
\begin{align}
	\int\limits_{(\mathbb{R} - \imath \epsilon)^n} \!\! d\bm{\rho}_n  \;  \hat{\mu}(\bm{\rho}_n) \, (2\beta)^{-\imath \underline{\bm{\gamma}}_n -\imath \underline{\bm{\rho}}_n} \prod_{j,k =1}^n \Gamma(\imath \gamma_j + \imath \rho_k) \, \Phi_{\bm{\rho}_n}(\bm{y}_n) = e^{- 2\beta e^{-y_1}} \, \Phi_{-\bm{\gamma}_n}(\bm{y}_n).
\end{align}
Consequently, the integral~\eqref{compl3} reduces to
\begin{align} \label{compl4}
	e^{\beta (e^{-x_1} - e^{-y_1})} \, \int\limits_{(\mathbb{R} - \imath \epsilon)^n} \!\! d\bm{\gamma}_n \; \hat{\mu}(\bm{\gamma}_n)  \, \Phi_{\bm{\gamma}_n}(\bm{x}_n) \, \Phi_{-\bm{\gamma}_n}(\bm{y}_n) .
\end{align}
Finally, changing integration variables $\gamma_j = \nu_j - \imath \epsilon$ (so that $\nu_j \in \mathbb{R}$) and using properties of $GL$ wave functions
\begin{align}
	\Phi_{\bm{\gamma}_n}(\bm{x}_n) = e^{\epsilon \underline{\bm{x}}_n } \, \Phi_{\bm{\nu}_n}(\bm{x}_n), \qquad \Phi_{-\bm{\gamma}_n}(\bm{y}_n) = e^{- \epsilon \underline{\bm{y}}_n } \, \Phi_{-\bm{\nu}_n}(\bm{y}_n) = e^{- \epsilon \underline{\bm{y}}_n } \, \overline{\Phi_{\bm{\nu}_n}(\bm{y}_n)}
\end{align}
we obtain
\begin{multline} 
	e^{\beta (e^{-x_1} - e^{-y_1})} \, \int\limits_{(\mathbb{R} - \imath \epsilon)^n} \!\! d\bm{\gamma}_n \; \hat{\mu}(\bm{\gamma}_n)  \, \Phi_{\bm{\gamma}_n}(\bm{x}_n) \, \Phi_{-\bm{\gamma}_n}(\bm{y}_n) \\
	= e^{\beta (e^{-x_1} - e^{-y_1}) + \epsilon (\underline{\bm{x}}_n - \underline{\bm{y}}_n)} \, \int\limits_{\mathbb{R}^n} d\bm{\nu}_n \; \hat{\mu}(\bm{\nu}_n) \, \Phi_{\bm{\nu}_n}(\bm{x}_n) \, \overline{\Phi_{\bm{\nu}_n}(\bm{y}_n)} = \delta(x_1 - y_1) \cdots \delta(x_n - y_n),
\end{multline}
where the last equality follows from the completeness relation for $GL$ wave functions~\eqref{compl-gl}. This concludes the derivation of~\eqref{compl-bc}. 

\section*{Acknowledgments} 
The second and the third authors thank BIMSA for hospitality. A big part of the work was done during their visit to BIMSA. The work of S. Derkachov (Section 2) was supported by RNF grant 23-11-00311. The work of S. Khoroshkin (Section 3) was supported by RNF grant 23-11-00150. 

\appendix

\section{Integral identities} \label{appA}
\subsection{Proof of flip relation} \label{app:flip}
Recall that flip relation pictured in Figure~\ref{fig:flip} is equivalent to the symmetry
\begin{align}
	F(x,y) = F(y, x)
\end{align}
of the function~\eqref{Fflip}
\begin{align}\label{Fflip-def}
	\begin{aligned}
		& F(x, y) = e^{\imath (\lambda - \rho) x -2\imath \rho y} \, (1 + \beta e^{-y})^{- \imath \rho - g} \; (1 - \beta e^{-y})^{- \imath \rho + g - 1} \\[6pt]
		& \times \int_{\ln \beta}^x dz \;\, e^{-2\imath \lambda z} \, (1 + \beta e^{-z})^{- \imath \lambda - g} \; (1 - \beta e^{-z})^{- \imath \lambda + g - 1} \, (e^z + e^y)^{\imath(\lambda + \rho)} \; (1 - e^{z - x})^{\imath (\lambda - \rho) - 1},
	\end{aligned}
\end{align}
where we assume $x > \ln \beta$, $y > \ln \beta$ and also for convergence $\Re(\imath \rho) < \Re(\imath \lambda) < g$.
To~prove this symmetry it is convenient to pass to the exponential variables
\begin{align}
	X = e^x > \beta, \qquad Y = e^y > \beta, \qquad Z = e^z.
\end{align}
Then the function can be rewritten in the form
\begin{align} \label{Fflip2}
	\begin{aligned}
		& F = X^{\imath(\lambda - \rho)} \, Y^{-2\imath \rho} \, (1 + \beta /Y )^{- \imath \rho - g} \; (1 - \beta / Y)^{- \imath \rho + g  - 1} \\[6pt]
		& \times \int_\beta^X dZ \;\, Z^{-2\imath \lambda - 1} \, (1 + \beta / Z )^{- \imath \lambda - g} \; (1 - \beta / Z )^{- \imath \lambda + g - 1}  \, (Z + Y)^{\imath(\lambda + \rho)} \; (1 - Z/X)^{\imath (\lambda - \rho) - 1}.
	\end{aligned}
\end{align}
Next to remove dependence on external parameters from the bounds of integral we change integration variable
\begin{align}
	w = \frac{X - Z}{Z - \beta}.
\end{align}
As a result, we arrive at the expression
\begin{align} 
	\begin{aligned}
		F & = X Y \, (X - \beta)^{-\imath \rho + g - 1} \, (Y - \beta)^{-\imath \rho + g - 1} \, (Y + \beta)^{- \imath \rho - g} \\[6pt]
		& \times \int_0^\infty dw \;\, w^{\imath(\lambda - \rho) - 1} \, (X + \beta + 2\beta w)^{-\imath \lambda - g} \; \bigl( X + Y + (Y + \beta) w \bigr)^{\imath(\lambda+ \rho)}.
	\end{aligned}
\end{align}
It is left to rescale integration variable 
\begin{align}
	w = (X + \beta) s
\end{align}
to obtain the formula
\begin{align}\label{Fflip3}
	\begin{aligned}
		F & = X Y \, (X - \beta)^{-\imath \rho + g - 1} \, (Y - \beta)^{-\imath \rho + g - 1} \, (X + \beta)^{- \imath \rho - g} \, (Y + \beta)^{- \imath \rho - g} \\[6pt]
		& \times \int_0^\infty ds \;\, s^{\imath(\lambda - \rho) - 1} \, (1 + 2\beta s)^{-\imath \lambda - g } \; \bigl( X + Y + (X + \beta) (Y + \beta) s \bigr)^{\imath(\lambda+ \rho)},
	\end{aligned}
\end{align}
which is clearly symmetric in $X, Y$. 

\subsection{Proof of reduced flip relation} \label{app:flip-lim}
The reduced flip relation pictured in Figure~\ref{fig:flip-lim} represents the following identity
\begin{multline}\label{flip-lim}
	e^{\imath (\lambda - \rho) x} \, \int_{\ln \beta}^x dz \;\, e^{-2\imath \lambda z} \, (1 + \beta e^{-z})^{- \imath \lambda - g} \; (1 - \beta e^{-z})^{- \imath \lambda + g - 1} \\[6pt]
	\times (e^z + \beta)^{\imath(\lambda + \rho)} \; (1 - e^{z - x})^{\imath (\lambda - \rho) - 1} = (2\beta)^{\imath (\rho -\lambda)} \; \frac{\Gamma(\imath \lambda - \imath \rho) \, \Gamma ( g - \imath \lambda )}{\Gamma ( g - \imath \rho )} \\[10pt]
	\times e^{-2 \imath \rho x} \; (1 + \beta e^{-x})^{- \imath \rho - g} \; (1 - \beta e^{-x})^{- \imath \rho + g - 1} \; (e^x + \beta)^{\imath(\lambda + \rho)},
\end{multline}
where we assume $x > \ln \beta$. Note also that the integral is absolutely convergent under condition $\Re(\imath \rho) < \Re(\imath \lambda) < g$.

To prove the above identity consider asymptotics of the function $F(x, y)$ defined by~\eqref{Fflip-def} as $y \to \ln\beta^+$. From its definition we have
\begin{multline}\label{Flim}
	F(x, y) \underset{y \to \ln \beta^+}{=} (y - \ln\beta)^{-\imath \rho + g - 1} \; 2^{-\imath \rho - g}  \; \beta^{-2\imath \rho} \\[6pt]
	\times e^{\imath (\lambda - \rho) x} \,  \int_{\ln \beta}^x dz \;\, e^{-2\imath \lambda z} \, (1 + \beta e^{-z})^{- \imath \lambda - g} \; (1 - \beta e^{-z})^{- \imath \lambda + g - 1} \\[6pt]
	\times (e^z + \beta)^{\imath(\lambda + \rho)} \; (1 - e^{z - x})^{\imath (\lambda - \rho) - 1} .
\end{multline}
Except the first line, this expression coincides with the left hand side of the identity~\eqref{flip-lim}. 

Next consider the same limit for the formula~\eqref{Fflip3} (which we rewrite in original variables)
\begin{multline}
	F(x, y)  \underset{y \to \ln \beta^+}{=} (y - \ln\beta)^{-\imath \rho + g - 1} \; 2^{-\imath \rho - g}  \; \beta^{-2\imath \rho} \\[6pt]
	\times e^{-2 \imath \rho x} \; (1 + \beta e^{-x})^{- \imath \rho - g} \; (1 - \beta e^{-x})^{- \imath \rho + g - 1} \; (e^x + \beta)^{\imath(\lambda + \rho)} \\[6pt]
	\times \int_0^\infty ds \;\, s^{\imath(\lambda - \rho) - 1} \, (1 + 2\beta s)^{\imath \rho - g } .
\end{multline}
The remaining integral after rescaling $s \to s/2\beta$ becomes Euler's beta integral
\begin{align}
	\int_0^\infty ds \;\, s^{\imath(\lambda - \rho) - 1} \, (1 + 2\beta s)^{\imath \rho - g } = (2\beta)^{\imath (\rho -\lambda)} \; \frac{\Gamma(\imath \lambda - \imath \rho) \, \Gamma ( g - \imath \lambda )}{\Gamma ( g - \imath \rho )},
\end{align}
which coincides with the coefficient from the right~\eqref{flip-lim}. Collecting all together we arrive at the claimed formula.

One subtle point is that to obtain~\eqref{Flim} we interchanged limit and integration. To justify it we use dominated convergence theorem. Namely, in addition to previous assumptions suppose $\Re(\imath \lambda), \Re(\imath \rho) < 0$, this restriction can be removed at the end by analytic continuation. Since $e^y \geq \beta$, we then can use inequality
\begin{align}
	\bigl| (e^z + e^y)^{\imath(\lambda + \rho)} \bigr| = (e^z + e^y)^{\Re(\imath\lambda) + \Re(\imath \rho)} \leq (e^z + \beta)^{\Re(\imath\lambda) + \Re(\imath \rho)}
\end{align}
to bound the integrand in~\eqref{Fflip2} by integrable function, which doesn't depend on $y$. Analogous arguments should be applied to the limit of the second formula~\eqref{Fflip3}.

\section{Gauss--Givental bounds} \label{sec:GG-bounds}

Denote by $\mathcal{P}_n$ the space of continuous polynomially bounded functions of $n$ variables
\begin{align} \label{Pspace}
	\mathcal{P}_n = \bigl\{ \phi \in C(\mathbb{R}^n) \colon \quad | \phi(\bm{x}_n) | \leq P(|x_1|, \dots, |x_n|), \quad P \text{ --- polynomial} \bigr\}.
\end{align}
In this appendix we show that Baxter and raising operators of $GL$ and $BC$ Toda systems act invariantly on this space.

For this we need the following three bounds. The last two are proved in~\cite{BDK}, although here we state their slightly weaker versions, which are sufficient in the present context.

\begin{lemma} \label{lem:line}
	Let $m \in \mathbb{N}_0$, $\kappa > 0$ and $x \in \mathbb{R}$. Then
	\begin{align}
		\int_{\mathbb{R}} dy \; |y|^m \, \exp \bigl( \kappa(x - y) - e^{x - y} \bigr) \leq P(|x|),
	\end{align}
	where $P(|x|) \equiv P(|x|; m, \kappa)$ is polynomial in $|x|$, and convergence of this integral is uniform in $x$ from compact subsets of $\mathbb{R}$.
\end{lemma}

\begin{lemma} \cite[Corollary 1]{BDK} \label{lem:2lines}
	Let $m \in \mathbb{N}_0$, $\kappa_1, \kappa_2 \geq 0$ and $x_1, x_2 \in \mathbb{R}$. Then
	\begin{align}
		\int_{\mathbb{R}} dy\; |y|^m \, \exp\bigl(\kappa_1 (x_1 - y) + \kappa_2 (y - x_2) - e^{x_1 - y} - e^{y - x_2} \bigr) \leq P(|x_1|, |x_2|),
	\end{align}
	where $P(|x_1|, |x_2|) \equiv P(|x_1|, |x_2|; m, \kappa_1, \kappa_2)$ is polynomial in $|x_j|$, and convergence of this integral is uniform in $x_1, x_2$ from compact subsets of $\mathbb{R}$.
\end{lemma}

\begin{lemma} \cite[Lemma 3]{BDK} \label{lem:boundary}
	Let $m \in \mathbb{N}_0$, $g > 0$, $\kappa \geq 0$ and $x \in \mathbb{R}$. Then
	\begin{align}\label{lem3}
		\int_{0}^\infty dy \; y^m  \, (1 + e^{-y})^{- g} \, (1 - e^{-y})^{g - 1} \, \exp\bigl(\kappa (y + x) - e^{y + x}\bigr) \leq P(|x|),
	\end{align}
	where  $P(|x|) \equiv P(|x|; m, g, \kappa)$ is polynomial in $|x|$, and convergence of this integral is uniform in~$x$ from compact subsets of $\mathbb{R}$.
\end{lemma}

\begin{proof}[Proof of Lemma~\ref{lem:line}]
	Change variable in the integral in question to $z = x - y$
	\begin{align}
		\int_{\mathbb{R}} dy \; |y|^m \, e^{ \kappa(x - y) - e^{x - y} } = \int_{\mathbb{R}} dz \; |x - z|^m \, e^{ \kappa z - e^{z} } \leq \int_{\mathbb{R}} dz \; (|x| + |z|)^m \, e^{ \kappa z - e^{z} }.
	\end{align}
	Then the desired bound follows from expanding brackets, since (recall that $\kappa > 0$)
	\begin{align}
		\int_{\mathbb{R}} dz \; |z|^\ell \, e^{\kappa z - e^{z}} < \infty.
	\end{align}
	The above bound is clearly uniform in $x$ from compact subsets.
\end{proof}

\subsection{$GL$ system}

Recall $GL$ Toda raising and Baxter operators
\begin{align} \label{LQdef}
	\begin{aligned}
		& \bigl[ \Lambda_n(\lambda) \, \phi \bigr] (\bm{x}_n) = \int_{\mathbb{R}^{n - 1}} d\bm{y}_{n - 1} \; \exp \biggl( \imath \lambda \bigl( \underline{\bm{x}}_n - \underline{\bm{y}}_{n - 1} \bigr) - \sum_{j = 1}^{n - 1} ( e^{x_j - y_j} + e^{y_j - x_{j + 1} } ) \biggr) \, \phi(\bm{y}_{n - 1}), \\[6pt]
		& \bigl[ Q_n(\lambda) \, \phi \bigr] (\bm{x}_n) = \int_{\mathbb{R}^n} d\bm{y}_n \; \exp \biggl( \imath \lambda \bigl( \underline{\bm{x}}_n - \underline{\bm{y}}_{n} \bigr) - \sum_{j = 1}^{n - 1} ( e^{x_j - y_j} + e^{y_j - x_{j + 1} } ) - e^{x_n - y_n} \biggr) \, \phi(\bm{y}_{n}),
	\end{aligned}
\end{align}
where we always assume $\bm{x}_n \in \mathbb{R}^n$. The following axillary operator
\begin{align} \label{L'def}
	\bigl[ \Lambda'_n(\lambda) \, \phi \bigr](\bm{x}_n) = e^{-\imath \lambda x_1} \, \bigl[ \Lambda_n(\lambda) \, \phi \bigr](\bm{x}_n) 
\end{align}
will be also useful when we pass to the case of $BC$ system.

\begin{proposition} \label{prop:QL-A-space}
	We have
	\begin{align} \label{L'space}
		& \Lambda'_n(\lambda) \colon \; \mathcal{P}_{n - 1} \; \to \; \mathcal{P}_n, && \hspace{-2cm} \Im \lambda \geq 0, \\[6pt] \label{Lspace}
		& \Lambda_n(\lambda) \colon \; \mathcal{P}_{n - 1} \; \to \; \mathcal{P}_n, && \hspace{-2cm} \lambda \in \mathbb{R}, \\[6pt] \label{Qspace}
		& Q_n(\lambda) \colon \; \mathcal{P}_n \; \to \; \mathcal{P}_n, && \hspace{-2cm} \Im \lambda < 0.
	\end{align}	
\end{proposition}

\begin{proof}
	First, notice that the second statement~\eqref{Lspace} follows from the first one~\eqref{L'space} due to definition~\eqref{L'def}. For the first statement by definition we have 
	\begin{align}
		\Bigl| \bigl[ \Lambda'_n(\lambda) \, \phi \bigr](\bm{x}_n) \Bigr| \leq \int_{\mathbb{R}^{n - 1}} d\bm{y}_{n - 1} \; P(|y_1|, \dots, |y_{n - 1}|) \, \prod_{j = 1}^{n - 1} \, e^{ \Im \lambda \, (y_j - x_{j + 1}) - e^{x_j - y_j} - e^{y_j - x_{j + 1}} } .
	\end{align}
	Expanding polynomial in terms of monomials $|y_1|^{m_1} \cdots |y_{n - 1}|^{m_{n - 1}}$ we obtain the sum with terms factorised into one dimensional integrals
	\begin{align}
		\int_{\mathbb{R}} dy_j \; |y_j|^{m_j} \, e^{ \Im \lambda \, (y_j - x_{j + 1}) - e^{x_j - y_j} - e^{y_j - x_{j + 1}} },
	\end{align}
	where $\Im \lambda \geq 0$. By Lemma~\ref{lem:2lines} these integrals are bounded by polynomials $P(|x_{j}|, |x_{j + 1}|)$ and converge uniformly in $x_j, x_{j + 1}$. This gives the needed polynomial estimate
	\begin{align}
		\Bigl| \bigl[ \Lambda'_n(\lambda) \, \phi \bigr](\bm{x}_n) \Bigr| \leq P(|x_1|, \dots, |x_n|),
	\end{align}
	while uniformity ensures that $\bigl[ \Lambda'_n(\lambda) \, \phi \bigr](\bm{x}_n)$ is continuous in $\bm{x}_n$.
	
	Now consider the third statement~\eqref{Qspace}. Due to definition 
	\begin{multline}
		\Bigl| \bigl[ Q_n(\lambda) \, \phi \bigr] (\bm{x}_n) \Bigr| \leq \int_{\mathbb{R}^n} d\bm{y}_n \; P(|y_1|, \dots, |y_n|) \, e^{- \Im \lambda \, (x_n - y_n) - e^{x_n - y_n}} \\
		\times \prod_{j = 1}^{n - 1} e^{ -\Im \lambda \, (x_j - y_j) - e^{x_j - y_j} - e^{y_j - x_{j + 1}} } .
	\end{multline}
	Again writing polynomial in terms of monomials we obtain the sum with terms factorised into integrals of two types
	\begin{align}
		& \int_{\mathbb{R}} dy_j \; |y_j|^{m_j} \, e^{ -\Im \lambda \, (x_j - y_j) - e^{x_j - y_j} - e^{y_j - x_{j + 1}} }  \qquad (j = 1, \dots, n - 1), \\[6pt]
		& \int_{\mathbb{R}} dy_n \; |y_n|^{m_n} \, e^{- \Im \lambda \, (x_n - y_n) - e^{x_n - y_n}},
	\end{align}
	where $\Im \lambda < 0$. All of them are polynomially bounded and converge uniformly in $x_j$ by Lemmas~\ref{lem:line},~\ref{lem:2lines}. 
\end{proof}

As a by-product, from the above proposition we get the following estimate for $GL$ Toda wave function
\begin{align}
	\Phi_{\bm{\lambda}_n}(\bm{x}_n) = \Lambda_n(\lambda_n) \cdots \Lambda_1(\lambda_1) \cdot 1.
\end{align}

\begin{corollary} \label{cor:Phi-bound}
	Let $\bm{\lambda}_n \in \mathbb{R}^n$. Then $\Phi_{\bm{\lambda}_n}(\bm{x}_n) \in \mathcal{P}_n$, and the corresponding integral is absolutely convergent.
\end{corollary}

\begin{proof}
	The fact that wave function belongs to $\mathcal{P}_n$ is clear from~\eqref{Lspace}. The absolute convergence follows from the property of the raising operators kernels
	\begin{align}
		\bigl| \Lambda_{\lambda_k}(\bm{x}_k | \bm{y}_{k - 1}) \bigr| \leq \Lambda_{0}(\bm{x}_k | \bm{y}_{k - 1}) , \qquad \lambda_k \in \mathbb{R},
	\end{align}
	 see~\eqref{LQdef}.
\end{proof}

\subsection{$BC$ system}

Introduce two axillary operators
\begin{align} \nonumber
	& \bigl[ V_n(\lambda) \, \phi \bigr] (\bm{x}_n) = \frac{(2\beta)^{\imath \lambda}}{\Gamma(g - \imath \lambda)} \, \int_{\mathbb{R}^n} d\bm{y}_n \; (1 + \beta e^{-y_1})^{- \imath \lambda - g} \, (1 - \beta e^{-y_1})^{- \imath \lambda + g - 1} \, \theta(y_1 - \ln \beta)\\[6pt]  \label{Vdef}
	& \hspace{2.35cm} \times \exp \biggl( \imath \lambda \bigl( \underline{\bm{x}}_n - \underline{\bm{y}}_n  \bigr) - \sum_{j = 1}^{n - 1} (e^{y_j - x_{j}} + e^{x_j - y_{j + 1}}) - e^{y_n - x_n} \biggr)  \; \phi(\bm{y}_n), \\[6pt] \nonumber
	& \bigl[ W_n(\lambda) \, \phi \bigr] (\bm{x}_n) = \frac{(2\beta)^{\imath \lambda}}{\Gamma(g - \imath \lambda)} \, \int_{\mathbb{R}^{n + 1}} d\bm{y}_{n + 1} \; (1 + \beta e^{-y_1})^{- \imath \lambda - g} \, (1 - \beta e^{-y_1})^{- \imath \lambda + g - 1} \, \theta(y_1 - \ln \beta)  \\[6pt] \label{Wdef}
	& \hspace{2.5cm}  \times  \exp \biggl( \imath \lambda \bigl( \underline{\bm{x}}_n - \underline{\bm{y}}_{n + 1} - y_1 \bigr) - \sum_{j = 1}^{n} (e^{y_j - x_{j}} + e^{x_j - y_{j + 1}}) \biggr) \; \phi(\bm{y}_{n + 1}).
\end{align}
Then $BC$ Toda raising and Baxter operators from Sections~\ref{sec:rais-op} and~\ref{sec:baxt-op} can be written as the products of $GL$ Toda operators with axillary ones
\begin{align} \label{BC-op-split}
	\begin{aligned}
		& \LLambda_n(\lambda) = V_n(\lambda) \, \Lambda_n(-\lambda), \\[6pt]
		& \QQ_n(\lambda) = W_n(\lambda) \, \Lambda'_{n + 1}(-\lambda), \\[6pt]
		& \QQr_n(\lambda) = V_n(-\lambda) \, Q_n(\lambda).
	\end{aligned}
\end{align}
This splitting is easily visualized in the language of diagrams from Section~\ref{sec:diagrams}.

By Proposition~\ref{prop:QL-A-space} $GL$ Toda operators act within the space of polynomially bounded continuous functions, so it is left to establish the same thing about axillary operators. 

\begin{lemma} \label{lem:VW-space}
	We have
	\begin{align}  \label{Vspace}
		& V_n(\lambda) \colon \; \mathcal{P}_{n} \; \to \; \mathcal{P}_n, && \hspace{-2cm} \Im \lambda \geq 0, \\[6pt]
		& W_n(\lambda) \colon \; \mathcal{P}_{n + 1} \; \to \; \mathcal{P}_n, && \hspace{-2cm} \Im \lambda \in (-g, 0).
	\end{align}
\end{lemma}

Before we prove this lemma let us remark that together with Proposition~\ref{prop:QL-A-space} and formulas~\eqref{BC-op-split} it gives the following statement. 

\begin{proposition} \label{prop:LQQr-space}
	We have
	\begin{align} \label{LLspace}
		& \LLambda_n(\lambda) \colon \; \mathcal{P}_{n - 1} \; \to \; \mathcal{P}_n, && \hspace{-2cm} \lambda \in \mathbb{R}, \\[6pt]
		& \QQ_n(\lambda) \colon \; \mathcal{P}_{n} \; \to \; \mathcal{P}_n, && \hspace{-2cm} \Im \lambda \in (-g, 0), \\[6pt]
		& \QQr_n(\lambda) \colon \; \mathcal{P}_{n} \; \to \; \mathcal{P}_n, && \hspace{-2cm} \Im \lambda < 0.
	\end{align}
\end{proposition}

\begin{proof}[Proof of Lemma~\ref{lem:VW-space}]
	First, consider the operator $V_n(\lambda)$. By definition~\eqref{Vdef} we have
	\begin{multline} \label{Vphi-b}
		\Bigl| \bigl[ V_n(\lambda) \, \phi \bigr] (\bm{x}_n) \Bigr| \leq \frac{(2\beta)^{-\Im \lambda}}{|\Gamma(g - \imath \lambda)| } \, \int_{\mathbb{R}^n} d\bm{y}_n \; P(|y_1|, \dots, |y_n|) \, \prod_{j = 2}^n \, e^{\Im \lambda\,  (y_j - x_j) - e^{x_{j - 1} - y_j} - e^{y_j - x_{j}}} \\[6pt]
		\times e^{\Im \lambda \, (y_1 - x_1) - e^{y_1 - x_1}} \, (1 + \beta e^{-y_1})^{\Im \lambda - g} \, (1 - \beta e^{-y_1})^{\Im \lambda + g - 1} \, \theta(y_1 - \ln \beta).
	\end{multline}
	Expanding polynomial in terms of monomials $|y_1|^{m_1} \cdots |y_n|^{m_n}$ we obtain sum with terms factorised into integrals of two types
	\begin{align}
		& \int_{\mathbb{R}} dy_j \; |y_j|^{m_j} \, e^{\Im \lambda\,  (y_j - x_j) - e^{x_{j - 1} - y_j} - e^{y_j - x_{j}}} \qquad (j = 2, \dots, n), \\[6pt] \label{int-b1}
		& \int_{\ln \beta}^\infty dy_1 \; |y_1|^{m_1} \, e^{\Im \lambda \, (y_1 - x_1) - e^{y_1 - x_1}} \, (1 + \beta e^{-y_1})^{\Im \lambda - g} \, (1 - \beta e^{-y_1})^{\Im \lambda + g - 1},
	\end{align}
	where $\Im \lambda \geq 0$. Integrals from the first line are polynomially bounded and converge uniformly in $x_{j - 1}, x_j$ due to Lemma~\ref{lem:2lines}. In the integral from the second line we can use inequality
	\begin{align}
		(1 + \beta e^{-y_1})^{\Im \lambda - g} =  (1 + \beta e^{-y_1})^{2\Im \lambda} \, (1 + \beta e^{-y_1})^{-\Im \lambda - g} \leq 4^{\Im \lambda} \, (1 + \beta e^{-y_1})^{- \Im \lambda - g}, 
	\end{align}
	and change integration variable to $z = y_1 - \ln\beta$, so that 
	\begin{multline}
		\int_{\ln \beta}^\infty dy_1 \; |y_1|^{m_1} \, e^{\Im \lambda \, (y_1 - x_1) - e^{y_1 - x_1}} \, (1 + \beta e^{-y_1})^{\Im \lambda - g} \, (1 - \beta e^{-y_1})^{\Im \lambda + g - 1} \\[6pt]
		\leq 4^{\Im \lambda} \, \int_{0}^\infty dz \; (|z| + |\ln \beta|)^{m_1} \, e^{\Im \lambda \, (z + \ln \beta - x_1) - e^{z + \ln \beta - x_1}} \, (1 +  e^{-z})^{-g'} \, (1 -  e^{-z})^{g' - 1},
	\end{multline}
	where we also denoted $g' = \Im \lambda + g > 0$. Then Lemma~\ref{lem:boundary} says that the last integral is bounded polynomially and converges uniformly in $x$. Hence, the whole right hand side~\eqref{Vphi-b} is bounded by polynomial in $\bm{x}_n$ and $\bigl[ V_n(\lambda) \, \phi \bigr] (\bm{x}_n)$ is continuous.
	
	Now consider the operator $W_n(\lambda)$. By definition~\eqref{Wdef} we have
	\begin{align} \label{Wphi-b}
		\begin{aligned}
			\Bigl| \bigl[ W_n(\lambda) \, \phi \bigr] (\bm{x}_n) & \Bigr| \leq \frac{(2\beta)^{-\Im \lambda}}{|\Gamma(g - \imath \lambda)| } \, \int_{\mathbb{R}^{n + 1}} d\bm{y}_{n + 1} \; P(|y_1|, \dots, |y_{n + 1}|) \\[6pt]
			&  \times e^{- \Im \lambda \, (x_n - y_{n + 1}) - e^{x_n - y_{n + 1}} } \, \prod_{j = 2}^n \, e^{ -\Im \lambda\,  (x_{j - 1} - y_j) - e^{x_{j - 1} - y_j} - e^{y_j - x_{j}}}  \\[6pt]
			& \times e^{2\Im \lambda \, y_1- e^{y_1 - x_1}} \, (1 + \beta e^{-y_1})^{\Im \lambda - g} \, (1 - \beta e^{-y_1})^{\Im \lambda + g - 1} \, \theta(y_1 - \ln \beta) .
		\end{aligned}
	\end{align}
	Again writing polynomial in terms of monomials we obtain sum with terms factorised into integrals of three types
	\begin{align}
		& \int_{\mathbb{R}} dy_{n + 1} \; |y_{n+1}|^{m_{n + 1}} \, e^{- \Im \lambda \, (x_n - y_{n + 1}) - e^{x_n - y_{n + 1}} }, \\[6pt]
		& \int_{\mathbb{R}} dy_j \; |y_j|^{m_j} \, e^{ -\Im \lambda\,  (x_{j - 1} - y_j) - e^{x_{j - 1} - y_j} - e^{y_j - x_{j}}} \qquad (j = 2, \dots, n), \\[6pt]
		& \int_{\ln\beta}^\infty dy_1 \; |y_1|^{m_1} \, e^{2\Im \lambda \, y_1- e^{y_1 - x_1}} \, (1 + \beta e^{-y_1})^{\Im \lambda - g} \, (1 - \beta e^{-y_1})^{\Im \lambda + g - 1},
	\end{align}
	where $\Im \lambda < 0$. The first two are polynomially bounded due to Lemmas~\ref{lem:line} and~\ref{lem:2lines} correspondingly. In the last integral we can use inequality 
	\begin{align}
		e^{2\Im \lambda \, y_1} = \beta^{2\Im \lambda} \, e^{2\Im \lambda \, (y_1 - \ln \beta)} \leq \beta^{2\Im \lambda}
	\end{align}
	to bound it by the integral of the same type, as we already encountered~\eqref{int-b1}. The remaining steps are the same, as before, so that we obtain polynomial bound for the whole expression~\eqref{Wphi-b}.
\end{proof}

\subsection{Kernels of operator products} \label{sec:ker-prod}

Raising and Baxter operators of $BC$ Toda split into products of two simpler operators~\eqref{BC-op-split}. Besides, in the paper we consider local relations between various products of the above operators, such as
\begin{align}
	Q_n(\lambda) \, Q_n(\rho) = Q_n(\rho) \, Q_n(\lambda), \qquad \QQ_n(\lambda) \, \QQ_n(\rho) = \QQ_n(\rho) \, \QQ_n(\lambda),
\end{align}
and others. Let us prove that the kernels of all these products are given by absolutely convergent integrals.

For example, consider the $BC$ Baxter operator~\eqref{BC-op-split}
\begin{align}
	\QQ_n(\lambda) = W_n(\lambda) \, \Lambda'_{n + 1}(-\lambda),
\end{align}
with $\Im \lambda \in (-g, 0)$. By Proposition~\ref{prop:LQQr-space} the following integral is convergent for $\phi \in \mathcal{P}_n$
\begin{align}
	\bigl[ \QQ_n(\lambda) \, \phi \bigr] (\bm{x}_n) = \int_{\mathbb{R}^{n + 1}} d\bm{y}_{n + 1} \int_{\mathbb{R}^n} d\bm{z}_n \; W_\lambda(\bm{x}_n | \bm{y}_{n + 1}) \,\Lambda'_{-\lambda}(\bm{y}_{n + 1} | \bm{z}_n) \, \phi(\bm{z}_n),
\end{align}
where in the integrand we have kernels of operators $W_n(\lambda)$, $\Lambda'_{n + 1}(-\lambda)$. Moreover, since
\begin{align} \label{WL-abs}
	\begin{aligned}
		& \bigl| W_\lambda(\bm{x}_n | \bm{y}_{n + 1}) \bigr| =  \biggl| \frac{\Gamma(g + \Im \lambda)}{\Gamma(g - \imath \lambda)} \biggr| \; W_{\imath \Im \lambda}(\bm{x}_n | \bm{y}_{n + 1}), \\[6pt]
		& \bigl| \Lambda'_{-\lambda}(\bm{y}_{n + 1} | \bm{z}_n)  \bigr| = \Lambda'_{-\imath \Im \lambda}(\bm{y}_{n + 1} | \bm{z}_n)
	\end{aligned}
\end{align}
and $| \phi |$ is still in $\mathcal{P}_n$, the above integral is absolutely convergent. By Fubini--Tonelli theorem this means that we can interchange the order of integrals and write
\begin{align}
	\bigl[ \QQ_n(\lambda) \, \phi \bigr] (\bm{x}_n) = \int_{\mathbb{R}^n} d\bm{z}_n \; \QQ_\lambda(\bm{x}_n | \bm{z}_n) \, \phi(\bm{z}_n),
\end{align}
where
\begin{align} \label{QQkerWL}
	\QQ_\lambda(\bm{x}_n | \bm{z}_n) = \int_{\mathbb{R}^{n + 1}} d\bm{y}_{n + 1} \; W_\lambda(\bm{x}_n | \bm{y}_{n + 1}) \,\Lambda'_{-\lambda}(\bm{y}_{n + 1} | \bm{z}_n).
\end{align}
Now notice that the kernel 
\begin{align}
	\Lambda'_{-\lambda}(\bm{y}_{n + 1} | \bm{z}_n) = \prod_{j = 1}^n \exp \bigl( \imath \lambda(z_j - y_{j + 1}) - e^{y_j - z_j} - e^{z_j - y_{j + 1}} \bigr),
\end{align}
as a function of $\bm{y}_{n + 1}$, belongs to $\mathcal{P}_{n + 1}$. Indeed, it is continuous and bounded by constant
\begin{align}
	\bigl| \Lambda'_{-\lambda}(\bm{y}_{n + 1} | \bm{z}_n) \bigr| \leq \prod_{j = 1}^n \exp \bigl( - \Im \lambda (z_j - y_{j + 1}) - e^{z_j - y_{j + 1}} \bigr) \leq C(\lambda),
\end{align}
since $\Im \lambda < 0$. By Lemma~\ref{lem:VW-space}, the operator $W_n(\lambda)$ is well-defined on $\mathcal{P}_{n + 1}$, hence, the integral~\eqref{QQkerWL} is convergent for $\Im \lambda \in (-g, 0)$. Furthermore, because of~\eqref{WL-abs} it is absolutely convergent. 

The same arguments can be applied to all other kernels of operator products appearing in the paper. For example, the kernel of $Q$-operators' product
\begin{align}
	\QQ_n(\lambda) \, \QQ_n(\rho) = W_n(\lambda) \, \Lambda'_{n + 1}(-\lambda) \, W_n(\rho) \, \Lambda'_{n + 1}(-\rho)
\end{align}
is given by convergent integral, since again the kernel of the last operator $\Lambda'_{-\rho}(\bm{y}_{n + 1} | \bm{z}_n)$ belongs to $\mathcal{P}_{n + 1}$ (as a function of $\bm{y}_{n + 1}$).

Let us also consider the operator product from Section~\ref{sec:gl-bc-rel}
\begin{multline} 
	Q^t_{n - 1}(\lambda) \, \Lambda^t_{n}(- \lambda) \, \exp(- \beta e^{-x_1}) \, \LLambda_n(\rho)  \\[6pt]
	= Q^t_{n - 1}(\lambda) \, \Lambda^t_{n}(- \lambda) \, \exp(- \beta e^{-x_1}) \, V_n(\rho) \, \Lambda_n(-\rho), \qquad \Im \lambda < 0, \quad \rho \in \mathbb{R},
\end{multline}
since it is slightly trickier. Here the transposed operators are defined as follows
\begin{align}
	& \bigl[ Q^t_{n - 1}(\lambda) \, \phi \bigr] (\bm{x}_{n - 1}) = \int_{\mathbb{R}^{n - 1}} d\bm{y}_{n - 1} \; Q_{\lambda}(\bm{y}_{n - 1} | \bm{x}_{n - 1}) \, \phi(\bm{y}_{n - 1}), \\[6pt]
	& \bigl[ \Lambda^t_n(-\lambda) \, \phi \bigr] (\bm{x}_{n - 1}) = \int_{\mathbb{R}^n} d\bm{y}_n \; \Lambda_{-\lambda}(\bm{y}_n | \bm{x}_{n - 1}) \, \phi(\bm{y}_n).
\end{align}
The kernel of the above operator product equals
\begin{multline}
	\int_{\mathbb{R}^{n - 1}} d\bm{y}_{n - 1} \int_{\mathbb{R}^n} d\bm{z}_n \int_{\mathbb{R}^n} d\bm{s}_n \; Q_{\lambda}(\bm{y}_{n - 1} | \bm{x}_{n - 1}) \, \Lambda_{-\lambda}(\bm{z}_n | \bm{y}_{n - 1}) \, \exp(-\beta e^{-z_1}) \\[6pt]
	\times  V_\rho(\bm{z}_n | \bm{s}_n) \, \Lambda_{-\rho}(\bm{s}_n | \bm{t}_{n - 1}).
\end{multline}
Note that $\Lambda_{-\rho}(\bm{s}_n | \bm{t}_{n - 1})$, as a function of $\bm{s}_n$, belongs to $\mathcal{P}_n$. So, again to prove that the last integral is absolutely convergent it is enough to show that the product 
\begin{align} \label{QtLteV}
	Q^t_{n - 1}(\lambda) \, \Lambda^t_{n}(- \lambda) \, \exp(- \beta e^{-x_1}) \, V_n(\rho)
\end{align}
is well defined on the space $\mathcal{P}_{n}$. 

From the explicit formulas~\eqref{LQdef} we have
\begin{align}
	\Lambda^t_n(-\lambda) \, \exp(-\beta e^{-x_1}) \colon \; \phi(\bm{x}_n) \; \mapsto \; \bigl[ Q_n(\lambda) \, \phi \bigr](\ln \beta, \bm{x}_{n - 1}).
\end{align}
Besides,
\begin{align}
	Q_{n - 1}^t(\lambda) = \mathcal{I}_{n - 1} \, Q_{n - 1}(\lambda) \, \mathcal{I}_{n - 1}, \qquad \mathcal{I}_{n - 1} \colon \; \phi(x_1, \dots, x_{n - 1}) \mapsto \phi(-x_{n - 1}, \dots, -x_1).
\end{align}
Therefore by Proposition~\ref{prop:QL-A-space} and Lemma~\ref{lem:VW-space} the product~\eqref{QtLteV} is well defined on $\mathcal{P}_{n}$.

\subsection{From equality of kernels to equality of operators} \label{sec:ker-op-eq}

All local relations between $GL$ and $BC$ Toda operators hold on the space of polynomially bounded continuous functions. In this section we demonstrate that to establish the equality of operators on this space it is sufficient to prove the equality of kernels.

For example, consider the commutativity
\begin{align} \label{QQcomm-op}
	\QQ_n(\lambda) \, \QQ_n(\rho) = \QQ_n(\rho) \, \QQ_n(\lambda), \qquad \Im \lambda, \, \Im \rho \in (-g, 0).
\end{align}
By Proposition~\ref{prop:LQQr-space} both sides of this relation are well defined on $\mathcal{P}_n$ and, as we argued in Section~\ref{sec:ker-prod}, the kernels of both products are given by absolutely convergent integrals.
In Section~\ref{sec:op-rel} using diagram technique we prove that these kernels are equal to each other
\begin{align}\label{QQker-eq}
	\int_{\mathbb{R}^n} d\bm{y}_n \; \QQ_{\lambda}(\bm{x}_n | \bm{y}_n) \, \QQ_{\rho}(\bm{y}_n | \bm{z}_n) =	\int_{\mathbb{R}^n} d\bm{y}_n \; \QQ_{\rho}(\bm{x}_n | \bm{y}_n) \, \QQ_{\lambda}(\bm{y}_n | \bm{z}_n).
\end{align}
Let us show that, as a consequence, the identity~\eqref{QQcomm-op} holds on $\mathcal{P}_n$. 

Acting on $\phi \in \mathcal{P}_n$ from the left we have
\begin{align} \label{QQphi}
	\bigl[ \QQ_n(\lambda) \, \QQ_n(\rho) \, \phi \bigr] (\bm{x}_n) = \int_{\mathbb{R}^n} d\bm{y}_n \int_{\mathbb{R}^n} d\bm{z}_n \; \QQ_{\lambda}(\bm{x}_n | \bm{y}_n) \, \QQ_{\rho}(\bm{y}_n | \bm{z}_n) \, \phi(\bm{z}_n).
\end{align}
From definition
\begin{align}
	\bigl| \QQ_\lambda(\bm{x}_n | \bm{y}_n) \bigr| = \biggl| \frac{\Gamma(g + \Im \lambda)}{\Gamma(g - \imath \lambda)} \biggr| \; \QQ_{\imath \Im \lambda}(\bm{x}_n | \bm{y}_n),
\end{align}
and clearly $| \phi | \in \mathcal{P}_n$. Hence, by Proposition~\ref{prop:LQQr-space} the integral~\eqref{QQphi} is absolutely convergent (in the initial order). Therefore, by Fubini--Tonelli theorem we can interchange the order of integrals
\begin{align}
	\bigl[ \QQ_n(\lambda) \, \QQ_n(\rho) \, \phi \bigr] (\bm{x}_n) = \int_{\mathbb{R}^n} d\bm{z}_n  \biggl[ \int_{\mathbb{R}^n} d\bm{y}_n  \; \QQ_{\lambda}(\bm{x}_n | \bm{y}_n) \, \QQ_{\rho}(\bm{y}_n | \bm{z}_n) \biggr] \, \phi(\bm{z}_n).
\end{align}
The same steps are applied to the right hand side of~\eqref{QQcomm-op}. As a result, this identity follows from the equality of kernels~\eqref{QQker-eq} in square brackets. 

All other local identities considered in the paper hold on the space of polynomially bounded continuous functions by the same arguments.

\subsection{Bounds on $BC$ Toda wave function} \label{sec:bc-bound}

For $\bm{k}_n \in \mathbb{N}_0^n$ denote standardly
\begin{align}
	\partial_{\bm{x}_n}^{\bm{k}_n} = \partial_{x_1}^{k_1} \cdots \partial_{x_n}^{k_n}.
\end{align}
Also, for brevity, we introduce
\begin{align}
	C(\bm{\lambda}_n) = \prod_{j = 1}^n \frac{1}{| \Gamma(g - \imath \lambda_j)| }.
\end{align}
In previous paper~\cite{BDK} we proved the following bound for $BC$ Toda wave function. 

\begin{proposition} \cite[Proposition 1]{BDK} \label{prop:bc-bound}
	Let $\bm{k}_n \in \mathbb{N}_0^n$ and $\bm{x}_n, \bm{\lambda}_n \in \mathbb{R}^n$. Then $\Psi_{\bm{\lambda}_n}(\bm{x}_n)$ is smooth in $\bm{x}_n$ and admits the bound
	\begin{multline} \label{dPsi-b}
		\Bigl| \partial_{\bm{x}_n}^{\bm{k}_n} \, \Psi_{\bm{\lambda}_n} (\bm{x}_n) \Bigr| \leq C(\bm{\lambda}_n) \, P(|x_1|, \dots, |x_n|, |\lambda_n|) \, \exp\Biggl( \bigl[ k_1 (\ln \beta - x_1) - \beta e^{-x_1} \bigr] \, \theta(\ln \beta - x_1) \\[2pt]
		+ \sum_{j = 1}^{n - 1} \biggl[ (k_j + k_{j + 1}) \, \frac{x_j - x_{j + 1}}{2} -  e^{ \frac{x_j - x_{j + 1} }{2} } \biggr] \, \theta(x_j - x_{j + 1} ) \Biggr),
	\end{multline}
	where $P$ is polynomial, whose coefficients depend on $\bm{k}_n, \beta, g$. In particular, for $\bm{k}_n = (0, \dots, 0)$ it doesn't depend on $\lambda_n$
	\begin{align} \label{Psin-b}
		\begin{aligned}
			\bigl| \Psi_{\bm{\lambda}_n}(\bm{x}_n) \bigr| & \leq C(\bm{\lambda}_n) \,  P(|x_1|, \dots, |x_n|) \\[6pt]
			& \times \exp\Biggl( - \beta e^{-x_1} \, \theta(\ln \beta - x_1) - \sum_{j = 1}^{n - 1}  e^{ \frac{x_j - x_{j + 1} }{2} }  \, \theta(x_j - x_{j + 1} ) \Biggr).
		\end{aligned}
	\end{align}
\end{proposition}

In Section~\ref{sec:gg-mb-equiv} to prove the equivalence of Gauss--Givental and Mellin--Barnes representations for $BC$ Toda wave function we use the inversion formula for $GL$ Toda wave function. 
According to Wallach~\cite{W} (see also~\cite[Section 7]{W2}), it holds for the functions from \textit{Whittaker Schwartz space}~$\mathcal{T}$, which consists of $\psi(\bm{x}_n) \in C^\infty(\mathbb{R}^n)$ such that
\begin{align}\label{sup}
	\sup_{\bm{x}_n \in \mathbb{R}^n} \exp\Biggl( \, \sum_{j = 1}^{n - 1} m_j \, \frac{x_j - x_{j + 1}}{2} \Biggr) \, (1 + |\bm{x}_n|)^d \, | \mathcal{D} \, \psi(\bm{x}_n) | < \infty,
\end{align}
for all $m_j, d \in \mathbb{N}_0$ and constant coefficient differential operators $\mathcal{D}$. From the above proposition we have the following statement.

\begin{corollary} \label{cor:Psi-Tspace}
	Let $\bm{\lambda}_n \in \mathbb{R}^n$ and $\epsilon > 0$. Then
	\begin{align}
		e^{- \epsilon \underline{\bm{x}}_n - \beta e^{-x_1}} \, \Psi_{\bm{\lambda}_n}(\bm{x}_n) \in \mathcal{T}.
	\end{align}
\end{corollary}

\begin{proof}
	For brevity, denote function in question
	\begin{align}
		\psi(\bm{x}_n) = e^{- \epsilon \underline{\bm{x}}_n - \beta e^{-x_1}} \, \Psi_{\bm{\lambda}_n}(\bm{x}_n).
	\end{align}
	First, by Proposition~\ref{prop:bc-bound} it is smooth. Second, let us calculate its derivatives. For any $\bm{k}_n \in \mathbb{N}_0^n$ we have
	\begin{align}
		\partial_{\bm{x}_n}^{\bm{k}_n} \, \psi(\bm{x}_n) 
		= \sum_i c_i(\epsilon) \, e^{\ell^{(i)} (\ln \beta - x_1) - \epsilon \underline{\bm{x}}_n - \beta e^{-x_1}} \, \partial_{\bm{x}_n}^{\bm{s}_n^{(i)}} \,  \Psi_{\bm{\lambda}_n}(\bm{x}_n),
	\end{align}
	with some coefficients $c_i$ and integers $\ell^{(i)} \in \mathbb{N}_0$, $\bm{s}_n^{(i)} \in \mathbb{N}_0^n$, such that $ \ell^{(i)} \leq k_1$, $s_j^{(i)} \leq k_j$. Since $\ell^{(i)} \geq 0$,
	\begin{align}
		e^{\ell^{(i)} (\ln \beta - x_1) - \beta e^{-x_1}}  \leq e^{ [\ell^{(i)} (\ln \beta - x_1) - \beta e^{-x_1}] \, \theta(\ln \beta - x_1)}.
	\end{align}
	Hence, using Proposition~\ref{prop:bc-bound} and the fact that $\ell^{(i)} \leq k_1$, $s_j^{(i)} \leq k_j$ we derive estimate
	\begin{align}\label{psi-b}
		\begin{aligned}
			\bigl| \partial_{\bm{x}_n}^{\bm{k}_n} \, \psi(\bm{x}_n) \bigr| & \leq P(|x_1|, \dots, |x_n|) \, \exp\Biggl( - \epsilon \underline{\bm{x}}_n + 2\bigl[ k_1 (\ln \beta - x_1) - \beta e^{-x_1} \bigr] \, \theta(\ln \beta - x_1) \\[2pt]
			& + \sum_{j = 1}^{n - 1} \biggl[ (k_j + k_{j + 1}) \, \frac{x_j - x_{j + 1}}{2} -  e^{ \frac{x_j - x_{j + 1} }{2} } \biggr] \, \theta(x_j - x_{j + 1} ) \Biggr),
		\end{aligned}
	\end{align}
	which is sufficient to prove~\eqref{sup}. Notice that function $(1 + |\bm{x}_n|)^d$ from~\eqref{sup} can be absorbed into polynomial $P$
	\begin{align}
		(1 + |\bm{x}_n|)^d \, P(|x_1|, \dots, |x_n|) \leq (1 + | \bm{x}_n|^2)^d \, P(|x_1|, \dots, |x_n|),
	\end{align}
	while in the exponent from~\eqref{sup} we can add step functions
	\begin{align}
		\exp\Biggl( \, \sum_{j = 1}^{n - 1} m_j \, \frac{x_j - x_{j + 1}}{2} \Biggr) \leq \exp\Biggl( \, \sum_{j = 1}^{n - 1} m_j \, \frac{x_j - x_{j + 1}}{2} \, \theta(x_j - x_{j + 1}) \Biggr),
	\end{align}
	so that it has the same form, as exponents from the third line~\eqref{psi-b}. So, these two factors doesn't essentially change estimate from the right hand side~\eqref{psi-b}. To see that this estimate is bounded for all $\bm{x}_n \in \mathbb{R}^n$, we pass to the variables
	\begin{align}
		y_1 = \ln \beta - x_1, \qquad y_2 = \frac{x_1 - x_2}{2}, \qquad \dots \qquad y_n = \frac{x_{n - 1} - x_n}{2}.
	\end{align}
	Expanding polynomial $P$ in monomials $|y_1|^{d_1} \cdots |y_n|^{d_n}$ we can rewrite the estimate~\eqref{psi-b} as the sum of factorised functions
	\begin{align}
		|y_1|^{d_1} \, e^{n \epsilon \, y_1 + 2[k_1 y_1 - e^{y_1}] \, \theta(y_1)} \, \prod_{j = 2}^n |y_j|^{d_j} \, e^{ 2(n + 1 - j) \epsilon y_j + [(k_j + k_{j + 1}) y_j - e^{y_j}] \, \theta(y_j)},
	\end{align}
	Since $\epsilon > 0$, these functions are bounded for all $\bm{y}_n \in \mathbb{R}^n$.
\end{proof}

\section{Mellin--Barnes bounds and analyticity} \label{sec:MB-bounds}

In this section we prove that Mellin--Barnes representations for $GL$ and $BC$ wave functions are absolutely convergent and entire in spectral parameters. For this we use bounds on gamma functions 
\begin{align} \label{gamma-b}
	& | \Gamma(a + \imath b) | \leq \sqrt{2\pi} \, e^{\frac{1}{6a}} \, |a + \imath b|^{a - \frac{1}{2}} \, e^{- \frac{\pi}{2} |b|}, \\[6pt] \label{inv-gamma-b}
	& \frac{1}{| \Gamma(a + \imath b) |} \leq \frac{e^{a + \frac{1}{6a}} }{\sqrt{2\pi}} \, |a + \imath b|^{-a +\frac{1}{2}} \, e^{\frac{\pi}{2} |b|},
\end{align}
where $a > 0$, $b \in \mathbb{R}$. These bounds follow from \cite[p. 34]{PK}.

\subsection{$GL$ system}

The Mellin--Barnes representation for $GL$ Toda wave function is defined by recursive formula
\begin{align} \label{Phi-MB}
	\Phi_{\bm{\lambda}_n}(\bm{x}_n) = \int\limits_{(\mathbb{R} + \imath r)^{n - 1}} d\bm{\gamma}_{n - 1} \; \hat{\mu}(\bm{\gamma}_{n - 1}) \, e^{\imath (\underline{\bm{\lambda}}_n - \underline{\bm{\gamma}}_{n - 1}) x_n } \, \prod_{j = 1}^n \prod_{k = 1}^{n - 1} \Gamma( \imath \lambda_j - \imath \gamma_k) \, \Phi_{\bm{\gamma}_{n - 1}}(\bm{x}_{n - 1})
\end{align}
with one particle wave function $\Phi_{\lambda_1}(x_1) = e^{\imath \lambda_1 x_1}$ and measure
\begin{align} \label{mu-gl}
	\hat{\mu}(\bm{\lambda}_n) = \frac{1}{n! \, (2\pi)^n} \prod_{\substack{j,k =1 \\ j\not= k}}^n \frac{1}{\Gamma(\imath \lambda_j - \imath \lambda_k)} = \frac{1}{n! \, (2\pi)^n} \prod_{1 \leq j < k \leq n} \frac{(\lambda_j - \lambda_k) \, \sh \pi(\lambda_j - \lambda_k)}{\pi}.
\end{align}
The parameter $r \in \mathbb{R}$ is chosen so that poles of gamma functions
\begin{align}
	\gamma_k = \lambda_j - \imath m, \qquad m \in \mathbb{N}_0
\end{align}
lie below the integration contours, that is $r > \Im \lambda_j$ for all $j = 1, \dots, n$. 

The analytic continuation of wave function in spectral variables is achieved by shifting integration contours. To prove that the above integral is absolutely convergent and, as a consequence, justify such shifts, we need the following inequality.

\begin{lemma} \cite[Lemma 2]{BDKK2}
	Let $(y_{m1}, \dots, y_{mm}) \in \mathbb{R}^m$ for $m = 1, \dots n$. Then for any $\delta > 0$ there exists $\epsilon > 0$ such that
	\begin{multline} \label{abs-val-ineq}
		2 \sum_{m =2 }^{n - 1} \, \sum_{1 \leq j < k \leq m} | y_{mj} - y_{mk} | - \sum_{m = 2}^{n} \sum_{j = 1}^m \sum_{k = 1}^{m - 1} |y_{mj} - y_{m - 1, k}| \\
		\leq - \sum_{1 \leq j < k \leq n} | y_{nj} - y_{nk} | + \delta \sum_{j = 1}^n | y_{nj} | - \epsilon \sum_{m = 1}^{n - 1} \sum_{j = 1}^m |y_{mj}|.
	\end{multline}
\end{lemma}

\begin{proposition} \label{prop:gl-mb-bound}
	Let $\Im \lambda_j < r$ for $j = 1, \dots, n$. Then the integral~\eqref{Phi-MB} is absolutely convergent uniformly in $\lambda_j$ from compact subsets. In particular, for $\bm{\lambda}_n \in \mathbb{R}^n$ and any $\delta > 0$ we have
	\begin{align} \label{Phi-MB-b}
		\bigl| \Phi_{\bm{\lambda}_n}(\bm{x}_n) \bigr| \leq P_\delta(\bm{\lambda}_n) \, \exp \biggl( - \frac{\pi}{2} \sum_{1 \leq j < k \leq n} |  \lambda_j -  \lambda_k | + \delta \sum_{j = 1}^n |  \lambda_j| \biggr),
	\end{align}
	where $P_\delta$ is polynomial.  
\end{proposition}

\begin{corollary} \label{cor:Phi-an}
	The function $\Phi_{\bm{\lambda}_n}(\bm{x}_n)$ can be analytically continued to $\bm{\lambda}_n \in \mathbb{C}^n$.
\end{corollary}

\begin{proof}[Proof of Proposition~\ref{prop:gl-mb-bound}]
	Denote $\bm{\gamma}_m = (\gamma_{m1}, \dots, \gamma_{mm})$ with identification $\bm{\gamma}_0 \equiv 0$ and $\bm{\gamma}_n \equiv \bm{\lambda}_n$. Then Mellin--Barnes representation in its full form is given by
	\begin{align} \label{Phi-full}
		\begin{aligned}
			\Phi_{\bm{\lambda}_n}(\bm{x}_n) = \Biggl( \prod_{m = 1}^{n - 1}  \, \int\limits_{(\mathbb{R} + \imath r_{m})^m} d\bm{\gamma}_m \Biggr) & \; e^{ \imath \sum\limits_{m = 1}^n (\underline{\bm{\gamma}}_m - \underline{\bm{\gamma}}_{m - 1} ) x_m } \, \prod_{m = 2}^{n - 1} \hat{\mu}(\bm{\gamma}_m) \\ 
			& \times  \prod_{m = 2}^n \prod_{j = 1}^m \prod_{k =1 }^{m - 1} \Gamma(\imath \gamma_{mj} - \imath \gamma_{m - 1, k}),
		\end{aligned}
	\end{align}
	where $r_1 > r_2 > \ldots > r_{n - 1} \equiv r > \Im \lambda_j$.
	From definition~\eqref{mu-gl} we have bound on the measure
	\begin{align} \label{mu-b}
		\bigl| \hat{\mu}(\bm{\gamma}_m) \bigr| \leq P(\bm{\gamma}_m) \, \exp \biggl( \pi \sum_{1 \leq j < k \leq m} | \Re \gamma_{mj} - \Re \gamma_{mk} | \biggr).
	\end{align}
	Here and in what follows by $P$ we denote polynomials. Combining this with inequality~\eqref{gamma-b} we bound absolute value of the whole integrand~\eqref{Phi-full} by the function
	\begin{multline}
		P( \Re\bm{\gamma}_1, \dots, \Re\bm{\gamma}_{n - 1} , \Re\bm{\lambda}_n) \, \exp \biggl( \, \sum_{m = 1}^{n - 1} (x_{m + 1} - x_m) m r_m - x_n \sum_{j = 1}^n \Im \lambda_j \\
		+ \pi \sum_{m =2 }^{n - 1} \, \sum_{1 \leq j < k \leq m} | \Re \gamma_{mj} - \Re \gamma_{mk} | - \frac{\pi}{2} \sum_{m = 2}^{n} \sum_{j = 1}^m \sum_{k = 1}^{m - 1} |\Re \gamma_{mj} - \Re \gamma_{m - 1, k}| \biggr)
	\end{multline}
	Due to~\eqref{abs-val-ineq} for any $\delta > 0$ this function is bounded by
	\begin{multline}
		P(\Re\bm{\gamma}_1, \dots, \Re\bm{\gamma}_{n - 1} , \Re\bm{\lambda}_n) \, \exp \biggl( \, \sum_{m = 1}^{n - 1} (x_{m + 1} - x_m) m r_m - x_n \sum_{j = 1}^n \Im \lambda_j \\
		 - \frac{\pi}{2} \sum_{1 \leq j < k \leq n} | \Re \lambda_{j} - \Re \lambda_{k} | + \delta \sum_{j = 1}^n | \Re \lambda_{j} | - \epsilon \sum_{m = 1}^{n - 1} \sum_{j = 1}^m |\Re \gamma_{mj}|\biggr)
	\end{multline}
	where $\epsilon > 0$. This gives convergence uniform in $\lambda_j$ from compact subsets. Besides, this estimate ensures that we can shift integration contours. Hence, for $\bm{\lambda}_n \in \mathbb{R}^n$ we can send $r_m \to 0$ (one by one), which gives the claimed bound~\eqref{Phi-MB-b}.
\end{proof}

\subsection{$BC$ system} \label{sec:MB-bounds-bc}

The first Mellin--Barnes representation for $BC$ wave function is given by
\begin{align} \label{Psi-MB3}
	\Psi_{\bm{\lambda}_n}(\bm{x}_n) = e^{\beta e^{-x_1}} \hspace{-0.3cm} \int\limits_{(\mathbb{R} - \imath r)^n} \!\! d\bm{\gamma}_n \; \hat{\mu}(\bm{\gamma}_n) \, \frac{ (2\beta)^{-\imath \underline{\bm{\gamma}}_n } \prod\limits_{j,k = 1}^n \Gamma(\imath \gamma_j \pm \imath \lambda_k) }{ \prod\limits_{1 \leq j < k \leq n} \Gamma(\imath \gamma_j + \imath \gamma_k) \, \prod\limits_{j = 1}^n \Gamma(g + \imath \gamma_j) } \; \Phi_{\bm{\gamma}_n}(\bm{x}_n),
\end{align}
where $r > 0$ is chosen so that poles of the integrand
\begin{align}
	\gamma_k = \pm \lambda_k + \imath m, \qquad m \in \mathbb{N}_0
\end{align}
are above the integration contours, that is $r > | \Im \lambda_j|$ for $j = 1, \dots n$. As before, to analytically continue this function in spectral parameters we shift contours, which is justified by the following proposition.

\begin{proposition} \label{prop:bc-mb-bound}
	Let $| \Im \lambda_j | < r$ for $j = 1, \ldots, n$. Then the integral~\eqref{Psi-MB3} is absolutely convergent uniformly in $\lambda_j$ from compact subsets.
\end{proposition}

\begin{corollary} \label{cor:Psi-an}
	The function $\Psi_{\bm{\lambda}_n}(\bm{x}_n)$ can be analytically continued to $\bm{\lambda}_n \in \mathbb{C}^n$.
\end{corollary}

\begin{proof}[Proof of Proposition~\ref{prop:bc-mb-bound}]
	First, change integration variables $\gamma_j = \rho_j - \imath r$, so that $\rho_j \in \mathbb{R}$, and use the relation 
	\begin{align}
		\Phi_{\bm{\gamma}_n}(\bm{x}_n) = e^{r \underline{\bm{x}}_n} \, \Phi_{\bm{\rho}_n}(\bm{x}_n).
	\end{align}
	Then with the help of inequalities~\eqref{gamma-b},~\eqref{inv-gamma-b},~\eqref{Phi-MB-b} and~\eqref{mu-b} we estimate the absolute value of the integrand in~\eqref{Psi-MB3} by the function
	\begin{multline}
		f(\bm{\lambda}_n, \bm{\rho}_n) \, \exp \biggl( \frac{\pi}{2} \sum_{1 \leq j < k \leq n} (| \rho_j - \rho_k| + |\rho_j + \rho_k|) \\
		- \frac{\pi}{2} \sum_{j, k = 1}^n( | \rho_j - \Re \lambda_k| +  | \rho_j + \Re \lambda_k|) + \Bigl(\frac{\pi}{2} + \delta \Bigr) \sum_{j = 1}^n | \rho_j| \biggr),
	\end{multline}
	where $f$ is regular (for $|\Im \lambda_j| < r$, $\rho_j \in \mathbb{R}^n$) and polynomially bounded. Since
	\begin{align}
		& \sum_{1 \leq j < k \leq n} | \rho_j \pm \rho_k| \leq \sum_{1 \leq j \not= k \leq n} |\rho_j| = (n - 1) \sum_{j = 1}^n |\rho_j|, \\[6pt]
		& - \sum_{j, k = 1}^n | \rho_j \pm \Re \lambda_k| \leq \sum_{j, k =1}^n (|\Re \lambda_k| - |\rho_j|) = n \sum_{j = 1}^n | \Re \lambda_j| - n \sum_{j = 1} |\rho_j|,
	\end{align}
	the exponential part is bounded by
	\begin{align}
		\exp \biggl( \Bigl(\delta - \frac{\pi}{2} \Bigr) \sum_{j = 1}^n | \rho_j| \biggr),
	\end{align}
	which gives the stated convergence. This bound also allows to shift integration contours and clearly uniform in $\lambda_j$ from compact sets.
\end{proof}

The second Mellin--Barnes representation is given by
\begin{align}\label{Psi-MB4}
	\Psi_{\bm{\lambda}_n}(\bm{x}_n) = e^{-\beta e^{-x_1}} \hspace{-0.3cm} \int\limits_{(\mathbb{R} - \imath r)^n} \!\! d\bm{\gamma}_n \; \hat{\mu}(\bm{\gamma}_n) \, \frac{ (2\beta)^{-\imath \underline{\bm{\gamma}}_n } \, \prod\limits_{j,k = 1}^n \Gamma(\imath \gamma_j \pm \imath \lambda_k) \, \prod\limits_{j = 1}^n \Gamma (g - \imath \gamma_j )}{ \prod\limits_{1 \leq j < k \leq n} \Gamma(\imath \gamma_j + \imath \gamma_k) \, \prod\limits_{j = 1}^n \Gamma (g \pm \imath \lambda_j ) } \,  \Phi_{\bm{\gamma}_n}(\bm{x}_n),
\end{align}
where $r$ is chosen so that integration contours pass between upward and downward series of integrand poles
\begin{align}
	\gamma_j = \pm \lambda_k + \imath m, \qquad \gamma_j = - \imath g - \imath m, \qquad m \in \mathbb{N}_0,
\end{align}
that is $r > | \Im \lambda_j|$ and $r < g$. Its convergence can be proven in analogous way, but this time we have analytic continuation only to the strip $| \Im \lambda_j | < g$.

\begin{proposition} \label{prop:bc-mb2-bound}
	Let $| \Im \lambda_j | < r < g$ for $j = 1, \ldots, n$. Then the integral~\eqref{Psi-MB4} is absolutely convergent uniformly in $\lambda_j$ from compact subsets.
\end{proposition}

\section{Gustafson integral reduction} \label{sec:gust-red}

The following identity is proved in \cite[Theorem 9.3]{G}
\begin{align} \label{Gust-int}
	\frac{1}{(4\pi)^n \, n!} \int_{\mathbb{R}^n} d\bm{\lambda}_n \; \frac{\prod\limits_{j = 1}^{2n + 2} \prod\limits_{k = 1}^n \Gamma(z_j \pm \imath \lambda_k)}{\prod\limits_{1 \leq j < k \leq n} \Gamma(\pm \imath \lambda_j \pm \imath \lambda_k) \, \prod\limits_{j = 1}^n \Gamma(\pm 2\imath \lambda_j)} = \frac{\prod\limits_{1 \leq j < k \leq 2n+2} \Gamma(z_j + z_k)}{\Gamma(z_1 + \ldots + z_{2n + 2})},
\end{align}
where it is assumed that $\Re z_j > 0$ for $j = 1, \dots, 2n+2$. Let us prove that in the limit $z_{2n+2} \to \infty$ it reduces to the formula
\begin{align} \label{Gust-red}
	\frac{1}{(4\pi)^n \, n!} \int_{\mathbb{R}^n} d\bm{\lambda}_n \; \frac{\prod\limits_{j = 1}^{2n + 1} \prod\limits_{k = 1}^n \Gamma(z_j \pm \imath \lambda_k)}{\prod\limits_{1 \leq j < k \leq n} \Gamma(\pm \imath \lambda_j \pm \imath \lambda_k) \, \prod\limits_{j = 1}^n \Gamma(\pm 2\imath \lambda_j)} = \prod\limits_{1 \leq j < k \leq 2n+1} \Gamma(z_j + z_k).
\end{align}
The absolute convergence of both integrals can be easily shown using inequality~\eqref{gamma-b} and reflection formula for the gamma function.

To take the limit we use the asymptotic formula~\cite[\href{http://dlmf.nist.gov/5.11.E12}{(5.11.12)}]{DLMF}
\begin{align}
	\frac{\Gamma(a + L)}{\Gamma(b + L)} \sim L^{a - b}, \qquad L \to \infty.
\end{align}
Namely, denote $z_{2n+2} = L$ and divide both sides of the identity~\eqref{Gust-int} by $\Gamma^{2n}(L)$. Then the right hand side tends to the desired expression
\begin{align}
	\lim_{L \to \infty} \frac{\prod\limits_{1 \leq j < k \leq 2n+1} \Gamma(z_j + z_k) \, \prod\limits_{j = 1}^{2n + 1} \Gamma(z_j + L)}{\Gamma(z_1 + \ldots + z_{2n + 1} + L) \, \Gamma^{2n}(L)} = \prod\limits_{1 \leq j < k \leq 2n+1} \Gamma(z_j + z_k).
\end{align}
The integrand of the left hand side also tends pointwise to the integrand in the reduced formula~\eqref{Gust-red}, because
\begin{align}
	\lim_{L \to \infty} \frac{\prod\limits_{k = 1}^n \Gamma(L \pm \imath \lambda_k)}{\Gamma^{2n}(L)} = 1.
\end{align}
Moreover, since $| \Gamma( z) | \leq |\Gamma(\Re z)|$ we have inequality
\begin{align}
	\Biggl| \frac{\prod\limits_{k = 1}^n \Gamma(L \pm \imath \lambda_k)}{\Gamma^{2n}(L)} \Biggr| \leq 1.
\end{align}
Hence, we can use dominated convergence theorem to interchange limit $L \to \infty$ and integration, which concludes the proof of~\eqref{Gust-red}.


\begin{thebibliography}{99}
	
	\bibitem[ADV1]{ADV} P. Antonenko, S. Derkachov, P. Valinevich, \textit{A-type open $SL(2, \mathbb{C})$ spin chain}, \href{https://doi.org/10.3842/SIGMA.2025.107}{SIGMA} \textbf{21} (2025) 107, \href{https://doi.org/10.48550/arXiv.2507.09568}{\tt [2507.09568]}.

	\bibitem[ADV2]{ADV2} P. Antonenko, S. Derkachov, P. Valinevich, \textit{BC-type open $SL(2, \mathbb{C})$ spin chain}, \href{https://doi.org/10.1007/s00023-025-01653-0}{Annales Henri Poincar\'e} (2026), \href{https://doi.org/10.48550/arXiv.2508.04972}{\tt [2508.04972]}.
	
	\bibitem[B]{B} O. Babelon, \textit{Equations in dual variables for Whittaker functions}, \href{https://doi.org/10.1023/B:MATH.0000010714.56215.2a}{Letters in Mathematical Physics} \textbf{65} (2003) 229--240, \href{https://doi.org/10.48550/arXiv.math-ph/0307037}{\tt [math-ph/0307037]}.
	
	\bibitem[BC]{BC} A. Borodin, I. Corwin, \textit{Macdonald processes}, \href{https://doi.org/10.1007/s00440-013-0482-3}{Probability Theory and Related Fields} \textbf{158} (2014) 225--400, \href{https://doi.org/10.48550/arXiv.1111.4408}{\tt [1111.4408]}.
	
	\bibitem[BCDK]{BCDK} N. Belousov, L. Cherepanov, S. Derkachov, S. Khoroshkin, \textit{Calogero-Sutherland hyperbolic system and Heckman-Opdam $\mathfrak{gl}_n$ hypergeometric function}, arXiv preprint \href{https://doi.org/10.48550/arXiv.2508.18864}{\tt [2508.18864]}.
	
	\bibitem[BDK]{BDK} N. Belousov, S. Derkachov, S. Khoroshkin, \textit{$BC$ Toda chain I: reflection operator and eigenfunctions}, to appear.
	
	\bibitem[BDKK]{BDKK2} N. Belousov, S. Derkachov, S. Kharchev, S. Khoroshkin, \textit{Baxter operators in Ruijsenaars hyperbolic system II: bispectral wave functions}, \href{https://doi.org/10.1007/s00023-023-01385-z}{Annales Henri Poincaré} \textbf{25} (2024) 3259--3296, \href{https://doi.org/10.48550/arXiv.2303.06382}{\tt [2303.06382]}.
	
	\bibitem[BK]{BK} N. Belousov, S. Khoroshkin, \textit{Ruijsenaars spectral transform}, \href{https://doi.org/10.1007/s11005-025-01957-6}{Letters in Mathematical Physics} \textbf{115} (2025), \href{https://doi.org/10.48550/arXiv.2411.19659}{\tt [2411.19659]}.
	
	\bibitem[BKP]{BKP} O. Babelon, K. K. Kozlowski, V. Pasquier, \textit{Baxter operator and Baxter equation for $q$-Toda and Toda$_2$ chains}, \href{https://doi.org/10.1142/S0129055X18400032}{Reviews in Mathematical Physics} \textbf{30}:6 (2018), \href{https://doi.org/10.48550/arXiv.1803.06196}{\tt [1803.06196]}. 
	
	\bibitem[DE1]{DE0} J. F. van Diejen, E. Emsiz, \textit{Difference equation for the Heckman-Opdam hypergeometric function and its confluent Whittaker limit}, \href{https://doi.org/10.1016/j.aim.2015.08.018}{Advances in Mathematics} \textbf{285} (2015) 1225--1240, \href{https://doi.org/10.48550/arXiv.1411.0463}{\tt [1411.0463]}.
	
	\bibitem[DE2]{DE} J. F. van Diejen, E. Emsiz, \textit{Bispectral dual difference equations for the quantum Toda chain with boundary perturbations}, \href{https://doi.org/10.1093/imrn/rnx219}{International Mathematics Research Notices} \textbf{2019}:12 (2019) 3740--3767, \href{https://doi.org/10.48550/arXiv.1903.01827}{\tt [1903.01827]}.
	
	\bibitem[DKM1]{DKM} S. \'E. Derkachov, K. K. Kozlowski, A. N. Manashov, \textit{On the separation of variables for the modular XXZ magnet and the lattice sinh-Gordon models}, \href{https://doi.org/10.1007/s00023-019-00806-2}{Annales Henri Poincaré} \textbf{20} (2019) 2623--2670, \href{https://doi.org/10.48550/arXiv.1806.04487}{\tt [1806.04487]}.
	
	\bibitem[DKM2]{DKM2} S. \'E. Derkachov, K. K. Kozlowski, A. N. Manashov, \textit{Completeness of SoV representation for $SL (2, \mathbb{R})$ spin chains}, \href{https://doi.org/10.3842/SIGMA.2021.063}{SIGMA} \textbf{17} (2021) 063, \href{https://doi.org/10.48550/arXiv.2102.13570}{\tt [2102.13570]}.
	
	\bibitem[DM]{DM} S. \'E. Derkachov, A. N. Manashov, \textit{Spin chains and Gustafson’s integrals}, \href{https://doi.org/10.1088/1751-8121/aa749a}{Journal of Physics A: Mathematical and Theoretical} \textbf{50} (2017), \href{https://doi.org/10.48550/arXiv.1611.09593}{\tt [1611.09593]}.
	
	\bibitem[DLMF]{DLMF} \textit{NIST Digital Library of Mathematical Functions}. \href{https://dlmf.nist.gov/}{https://dlmf.nist.gov/}, Release 1.2.4 of 2025-03-15. F. W. J. Olver, A. B. Olde Daalhuis, D. W. Lozier, B. I. Schneider, R. F. Boisvert, C. W. Clark, B. R. Miller, B. V. Saunders, H. S. Cohl, and M. A. McClain, eds.
	
	\bibitem[F]{F} L. D. Faddeev, \textit{Quantum completely integrable models in field theory}, Mathematical Physics Reviews \textbf{1} (1980) 107--155, reprinted in \textit{40 years in mathematical physics}, World Scientific (1995).
	
	\bibitem[GaKL]{GaKL} A.Galiullin,  S. Khoroshkin,  M. Lyachko, \textit{Zhelobenko–Stern formulas and $B_n$ Toda wave functions}, \href{https://doi.org/10.1007/s11005-024-01824-w}{Letters in Mathematical Physics} \textbf{114} (2024), \href{https://doi.org/10.48550/arXiv.2402.16120}{\tt [2402.16120]}.
	
	\bibitem[G]{Gaud} M. Gaudin, \textit{La fonction d'onde de Bethe} (Masson, Paris, 1983); English translation: \textit{The Bethe Wavefunction}, trans. Jean-Sébastien Caux (Cambridge University Press, Cambridge, 2014).
	
	\bibitem[GKL]{GKL} A. Gerasimov, S. Kharchev, D. Lebedev, \textit{Representation theory and quantum inverse scattering method: the open Toda chain and the hyperbolic Sutherland model}, \href{https://doi.org/10.1155/S1073792804132595}{International Mathematics Research Notices} {\bf 2004}:17 (2004) 823--854, \href{https://doi.org/10.48550/arXiv.math/0204206}{\tt [math/0204206]}.
	
	\bibitem[GLO1]{GLO1} A. Gerasimov, D. Lebedev, S. Oblezin, \textit{Baxter operator and Archimedean Hecke algebra}, \href{https://doi.org/10.1007/s00220-008-0547-9}{Communications in Mathematical Physics} \textbf{284} (2008) 867--896, \href{https://doi.org/10.48550/arXiv.0706.3476}{\tt [0706.3476]}.
	 
	\bibitem[GLO2]{GLO} A. A. Gerasimov, D. R. Lebedev, S. V. Oblezin, \textit{New integral representations of Whittaker functions for classical Lie groups}, \href{https://doi.org/10.1070/RM2012v067n01ABEH004776}{Russian Mathematical Surveys} \textbf{67}:1 (2012), \href{https://doi.org/10.48550/arXiv.0705.2886}{\tt [0705.2886]}.
	
	\bibitem[G]{G}  R. A. Gustafson, \textit{Some $q$-beta and Mellin-Barnes integrals on compact Lie groups and Lie algebras}, \href{https://www.jstor.org/stable/2154615}{Transactions of the American Mathematical Society} \textbf{341}:1 (1994) 69--119.
	
	\bibitem[HC]{HC} Harish-Chandra, \textit{Spherical functions on a semisimple Lie group, I}, \href{https://doi.org/10.2307/2372786}{American Journal of Mathematics} \textbf{80}:2 (1958) 241--310.
	
	\bibitem[HR]{HR3} M. Halln\"as, S. Ruijsenaars, \textit{Joint eigenfunctions for the relativistic Calogero–Moser Hamiltonians of hyperbolic type. III. Factorized asymptotics}, \href{https://doi.org/10.1093/imrn/rnaa193}{International Mathematics Research Notices} \textbf{2021}:6 (2021) 4679--4708, \href{https://doi.org/10.48550/arXiv.1905.12918}{\tt [1905.12918]}.
	
	\bibitem[I]{I} V. I. Inozemtsev, \textit{The finite Toda lattices}, \href{https://doi.org/10.1007/BF01218159}{Communications in Mathematical Physics} \textbf{121} (1989) 629--638.
	
	\bibitem[IS]{IS} N. Z. Iorgov, V. N. Shadura, \textit{Wave functions of the Toda chain with boundary interaction}, \href{https://doi.org/10.1007/s11232-005-0075-0}{Theoretical and mathematical physics} \textbf{142} (2005) 289--305, \href{https://doi.org/10.48550/arXiv.nlin/0411002}{\tt [nlin/0411002]}.
	
	\bibitem[K]{Kozl} K. K. Kozlowski, \textit{Unitarity of the SoV transform for the Toda chain}, \href{https://doi.org/10.1007/s00220-014-2134-6}{Communications in Mathematical Physics} \textbf{334} (2015) 223--273, \href{https://doi.org/10.48550/arXiv.1306.4967}{\tt [1306.4967]}. 
	
	\bibitem[KL1]{KL} S. Kharchev, D. Lebedev, \textit{Integral representations for the eigenfunctions of quantum open and periodic Toda chains from the QISM formalism}, \href{https://doi.org/10.1088/0305-4470/34/11/317}{Journal of Physics A: Mathematical and General} \textbf{34} (2001), \href{https://doi.org/10.48550/arXiv.hep-th/0007040}{\tt [hep-th/0007040]}.
	
	\bibitem[KL2]{KL2} S. Kharchev, D. Lebedev, \textit{Eigenfunctions of $GL(N, \mathbb{R})$ Toda chain: Mellin-Barnes representation}, \href{https://doi.org/10.1134/1.568323}{Journal of Experimental and Theoretical Physics Letters } \textbf{71} (2000) 235--238, \href{https://doi.org/10.48550/arXiv.hep-th/0004065}{\tt [hep-th/0004065]}.
	
	\bibitem[O]{O} E. M. Opdam, \textit{Root systems and hypergeometric functions. IV}, Compositio Mathematica \textbf{67} (1988) 191--209.

	\bibitem[PK]{PK} R. B. Paris, D. Kaminski, \textit{Asymptotics and Mellin-Barnes Integrals}, \href{https://doi.org/10.1017/CBO9780511546662}{Cambridge University Press} (2001).
	
	\bibitem[Sil]{Sil} A. V. Silantyev, \textit{Transition function for the Toda chain}, \href{https://doi.org/10.1007/s11232-007-0024-1}{Theoretical and Mathematical Physics} \textbf{150} (2007) 315--331, \href{https://doi.org/10.48550/arXiv.nlin/0603017}{\tt [nlin/0603017]}.

	\bibitem[Skl1]{S} E. K. Sklyanin, \textit{The quantum Toda chain}, In: \href{https://doi.org/10.1007/3-540-15213-X_80}{Non-Linear Equations in Classical and Quantum Field Theory. Lecture Notes in Physics} \textbf{226} (1985) 196--233.

	\bibitem[Skl2]{Skl} E. K. Sklyanin, \textit{Boundary conditions for integrable quantum systems}, \href{https://doi.org/10.1088/0305-4470/21/10/015}{Journal of Physics~A} \textbf{21} (1988) 2375--2389.
	
	\bibitem[Skl3]{Skl2} E. K. Sklyanin, \textit{B\"acklund transformations and Baxter’s Q-operator}, in: \textit{Integrable systems: from classical to quantum}, CRM Proceedings \& Lecture Notes \textbf{26}, \href{https://doi.org/10.48550/arXiv.nlin/0009009}{\tt [nlin/0009009]}.
	
	\bibitem[STS]{STS} M. A. Semenov-Tian-Shansky, \textit{Quantization of Open Toda Lattices}, in: \textit{Dynamical Systems VII}, \href{https://doi.org/10.1007/978-3-662-06796-3_8}{Encyclopaedia of Mathematical Sciences} \textbf{16} (1994) 226--259.
	
	\bibitem[W1]{W} N. R. Wallach, \textit{The Whittaker Plancherel theorem}, \href{https://doi.org/10.1007/s11537-023-2230-5}{Japanese Journal of Mathematics} \textbf{19} (2024) 1-65, \href{https://doi.org/10.48550/arXiv.1705.06787}{\tt [1705.06787]}.
	
	\bibitem[W2]{W2} N. R. Wallach, \textit{The spherical Whittaker Inversion Theorem and the quantum non-periodic Toda Lattice}, arXiv preprint \href{https://doi.org/10.48550/arXiv.2303.11256}{\tt [2303.11256]}.
	
\end{thebibliography}
\end{document}